\documentclass[pdftex,phd,nohyperref]{pittetd} %hidelinks,,nohyperref
\usepackage{amsmath}
\usepackage{amssymb}
\usepackage{mathtools}
\usepackage{amsthm}
\usepackage{bbm}
\usepackage[utf8]{inputenc} % allow utf-8 input
\usepackage[T1]{fontenc}    % use 8-bit T1 fonts

\usepackage{color}
\usepackage{xcolor}
\usepackage[colorlinks,linkcolor=black,urlcolor=black,citecolor=black]{hyperref}

\usepackage{graphicx}
\usepackage{indentfirst}
\usepackage{pdflscape}

\usepackage{multirow}
\usepackage{booktabs} % Allows the use of \toprule
\usepackage[round]{natbib}
\usepackage{subcaption}
\usepackage{enumitem}
\usepackage{algorithm}% http://ctan.org/pkg/algorithm
\usepackage{algpseudocode}[noend] %http://ctan.org/pkg/algorithmicx
\usepackage[section]{placeins}
\FloatBarrier

\def\bSig\mathbf{\Sigma}

\DeclareMathOperator{\logit}{logit}
\DeclareMathOperator{\rank}{rank}
\DeclareMathOperator{\dimselfdefined}{dim}

\DeclareMathOperator{\expit}{expit}

\theoremstyle{plain}
\newtheorem{theorem}{Theorem}[section]
\newtheorem{proposition}[theorem]{Proposition}
\newtheorem{lemma}[theorem]{Lemma}

\theoremstyle{definition}
\newtheorem{definition}[theorem]{Definition}
\newtheorem{assumption}[theorem]{Assumption}
\theoremstyle{remark}

\newcommand{\red}[1]{{\color{red} #1}}

\newcommand{\bX}{\boldsymbol{X}}
\newcommand{\bx}{\boldsymbol{x}}
\newcommand{\bbeta}{\boldsymbol{\beta}}
\newcommand{\btheta}{\boldsymbol{\theta}}
\newcommand{\Tau}{\mathcal{T}}

\makeatletter
\renewcommand{\@biblabel}[1]{[#1]\hfill}
\makeatother
\makeatletter
\def\@hangfrom#1{\setbox\@tempboxa\hbox{{#1}}%
      \hangindent 0pt%\wd\@tempboxa
      \noindent\box\@tempboxa}
\makeatother
\title[Causal Inference under Data Restrictions]{
\large Causal Inference under Data Restrictions
}
\author{Xiaoqing Tan}
\degree{B.S., Sun Yat-sen University, 2018}

\school[]{School of Public Health}
\date{July 25, 2022}

\year{2022}   

\subject{Dissertation}
\begin{document}
\maketitle
% If moved around the document, it might generate errors and warnings while compiling.
%==========================================================================================%
%==========================================================================================%

%==========================================================================================%
%CREATING THE COMMITTEE PAGE
%==========================================================================================%
% For the committee membership page, you have to provide the names and affiliations of the members. The first one will 

% THESIS ADVISOR (First member of the committee)
\committeemember{\textbf{Lu Tang}, PhD, Assistant Professor\\ Department of Biostatistics, University of Pittsburgh}

% THESIS CO-ADVISOR
% \coadvisor{Second advisor, Co-advisor Departmental Affiliation}
%Uncomment to add a 'Second Advisor' into the document
% COMMITTEE MEMBERS
% \committeemember{Second member's name, Departmental Affiliation}
% \committeemember{Third member's name, Departmental Affiliation}
\committeemember{\textbf{Gong Tang}, PhD, Professor\\ Department of Biostatistics, University of Pittsburgh}
\committeemember{\textbf{Chung-Chou H. Chang}, PhD, Professor\\ Department of Medicine and Biostatistics, University of Pittsburgh}
% To add more committee members
\committeemember{\textbf{Emily Brant}, MD, Assistant Professor\\ Department of Medicine, University of Pittsburgh} 
\committeemember{\textbf{Zhengling Qi}, PhD, Assistant Professor\\ School of Business, George Washington University}
% \committeemember{}
% To use uncommon the different committee members or add '\committeemember' commands as needed

%Special Option
% For master's theses, the committee may be omitted, naming only the advisor.

% DEPARTMENT INFORMATION 
\school{School of Public Health}
\makecommittee
%==========================================================================================%
%==========================================================================================%

%==========================================================================================%
%CREATING THE COPYRIGHT PAGE
%==========================================================================================%

% Create a copyright page with the author and year specified in the 'Title Page'
\copyrightpage
%Uncomment to get a copyright page.
%==========================================================================================%
%==========================================================================================%

%==========================================================================================%
%CREATING THE ABSTRACT
%==========================================================================================%

\begin{abstract}

This dissertation focuses on modern causal inference under uncertainty and data restrictions, with applications to neoadjuvant clinical trials, distributed data networks, and robust individualized decision making.

In the first project, we propose a method under the principal stratification framework to identify and estimate the average treatment effects on a binary outcome, conditional on the counterfactual status of a post-treatment intermediate response. Under mild assumptions, the treatment effect of interest can be identified. We extend the approach to address censored outcome data. The proposed method is applied to a neoadjuvant clinical trial and its performance is evaluated via simulation studies.

In the second project, we propose a tree-based model averaging approach to improve the estimation accuracy of conditional average treatment effects at a target site by leveraging models derived from other potentially heterogeneous sites, without them sharing subject-level data. To our best knowledge, there is no established model averaging approach for distributed data with a focus on improving the estimation of treatment effects. The performance of this approach is demonstrated by a study of the causal effects of oxygen therapy on hospital survival rate and backed up by comprehensive simulations.

In the third project, we propose a robust individualized decision learning framework with sensitive variables to improve the worst-case outcomes of individuals caused by sensitive variables that are unavailable at the time of decision. Unlike most existing work that uses mean-optimal objectives, we propose a robust learning framework via finding a newly defined quantile- or infimum-optimal decision rule. From a causal perspective, we also generalize the classic notion of (average) fairness to conditional fairness for individual subjects. The reliable performance of the proposed method is demonstrated through synthetic experiments and three real-data applications. 

\textbf{Public health significance:} The dissertation addresses several aspects of causal inference: 1) identify principal stratum treatment effects; 2) enhance the estimation of treatment effects via heterogeneous data integration; 3) derive robust individualized decision rules considering worst-case scenarios. It has the potential to fundamentally improve the current practice in drug development and precision medicine.

\end{abstract}

% SPECIAL OPTIONS
% Include Keywords
% To include keywords as part of the abstract include the option '[keywords]' (e.g., \begin{abstract}[Keywords:]
% The list comes from the '\keywords' specified in the 'Title Page'

% Include the word 'ABSTRACT'
% Use '\begin{abstract*}' and '\end{abstract*} instead of '\begin{abstract}' and '\end{abstract}
% The word `ABSTRACT' appears on the top of the page
%==========================================================================================%
%==========================================================================================%

%==========================================================================================%
% TABLE OF CONTENTS, FIGURES, AND TABLES
%==========================================================================================%

% Table of contents
\tableofcontents
% Comment (Use the '%' character) to omit

% List of Tables
\listoftables                      
% Comment (Use the '%' character) to omit

% List of Figures
\listoffigures                
% Comment (Use the '%' character) to omit

% **If no figures and/or tables are included in the document, 'PITETD' will still create an empty page for the figures and/or tables **

% ** LaTex automatically includes all the figures and tables from the figures/tables included in the document **%
%==========================================================================================%
%==========================================================================================%

%==========================================================================================%
% INCLUDING A PREFACE
%==========================================================================================%
\phantomsection
\preface

% To include a preface for the document uncomment the '\preface' command
% Include the text of the preface after the command

% Text of the preface

This work, as with most research, is an effort in collaboration. 
First and foremost, I would like to express my deepest gratitude to my Ph.D. advisors Dr. Gong Tang and Dr. Lu Tang.  
This dissertation would not have been possible without their continuous guidance, support, and encouragement. 
%% Gong
Gong is an exemplary scholar, who is always hardworking and dedicated to research. 
I am constantly amazed by his sharpness on research and his deep insights about many different topics, ranging from medical science, applied math, to statistics, and computer science. 
His ability to persevere through difficulties and ask the right questions are abilities that I will continue to
aspire to throughout my career. 
%% Lu
I am also fortunate enough to work with Lu, an inspirational advisor, who is incredibly generous with his ideas and time; and a caring mentor, who constantly offers his support and empathy.  
Our meetings and discussions have been a great source of inspiration and greatly shaped my way of approaching research problems. 
I am deeply indebted to him for devoting so much time and energy to mentoring me and guiding me through the transition from a student to a researcher. 
His passion for research have profoundly influenced me from both professional and personal perspectives.

I am also very grateful to have Dr. Joyce Chang, Dr. Zhengling Qi, and Dr. Emily Brant serving on my doctoral dissertation committee and providing me with invaluable comments. 
%% Joyce
I am fortunate to collaborate with Joyce on the project of heterogeneous causal models integration. I am grateful for her training and guidance on causal inference, and I have learnt a lot about how to present research ideas, story telling, as well as critical thinking. 
%% Zhengling
I am grateful to Zhengling, who has helped me tremendously in developing the project of robust individualized decision learning. 
His constant support and humor make all the research meetings interesting and enjoyable. 
Zhengling’s enthusiasm for research has left a lot of positive impacts on me. 
%% Emily
Emily has provided me with a wonderful opportunity to work on the exciting application of sepsis disease for interdisciplinary collaboration. 
Her domain expertise in medicine has been invaluable in leading me to find immense motivation for methodology development that has shaped my dissertation.

Additionally, I would like to express my gratitude and appreciation to a number of wonderful researchers who I have been fortunate to collaborate with through my Ph.D. 
%% GSR
Specifically, I would like to thank Dr. Timothy Girard, Dr. David Samuels from Vanderbilt University, and others in their team. 
They have exposed me to the world of interesting applications of critical care medicine and provided me with numerous opportunities to apply meaningful and interpretable statistical tools in medical problems. 
%% Jiebiao
I would like to thank Dr. Jiebiao Wang and Dr. Qi Yan from Columbia University for their guidance and support on our genetic side project. 
They have provided me with a great opportunity to discuss and learn causal inference problems under genetic context. 
%% lilly
I would like to thank the amazing research team I met when I interned at Eli Lilly and Company, including Dr. Shu Yang from North Carolina State University, Dr. Ilya Lipkovich, Dr. Douglas Faries, Dr. Wendy Ye, and Zbigniew Kadziola. 
My journey would not have been so rewarding without them, and I have learned a lot from our interactions.

I would also like to thank all my friends that I made and all the people that I met throughout the Ph.D. studies.  
My thanks also go to staff members at the Department of Biostatistics, University of Pittsburgh, who always patiently helped with my questions and warmly welcomed me into the office with big smiles.

Last but not least, I would like to thank my family, Yu Wang the duck, Kitty and Bunny the cats, Larry the chinchilla, and Bibi the bunny for their unwavering support and unconditioned love. 

%==========================================================================================%
%==========================================================================================%

%==========================================================================================%
% STARTING THE DOCUMENT
%==========================================================================================%

% \chapter{Second chapter}
\chapter{Introduction}

\section{Challenges in Causal Inference under Data Restrictions}

Modern statistical and machine learning methods are capable of capturing correlations between variables but often fail to inform us the causes behind. 
Causal inference, on the other hand, helps to understand the underlying data-generating process, which is critical when analyzing data from our contemporary world, particularly in fields such as healthcare, political science, and economics.

\subsection{Data Collection Concerns in Traditional Clinical Trials}

In traditional clinical trials, data collection often requires long years of follow-up, leading to problems such as patients' withdrawal and high cost of the study. Recently, in neoadjuvant clinical trials, early efficacy of a treatment is assessed first via an intermediate post-treatment response and the eventual efficacy is assessed via long-term outcomes such as survival. 
Although strongly associated with survival, this intermediate response has not been confirmed as a surrogate endpoint. To fully understand its clinical implication, it is important to establish causal estimands such as the causal effect in survival for patients who would obtain a certain intermediate response under treatment. 
In Chapter~\ref{chap:ps}, driven by a recent neoadjuvant clinical trial, a method is developed under the principal stratification framework to identify and estimate the average treatment effects on the long-term outcome, conditional on the counterfactual status of the post-treatment intermediate response.

\subsection{Privacy Concerns in Distributed Data Networks}

In the modern context, new challenges arise in the research of causal inference due to data restrictions. 
Data privacy has become an important issue with the establishment of multiple distributed research networks in large scale studies \citep{fleurence2014launching,hripcsak2015observational,platt2018fda,donohue2021use}. These distributed networks collect sensitive subject-level data and store them at individual research sites (e.g., hospitals). Effective statistical and machine learning approaches are hence needed to be developed to jointly analyze data across sites, without directly utilizing subject-level information. 
Chapter~\ref{chap:hetero} introduces a tree-based model averaging approach to improve the estimation accuracy of conditional average treatment effects at a target site by leveraging models derived from other potentially heterogeneous sites, without them sharing subject-level data.

\subsection{Timeliness and Fairness Concerns in Decision Making}

There's been a growing concern around the timeliness and fairness of individualized decision making algorithms. 
For example, there may exist sensitive variables that are important to the intervention decision, but their inclusion in decision making is prohibited due to reasons such as delayed availability, fairness, or other concerns. 
Robust individualized decision rules that take into account the variation caused by the unavailability of these sensitive variables are needed. 
On the other hand, in most existing works such as \cite{manski2004statistical,qian2011performance}, individualized decision rules aim to maximize the potential average performance. 
Consequently, certain groups may get unfairly or unsafely treated due to the heterogeneity in their response to the treatment. 
These problems can impact people's lives in direct and important ways like loan approvals or the length of a sentence in a court case. 
It is therefore an imperative task to develop fairness-aware decision learning methods. 
In Chapter~\ref{chap:itr}, we propose a robust individualized decision learning framework with sensitive variables to improve the worst-case outcomes of individuals caused by sensitive variables that are unavailable at the time of decision.

\section{Outline and Contributions}

This section lists the chapters and corresponding contributions. Each chapter aims to be a self contained exposition on a specific topic; as a result, some introductory material for particular chapters are similar in scope. In the following of the dissertation, I develop several statistical and machine learning methods to address the aforementioned challenges in causal inference under uncertainty and data restrictions, with applications to neoadjuvant randomized trials, distributed data networks, and robust individualized decision making.

Chapter~\ref{chap:ps} concerns identifying and estimating causal effects that involve a post-treatment intermediate response in neoadjuvant randomized clinical trials. 
In neoadjuvant trials, early efficacy of a treatment is assessed via the binary pathological complete response (pCR) and the eventual efficacy is assessed via long-term clinical outcomes such as survival. Although pCR is strongly associated with survival, it has not been confirmed as a surrogate endpoint. To fully understand its clinical implication, it is important to establish causal estimands such as the causal effect in survival for patients who would achieve pCR under the new regimen. 
Under the principal stratification framework, previous studies focus on sensitivity analyses by varying model parameters in an imposed model on counterfactual outcomes. Under mild assumptions, we propose an approach to identify and estimate those model parameters using empirical data and subsequently the causal estimand of interest. We also extend our approach to address censored outcome data. The proposed method is applied to a recent clinical trial and its performance is evaluated via simulation studies. 
This chapter has been accepted for publication in the \textit{Proceedings of the First Conference on Causal Learning and Reasoning} \citep[\textit{CLeaR'22},][]{tan2022identifying}.

Chapter~\ref{chap:hetero} concerns improving the estimation accuracy of personalized treatment effects by leveraging models rather than subject-level data from heterogeneous data sources. 
Accurately estimating personalized treatment effects within a study site (e.g., a hospital) has been challenging due to limited sample size. Furthermore, privacy considerations and lack of resources prevent a site from leveraging subject-level data from other sites. We propose a tree-based model averaging approach to improve the estimation accuracy of conditional average treatment effects (CATE) at a target site by leveraging models derived from other potentially heterogeneous sites, without them sharing subject-level data. To our best knowledge, there is no established model averaging approach for distributed data with a focus on improving the estimation of treatment effects. Specifically, under distributed data networks, our framework provides an interpretable tree-based ensemble of CATE estimators that joins models across study sites, while actively modeling the heterogeneity in data sources through site partitioning. The performance of this approach is demonstrated by a real-world study of the causal effects of oxygen therapy on hospital survival rate and backed up by comprehensive simulation results. 
This chapter has been accepted for publication in the \textit{Proceedings of the $\mathit{39}^{th}$ International Conference on Machine Learning}
\citep[\textit{ICML'22},][]{tan2021tree}.

Chapter~\ref{chap:itr} introduces RISE, a \underline{r}obust \underline{i}ndividualized decision learning framework with \underline{se}nsitive variables, where sensitive variables are collectible data and important to the intervention decision, but their inclusion in decision making is prohibited due to reasons such as delayed availability or fairness concerns. A naive baseline is to ignore these sensitive variables in learning decision rules, leading to significant uncertainty and bias. 
To address this, we propose a decision learning framework to incorporate sensitive variables during \textit{offline} training but not include them in the input of the learned decision rule during model deployment. 
Specifically, from a causal perspective, the proposed framework intends to improve the worst-case outcomes of individuals caused by sensitive variables that are unavailable at the time of decision. 
Unlike most existing literature that uses mean-optimal objectives, we propose a robust learning framework by finding a newly defined quantile- or infimum-optimal decision rule. 
The reliable performance of the proposed method is demonstrated through synthetic experiments and three real-world applications. 
This chapter has been accepted for publication in the \textit{Advances in Neural Information Processing Systems}
\citep[\textit{NeurIPS'22},][]{tan2022rise}.

\chapter{Identifying Principal Stratum Causal Effects Conditional on a Post-treatment Intermediate Response}\label{chap:ps}

\section{Introduction}

We have seen a major shift in the conduct of breast cancer clinical trials in recent years. Traditionally, breast cancer patients are randomly assigned to control or treatment after the primary surgery. 
Patients from the two groups are then followed over years for comparison of their long-term outcomes such as disease-free survival and overall survival. However, in recent years, there have been an increasing number of neoadjuvant trials where many of the systemic therapies are administered prior to the breast surgery \citep{food2013guidance}.

The primary endpoint in neoadjuvant breast cancer clinical trials is pathological complete response (pCR), a binary indicator of absence of invasive cancer in the breast and auxiliary nodes \citep{food2013guidance}. The rationale for using pCR is that efficacy of a treatment can be assessed at the time of surgery  
instead of the typical 5-10 years of follow-up on survival endpoints in the adjuvant setting. Strong association between pCR and survival has been well documented  \citep{cortazar2014pathological, von2012definition, song2021association}, making pCR an attractive candidate surrogate. In the latest guidance of the U.S. Food and Drug administration (FDA), pCR is accepted as an endpoint to support accelerated drug approvals, provided certain requirements are met \citep{food2013guidance}. 
It is important to decipher the causal relationship among treatment, pCR, and survival in order to interpret the efficacy in survival when pCR is involved. 

In the recently published National Surgical Adjuvant Breast and Bowel Project (NSABP) B-40 trial, patients with operable human epidermal growth factor receptor 2 (HER2)-negative breast cancer were randomly assigned to receive or not to receive bevacizumab along with their neoadjuvant chemotherapy regimens \citep{bear2012bevacizumab}. The addition of bevacizumab significantly increased the rate of pCR (28.2\% without bevacizumab vs. 34.5\% with bevacizumab, p-value $=0.02$). In terms of the long-term outcomes, patients on bevacizumab showed improvements in event-free survival (EFS) and overall survival (OS) compared to the control patients (EFS: hazard ratio 0.80, p-value $=0.06$; OS: hazard ratio 0.65, p-value $=0.004$) \citep{bear2015neoadjuvant}. Some investigators are interested in the comparison of survival between pCR patients in the treatment group and pCR patients in the control group. Such comparison, however, is problematic because these two groups of pCR patients are different and any direct comparison between them lacks causal interpretation. 

Under the counterfactual framework \citep{rubin1974estimating}, potentially a patient has a pCR status after taking the control regimen and a pCR status after taking the treatment. 
Similarly, one can define counterfactual outcomes and causal effects in survival status (0/1) after a certain time period such as three years. 
The principal stratum framework proposed by \cite{frangakis2002principal} can be used to describe causal effect in long-term outcomes (such as EFS) with an intermediate outcome (such as pCR) involved. 
Each principal stratum consists of subjects with the same pair of potential pCR status: the pCR status under the control regimen and the pCR status under the treatment regimen. 
One can then define the causal effect of treatment in EFS on each principal stratum.

In this chapter, we propose a method to identify and estimate principal stratum causal effects for a binary outcome and later extend our method for censored outcome data. 
The causal estimand of interest is the treatment efficacy in 3-year EFS and OS among patients who would achieve pCR under chemotherapy plus bevacizumab as in our motivating study, the NSABP B-40 trial. 
A model of counterfactual outcome given the observed data is imposed. Using some probabilistic arguments, we connect the model parameters with quantities that can be empirically estimated from the observed data. The resulting equations allow us to estimate the model parameters and subsequently the causal estimand of interest, and resolve the identifiability issue.

The remaining chapter is organized as follows. Section~\ref{s:related} presents related work in principal stratum causal effects. Section~\ref{s:method} introduces the standard data settings, causal estimands of interest, and a regression model in the context of a randomized neoadjuvant trial. In Section~\ref{ss:estimation}, we provide key assumptions for identification of the causal estimand and introduce the proposed method. In Section~\ref{s:sim}, we conduct a simulation study to assess the performance of our method in terms of bias and coverage of bootstrap confidence intervals. In Section~\ref{appb40}, we apply the proposed method to the motivating NSABP B-40 study. We conclude with a discussion of the proposed method and future work in Section~\ref{s:discuss}.

\section{Related Work}\label{s:related}

\cite{frangakis2002principal} propose to split study population into principal strata.  
Each principal stratum is by definition independent of treatment assignment since it contains information on counterfactual, or potential outcomes rather than the observed outcome for a specific treatment assignment. One can then define treatment effects 
on each principal stratum. Additionally, any union of the basic principal strata would also be a valid principal stratum as it leads to comparisons among a common set of individuals.  \cite{gilbert2015surrogate} show the principal stratification framework is useful for evaluating whether and how treatment effects differs across subgroups characterized by the intermediate variable, thus being firmly associated with the utility of the treatment marker.

Identification of principal stratum causal effects is in general difficult. A major challenge is that we do not observe the individual membership of principal stratum because of its counterfactual nature \citep{gilbert2008evaluating, wolfson2010statistical}. 
Under the principal stratification framework, \cite{gilbert2003sensitivity} propose to perform sensitivity analyses by varying model parameters in an imposed parametric model for counterfactual outcomes.  
\cite{shepherd2006sensitivity} and \cite{jemiai2007semiparametric} extend this sensitivity analyses approach by including baseline covariates in the model. 
These sensitivity analyses can provide researchers with a range of causal estimates under different values of the sensitivity parameters. 
In reality, however, it is often unclear what the plausible values are for these sensitivity parameters and the selected combinations may not be exhaustive. 
\cite{li2010bayesian} and \cite{zigler2012bayesian} use Bayesian approaches to model the joint distribution of the counterfactual intermediate outcomes and long-term outcomes and incorporate prior information regarding non-identifiable associations. 
The lack of identifiability, however, still exists and is reflected by the over-coverage of confidence intervals in their simulation studies. 

Principal stratum causal effects with regards to outcomes truncated by death are not identifiable without further assumptions \citep{zhang2003estimation, kurland2009longitudinal, lee2010causal}. 
\cite{tchetgen2014identification} identify causal effects by borrowing information from post-treatment risk factors of the intermittent outcome and the causal estimand may vary according to the selected risk factors. 
Instrumental variables are also introduced to provide information on the unobserved principal strata and the justification of that exclusion restriction assumption is often challenging \citep{ding2011identifiability, wang2017identification}. 

All the above methods either fall into sensitivity analyses or require exclusion restriction assumptions. In this chapter, we propose a method to identify and estimate principal stratum causal effects under data settings as \cite{shepherd2006sensitivity} for a binary outcome and later extend our method to address issues of censored outcome data under mild assumptions. Identification of the causal effect is achieved with the bias minimal and the coverage probabilities close to the nominal levels.

\section{The Principal Stratification Framework of Interest}\label{s:method}

\subsection{Standard Setting for Neoadjuvant Studies}

Consider a neoadjuvant breast cancer clinical trial where patients are randomized to two treatment groups. For subject $i = 1, 2, \ldots, n$, let $Z_i \in \{0, 1\}$ be the binary treatment assignment; $X_i \in \varGamma = \{0, 1, \ldots, K\}$ be a baseline discrete covariate. A continuous baseline variable $X_i$ such as clinical tumor size, would be grouped into $K+1$ categories based on scientific knowledge. We will discuss extensions to the scenarios with a continuous $X_i$ in Section~\ref{s:discuss}. Throughout this paper, we assume that the stable unit treatment value assumption (SUTVA) \citep{rubin1980randomization} holds: the potential outcomes of any individual $i$ are unrelated to the treatment assignment of other individuals. Then we can denote $S_i(Z_i) \in \{0, 1\}$ as a binary post-randomization intermediate response such as the pCR status for subject $i$ under treatment $Z_i$ (possibly counterfactual). And denote $Y_i\{Z_i,S_i(Z_i)\}=Y_i(Z_i) \in \{0, 1\}$ as a binary long-term outcome of interest such as the EFS status at 3-year after study entry for subject $i$ under treatment $Z_i$ (possibly counterfactual). For  individual $i$, $\{Z_i, X_i, S_i(Z_i), Y_i(Z_i)\}$ represents the observed data of treatment assignment, baseline covariate, intermediate response and long-term outcome. If $Z_i = 0$, $\{S_i(0),Y_i(0)\}$ are observed and $\{S_i(1),Y_i(1)\}$ are counterfactual. If $Z_i = 1$, then $\{S_i(1),Y_i(1)\}$ are observed and $\{S_i(0),Y_i(0)\}$ are counterfactual. Thus for individual $i$, the complete counterfactual data would be $\{Z_i, X_i, S_i(0), S_i(1), Y_i(0), Y_i(1)\}$. Another important assumption is the monotonicity assumption: $S_i(0) \leq S_i(1)$ \citep{angrist1996identification}, as in the motivating NSABP B-40 study, addition of bevacizumab led to improved pCR \citep{bear2012bevacizumab}. We also assume for subject $i$, the treatment assignment $Z_i$ is independent of $X_i$ and the potential outcomes. 

Under the principal stratification framework, denote the principal strata to be $E_{jk} = \{i: S_i(0) = j, S_i(1) = k\}$,
$j,k=0,1$. The principal stratum causal effects of interest are
\begin{eqnarray*}
\theta_{jk}=\mathbb{E}\{Y_i(1)-Y_i(0)|i\in E_{jk}\},~~j,k=0,1.
\end{eqnarray*}

Under the monotonicity assumption, the principal stratum $E_{10}$ is empty. In the NSABP B-40 study, we are interested in the causal effect in $E_{01} \cup E_{11}$, those who would achieve pCR had they been treated with chemotherapy plus bevacizumab: 
% equivalently,
\begin{eqnarray*}
\theta=\mathbb{E}\{Y_i(1)-Y_i(0)|i\in E_{+1}=E_{01} \cup E_{11}\} = \mathbb{E}\{Y_i(1)-Y_i(0)|S_i(1)=1\}.
\end{eqnarray*}

Other principal stratum causal effects such as $\theta_{jk}$ can be estimated using a similar approach as we outline in Section~\ref{ss:estimation}.

\subsection{Modeling a Counterfactual Outcome}

In order to estimate the principal stratum causal effects, \cite{gilbert2003sensitivity} propose to use a logistic regression model for $\Pr \{S_i(1)=1|S_i(0)=0, Y_i(0)\}$ as
\begin{align*}
    \Pr \{S_i(1)=1|S_i(0)=0, Y_i(0)\}&=\logit^{-1}\{\beta_0+\beta_1 Y_i(0)\}.
\end{align*}

\cite{shepherd2006sensitivity} further extend the logistic regression by incorporating baseline covariates $X_i$ as
\begin{align}
\Pr \{S_i(1)=1|S_i(0)=0, Y_i(0), X_i = x\}&=\logit^{-1}\{\beta_0+\beta_1 Y_i(0)+\beta_2 x\} \nonumber \\
&=\frac{\exp\{\beta_0+\beta_1 Y_i(0)+\beta_2 x\}}{1+\exp\{\beta_0+\beta_1 Y_i(0)+\beta_2 x\}} \label{Model}.
\end{align}

\cite{jemiai2007semiparametric} consider a more general model framework:
\begin{align*}
\Pr\{S_i(1)=1|S_i(0)=0, Y_i(0),X_i=x\}=w[r(x)+g\{Y_i(0),x\}]
\end{align*}
where $w(u)\equiv\{1+\exp(-u)\}^{-1}$ and $g(\cdot,\cdot)$ is a known function. In the case of \cite{shepherd2006sensitivity}, $g(u,v)=\beta_1 u$ with $\beta_1$ known. 
\cite{jemiai2007semiparametric} show that under the monotonicity assumption, inference could be made on $\theta$ for any fixed function $g$ and sensitivity analyses could be performed by varying $g$.

\section{The Proposed Method} \label{ss:estimation}

\subsection{Key Identification Assumptions}
\label{s:assump}

Identification of causal effects is achieved through two key assumptions. First, the monotonicity assumption: $S_i(0) \leq S_i(1)$ \citep{angrist1996identification}. That is, a subject who responds under the control would respond if given the treatment. This monotonicity assumption could prove valuable \citep{bartolucci2011modeling} and can be justified in many scenarios that the additional therapy would help to improve the response. In the motivating NSABP B-40 study, addition of bevacizumab led to improved pCR \citep{bear2012bevacizumab}. 
Second, a parametric model is used to describe the counterfactual response under the treatment for a control non-respondent. Both the future long-term outcome and a baseline covariate are predictors in this parametric model. It is required that the level of the covariates is at least of the same dimension of model parameters and the imposed linearity assumption is critical to identify and estimate those regression parameters. We will elaborate the second assumption in Section~\ref{s:identify}.

\subsection{Identification of Model Parameters and Causal Estimands}
\label{s:identify}
As mentioned in \cite{shepherd2006sensitivity} and will be described in Section \ref{ss:estparam}, when the parameters of model~\eqref{Model} are identified, the causal estimands can be identified.

\begin{lemma} \label{lemma_idf}
For any $x \in \Gamma=\{ 0,1, \ldots, K\}$ and $y \in \{0,1\}$, let $a_x=\Pr\{S(1)=1|S(0)=0,X=x\}$ and  $b_{xy}=\Pr\{Y(0)=y|S(0)=0, X=x\}$. 
Let ${\mathbf a}=(a_0,a_1,\ldots,a_K)^T$ and ${\mathbf b}_y=(b_{y0},b_{y1},\ldots,b_{yK})^T$. 
Define $h_x({\bbeta},{\mathbf a},{\mathbf b}_0,{\mathbf b}_1)=a_x-\sum_{y=0}^1 b_{xy}\logit^{-1}\{\beta_0+\beta_1 y+\beta_2 x\}$, 
and $H({\bbeta}, {\mathbf a},{\mathbf b}_0,{\mathbf b}_1)=\{h_0(\beta,{\mathbf a},{\mathbf b}_0,{\mathbf b}_1),\ldots,h_K(\beta,{\mathbf a},{\mathbf b}_0,{\mathbf b}_1)\}^T$. 

If $\rank\{{\partial H({\bbeta},{\mathbf a},{\mathbf b}_0,{\mathbf b}_1)}/{\partial\bbeta}\}=3$, within the neighborhood of $\bbeta$ there is a unique solution $\bbeta=\psi({\mathbf a},{\mathbf b}_0,{\mathbf b}_1)$ such that
$H\{\psi({\mathbf a},{\mathbf b}_0,{\mathbf b}_1),{\mathbf a},{\mathbf b}_0,{\mathbf b}_1\}=0.$
\end{lemma}

\begin{proof}
For all $x\in \Gamma$, we have
\begin{align*}
a_x&=\Pr\{S(1)=1|S(0)=0,X=x\} =  \sum_{y=0}^1 \Pr\{S(1)=1,Y(0)=y|S(0)=0,X=x\} \\
&=  \sum_{y=0}^1  \Pr\{Y(0)=y|S(0)=0,X=x\} \Pr\{S(1)=1|Y(0)=y,S(0)=0,X=x\} \\
&=  \sum_{y=0}^1  b_{xy} \logit^{-1}(\beta_0+\beta_1 y+\beta_2 x). 
\end{align*}

Hence, $H({\bbeta}, {\mathbf a},{\mathbf b}_0,{\mathbf b}_1)=0$ 
and $H(\cdot)$ is a smooth function of $\bbeta, {\mathbf a},{\mathbf b}_0, \text{and } {\mathbf b}_1$. 
By invoking the implicit function theorem, 
when $\rank({\partial H}/{\partial\bbeta})=3$, there exists a smooth function $\psi$ such that $\bbeta=\psi({\mathbf a},{\mathbf b}_0,{\mathbf b}_1)$ and $H\{\psi({\mathbf a},{\mathbf b}_0,{\mathbf b}_1),{\mathbf a},{\mathbf b}_0,{\mathbf b}_1\}=0.$ 
% This completes the proof of Lemma \ref{lemma_idf}.
\end{proof}

The identifiability of model parameter $\bbeta$ depends on the availability of $a_x=\Pr\{S(1)=1|S(0)=0,X=x\}$ and $b_{xy}=\Pr\{Y(0)=y|S(0)=0, X=x\}$, for $x\in \Gamma; y=0,1$. The linearity in $X=x$ in model~\eqref{Model} also plays an important role. In general, when $\beta_2\neq 0$ and $K\geq 2$, there are equal or more equations than the number of unknown parameters in $\bbeta$, Lemma~\ref{lemma_idf} would hold. 
In practice, given $({\mathbf a},{\mathbf b}_0,{\mathbf b}_1)$, one solves for $\bbeta$ such that $H({\bbeta}, {\mathbf a},{\mathbf b}_0,{\mathbf b}_1) = 0$. Then verify that  $\rank\{{\partial H({\bbeta},{\mathbf a},{\mathbf b}_0,{\mathbf b}_1)}/{\partial\bbeta}\}=3$ at the solution.

\subsection{Estimation of Causal Estimands}
\label{s:Estimation}

The causal estimand of interest is
\begin{align}
    \theta = \mathbb{E}\{Y_i(1) - Y_i(0)|S_i(1)=1\}
    % &= \mathbb{E}[Y_i(1)|S_i(1) = 1] - \mathbb{E}[Y_i(0)|S_i(1) = 1] \nonumber \\
    &= \mathbb{E}\{Y_i(1)|S_i(1)=1\} - \mathbb{E}\{Y_i(0)|S_i(1)=1\}. \label{trtDiff}
\end{align}

Because $\{Y_i(1),S_i(1)\}$ are observed for subjects in the treatment arm, $\Pr \{Y_i(1)=1|S_i(1)=1\}$ can be estimated by
\begin{align}
    \widehat{\Pr} \{Y_i(1)=1|S_i(1)=1\} 
    % &= \Pr \{Y_i(1)=1|i \in E_{+1}, Z_i=1\} \\
    &= \frac{\sum_i \mathbbm{1}\{Z_i=1, S_i(1)=1, Y_i(1)=1\}}{\sum_i \mathbbm{1}\{Z_i=1, S_i(1)=1\}}. \label{prY11_E+1}
\end{align}
where $\mathbbm{1}\{\cdot\}$ is the indicator function. 

Meanwhile,
\begin{align}
    \Pr \{Y_i(0) = 1|S_i(1)=1\}
    &= \frac{\Pr \{S_i(1) = 1, Y_i(0) = 1\}}{\Pr \{S_i(1) = 1\}} \nonumber \\
    &=\frac{\sum_x \Pr \{S_i(1) = 1, Y_i(0) = 1|X_i = x\} \cdot \Pr \{X_i = x\}}{\sum_x \Pr \{S_i(1) = 1|X_i = x\} \cdot \Pr \{X_i = x\}} \label{prY01}
\end{align}

In equation \eqref{prY01}, $\Pr \{X_i=x\}$ can be estimated by $\widehat{\Pr} \{X_i=x\} = {{\sum_i \mathbbm{1}(X_i=x)}/{n}}$ and
\begin{align}
    &\Pr \{S_i(1) = 1, Y_i(0) = 1|X_i = x\} \nonumber \\
     &=  \sum_{j=0}^{1} \Pr \{S_i(1) = 1, Y_i(0) = 1, S_i(0) = j|X_i = x\} \nonumber \\
     &=  \sum_{j=0}^{1} \Pr \{S_i(1) = 1, Y_i(0) = 1 | S_i(0) = j, X_i = x\} \cdot \Pr \{S_i(0) = j|X_i = x\} \nonumber \\
    &=  \sum_{j=0}^{1} \Big [ \Pr \{S_i(1)=1|S_i(0)=j, Y_i(0)=1, X_i = x\} \nonumber \\
    &\quad \cdot \Pr \{Y_i(0)=1|S_i(0)=j, X_i = x\} \cdot \Pr \{S_i(0) = j|X_i = x\} \Big ]. \label{factor}
    % &= \Pr \{S_i(1)=1|S_i(0)=0,Y_i(0)=1, X_i = x\} \\
    % &\quad \cdot \Pr \{Y_i(0)=1|S_i(0)=0, X_i = x\} \cdot \Pr \{S_i(0)=0|X_i = x\} \\
    % &\quad + \Pr \{S_i(1)=1|S_i(0)=1,Y_i(0)=1, X_i = x\} \\
    % &\quad \cdot \Pr \{Y_i(0)=1|S_i(0)=1, X_i = x\} \cdot \Pr \{S_i(0)=1|X_i = x\}
\end{align}

In equation \eqref{factor}, $\Pr \{Y_i(0)=1|S_i(0)=j, X_i = x\}$, $j=0,1$, can be estimated by
\begin{gather*}
    \begin{split}
        \widehat{\Pr} \{Y_i(0)=1|S_i(0)=j, X_i = x\} 
        % &= \widehat{\Pr} \{Y_i(0)=1|S_i(0)=0, X_i = x, Z_i=0\} \\
        &= \frac{\sum_i \mathbbm{1}\{Z_i=0, S_i(0)=j, Y_i(0)=1, X_i=x\}}{\sum_i \mathbbm{1}\{Z_i=0, S_i(0)=j, X_i=x\}}.
    \end{split} 
\end{gather*}

By the monotonicity assumption, $\Pr \{S_i(1)=1|S_i(0)=1,Y_i(0)=1, X_i = x\}\equiv 1$. 

The estimation of $\Pr \{S_i(j) = 1|X_i = x\}$, $j=0,1$, is described in Lemma \ref{lemma_in}.
\begin{lemma} \label{lemma_in}
Under the monotonicity assumption, for any $x$,
we denote 
\begin{eqnarray*}
\widehat{q}_j(x)=\frac{\sum_i \mathbbm{1}\{Z_i=j,S_i(j)=1,X_i=x\}}{\sum_i \mathbbm{1}\{Z_i=j,X_i=x\}},~~j=0,1;
\end{eqnarray*}
the observed proportions of responders in the control group and the treatment group with $X=x$, respectively. 

We use maximum likelihood estimation to estimate $\Pr \{S_i(j) = 1|X_i = x\}$, $j=0,1$.
\begin{itemize}
  \item[(a)] when $\widehat{q}_1(x)\geq \widehat{q}_0(x)$, the maximum likelihood estimate of 
$\Pr\{S_i(j)=1|X_i=x\}$ is $\widehat{q}_j(x)$, $j=0,1;$
  \item[(b)] when $\widehat{q}_1(x)< \widehat{q}_0(x)$, the maximum likelihood estimate of 
$\Pr\{S_i(j)=1|X_i=x\}$ is \\ ${{\sum_i \mathbbm{1}(S_i=1,X_i=x)}/{\sum_i \mathbbm{1}(X_i=x)}}$, $j=0,1$.
\end{itemize}
\end{lemma}

In the second scenario, the estimates are the same as the pooled proportion of responders among patients with $X=x$. The proof of Lemma \ref{lemma_in} is presented in Appendix~\ref{supplA}.

The last item in equation \eqref{prY01} needed for estimating the causal estimand is $\Pr \{S_i(1)=1|S_i(0)=0,Y_i(0)=1, X_i = x\}$.  \cite{gilbert2003sensitivity} and \cite{shepherd2006sensitivity} conduct sensitivity analyses by varying the values of the $\bbeta$ in model~\eqref{Model}. In Section~\ref{ss:estparam}, we will discuss how to estimate $\bbeta$ using a probabilistic equation.

\subsection{Estimation of Model Parameters}\label{ss:estparam}

Let
\begingroup
\allowdisplaybreaks
\begin{align*}
    & G_L(x) = \Pr \{ S_i(1)=1|S_i(0)=0, X_i=x \}  \\
    & G_R(x,y) = \Pr \{ Y_i(0)=y|S_i(0)=0, X_i=x \}  \\
    & G_M(x,y;\bbeta) = \Pr \{ S_i(1)=1|S_i(0)=0, Y_i(0)=y, X_i=x \}.
\end{align*}
\endgroup
This leads to an equation system:
\begin{equation*}
    G_L(x) = \sum_{y=0}^1 G_M(x,y;\bbeta) \cdot G_R(x,y); x \in \Gamma \label{GM}
\end{equation*}

We can estimate $G_L(x)$ with the following empirical estimates from the observed data by
\begin{align*}
    \widehat{G}_L(x) 
    % &= \widehat{\Pr} \{ S_i(1)=1|S_i(0)=0, X_i=x \} \\
    &= \frac{\widehat{\Pr}\{S_i(0)=0, S_i(1)=1|X_i=x\}}{\widehat{\Pr} \{ S_i(0)=0|X_i=x \}}
\end{align*}
where the numerator and the denominator are derived from Lemma \ref{lemma_in}. The details are presented in Appendix~\ref{supplA}. 

Because $\{X_i, S_i(0), Y_i(0)\}$ are observed for subjects in the control arm, $G_R(x,y)$ can be estimated by
\begin{align*}
    \widehat{G}_R(x,y) 
    % &= \widehat{\Pr} \{ Y_i(0)=y|S_i(0)=0, X_i=x \} \\
    = \frac{\sum_i \mathbbm{1}\{Z_i=0, S_i(0)=0, Y_i(0)=y, X_i=x\}}{\sum_i \mathbbm{1}\{Z_i=0, S_i(0)=0, X_i=x\}}
\end{align*}

With $\widehat{G}_L(x)$ and $\widehat{G}_R(x,y)$ estimated from the observed data and $G_M(x,y;\bbeta)$ specified as the regression model in equation~\eqref{Model}, we have
\begin{equation}
    \widehat{G}_L(x) = \sum_{y=0}^1 G_M(x,y;\bbeta) \cdot \widehat{G}_R(x,y); ~~x \in \Gamma \label{GM-estimated}
\end{equation}

The number of unknown parameters $\bbeta$ in system of equations \eqref{GM-estimated} is three and the number of equations is $(K+1)$, for $X_i \in \varGamma = \{0, 1, \ldots, K\}$. For \eqref{GM-estimated}, when $K+1<3$, we cannot uniquely solve for $\bbeta$. When $K+1=3$, the number of equations is the same as the number of unknown parameters and in general we can solve for $\bbeta$. When $K+1>3$, there are more equations than the number of unknown parameters, and there are generally no exact solutions to the equation systems (\ref{GM-estimated}). In that case, we propose to estimate $\bbeta$ by 
\begin{equation}
    \widehat{\bbeta} = \operatorname*{\arg\,min}_{\bbeta} \sum_{x=0}^K\{\widehat{G}_L(x) - \sum_{y=0}^1 G_M(x,y;\bbeta) \cdot \widehat{G}_R(x,y)\}^2 \label{betabeta}
\end{equation}
where $\widehat{G}_L(x)$, $\widehat{G}_R(x,y)$ and $G_M(x,y;\bbeta)$ are probabilities bounded between 0 and 1.

With $\bbeta$ estimated, we can estimate the causal estimand $\theta$ via the procedure outlined in Section~\ref{s:Estimation}.

% \subsection{Consistency of \texorpdfstring{$\widehat{\bbeta}$}{TEXT} and \texorpdfstring{$\widehat{{\theta}}$}{TEXT}} \label{consistency}
\subsection{Consistency of Model Parameters and Causal Estimands} \label{consistency}

Here we provide the theoretical guarantee of our estimators $\bbeta$ and $\theta$.

Let 
\begin{eqnarray*}
&& Q_0^{(x)}(\bbeta)= \{{G}_L(x) - {\sum_{y=0}^1} G_M(x,y;\bbeta) \cdot {G}_R(x,y)\}^2,~~ x \in \Gamma; y=0,1\\
&& \tilde{Q}_0(\bbeta)=\{Q_0^{(0)}(\bbeta),Q_0^{(1)}(\bbeta),\ldots, Q_0^{(K)}(\bbeta)\}^T, \\
&&Q_n(\bbeta) = \sum_{x=0}^K Q_n^{(x)}(\bbeta) = \sum_{x=0}^K \{\widehat{G}_L(x) - \sum_{y=0}^1 G_M(x,y;\bbeta) \cdot \widehat{G}_R(x,y)\}^2
\end{eqnarray*}

\begin{theorem}\label{Theorem1}
Under the following conditions:
\begin{itemize}
  \item[(a)] $\bbeta$ satisfies $Q_0^{(x)}(\bbeta)=0$, $\forall x\in \Gamma=\{0,1,\ldots,K\}$.
  \item[(b)] $\rank | {{\partial \tilde{Q}_0(\bbeta)}/{\partial \bbeta}}  | \geq \dimselfdefined(\bbeta)$. 
  \item[(c)] ${\widehat{G}_L(x) \overset{p}{\to} G_L(x), \widehat{G}_R(x, y) \overset{p}{\to} G_R(x, y)}$, as ${n \rightarrow \infty}$, $\forall x \in \Gamma; \forall y=0,1$.
\end{itemize} 
Then ${\widehat{\bbeta}=\operatorname*{\arg\,min}_{\bbeta} Q_n(\bbeta) \overset{p}{\to} \bbeta}$ and the causal estimand $\widehat{\theta} \overset{p}{\to} \theta$ as $n \to \infty$. 
\end{theorem}
The detailed proof of Theorem \ref{Theorem1} is presented in Appendix~\ref{supplB}.

\subsection{Extension to Censored Data}

As in the motivating NSABP B-40 study, the long-term outcome $Y_i$ may be subject to right censoring. For any time $T=t_0$ of interest, the binary counterfactual outcomes would be $\{Y_i(0;t_0),Y_i(1;t_0)\}$ and the causal estimand can be formulated as
\begin{align*}
\theta(t_0)=\mathbb{E}\{Y_i(1;t_0)-Y_i(0;t_0)|i\in E_{+1}\}.
\end{align*}
With $Y_i$ subject to censoring, $\Pr\{Y_i(1;t_0)=1|i\in E_{+1}\}$ can be estimated by the Kaplan-Meier (KM) estimates at time $T=t_0$. The estimation is similar for other relevant quantities such as $\Pr\{Y_i(0;t_0)=1|S_i(0)=j,X_i=x\}$ in equation \eqref{factor} under the scenario where $Y_i(Z_i)$ is always observed.

\section{Simulation Studies}\label{s:sim}

A simulation study is used to assess the performance of the proposed method. The setup is chosen to resemble the NSABP B-40 study by simulating treatment assignment, baseline tumor size category, binary pCR response status, and binary survival status, specifically:
\begin{equation*}
    \mathfrak{D} = [ D_i = \{Z_i, X_i, S_i(0), S_i(1), Y_i(0), Y_i(1)\} , \text{ } i=1, \ldots ,n].
\end{equation*}

We simulate the subject-level data as follows. First, we simulate the categorical baseline tumor category $X_i$ from a multinomial distribution with $\Pr \{X_i=x \} = 0.25, x \in \{0,1,2,3\}$. Next, we simulate $S_i(0)$ given $X_i$ from a Bernoulli distribution with $\Pr \{S_i(0) = 1|X_i = x\} = p(x) \text{ with } p(0),\,p(1),\,p(2),\,p(3) = 0.3,\, 0.25, \,0.25, \,0.2$, respectively. We then simulate the survival status under control, $Y_i(0)$, with a Bernoulli draw with $\Pr\{Y_i(0)=1|S_i(0) = 0, X_i=x\} = 0.7, \,0.65, \,0.6, \,0.55$ for $x=0,\,1,\,2,\,3$, respectively and $\Pr\{Y_i(0)=1|S_i(0) = 1, X_i=x\} = 0.84, \,0.78, \,0.72, \,0.66$ for $x=0,\,1,\,2,\,3$, respectively. The choice of these numbers reflects a 20\% improvement in 3-year EFS for respondents over nonrespondents under the control regimen. 

Next, we simulate the conditional distribution $\{S_i(1)|S_i(0),Y_i(0),X_i\}$. For subjects with $S_i(0) = 1$ we set $S_i(1)$ to be 1 to enforce the monotonicity assumption. For subjects with $S_i(0) = 0$ we draw $S_i(1)$ from a Bernoulli distribution: $\Pr \{ S_i(1)=1|S_i(0)=0, Y_i(0)=y, X_i=x \} = \logit^{-1} \{\beta_0+\beta_{1}y+\beta_{2}x\}$. We try different settings for $\bbeta$ = (-3, -5, 0.2), (-5, -1, -2), and (-7, 3, 0.2).

We then simulate the survival status under treatment, $Y_i(1)$, according to the following probability distributions:
% \vspace*{-0.4cm}
\begingroup
\allowdisplaybreaks
\begin{eqnarray*}
&&\Pr \{ Y_i(1)=1|S_i(0)=0, S_i(1)=0, Y_i(0)=0\} = 0.5, \\
&&\Pr \{ Y_i(1)=1|S_i(0)=0, S_i(1)=0, Y_i(0)=1\} = 0.6, \\
&&\Pr \{ Y_i(1)=1|S_i(0)=0, S_i(1)=1, Y_i(0)=0\} = 0.85, \\
&& \Pr \{ Y_i(1)=1|S_i(0)=0, S_i(1)=1, Y_i(0)=1\} = 0.9, \\
&& \Pr \{ Y_i(1)=1|S_i(0)=1, S_i(1)=1, Y_i(0)=0\} = 0.85, \\
&& \Pr \{ Y_i(1)=1|S_i(0)=1, S_i(1)=1, Y_i(0)=1\} = 0.9.
\end{eqnarray*}
\endgroup
% \vspace*{-0.5cm}
These probabilities are chosen to make the 3-year EFS under treatment greater for those who would obtain pCR under treatment than those who would not, and have a greater 3-year EFS for those patients who would be event-free under control than those who would not be event-free under control. We set these probabilities to be independent of the baseline tumor size given the potential outcomes $\{S_i(0),S_i(1),Y_i(0)\}$.

Lastly we simulate the treatment assignment with equal probability for each arm as a Bernoulli draw with $\Pr\{Z_i = 0\}$ and $\Pr\{Z_i = 1\}$ both equal to 0.5 to ensure that independence between potential outcomes and treatment assignment. 
For the simulated data the true average causal effect for principal stratum $S_i(1)=1$, $\mathbb{E}\{Y_i(1)-Y_i(0)|S_i(1) = 1\}$, can be calculated using the above parameters for simulations. The detailed calculations is given in Appendix~\ref{supplC}. 
Under the three parameter settings the true values of the causal estimands are $\theta$=0.179, 0.130, and 0.120, respectively. This means that under the three different settings, if the treatment was administered to all subjects who would achieve pCR under treatment there would be a 17.9\%, 13.0\%, 12.0\% increment in survival respectively, within the time frame under consideration, than had all of them taken the control instead. 

Under each parameter setting and a chosen sample size $n$=1000, 2000, or 4000, we simulate $R$=1000 replicates. A quasi-Newton method, the Broyden-Fletcher-Goldfarb-Shanno algorithm, is used for the optimization. 
We create $B$=500 bootstrap samples to obtain the 95\% confidence interval for the causal estimates. Let $\widehat{\theta}^{(r)}$ be the mean estimate among bootstrap samples from the $r$ replicate, $r = 1, \ldots ,R$. 

We construct bootstrap confidence intervals to account for the variability introduced by estimating model parameters. 
We use the basic bootstrap CI, or the pivotal CI \citep{davison1997bootstrap} for constructing CIs from bootstrap estimates. 
Let 
$\{\widehat{\theta}^{(1)}, \widehat{\theta}^{(2)}, \ldots, \widehat{\theta}^{(B)}\}$ are the causal effect estimates from $B$ bootstrap samples. Denote $\theta^*_{(1-\alpha/2)}$ and $\theta^*_{(\alpha/2)}$ as the $100(1 - \alpha/2)\%$ and $100(\alpha/2)\%$ of the bootstrap causal effect estimates. The $100(1-\alpha)\%$ bootstrap confidence interval is given by $(2\widehat{\theta} - \theta^*_{(1-\alpha/2)}, 2\widehat{\theta} - \theta^*_{(\alpha/2)})$ 
where $\widehat{\theta}$ is the estimate from the data.

We report the empirical bias, mean squared error (MSE), average length of 95\% CIs, and the coverage of those CIs, where
$$\text{Bias}(\widehat{\theta}) = R^{-1}\sum_{r=1}^R \{\widehat{\theta}^{(r)} - \theta \},$$ 
$$\text{MSE}(\widehat{\theta}) =  R^{-1}\sum_{r=1}^R \{\widehat{\theta}^{(r)} - \theta \}^2,$$
$$\text{$95\%$ CI width} = R^{-1}\sum_{r=1}^R|\widehat{\theta}_{U,0.05}^{(r)} - \widehat{\theta}_{L,0.05}^{(r)}|,$$
$$\text{$95\%$ CI coverage} = R^{-1}\sum_{r=1}^R \mathbbm{1}\{\theta \in (\widehat{\theta}_{L,0.05}^{(r)}, \widehat{\theta}_{U,0.05}^{(r)})\}$$ 
with $\widehat{\theta}_{L,0.05}^{(r)}$ and $\widehat{\theta}_{U,0.05}^{(r)}$ the lower bound and upper bound of the $95\%$  bootstrap CIs of $\widehat{\theta}$ from the $r^{th}$ simulated dataset. 
Table \ref{t:table_simulate} shows the simulation results of the proposed method under three different parameter settings and various sample sizes. Our simulation results show the identification of causal effects is achieved with the bias negligible and the coverage probabilities close to the nominal levels.

\clearpage
\begin{table}[hbt!]%[htp]%[!p]
 \centering
  \caption{Simulation results of the proposed method under three different parameter settings and various sample sizes.}
    \label{t:table_simulate}
    % \fontsize{11}{13.2}\selectfont  %\fontsize{size}{skip}\selectfont
    %Set font size. The first parameter is the font size to switch to; the second is the \baselineskip to use. The unit of both parameters defaults to pt. A rule of thumb is that the baselineskip should be 1.2 times the font size.
\begin{normalsize}
\begin{tabular}{ccccc}
\toprule
{Sample size} & {Empirical bias} & MSE &  {95\% CI width} & {95\% CI coverage}\\
\midrule
\multicolumn{5}{c}{{Setting 1: $\bbeta$=(-3, -5, 0.2), $\theta$=0.179}}  \\
 \addlinespace[0.5ex]
 1000 & -0.011 & 3.001e-3 & 0.206 & 0.952 \\
 2000 & -0.006 & 1.539e-3 & 0.155 & 0.955 \\
 4000 & -0.002 & 6.755e-4 & 0.116 & 0.962 \\
 \addlinespace[1.5ex]
 \multicolumn{5}{c}{{Setting 2: $\bbeta$=(-5, -1, -2), $\theta$=0.130}}  \\
 \addlinespace[0.5ex]
 1000 & -6.011e-5 & 2.496e-3 & 0.185 & 0.943 \\
 2000 & 9.358e-4 & 1.137e-3 & 0.130 & 0.948 \\
 4000 & 1.086e-4 & 5.462e-4 & 0.093 & 0.950 \\
 \addlinespace[1.5ex]
 \multicolumn{5}{c}{{Setting 3: $\bbeta$=(-7, 3, 0.2), $\theta$=0.120}}  \\
 \addlinespace[0.5ex]
 1000 & 0.008 & 2.547e-3 & 0.194 & 0.955 \\
 2000 & 0.006 & 1.319e-3 & 0.141 & 0.957 \\
 4000 & 0.003 & 6.363e-4 & 0.100 & 0.953 \\
\bottomrule
\end{tabular}
\end{normalsize}
\end{table}

\section{Application to NSABP B-40 Trial}\label{appb40}

\subsection{B-40 Data Analysis}
\label{b40-intro}

Here we apply the proposed method to the NSABP B-40 study \citep{bear2012bevacizumab, bear2015neoadjuvant}. Among the 1206 enrolled participants, 13 withdrew consent, 7 had missing data and 2 had had inoperable disease after chemotherapy. Another 15 patients did not have nodal assessment so their pCR status was not ascertained. We conduct our analysis among the rest 1169 patients. Our purpose is to estimate the causal treatment effect in 3-year EFS and OS among patients who would obtain a pCR had bevacizumab been added to their treatment regimen. KM estimates are used since there are 61 patients censored at 3 years.

To apply our method, the clinical tumor size is used as the baseline auxiliary covariate $X$. Patients are grouped into four nearly equal-sized groups: 2-3 cm, 3.1-4 cm, 4.1-6 cm and $>$6 cm, based on breast cancer expert knowledge. We code these four tumor size groups into $\{0,1,2,3\}$, respectively. Among the 589 patients in the control arm, the proportions of those who achieved pCR in each patient group are 28\%, 23\%, 22\% and 17\%, respectively; among the 580 patients in the treatment arm, the proportions of those who achieved pCR are 31\%, 26\%, 25\% and 27\%, respectively. This does not violate the monotonicity assumption $S_i(0) \leq S_i(1)$. The 3-year long-term outcome status $Y_i=1$ if the patient $i$ survived within the first 3 years and 0 otherwise.

We calculate the 95\% bootstrap confidence intervals from $500$ bootstrap samples. 
The estimated causal treatment effect in 3-year EFS among those who would obtained pCR under treatment is $\widehat{\theta}_{\text{EFS}}=0.180$ (95\% CI=(0.056, 0.377)) with $\widehat{\bbeta}=(-1.797,-5.874,0.285)$. The estimated causal treatment effect in 3-year OS among those who would obtained pCR under treatment is $\widehat{\theta}_{\text{OS}}=0.175$ (95\% CI=(0.062, 0.354)) with $\widehat{\bbeta}=(-1.85,-4.764,0.289)$. For both scenarios, because 0 is outside of the 95\% CIs, we would claim that the addition of bevacizumab improves 3-year EFS and OS among patients who would respond to neoadjuvant chemotherapy plus bevacizumab at a 95\% confidence level.

\subsection{Sensitivity of Initial Parameters in Optimization}

For the real data application, the initial estimate $\bbeta_{init} = (\beta_0, \beta_1, \beta_2)$ is set at $(0, 0, 0)$. %since we do not have any information on the model parameters. 
To see the sensitivity of initial parameters, we try $9261 = 21\times 21 \times 21$ different initial values of $\bbeta_{init}$, with $\beta_0$, $\beta_1$, and $\beta_2$ on the integer grids of $[-10,10]\times [-10,10]\times [-10,10]$. The corresponding histograms of causal estimates in 3-year EFS and 3-year OS at convergence are presented in Figure~\ref{figure_hist}. Our estimated model parameters $\widehat{\bbeta}$ in Section~\ref{b40-intro} achieves the minimum loss of equation (\ref{betabeta}). Except for some extreme initialization such as (10,10,10), most of the $\widehat{\theta}$ are the same or very close to the causal estimates calculated by using $\bbeta_{init}$ = (0,0,0) as initial parameters. Therefore, we conclude that the causal estimand is not sensitive to the initial parameter settings in optimization. In practice, we suggest running optimization with various initial values and identify the right estimate.

\begin{figure}[!htb]%[!p]
% \vspace{1cm}
% \centering\includegraphics[width=\linewidth]{figures/hist_EFS_OS.pdf}
\centering
\centerline{\includegraphics[width=\linewidth]{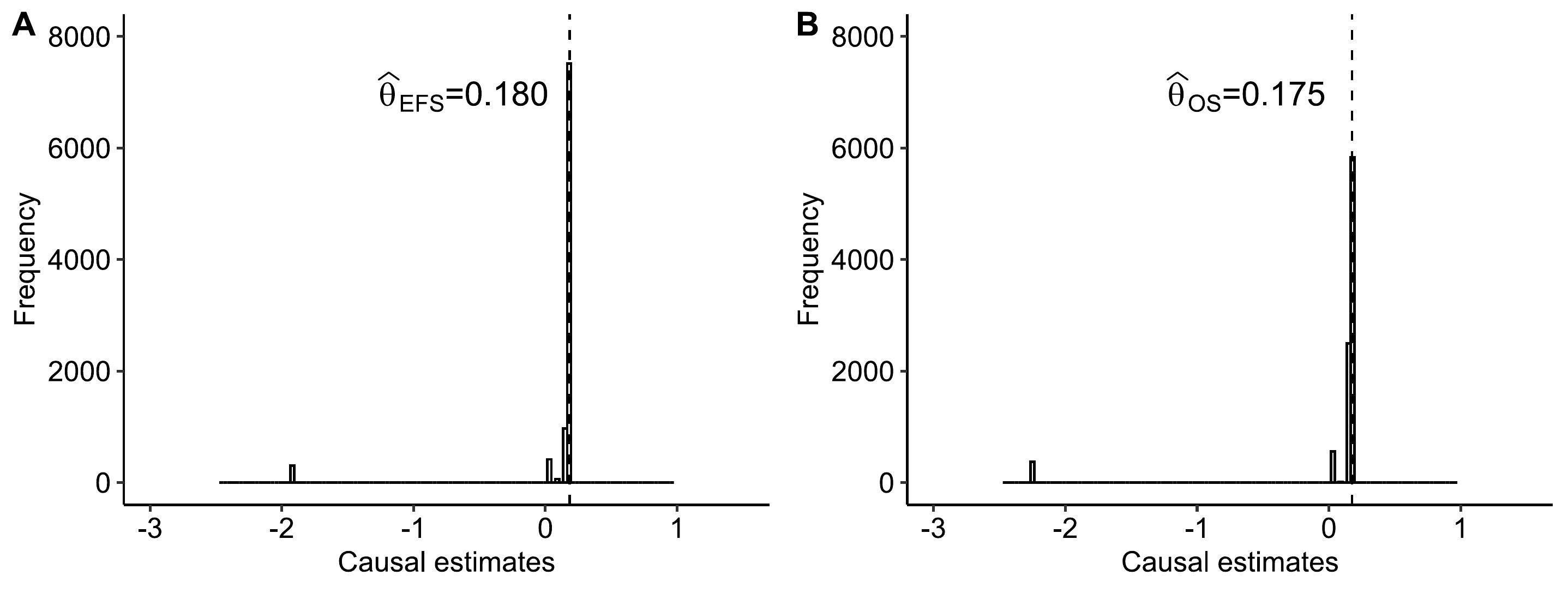}}
% \vspace{-1cm}
\caption{
Histogram of the causal estimates obtained from different initial values of model parameters in the optimization process for 3-year EFS (Figure A) and 3-year OS (Figure B), respectively. 
Except for some extreme initialization, most of the causal estimates are the same or very close to the causal estimate calculated by using zeros as initial parameters.
}
\label{figure_hist}
\end{figure}

\subsection{Comparisons to Sensitivity Analysis Method}

We compare the performance of our method with that of the sensitivity analysis similar to \cite{gilbert2003sensitivity} and \cite{shepherd2006sensitivity}. Recall that for $X=x \in \varGamma = \{0, 1, \ldots, K\}$, we have an equation system:
\begin{align*}
    \widehat{G}_L(x) = \sum_{y=0}^1 G_M(x,y;\bbeta) \cdot \widehat{G}_R(x,y); ~~x \in \Gamma  = \{0, 1, \ldots, K\}
\end{align*}
where $G_M(x,y;\bbeta) = \logit \{ \beta_0+\beta_{1}y+\beta_{2}x \}$. In the sensitivity analysis we vary the value of $\beta_1$ from -7 to -3. Then for each category of $x$ we define $\beta_x = \beta_0 + \beta_2 x$. Under this reparameterization we have only one unknown parameter, $\beta_x$, for each equation. We then solve for $\beta_x$ for each equation independently and obtain the causal estimand subsequently.

By varying values of $\beta_1$ around the estimated $\widehat{\beta_1}$ from Section~\ref{b40-intro}, the corresponding causal estimands in 3-year EFS and 3-year OS are presented in Table \ref{t:table_sensi}. The estimated causal effects in 3-year EFS vary from 0.159 to 0.181 with none of the 95\% CIs including 0; the estimated causal effects in 3-year OS vary from 0.132 to 0.176 with none of the 95\% CIs including 0. These intervals overlap a lot with the confidence intervals of real data. %in Table \ref{t:table_real}. 
These results suggest the addition of bevacizumab may improve 3-year EFS and 3-year OS among patients who would respond to neoadjuvant chemotherapy plus bevacizumab.

\begin{table}[!htb]
\caption{Sensitivity analysis for the estimated causal effect of bevacizumab in 3-year survival among those who would obtain pCR under chemotherapy plus bevacizumab.}
\label{t:table_sensi}
\begin{center}
\begin{tabular}{llcc}
\toprule
Long-term survival  & $\beta_1$ &$\widehat{\theta}$ &{95\% CI for $\widehat{\theta}$} \\ \midrule 
EFS  & -7 & 0.181 & (0.025, 0.290) \\
     & -6 & 0.180 & (0.043, 0.289) \\
     & -5 & 0.178 & (0.040, 0.282) \\
     & -4 & 0.172 & (0.058, 0.272) \\
     & -3 & 0.159 & (0.065, 0.267) \\
\addlinespace[1ex]
OS   & -7 & 0.176 & (0.055, 0.278) \\
     & -6 & 0.172 & (0.067, 0.267) \\
     & -5 & 0.166 & (0.069, 0.267) \\
     & -4 & 0.153 & (0.066, 0.235) \\
     & -3 & 0.132 & (0.064, 0.200) \\
\bottomrule
\end{tabular}
\end{center}
\end{table}

\section{Discussion and Future Work}\label{s:discuss}

We have proposed a method under the principal stratification framework to estimate causal effects of a treatment on a binary long-term endpoint conditional on a post-treatment binary marker in randomized controlled clinical trials. We also extend our method to address censored outcome data. In our motivating study, we demonstrate the causal effect of the new regimen in the long-term survival for patients who would achieve pCR. Other principal stratum causal effects can be estimated in a similar fashion. Our approach can play an important role in a sensitivity analysis.

Identification of causal effects is achieved through two assumptions. First, a subject who responds under the control would respond if given the treatment. This monotonicity assumption could prove valuable \citep{bartolucci2011modeling} and can be justified in many scenarios that the additional therapy would help to improve the response. When the auxiliary variable $X$ is discrete, we can identify and estimate $\Pr\{S(1)=1|S(0)=0,X\}$ under the monotonicity assumption. Second, a parametric model is used to describe the counterfactual response under the treatment for a control non-respondent \citep{shepherd2006sensitivity}. Both the future long-term outcome and a baseline covariate are predictors in this parametric model. 
\cite{shepherd2006sensitivity} does not consider when the auxiliary $X$ is discrete, the parameters of model~\eqref{Model} can be identified when the level of the discrete covariate is at least of the same dimension of model parameters. Instead they perform sensitivity analyses by varying the values of those model parameters in order to estimate the causal estimands. It is recognized that no diagnostic tool is available to verify the validity of this counterfactual model. 

In the motivating dataset, we discretize a continuous baseline variable into several levels. In practice, the linearity assumption may not hold. We would consider a two-pronged approach: 1) to estimate $G_L(x)$ and $G_R(x,y)$ by nonparametric estimates such as spline or kernel density estimates for a univariate continuous $X$; 2) to use a more flexible model for the counterfactual response such as a logistic regression with natural cubic spline with fixed and even-spaced knots along the domain of $X$. For each given $x$, we can still use the same probabilistic argument to link those estimates and the model parameters. The objective function would be a weighted sum of the squared difference of those probabilistic estimates.

\chapter{A Tree-based Model Averaging Approach for Personalized Treatment Effect Estimation from Heterogeneous Data Sources}\label{chap:hetero}

\section{Introduction}

Estimating individualized treatment effects has been a hot topic because of its wide applications, ranging from personalized medicine, policy research, to customized marketing advertisement. Treatment effects of certain subgroups within the population are often of interest. Recently, there has been an explosion of research devoted to improving estimation and inference of covariate-specific treatment effects, or conditional average treatment effects (CATE) at a target research site \citep{athey2016recursive,wager2018estimation,hahn2020bayesian,kunzel2019metalearners,nie2020quasioracle}.  
However, due to the limited sample size in a single study, improving the accuracy of the estimation of treatment effects remains challenging.

Leveraging data and models from various research sites to conduct statistical analyses is becoming increasingly popular \citep{reynolds2020leveraging, cohen2020leveraging, berger2015optimizing}.
Distributed research networks have been established in many large scale studies \citep{fleurence2014launching,hripcsak2015observational,platt2018fda,donohue2021use}.  
A question often being asked is whether additional data or models from other research sites could bring improvement to a local estimation task, especially when a single site does not have enough data to achieve a desired statistical precision. 
This concern is mostly noticeable in estimating treatment effects where sample size requirement is high yet observations are typically limited. 
Furthermore, information exchange between data sites is often highly restricted due to privacy, feasibility, or other concerns, prohibiting centralized analyses that pool data from multiple sources \citep{maro2009design,brown2010distributed,toh2011comparative,raghupathi2014big,deshazo2015comparison,donahue2018veterans,dayan2021federated}. 
One way to tackle this challenge is through model averaging \citep{raftery1997bayesian}, where multiple research sites collectively contribute to the tasks of statistical modeling without sharing sensitive subject-level data.
Although this idea has existed in supervised learning problems  \citep{dai2011greedy,mcmahan2017communication}, 
to our best knowledge, there are no established model averaging approach and theoretical results on estimating CATE in a distributed environment. The extension is non-trivial because CATE is unobserved in nature, as opposed to prediction problems where labels are given.

This chapter focuses on improving the prediction accuracy of CATE concerning a target site by leveraging models derived from other sites 
where 
\textit{transportability}  \citep[to be formally defined in Section~\ref{sec:assumption},][]{pearl2011transportability,stuart2011use,pearl2014external,bareinboim2016causal,buchanan2018generalizing,dahabreh2019generalizing} may not hold. 
Specifically, there may exist heterogeneity in treatment effects. In the context of our multi-hospital example, these are: 
1) \textbf{local heterogeneity}: within a hospital, patients with different characteristics may have different treatment effects. 
This is the traditional notion of CATE; 
and 2) \textbf{global heterogeneity}: where the same patient may experience different treatment effects at different hospitals. The second type of heterogeneity is driven by site-level confounding, and hampers the transportability of models across hospital sites. 
We also note that these two types of heterogeneity may interact with each other in the sense that transportability is dependent on patient characteristics, which we will address.

We propose a model averaging framework that uses a flexible tree-based weighting scheme to combine learned models from sites that takes into account heterogeneity. The contribution of each learned model to the target site depends on subject characteristics. This is achieved by applying tree splittings \citep{breiman1984classification} at both the site and the subject levels. 
For example, effects of a treatment in two hospitals may be similar for female patients but not for male, suggesting us to consider borrowing information across sites only on selective subgroups. 
Our approach extends the classic model averaging framework \citep{raftery1997bayesian,wasserman2000bayesian,hansen2007least,yang2001adaptive} by allowing data-adaptive weights, which are interpretable in a sense that they can be used to lend credibility to transportability. 
For example, in the case of extreme heterogeneity where other sites merely contribute to the target, the weights can be used as a diagnostic tool to inform the decision against borrowing information.

Our primary contributions are summarized as follows. 
{{1)}} We propose a model averaging scheme with interpretable weights that are adaptive to both local and global heterogeneity via tree-splitting dedicated to improving CATE estimation under distributed data networks. 
{{2)}} We generalize model averaging techniques to study the transportability of causal inference. Causal assumptions with practical implications are explored to warrant the use of our approach. 
{{3)}} We provide an extensive empirical evaluation of the proposed approach with a concrete real-data example on how to apply the method in practice. 
{{4)}} Compared to other distributed learning methods, the proposed framework enables causal analysis without sharing subject-level data, is easy to implement, offers ease of operations, and minimizes infrastructure, which facilitates practical collaboration within research networks.

The remaining chapter is organized as follows. In Section~\ref{sec:related}, we present a general formulation of the problem and discuss related work on model averaging and data fusion. We describe the proposed method and assumptions in detail in Section~\ref{sec:method}. 
%and estimators for comparison in Section~\ref{sec:estimators}. 
The performance of the proposed method is assessed by simulation experiments in Section~\ref{sec:simulation} and illustrated through a multi-hospital electronic health data application for critical care medicine in Section~\ref{sec:application} to estimate conditional treatment effects for oxygen therapy. 
We conclude the chapter in Section~\ref{sec:disc}.

\section{Related Work} \label{sec:related}

There are two types of construct of a distributed database \citep{breitbart1986database}: \emph{homogeneous} versus \emph{heterogeneous}. 
For homogeneous data sources, data across sites are random samples of the global population. 
Recent modeling approaches \citep{lin2010relative, lee2017communication, mcmahan2017communication, battey2018distributed, jordan2018communication, tang2020distributed,wang2021tributarypca} all assume samples are randomly partitioned, which guarantees %equal sample size and 
identical data distribution across sites. 
The goal of these works is to improve overall prediction by averaging results from homogeneous sample divisions. 
The classic random effects meta-analysis (see, e.g.,  \citet{whitehead2002meta,sutton2000methods,borenstein2011introduction} describes heterogeneity using modeling assumptions, but its focus mostly is still on global patterns.

\subsection{Heterogeneous Models} 
In practice, however, there is often too much global heterogeneity 
in a distributed data network to warrant direct aggregation of models obtained from local sites. The focus shifts to improving the estimation of a target site by selectively leveraging information from other data sources. 
There are two main classes of approaches.
The first class is based on comparison of the learned model parameters $\{\widehat\btheta_1,\dots,\widehat\btheta_K\}$ from $K$ different sites where for site $k$ we adopt model $f_k(\bx) = f(\bx; \btheta_k)$ with subject features $\bx$ to approximate the outcome of interest $Y$. 
Clustering and shrinkage approaches are then used
by merging data or models that are similar \citep{ke2015homogeneity,smith2017federated,ma2017concave,wang2020sylvester,tang2020individualized}. 
Most of these require the pooling of subject-level data.
The second class of approaches falls in the \textit{model averaging} framework \citep{raftery1997bayesian} with weights directly associated with the local prediction. 
Let site 1 be our target site, and the goal is to improve $f_{1}$ using a weighted estimator $f^*(\bx) = \sum_{k=1}^K {\omega}_{k} f_k(\bx)$ with weights $\omega_k$ to balance the contribution of each model and $\sum_k \omega_{k} = 1$. 
It provides an immediate interpretation of usefulness of each data source. 
When the weights are proportional to the prediction performance of $f_k$ on site 1, for example, 
$${\omega}_{k} =\frac{\exp\{- \sum_{i \in \mathcal{I}_1}(f_k(\bx_i) - y_i)^2\}}{ \sum_{\ell=1}^{K} \exp\{- \sum_{i \in \mathcal{I}_1}(f_\ell(\bx_i) - y_i)^2\} },$$
with $y_i$ being the observed outcome of subject $i$ in site 1, indexed by $\mathcal{I}_1$, the method is termed as the exponential weighted model averaging (EWMA).
Several variations of ${\omega}_{k}$ can be found in
\citet{yang2001adaptive, dai2011greedy,yao2018using,dai2018bayesian}. 
In general, separate samples are used to obtain the estimates of $\omega_k$'s and $f_k$'s, respectively. 
% However, existing methods only concern prediction. 

Here we focus on the literature review of model averaging. 
We note that our framework is also related to federated learning \citep{mcmahan2017communication}. But the latter often
involves iterative updating  rather than a one-shot procedure, and could be hard to apply to nonautomated distributed
research networks. Besides, it has been developed mainly to estimate a global prediction model by leveraging distributed data, and is not designed to target any specific site. We further discuss these approaches and other related research topics and their distinctions with model averaging in Appendix~\ref{suppl-related}.

\subsection{Transportability}
In causal inference, there is a lot of interest in identifying subgroups with enhanced treatment effects, targeting  at  the  feasibility  of  customizing  estimates  for  individuals \citep{athey2016recursive,wager2018estimation,hahn2020bayesian,kunzel2019metalearners,nie2020quasioracle}. These methods aim to estimate the CATE function $\tau(\bx)$, denoting the difference in potential outcomes between treatment and control, conditional on subject characteristics $\bx$. 
To reduce uncertainty in estimation of personalized treatment effects, incorporating additional data or models are sought after. 
\citet{pearl2011transportability,pearl2014external,bareinboim2016causal} introduced the notion of transportability to warrant causal inference models be generalized to a new population. 
The issue of generalizability is common in practice due to the non-representative sampling of participants in randomized controlled trials \citep{cook2002experimental,druckman2011cambridge,allcott2015site,stuart2015assessing,egami2020elements}. 
Progress on bridging the findings from an experimental study with observational data can be found in, e.g., \citet{stuart2015assessing,kern2016assessing,stuart2018generalizability,ackerman2019implementing,yang2020elastic,harton2021combining}. 
See \citet{tipton2018review,colnet2020causal,degtiar2021review} and references therein for a comprehensive review. 
However, most methods require fully centralized data. In contrast, we leverage the distributed nature of model averaging to derive an integrative CATE estimator.

\section{A Tree-based Model Averaging Framework} \label{sec:method}

We first formally define the {conditional average treatment effect} (CATE). 
Let $Y$ denote the outcome of interest, $Z \in \{0,1\}$ denote a binary treatment indicator, and $\bX$ denote subject features. Correspondingly, let $y$, $z$ and $\bx$ denote their realizations. 
Using the potential outcome framework \citep{neyman1923applications,rubin1974estimating}, we define
CATE as $\tau(\bx)=E[Y^{(Z=1)}-Y^{(Z=0)} |\bX=\bx],$
where $Y^{(Z=1)}$ and $Y^{(Z=0)}$ are the potential outcomes under treatment arms $Z=1$ and $Z=0$, respectively. 
The expected difference of the potential outcomes is dependent on subject features $\bX$. By the causal consistency assumption, the observed outcome is $Y = ZY^{(Z=1)} + (1 - Z)Y^{(Z=0)}$.

Now suppose the distributed data network 
consists of $K$ sites, each with sample size of $n_k$. 
Site $k$ contains data $\mathcal{D}_k = \{y_i, z_i, \bx_i\}_{i \in \mathcal{I}_k}$, where $\mathcal{I}_k$ denotes its index set. 
Its CATE function is given by $\tau_k (\bx) = E_k[Y^{(Z=1)}-Y^{(Z=0)} |\bX=\bx],$ where the expectation is taken over the data distribution in site $k$. 
Without loss of generality, we assume the goal is to estimate the CATE function in site 1, $\tau_1$.

\subsection{Causal Assumptions} \label{sec:assumption}

To ensure information can be properly borrowed across sites, we first impose the following idealistic assumptions, and then present relaxed version of Assumption~\ref{assump:transportability}. 
Let $S$ be the site indicator taking values in $\mathcal{S} = \{1, \dots, K\}$ such that $S_i = k$ if $i \in \mathcal{I}_k$. 

\begin{assumption}[Unconfoundedness] \label{assump:unconfounded}
$$\{Y^{(Z=0)}, Y^{(Z=1)}\} \perp Z | \bX, S;$$
\end{assumption}
\begin{assumption}[Transportability] \label{assump:transportability}
$$\{Y^{(Z=0)}, Y^{(Z=1)}\} \perp S | \bX;$$
\end{assumption}
\begin{assumption}[Positivity] \label{assump:positivity}
$$	0< P(S = 1| \bX) < 1 \mbox{ and } 0< P(Z = 1| \bX, S) < 1  \quad \mbox{for all } \bX \mbox{ and } S.$$
\end{assumption}

Assumption~\ref{assump:unconfounded} 
ensures treatment effects are unconfounded within sites so that $\tau_k(\bx)$ can be consistently identified. 
It holds by design when data are randomized controlled trials or when treatment assignment depends on $\bX$. 
By this assumption, we have $\tau_{k}(\bx) = 
E[Y |\bX=\bx, S = k, Z = 1] - E[Y |\bX=\bx, S = k, Z = 0]$.
The equality directly results from the assumption. 
Assumption~\ref{assump:transportability}
essentially states that the CATE functions are transportable, i.e., $\tau_k(\bx) = \tau_{k'}(\bx)$ for $k, k'\in \{1,\dots, K\}$. See also \citet{stuart2011use}, \citet{buchanan2018generalizing} and \citet{yang2020elastic} for similar consideration.
This assumption may not be satisfied due to heterogeneity across sites. In other words, site can be a confounder which prevents transporting of CATE functions across sites.
Our method allows Assumption~\ref{assump:transportability} 
to be violated and use model averaging weights to determine transportability. Explicitly, we consider a relaxed Assumption~\ref{assump:partial-trans} to hold for a subset of sites that contains site 1.

\begin{assumption}[Partial Transportability] \label{assump:partial-trans}
$$\{Y^{(Z=0)}, Y^{(Z=1)}\} \perp S_1 | \bX.$$
\end{assumption}
Here, $S_1$ takes values in $\mathcal{S}_1 = \{k: \tau_k(\bx) = \tau_1(\bx)\}$ and $\{1\} \subset \mathcal{S}_1 \subset \mathcal{S}$. 
We denote $\mathcal{S}_1$ as the set of transportable sites with regard to site 1. 
Hence, transportability holds across some sites and specific subjects.
In a special case in Section~\ref{sec:simulation} where $\mathcal{S}_1 = \{1\}$, bias may be introduced to by model averaging. However, our approach is still able to exploits the bias and variance trade off to improve estimation. 
Assumption~\ref{assump:positivity} 
ensures that all subjects are possible to be observed in site 1 and all subjects in all sites are possible to receive either arm of treatment. 
The former ensures a balance of covariates between site 1 population and the population of other sites. 
Violation of either one may result in extrapolation and introduce unwanted bias to the ensemble estimates for site 1. This assumption is also used, e.g., in \citet{stuart2011use}.

\subsection{Model Ensemble} \label{sec:adaptiveMA}

We consider an adaptive weighting of $\{\tau_1, \dots, \tau_K\}$ by
\begin{equation} \label{eq:agg}
   {\tau}^*(\bx) = \sum_{k=1}^K \omega_{k}(\bx) \tau_k(\bx)
\end{equation}
where ${\tau}^*$ is the weighted model averaging estimator. The weight functions $\omega_{k}(\bx)$'s are not only site-specific, but also depend on $\bx$, and follow $\sum_{k=1}^{K} \omega_{k}(\bx) = 1$.
It measures the importance of $\tau_k$ in assisting site 1 when subjects with characteristics $\bx$ are of interest. 
We rely on each of the sites to derive their respective $\widehat\tau_k$ from $\mathcal{D}_k$ so that $\mathcal{D}_1, \dots, \mathcal{D}_K$ do not need to be pooled. Only the estimated functions $\{\widehat \tau_2, \dots, \widehat \tau_K\}$ are passed to site 1. 
We will describe the approaches to estimate $\widehat{\tau}_k$ in Section \ref{sec:local}.

A two-stage model averaging approach is proposed.
We first split $\mathcal{D}_1$, the data in the target site, into a training set and an estimation set indexed by $\{i \in \mathcal{I}_1^{(1)}\}$ and $\{i \in \mathcal{I}_1^{(2)}\}$, respectively.
\emph{1) Local stage:} 
Obtain $\widehat{\tau}_1$ from subjects in $\mathcal{I}_1^{(1)}$. 
Obtain $\widehat{\tau}_k$ from local subjects in $\mathcal{I}_k$, $k = 2, \dots, K$. These $\{\widehat{\tau}_k\}_{k=1}^K$ are then passed to site 1 to get $K$ predicted treatment effects for each subject in $\mathcal{I}_1^{(2)}$, resulting in an augmented data set as shown in Figure~\ref{fig:framework}(b).
\emph{2) Ensemble stage:} A tree-based ensemble model is trained on the augmented data by either an ensemble tree (ET) or an ensemble random forest (EF), with the predicted treatment effects from the previous stage, i.e., $\widehat{\tau}_k(\bx_i)$ as the \emph{outcome}. The site indicator $S$ of which local model is used as well as the subject features $\bx_i$ are fed into the ensemble model as \emph{predictors}. The resulting model will be used to compute our proposed model averaging estimator.
Figure~\ref{fig:framework}(a) illustrates a conceptual diagram of the proposed model averaging framework and structure of the augmented data. 
Note the idea of data augmentation has been used in, e.g., computer vision \citep{wang2017effectiveness,mo2020towards,mo2021point}, statistical computing \citep{van2001art}, and imbalanced classification \citep{chawla2002smote}. 
Here the technique is being used to construct weights for model averaging, which will be discussed in the following paragraph. 
Algorithm~\ref{algo:code} provides an algorithmic overview. 
Our method has been implemented as an R package \texttt{ifedtree} available on GitHub (\url{https://github.com/ellenxtan/ifedtree}).

\begin{figure}[!htb]%[tb]%tp]
\centering
 \begin{subfigure}{0.8\textwidth}
  \centerline{\includegraphics[width=\linewidth]{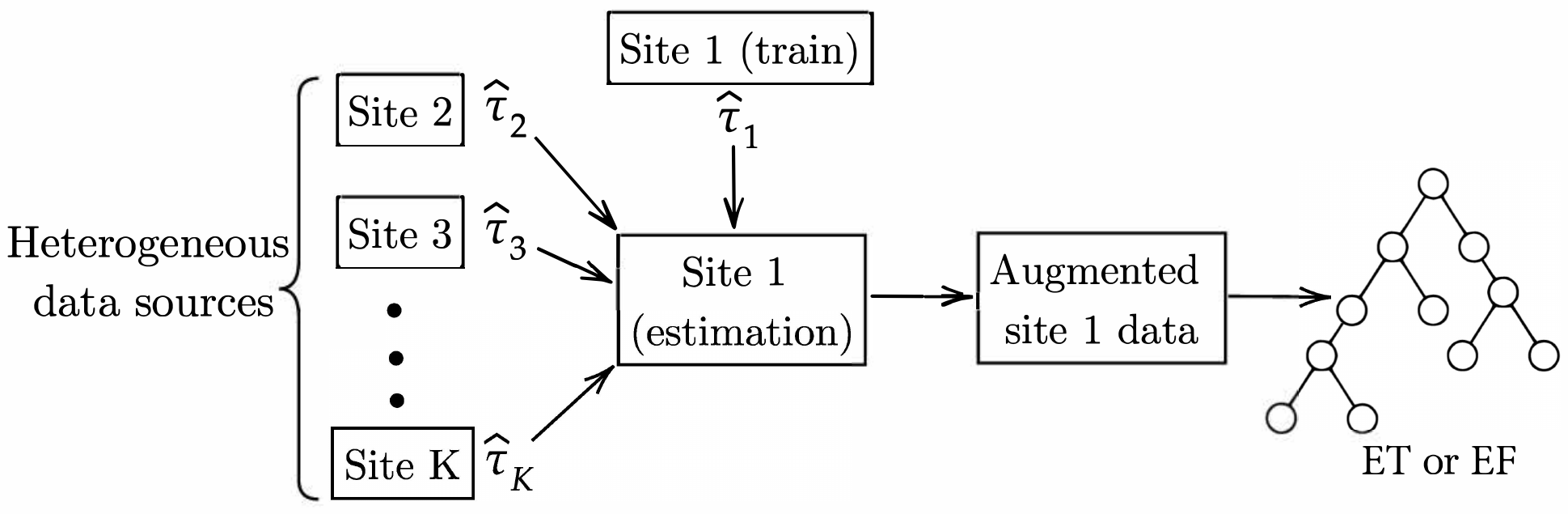}}
%   \vspace{-0.3cm}
  \caption{}
%   \vspace{0.2pc}
 \end{subfigure}\\
 \begin{subfigure}{0.4\textwidth}
  \centerline{\includegraphics[width=\linewidth]{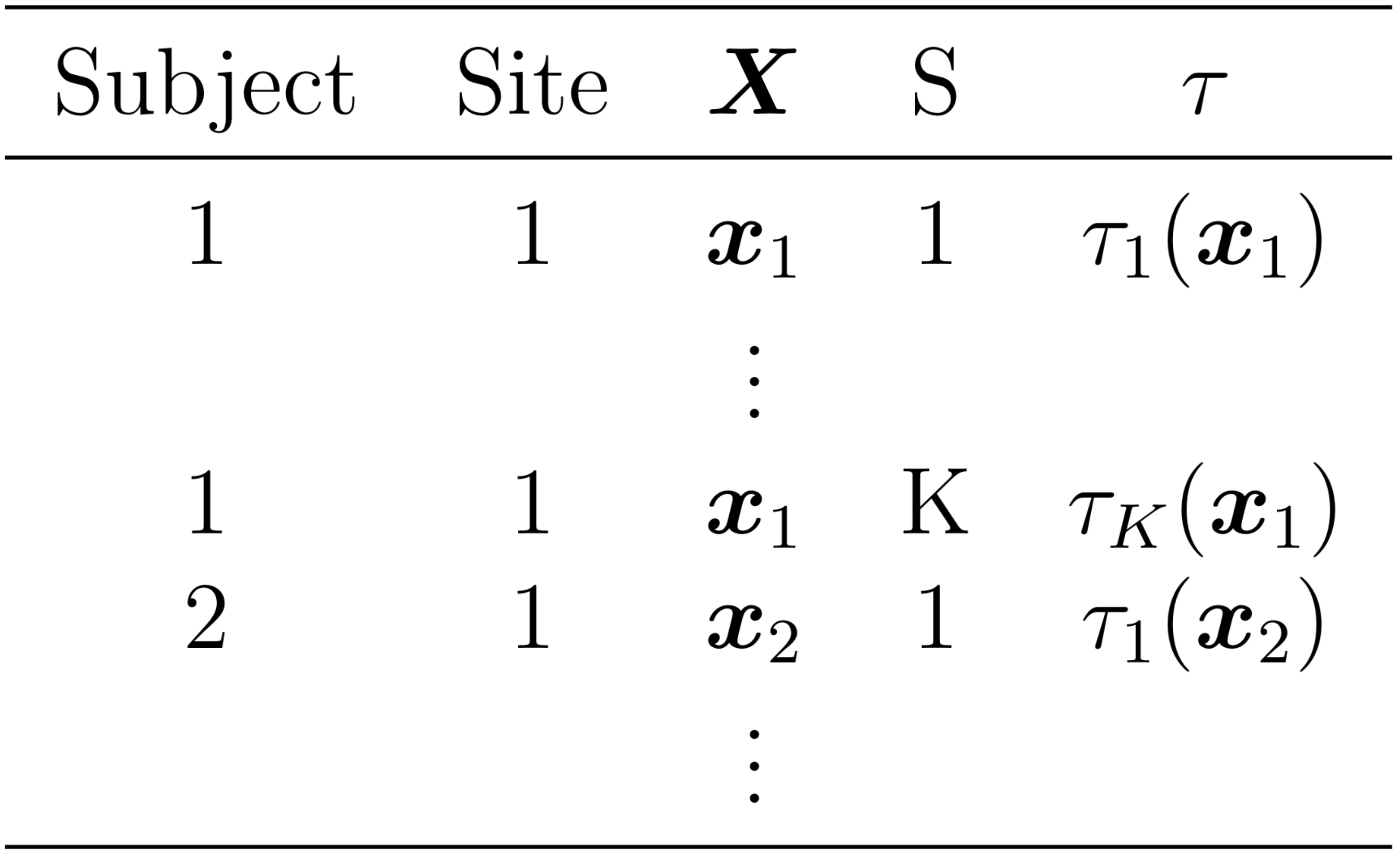}}
  \caption{}
 \end{subfigure}
%  \vspace{-0.3cm}
 \caption{(a) Schema of the proposed algorithm. (b) Illustration of the augmented data constructed from the estimation set of site 1.} 
 \label{fig:framework}
\end{figure}

\begin{algorithm}[tb]
% \small%\scriptsize%
  \caption{Tree-based model averaging for heterogeneous data sources}
  \label{algo:code}
\begin{algorithmic}%[1]
    \For{$k=1$ {\bfseries to} $K$} \Comment{Loop through $K$ sites. Can be run in parallel.}
    \State Build a local model using site $k$ data. Site 1 model uses its training set only.
    \EndFor
    \For{$i \in \mathcal{I}_1^{(2)}$}  \Comment{Loop through subjects in site 1 estimation set.}
        \For{$k=1$ {\bfseries to} $K$}  \Comment{Loop through $K$ local models.}
            \State Predict $\widehat{\tau}_k(\bx_i)$ using local model $k$.
            \State $D_{i, k} = [\bx_i, k, \widehat{\tau}_k(\bx_i)]$.
        \EndFor
    \EndFor
    \State Create augmented site 1 data $\mathfrak{D}_{aug, 1}$ by concatenating $D_{i, k}$ vectors.
    % \State $\widehat{\tau}_{ET} (\bx, s) \gets$ \Call{EnsembleTree}{$\mathfrak{D}_{aug, 1}$}
    \State $\widehat{\Tau}_{\text{EF}} (\bx, s) =$ {\sc EnsembleForest}($\mathfrak{D}_{aug, 1}$)
    % \Call{EnsembleForest}{$\mathfrak{D}_{aug, 1}$}
    \Comment{Or {\sc EnsembleTree} when $B = 1$.}
\end{algorithmic}
\end{algorithm}

\subsection{Construction of Weights} 
A tree-based ensemble is constructed to estimate the weighting functions $\{\omega_{k}\}_{k=1}^K$.
Heterogeneity across sites is explained by including the site index into an augmented training set when building trees.
An intuition of our approach is that sites that are split away from site 1 (by tree nodes) are ignored and the sites that fall into the same leaf node are considered homogeneous to site 1 hence contribute to the estimation of $\tau_1(\bx)$. A splitting by site may occur in any branches of a tree, resulting in an information sharing scheme across sites that is dependent on $\bx$. 
We construct the ensemble by first creating an augmented data 
% \begin{equation*}
    $\mathfrak{D}_{aug, 1} = \{\bx_i, k, \widehat{\tau}_k(\bx_i)\}_{i \in \mathcal{I}_1^{(2)}, k \in \mathcal{S}}$,
% \end{equation*}
for subjects in $\mathcal{I}_1^{(2)}$.
The illustration of this augmented site 1 data is given in Figure~\ref{fig:framework}(b). An ensemble is then trained on this data by either a tree or a random forest, 
with the estimated treatment effects $\widehat{\tau}_k(\bx_i)$ as the outcome, and a categorical site indicator of which local model is used along with all subject-level features as predictors, i.e., $(\bx_i, k)$. 
We denote the resulting function as $\Tau(\bx, s)$ which depends on both $\bx$ and site $s$, specifically, 
$\Tau_{\text{ET}} (\bx, s)$ and $ \Tau_{\text{EF}} (\bx, s)$ for ensemble tree (ET) and ensemble forest (EF), respectively. Let $\mathcal{L}(\bx,s)$ denote the final partition of the feature space by the tree to which the pair $(\bx, s)$ belongs. The ET estimate based on the augmented site 1 data can be derived by
\begin{align}\label{eq:ET}
    \widehat \Tau_{\text{ET}} (\bx, s)
    &=  \{|\{(i, k): (\bx_i, k) \in \mathcal{L}(\bx, s)\}_{i \in \mathcal{I}_1^{(2)}, k \in \mathcal{S}}|\}^{-1} \sum_{\{(i, k): (\bx_i, k) \in \mathcal{L}(\bx, s)\}_{i \in \mathcal{I}_1^{(2)}, k \in \mathcal{S}}} \widehat{\tau}_k(\bx_i) \nonumber\\
    &=  \sum_{i\in\mathcal{I}_1^{(2)}} \sum_{k=1}^{K} \frac{\mathbbm{1}\{(\bx_i, k)\in \mathcal{L}(\bx, s)\}}{|\mathcal{L}(\bx, s)|} \widehat{\tau}_k(\bx_i).
\end{align}
Intuitively, observations with similar characteristics ($\bx$ and $\bx'$) and from similar sites ($s$ and $s'$) are more likely to fall in the same partition region in the ensemble tree, i.e., $(\bx, s) \in \mathcal{L}(\bx', s')$ or $(\bx', s') \in \mathcal{L}(\bx, s)$. This resembles a \emph{non-smooth kernel} where weights are $1/|\mathcal{L}(\bx, s)|$ for observations that are within the neighborhood of $(\bx, s)$, and 0 otherwise.
The estimator borrows information from neighbors in the space of $\bX$ and $S$. 
The splits of the tree are based on minimizing in-sample MSE of $\widehat\tau$ within each leaf 
and pruned by cross-validation over choices of the complexity parameter. 
Since a single tree is prone to be unstable,
in practice, we use random forest to reduce variance and smooth the partitioning boundaries.
By aggregating $B$ ET estimates each based on a subsample of the augmented data, $\{\widehat{\Tau}^{(b)}\}_{b=1}^B$, an EF estimate can be constructed by

\begin{align}
       \widehat \Tau_{\text{EF}} (\bx, s) 
& = \frac{1}{B} \sum_{b=1}^{B} \widehat \Tau^{(b)} (\bx, s) \nonumber \\
& = \sum_{i\in \mathcal{I}_1^{(2)}} \sum_{k=1}^{K} \lambda_{i,k}(\bx, s) \widehat{\tau}_k(\bx_i), \label{eq:EF}\\
    \text{where~} &  \lambda_{i,k}(\bx, s) = \frac{1}{B} \sum_{b=1}^{B} \frac{\mathbbm{1}\{(\bx_i, k)\in \mathcal{L}_b(\bx, s)\}}{|\mathcal{L}_b(\bx, s)|}.\nonumber
\end{align}
The form of $\widehat \Tau^{(b)} (\bx, s)$ closely follows \eqref{eq:ET} but is based on a subsample of $\mathfrak{D}_{aug, 1}$. The weights, $\lambda_{i,k}(\bx, s)$, are similar to that in \eqref{eq:ET}, and can be viewed as kernel weighting that defines an adaptive neighborhood of $\bx$ and $s$.  
We then obtain the model averaging estimates defined in \eqref{eq:agg} by fixing $s = 1$ such that $\widehat \tau_{\text{ET}}^*(\bx) = \widehat\Tau_{\text{ET}} (\bx, s=1)$ or $\widehat \tau_{\text{EF}}^*(\bx) = \widehat\Tau_{\text{EF}} (\bx, s=1)$. 
The weight functions $\{\omega_{k}(\bx)\}_{k=1}^K$ for $\widehat\tau^*(\bx)$ can be immediately obtained from the ET or EF by
\begin{align*}
    \widehat\tau_{\text{ET}}^*(\bx) = \ &  \widehat \Tau_{\text{ET}} (\bx, 1) = \sum_{k=1}^K \widehat \omega_{k}(\bx) \widehat \tau_k(\bx),  \\
    \text{where~~} & \widehat \omega_{k}(\bx) = \sum_{i\in\mathcal{I}_1^{(2)}} \frac{\mathbbm{1}\{(\bx_i, k)\in \mathcal{L}(\bx, 1)\}}{|\mathcal{L}(\bx, 1)|}; \nonumber \\
%   \text{or }
   \widehat\tau_{\text{EF}}^*(\bx) = \ &  \widehat \Tau_{\text{EF}} (\bx, 1) = \sum_{k=1}^K \widehat \omega_{k}(\bx) \widehat \tau_k(\bx),  \\
    \text{where~~} & \widehat \omega_{k}(\bx) = \sum_{i\in\mathcal{I}_1^{(2)}} \lambda_{i,k}(\bx, 1). %\label{eq:weights}
\end{align*}
It can be verified that $\sum_{k=1}^{K} \widehat \omega_k(\bx) = 1$ for all $\bx$. 
As our simulations in Section \ref{sec:simulation} show, $\widehat \tau^*$ improves the local functional estimate $\widehat\tau_1$.
We set $B=2,000$ throughout the paper.
Tree and forest estimates are obtained by R packages \verb|rpart| and \verb|grf|, respectively.

\subsection{Interpretability of Weights}

The choice of tree-based models naturally results in such kernel weighting $w_k(\bx)$ \citep{athey2019generalized}, which are not accessible by other ensemble techniques.
Such explicit and interpretable weight functions could deliver meaningful rationales for data integration.
For example, under scenarios where there exists extreme global heterogeneity (as shown in Section~\ref{sec:simulation} when $c$ is large), $w_k(\bx)$ can be used as a diagnostic tool to decide which external data sources should be co-used. 
Weights close to $0$ inform against model transportability, and they are adaptive to subject-level features $\bx$ so that decisions can be made based on the subpopulations of interest.

\subsection{Local Models: Obtaining $\widehat\tau_k$} \label{sec:local}

Estimate of $\tau_k(\bx)$ at each local site must be obtained separately before the ensemble. 
Our proposed ensemble framework can be applied to a general estimator of $\tau_k(\bx)$. For each site, the local estimate could be obtained using different methods. 
Recently, there has been many work dedicated to the estimation of individualized treatment effects \citep{athey2016recursive,wager2018estimation,hahn2020bayesian,kunzel2019metalearners,nie2020quasioracle}. 
As an example, we consider using the causal tree (CT) \citep{athey2016recursive} to estimate the local model at each site. CT is a non-linear learner that \textit{(i)} allows different types of outcome such as discrete and continuous, and can be applied to a broad range of real data scenarios; \textit{(ii)} can manage hundreds of features and high order interactions by construction; \textit{(iii)} can be applied to both experimental studies and observational studies by propensity score weighting or doubly robust methods.  
CT is implemented in the R package \verb|causalTree|. We also explore another estimating option for local models in Appendix~\ref{suppl-sec:sim}.

\subsection{Asymptotic Properties} \label{sec:consist}

We provide consistency guarantee of the proposed estimator $\widehat \Tau_{\text{EF}}$ for the true target $\tau_1$. 
Assuming point-wise consistent local estimators are used for $\{\tau_k\}_{k=1}^{K}$, EF with subsampling procedure described in Appendix~\ref{suppl-sec:consist} is consistent. 
% \vspace{-0.2cm}
\begin{theorem}\label{Theorem}
Suppose the subsample used to build each tree in an ensemble forest is drawn from different subjects of the augmented data and the following conditions hold:

\begin{itemize}
  \item[(a)] Bounded covariates: Features $\bX_i$ and the site indicator $S_i$ are independent and have a density that is bounded away from 0 and infinity.
  \item[(b)] Lipschitz response: the conditional mean function $\mathbb{E}[ \Tau|\bX=\bx,S=1]$ is Lipschitz-continuous.
  \item[(c)] Honest trees: trees in the random forest use different data for placing splits and estimating leaf-wise responses.
\end{itemize} 

Then $\widehat \Tau_{\text{EF}}(\bx,1) \overset{p}{\to} \tau_{1}(\bx)$, for all $\bx$, as $\min_k n_k \to \infty$. Hence, $\widehat \tau_{\text{EF}}^*(\bx) \overset{p}{\to} \tau_1(\bx).$
\end{theorem}
% \vspace{-0.2cm}
% Under some mild conditions on the forest-growing scheme, 
The conditions and a proof of Theorem \ref{Theorem} is given in Appendix~\ref{suppl-sec:consist}.
To demonstrate the consistency properties of our methods, we add in Appendix~\ref{suppl-sec:sim} oracle versions of ET and EF estimators, denoted as ET-oracle and EF-oracle, which use the ground truth of local models $\{\tau_k\}_{k=1}^K$ in estimating $\{\widehat \omega_k\}_{k=1}^K$.
This removes the uncertainty in local models. 
The remaining uncertainty only results from the estimation of the ensemble weights, and we see both oracle estimators achieve minimal MSE.
Section~\ref{sec:simulation} gives a detailed evaluation of the finite sample performance.

\section{Simulation Studies}
\label{sec:simulation}

Monte Carlo simulations are conducted to assess the proposed methods. We specify $m(\bx, k)$ as the conditional outcome surface and $\tau(\bx, k)$ as the conditional treatment effect for individuals  with features $\bx$ in site $k$. The treatment propensity is specified as $e(\bx) = \Pr(Z=1|\bX=\bx)$. 
The potential outcomes can be written as $Y_i = m(\bX_i, S_i) + \{ Z_i - e(\bX_i) \} \tau(\bX_i, S_i) + \epsilon_i$, following notations in \citet{robinson1988root,athey2016recursive,wager2018estimation,nie2020quasioracle}. 
The mean function is $m(\bx, k) = \frac{1}{2} x_1 + \sum_{d=2}^4 x_d +(x_1 - 3) \cdot c \cdot U_k$, and the treatment effect function is specified as 
$$ \tau(\bx, k) = \mathbbm{1}\{x_1 > 0\} \cdot x_1 +(x_1 - 3) \cdot c \cdot U_k,$$
where $z=0,1$, $U_k$ denotes the global heterogeneity due to site-level confounding, controlled by a scaling factor $c$, and $\epsilon_i \sim N(0, 1)$. 
Features follow $\bX_i \sim {N}(\boldsymbol{0},\boldsymbol{I}_D)$, where $D=5$, and are independent of $\epsilon_i$.
The simulation setting within each site (with $k$ fixed) is motivated by designs in \citet{athey2016recursive}. Features in $\tau$ are determinants of treatment effect while those in $m$ but not in $\tau$ are prognostic only. 
%Features that do not affect outcomes are noise covariates. 
% and hold $m(\bX_i, S_i)=0$ fixed
The data are generated under a distributed data networks. 
We assume there are $K=20$ sites in total, each with a sample size $n=500$. 
In our main exposition, we consider an experimental study design where treatment propensity is $e(\bx) = 0.5$, i.e., individuals are randomly assigned to treatment and control.
Variations of the settings above are discussed, with results presented in Appendix~\ref{suppl-sec:sim}.

Two types for global heterogeneity are considered by the choice of $U_k$. 
% discrete
For \textbf{\emph{discrete grouping}}, we assume there are two underlying groups among the $K$ sites $U_k \sim Bernoulli(0.5)$. Specifically, we assume odd-index sites and even-index sites form two distinct groups $\mathcal{G}_1 = 
%\{k: k \bmod 2 = 1\} =  
\{1, 3, \dots, K-1\}$; $\mathcal{G}_2 = 
%\{k: k \bmod 2 = 0\} =  
\{2, 4\dots, K\}$ such that $U_{k\in\mathcal{G}_1} = 0$ and $U_{k\in\mathcal{G}_2} = 1$. 
% continuous
Sites from similar underlying groupings have similar treatment effects and mean effects, while sites from different underlying groupings have different treatment effects and mean effects. 
For \textbf{\emph{continuous grouping}}, we consider $U_k \sim Unif[0,1]$. 
We vary the scales of the global heterogeneity under the discrete and continuous cases, respectively, with $c$ taking values $c \in \{0,0.6,1,2\}$. A $c=0$ implies all data sources are homogeneous. 
In other words, Assumption~\ref{assump:transportability} is satisfied when $c = 0$ but not when $c > 0$.

\subsection{Compared Estimators and Evaluation} 
The proposed approaches ET and EF are compared with several competing methods. 
% LOC
\textbf{LOC:} A local CT estimator
that does not utilize external information. 
It is trained on $\mathcal{I}_1$ only, combining training and estimation sets.
% MA
\textbf{MA:} A naive model averaging method
with weights ${\omega}_{k}^{\text{MA}} = 1/k$. This approach assumes models are homogeneous. 
% EWMA
\textbf{EWMA:}
We consider a modified version of EWMA that can be used for CATE.  
We obtain an approximation of $\tau_1(\bx)$ by fitting another local model using the estimation set of site 1, denoted by $\widetilde\tau_1(\bx)$. 
Its weights are given by $${\omega}_{k}^{\text{EWMA}} = \frac{\exp\{- \sum_{i \in \mathcal{I}_1^{(2)}}(\widehat\tau_k(\bx_i) - \widetilde\tau_1(\bx_i))^2\} }{ \sum_{\ell=1}^{K} \exp\{- \sum_{i \in \mathcal{I}_1^{(2)}}(\widehat\tau_\ell(\bx_i) - \widetilde\tau_1(\bx_i))^2\} }.$$ 
% STACK
\textbf{STACK:} A stacking ensemble, which is a linear ensemble of predictions of several models \citep{breiman1996stacked}. 
To our end, we regress $\widetilde\tau_1(\bx)$ on the predictions of the estimation set in site 1 from each local model, $\{\widehat\tau_1(\bx), \dots, \widehat\tau_k(\bx) \}$. The stacking weights are not probabilistic hence not directly interpretable. 
We report the empirical mean squared error (MSE) of these methods over an independent testing set of sample size $n_{te}=2000$ from site 1. 
$\mbox{MSE}(\widehat{\tau}) = n_{te}^{-1}\sum_{i=1}^{n_{te}} \{ \widehat{\tau}(\bx_i) - \tau_1(\bx_i) \}^2. $
Each simulation scenario is repeated for 1000 times.

\clearpage
\begin{figure*}[!htb]%[tb]%[htp]
\centering
% \vspace{1cm}
 \begin{subfigure}{0.49\textwidth}
  \centerline{\includegraphics[width=\linewidth]{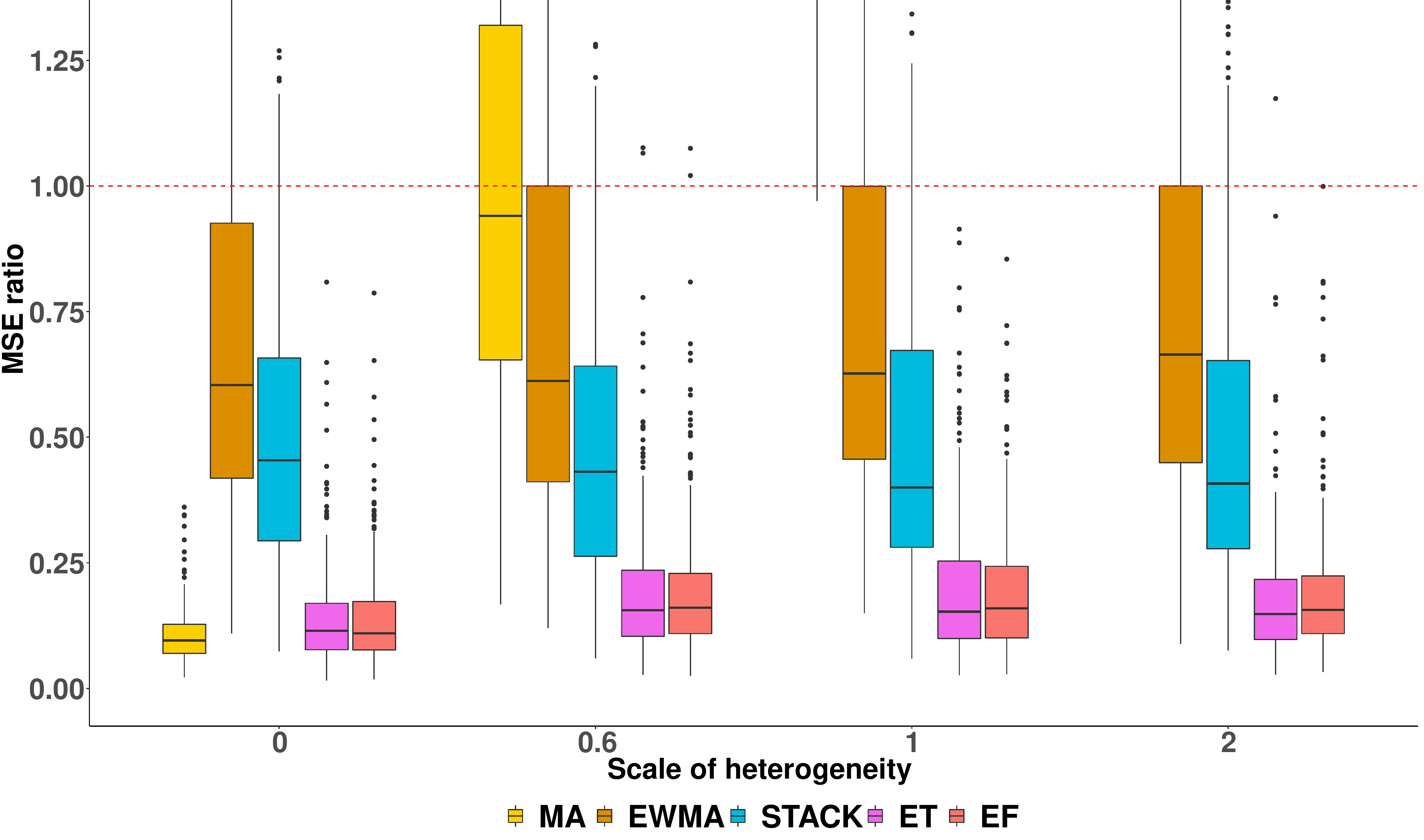}}
  \caption{}
 \end{subfigure}
 \begin{subfigure}{0.49\textwidth}
  \centerline{\includegraphics[width=\linewidth]{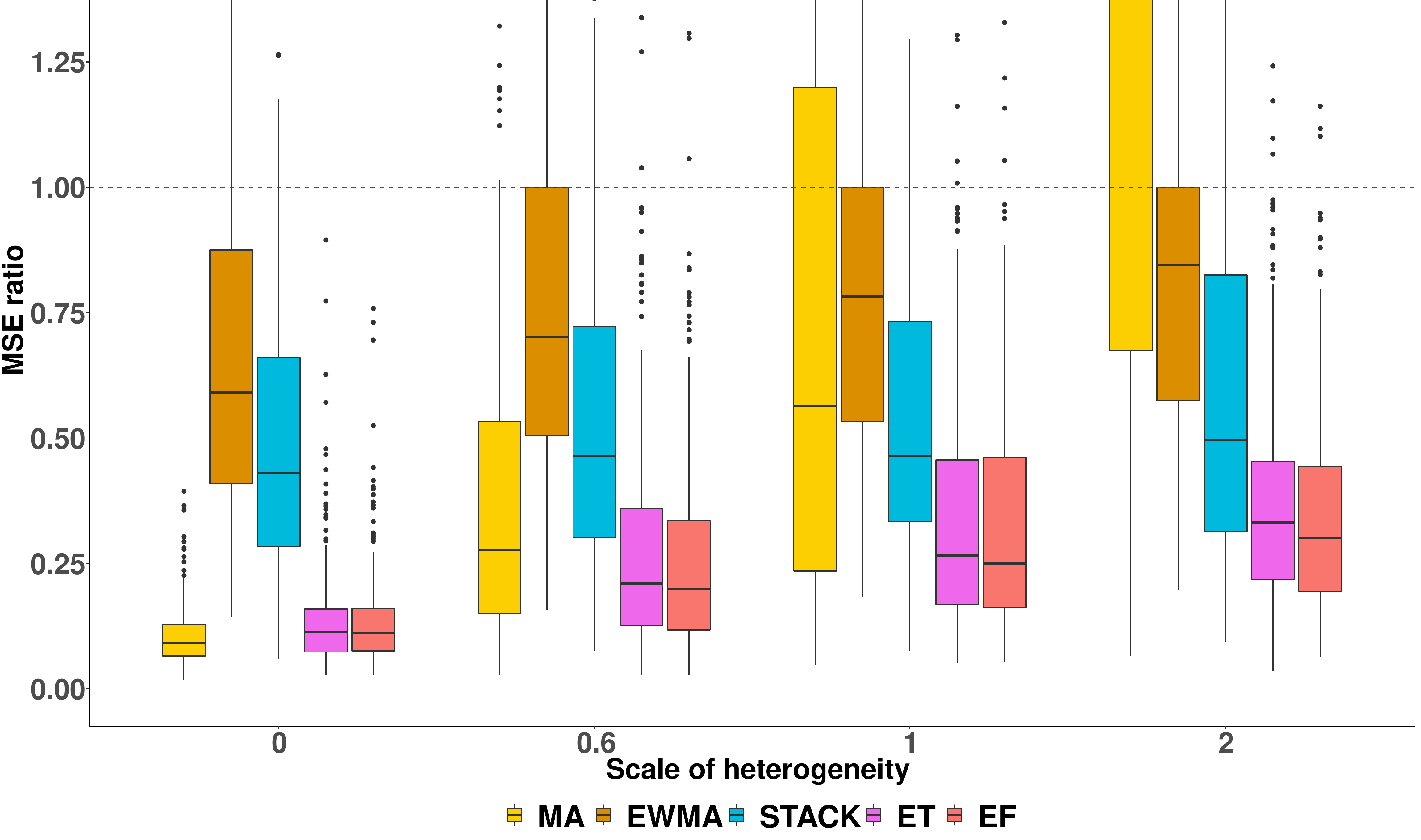}}
  \caption{}
 \end{subfigure}
%  \vspace{-1cm}
 \caption{
 Box plots of MSE ratios of CATE estimators, respectively, over LOC, for \textbf{(a) discrete grouping} and \textbf{(b) continuous grouping} across site. 
  Different colors imply different estimators, and x-axis, i.e., the value of $c$, differentiates the scale of global heterogeneity. The red dotted line denotes an MSE ratio of 1. 
  MA performance is truncated due to large MSE ratios. 
The proposed ET and EF achieve smaller MSE ratios compared to standard model averaging or ensemble methods and are robust to heterogeneity across settings. 
}
\label{fig:sim_box}
\end{figure*}

\begin{table}[!htb]%[tb]
\centering
%\footnotesize
\caption{MSE ratios of EF over LOC. As $n$ increases, model averaging becomes more powerful due to better estimation of $\tau_k$, and is more pronounced when $c$ is small.}
\label{fig:ratio}
\resizebox{0.6\columnwidth}{!}{
\begin{tabular}{@{}cccccc@{}}
\toprule
\multicolumn{1}{l}{}  &  & $c=0$  & $c=0.6$ & $c=1$  & $c=2$  \\ \midrule
\multirow{3}{*}{\shortstack{Discrete\\ grouping}} & $n=100$  & 0.57 & 0.59  & 0.61 & 0.59 \\
& $n=500$  & 0.12 & 0.17  & 0.17 & 0.16 \\
& $n=1000$ & 0.07 & 0.12  & 0.12 & 0.13 \\
\multirow{3}{*}{\shortstack{Continuous\\ grouping}} & $n=100$  & 0.54 & 0.59  & 0.63 & 0.69 \\
& $n=500$  & 0.11 & 0.24  & 0.31 & 0.34 \\
& $n=1000$ & 0.08 & 0.17  & 0.21 & 0.26 \\ \bottomrule
\end{tabular}
}
% \vskip -0.2in
% \vspace{-0.1cm}
\end{table}

\subsection{Estimation Performance} 
Figure~\ref{fig:sim_box} shows the performance of the proposed estimators and the competing estimators, using LOC as the benchmark. 
The proposed ET and EF show the best performance in terms of the mean and variation of MSE among other estimators when $c > 0$, and comparable to equal weighting MA when $c = 0$.
Although, a forest is more stable than a tree in practice,
both ET and EF give similar results because the true model is relatively simple and can be accurately estimated by a single ensemble tree under the given sample size.

Although asymptotically consistency, under finite sample, bias exists in local models and leads to biased model averaging estimates. 
While explicit quantification of bias and variance remains challenging due to extra uncertainty carried forward from the local estimates, we demonstrated that the proposed estimators can improve upon the local models under small sample size via Table~\ref{fig:ratio}.
It shows the MSE ratio of EF over LOC as a measure of gain resulting from model averaging by varying $n =100, 500, 1000$. 
The decrease in MSE ratio as $n$ increases, regardless of the choice of $c$, is consistent with our asymptotic results in Theorem~\ref{Theorem}.
This is due to a bias-and-variance trade-off in the ensemble that ensures a small MSE, which remains smaller than that in LOC despite varying $n$.
It also shows our method is robust to the existence of local uncertainty.

\subsection{Visualization of Information Borrowing} Figure~\ref{fig:sim_visual} visualizes the proposed ET and EF.  
In (a) and (d), the site indicator and $X_1$ appear as splitting variables in the ETs, which is consistent with the data generation process. 
The estimated treatment effect (b) and (e) reveals the pattern of transportability across sites and with respect to $X_1$. 
Panels (c) and (f) plot the model averaging weights in EFs over $X_1$. Site 1 has a relatively large contribution to the weighted estimator while models from other sites have different contributions at different values of $X_1$ depending on their similarity in $\tau(\bx,k)$ to that in site 1. 
Corresponding ET and EF show consistent patterns.

\clearpage
\begin{figure*}[!htb]%[tb]
  \begin{subfigure}{0.3\textwidth}
    \centering
    \includegraphics[width=\linewidth]{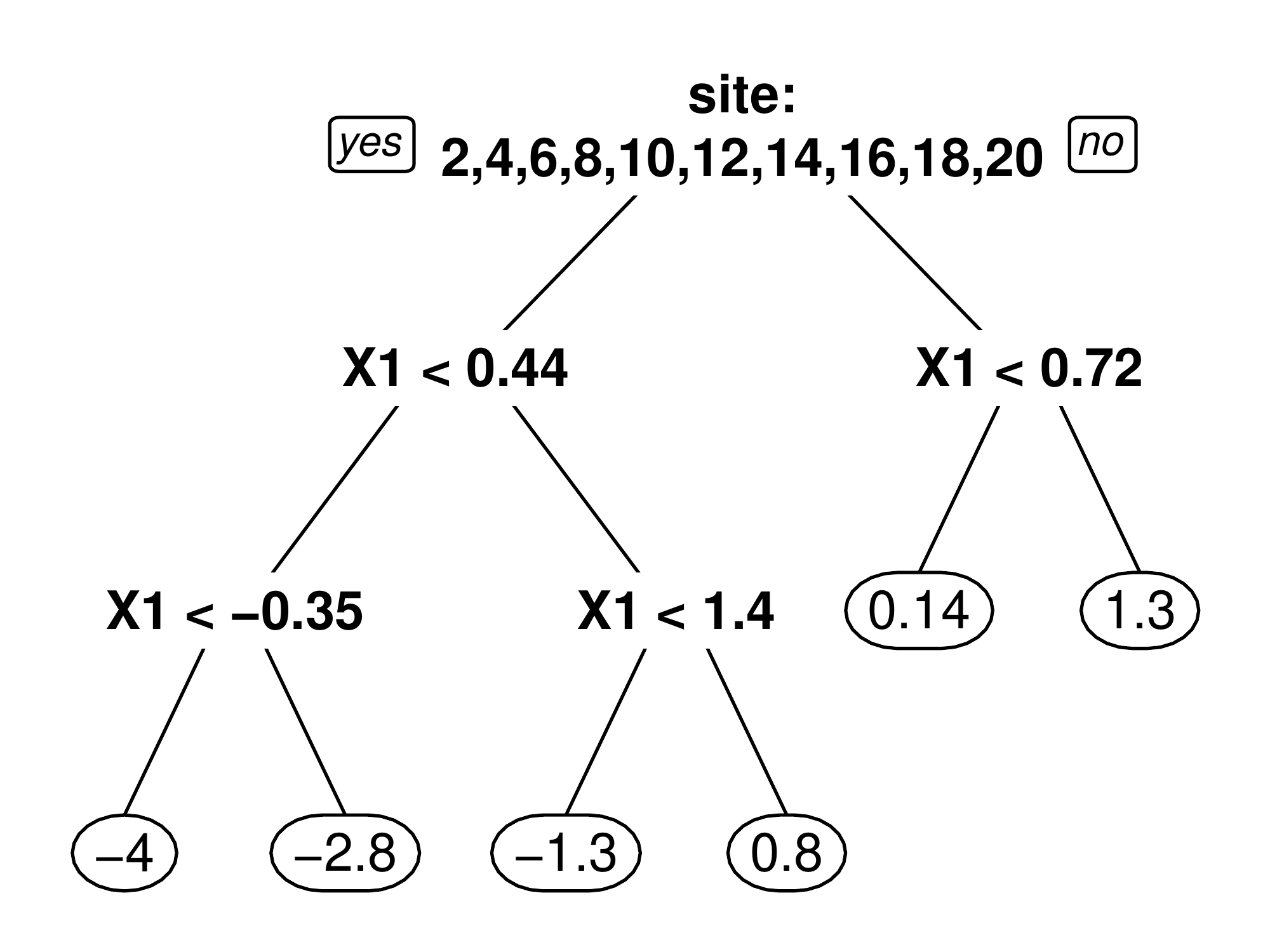}
    \caption{}
  \end{subfigure}
  \begin{subfigure}{0.325\textwidth}
    \centering
    \includegraphics[width=\linewidth]{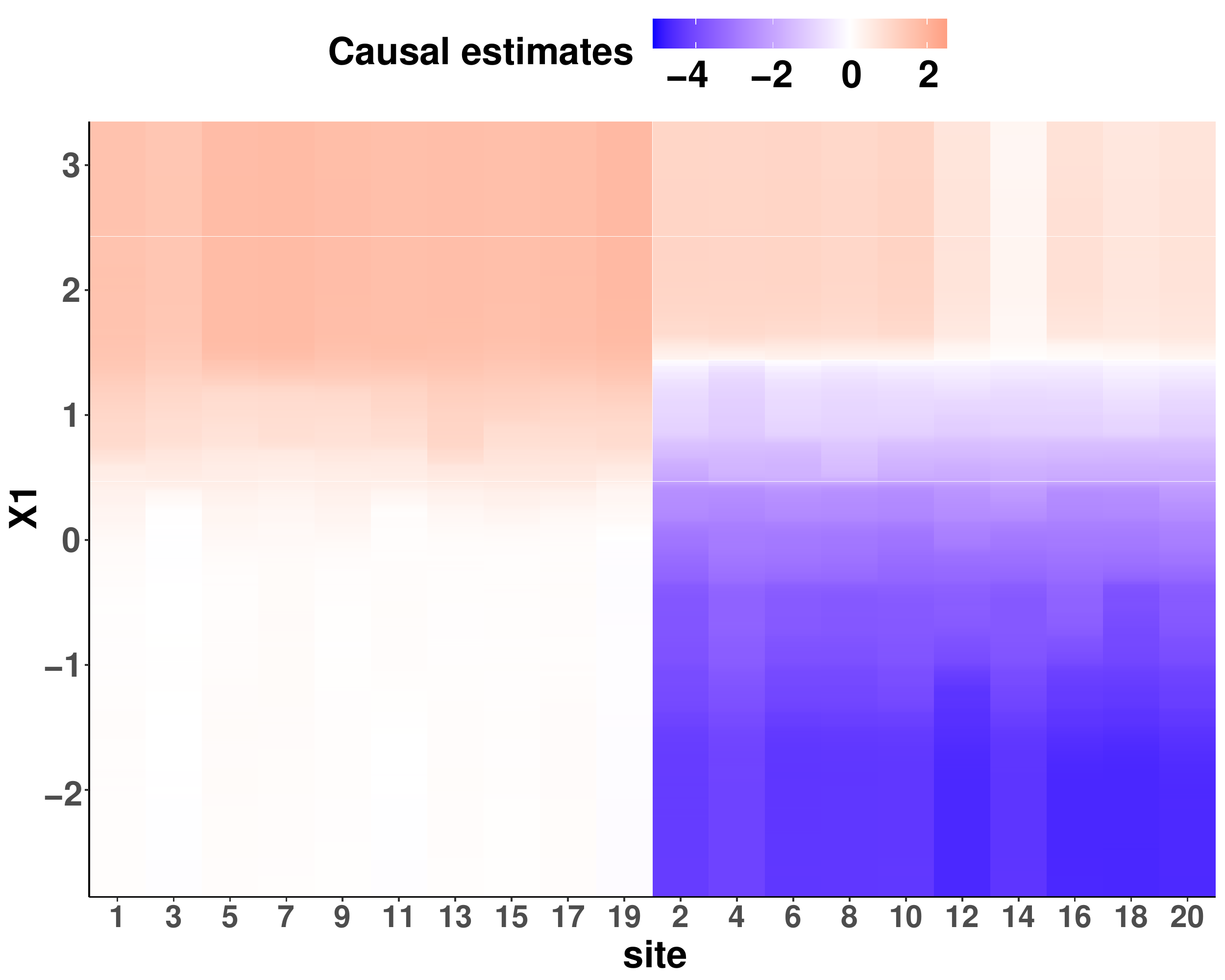}
    \caption{}
  \end{subfigure}
  \begin{subfigure}{0.37\textwidth}
    \centering
    \includegraphics[width=\linewidth]{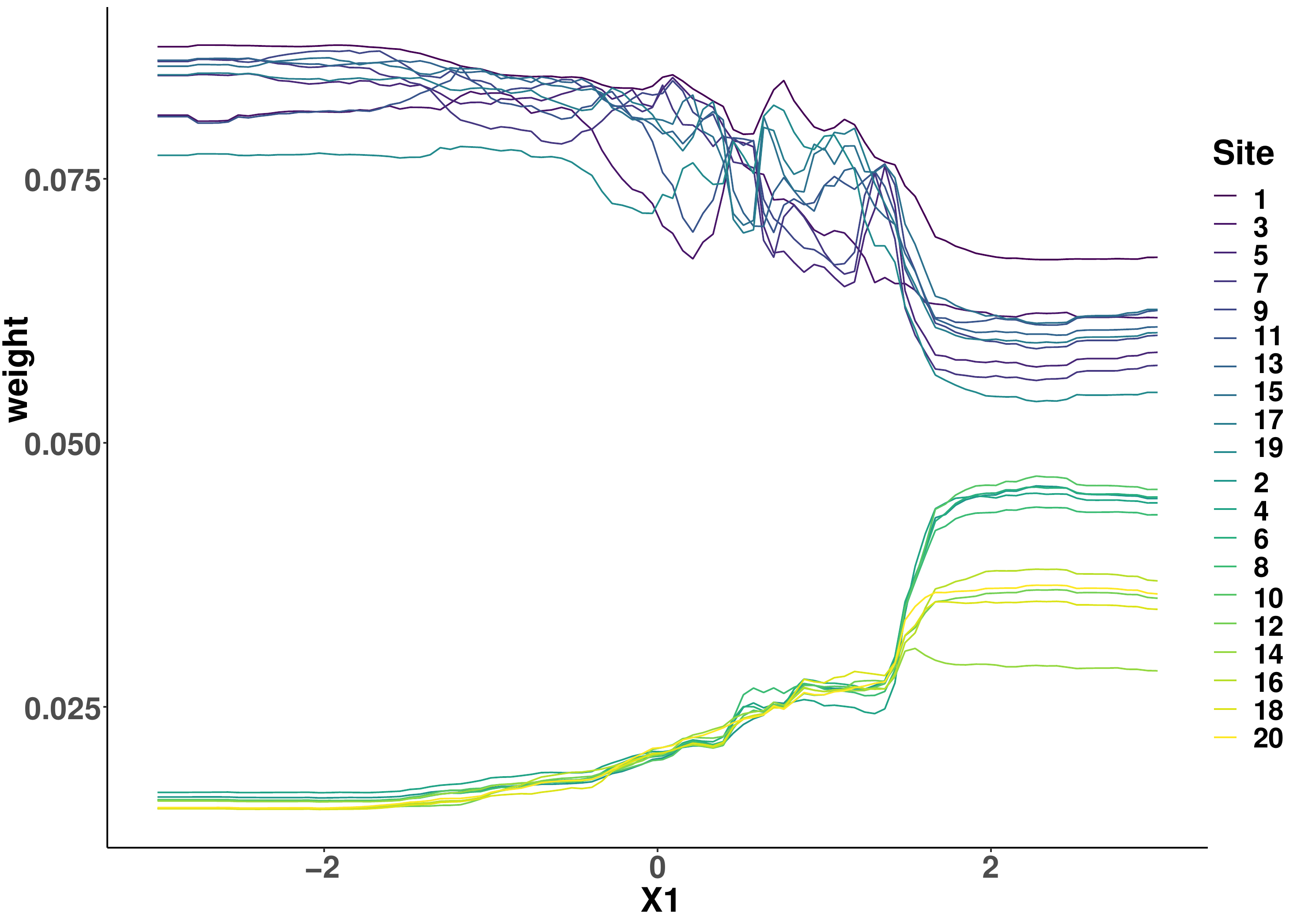}
    \caption{}
  \end{subfigure}
  \medskip

  \begin{subfigure}{0.3\textwidth}
    \centering
    \includegraphics[width=\linewidth]{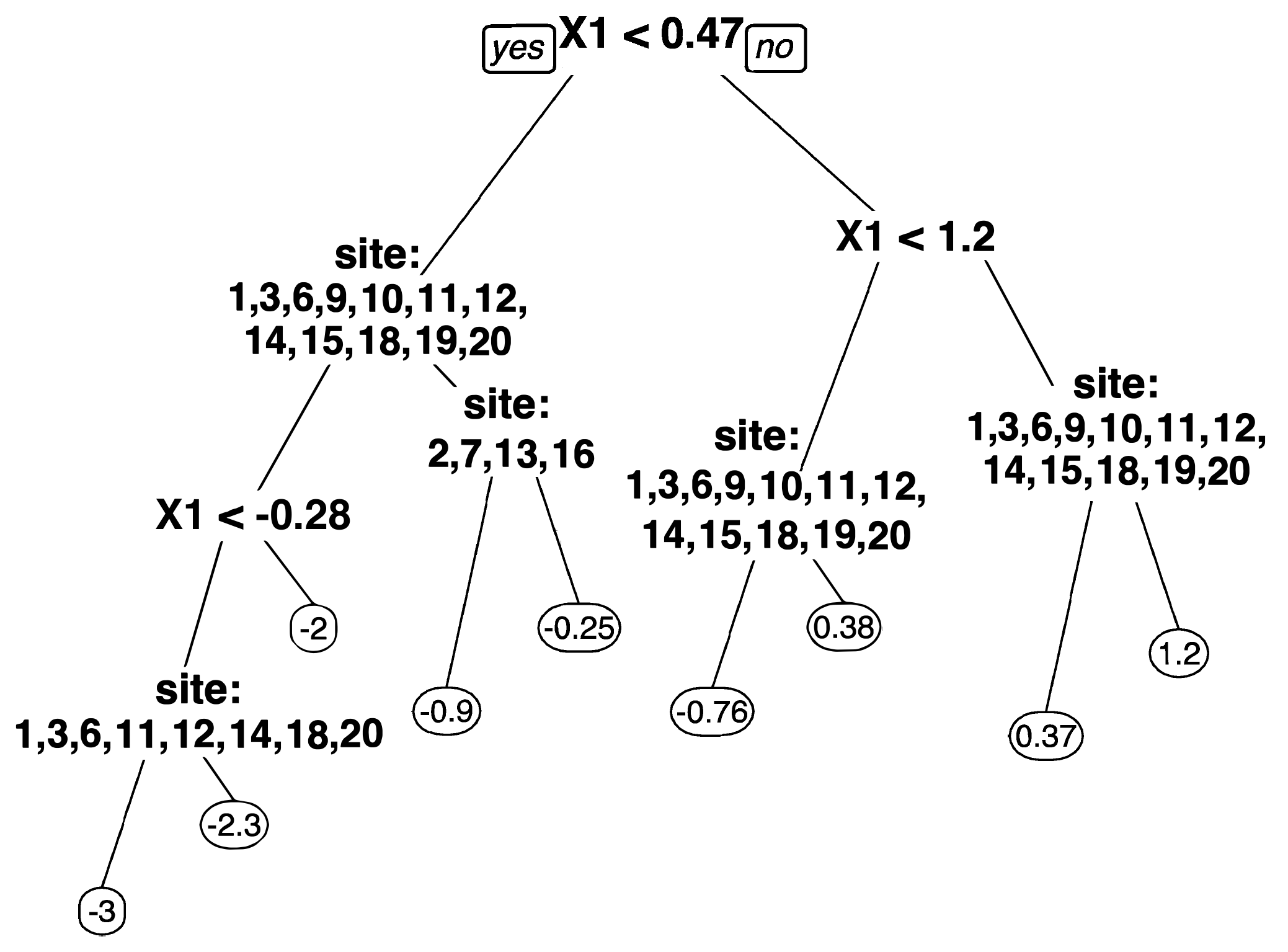}
    \caption{}
  \end{subfigure}
  \begin{subfigure}{0.325\textwidth}
    \centering
    \includegraphics[width=\linewidth]{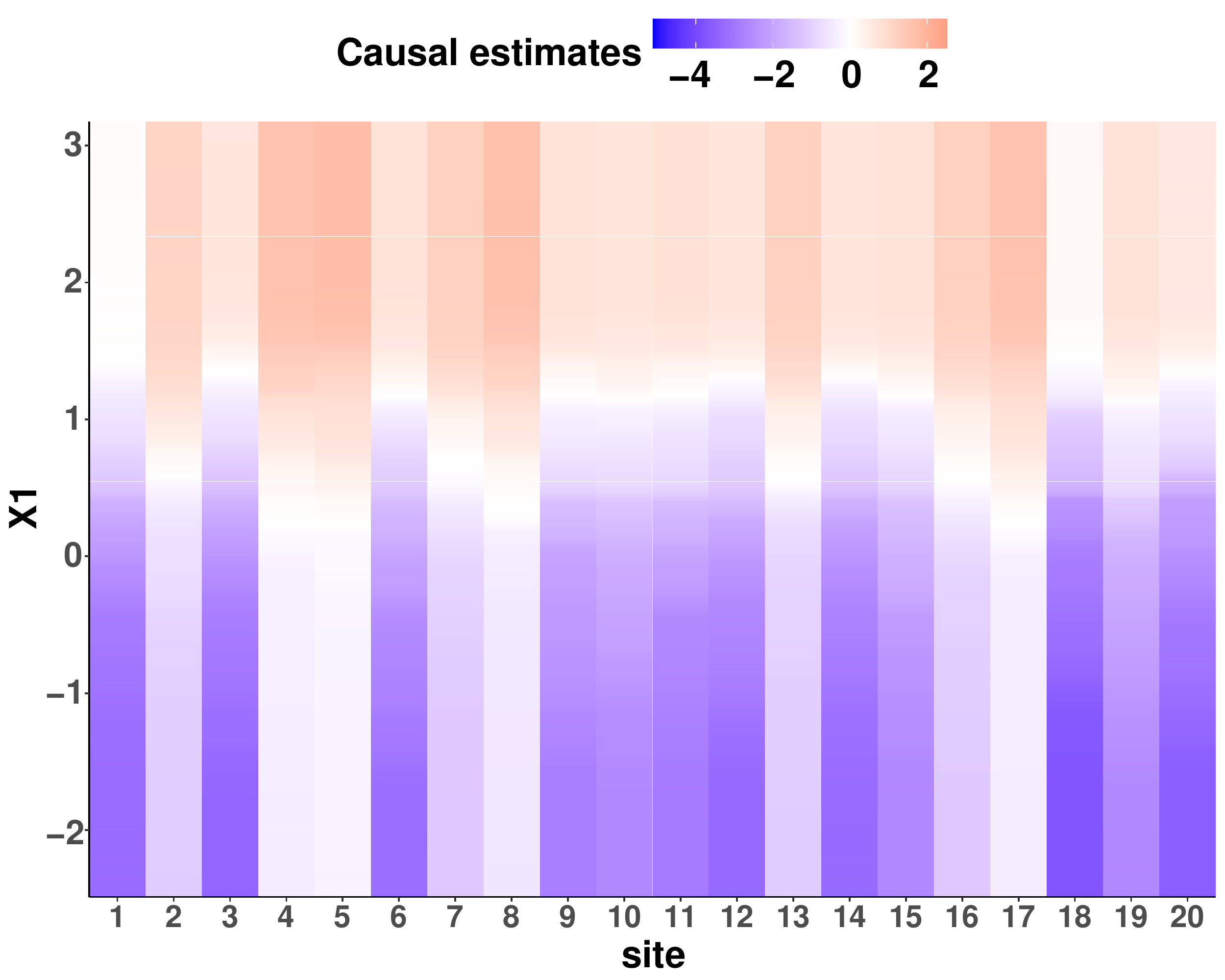}
    \caption{}
  \end{subfigure}
  \begin{subfigure}{0.37\textwidth}
    \centering
    \includegraphics[width=\linewidth]{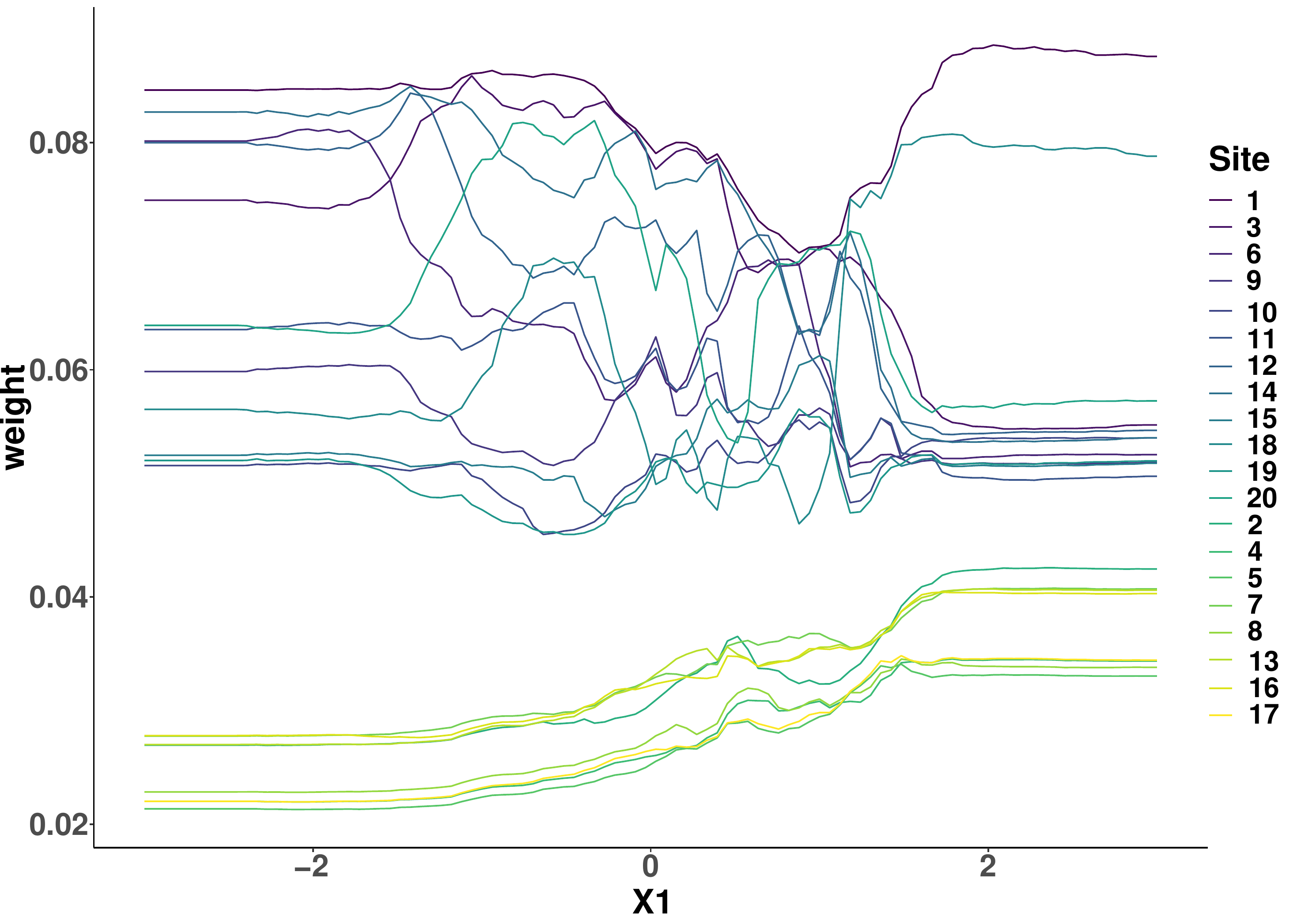}
    \caption{}
  \end{subfigure}
% \vspace{-1cm}
  \caption{
 Visualization of simulation 
 results under \textbf{discrete grouping (a,b,c)} and \textbf{continuous grouping (d,e,f)} when  $c=1$. (a) and (d) visualize the proposed ETs where the site indicator and $X_1$ are selected as splitting variables, which is consistent with the underlying data generation process. 
 (b) and (e) show the predicted treatment effects of the proposed EFs varying $X_1$ in each site, marginalized over all other features. (b) is arranged according to the true grouping, odd sites versus even sites. 
 The plot recovers the pattern of local and global heterogeneity. 
 (c) and (f) plot the interpretable model averaging weights in EFs over $X_1$. 
 The weights of site 1 have a relatively large contribution to the weighted estimator while models from other sites have different contributions for different $X_1$ depending on their similarity in $\tau(\bx,k)$ to that in site 1. 
 Corresponding ET and EF show consistent patterns and recover the true grouping.
  }
  \label{fig:sim_visual}
\end{figure*}

\subsection{Additional Simulations}
The detailed results of these additional simulations are included in Appendix~\ref{suppl-sec:sim}.

\emph{\textbf{1) Connection to supervised learning.}} The uniqueness of averaging $\tau_k(\bx)$ as opposed to supervised learning that averages prediction models $f_k(\bx)$ is that the outcome of $f_k(\bx)$ is immediately available. 
In our case, an additional estimation step is needed to construct the model averaging weights. 
We provide a comparison among estimators that utilize the ground truth $\{\tau_k(\bx)\}_{k=1}^K$ (denoted as ``-oracle'') when computing ensemble weights.
This mimics the case of supervised learning where weights are based on observed outcomes. 
Oracle methods achieve smaller MSE ratios; the pattern is consistent with Table~\ref{fig:ratio}.

\emph{\textbf{2) Simulation under observational studies.}} 
We also consider the treatment generation mechanism under an observational design. 
Specifically, the propensity is given as $e(\bx) = \text{expit}(0.6x_1)$. We consider both a correctly specified propensity model using a logistic regression of $Z$ on $X_1$ and a misspecified propensity model with a logistic regression of $Z$ on all $\bX$. In general, the proposed estimators obtain the best performance with similar results as in Figure~\ref{fig:sim_box}. With the correctly specified propensity score model, the local estimator is consistent in estimating $\tau_k(\bx)$, the proposed framework is valid. When the propensity model is misspecified, extra uncertainty is carried forward from the local estimates, but the proposed estimators can still improve upon LOC. This is due to a bias-and-variance trade-off that leads to small MSE, which remains smaller than the local models. 

\emph{\textbf{3) Covariate dimensions.}} 
Besides $D=5$, we consider other choices of covariate dimension including $D=20, 50$. With a higher dimension, the MSE ratio between the proposed estimates and LOC estimates increases but the same pattern across methods persists.

\emph{\textbf{4) Unequal sample size at each site.}} 
In the distributed date network, different sites may have a different sample size $n_k$. Those with a smaller sample size may not be representative of their population, leading to an uneven level of precision for local causal estimates. We consider a simulation setting where site 1 has a sample size of $n_1=500$ while other site $n_2,\ldots,n_K$ has a sample size of 200. Results show that the MSE ratio between the proposed estimates and LOC estimates increases compared to the scenario where the sample size in all sites are 500. However, the proposed estimators still enjoy the most robust performance. 
This also shows our method is robust to the existence of local uncertainty. 

\emph{\textbf{5) Different local estimators.}} We stress that other consistent estimators could be used as the local model. Options such as causal forest \citep{wager2018estimation} are explored varying the sample size at local sites. Similar performance is observed as in Figure~\ref{fig:sim_box}.

\emph{\textbf{6) Further comparisons to non-adaptive ensemble.}} Here we provide a brief discussion of the implications of the proposed method and how it differs from non-adaptive methods such as stacking. 
Although unrealistic, when the true weights are non-adaptive, the performance may be similar. Plus, our learned weights can be used to examine adaptivity, as shown in Figure~\ref{fig:sim_visual}(c,f) and Figure~\ref{fig:real}(c). 
Stacking is shown to be more robust than non-adaptive model averaging in case of model misspecification. See discussion in \citet{clarke2003comparing}. 
Our additional simulation results show that in case of a large global heterogeneity, as $c$ increases, the heterogeneity across sites gets larger, reducing the influence of important covariates on heterogeneity, hence the weights become more non-adaptive. However, the proposed methods still enjoy a comparable performance to STACK, which further indicates the robustness of the proposed methods.

\section{Example: A Multi-Hospital Data Network}
\label{sec:application}

Application with contextual insights is provided based on an analysis of the eICU Collaborative Research Database, a multi-hospital database published by Philips Healthcare \citep{pollard2018eicu}. 
The analysis is motivated by a recent retrospective study that there is a higher survival rate when SpO$_2$ is maintained at 94-98\% among patients requiring oxygen therapy \citep{van2020search}, not ``the higher the better''. 
We use the same data extraction code to create our data. 
We consider SpO$_2$ within this range as treatment ($Z=1$) and outside of this range as control ($Z=0$). 
A total of 7,022 patients from 20 hospitals, each with at least 50 patients in each treatment arm, are included with a randomly selected target (hospital 1).
Hospital-level summary information is provided in Appendix~\ref{suppl-sec:real}.
Patient-level features include age, BMI, sex, Sequential Organ Failure Assessment (SOFA) score, and duration of oxygen therapy.
The outcome is hospital survival ($Y=1$) or death ($Y=0$).

Figure~\ref{fig:real} visualizes the performance of EF-based estimated effect of oxygen therapy setting on in-hospital survival. 
CT is used as the local model with propensity score modeled by a logistic regression. 
Figure~\ref{fig:real}(a) shows the propensity score-weighted \textbf{average survival} for those whose received treatment is consistent with the estimated decision. 
Specifically, the expected reward is given by $$\frac{\sum_i Y_i 1(Z_i = Z^{est}_i) / \pi(Z_i,\bX_i)}{\sum_i 1(Z_i = Z^{est}_i) / \pi(Z_i,\bX_i)},$$ where $Z^{est}_i = 1(\widehat\tau>0)$ denotes the estimated treatment rule and $\pi(Z_i,\bX_i)$ is the probability of receiving the actual treatment.
We provide expected reward for the 1) observed treatment assignment (baseline), 2) LOC-based rule, and 3) EF-based rule. 
The treatment rule based on our method can increase mean survival by 3\% points compared to baseline, and is more promising than LOC.

% \clearpage
\begin{figure*}[!htb]%[htp]
\centering
%  \vspace{1cm}
 \begin{subfigure}{0.2\textwidth}
  \centerline{\includegraphics[width=\linewidth]{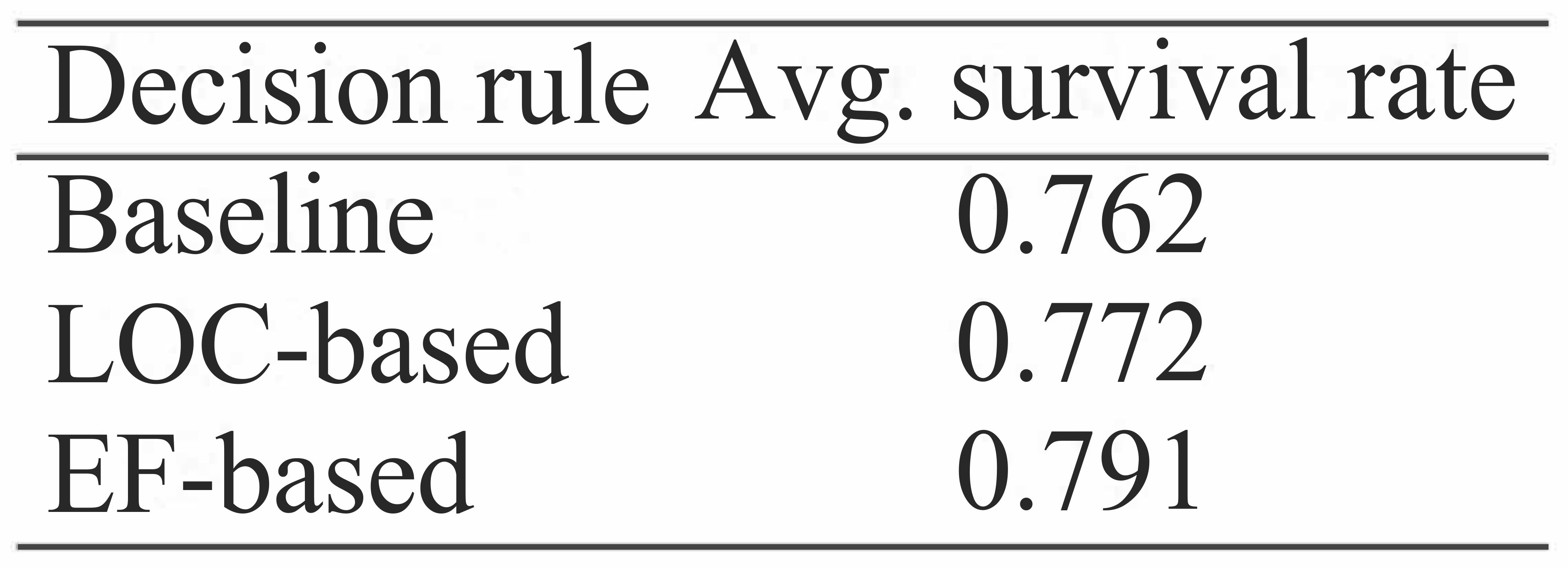}}
  \caption{}
 \end{subfigure}
 \begin{subfigure}{0.39\textwidth}
  \centerline{\includegraphics[width=\linewidth]{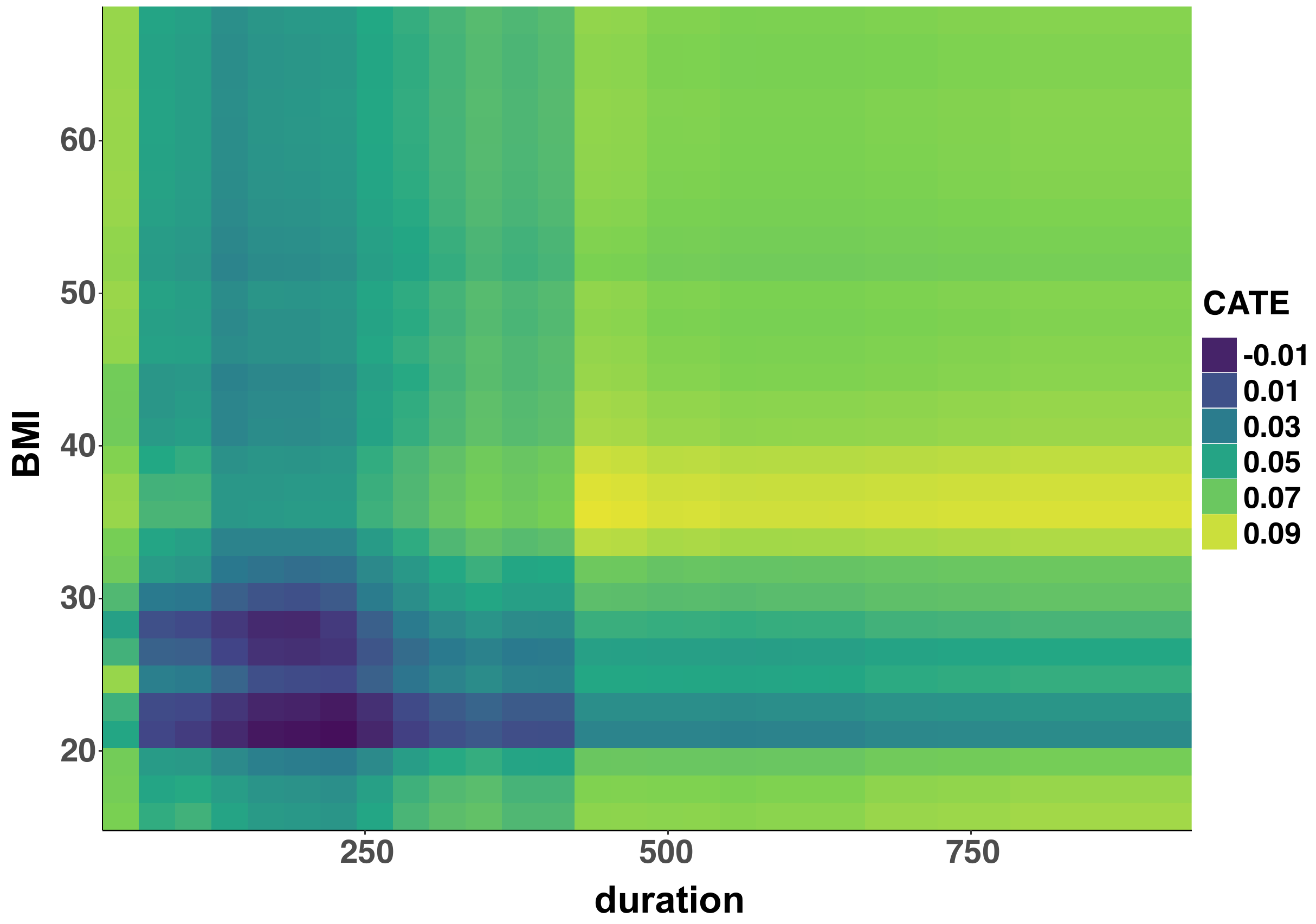}}
  \caption{}
 \end{subfigure}
 \begin{subfigure}{0.39\textwidth}
  \centerline{\includegraphics[width=\linewidth]{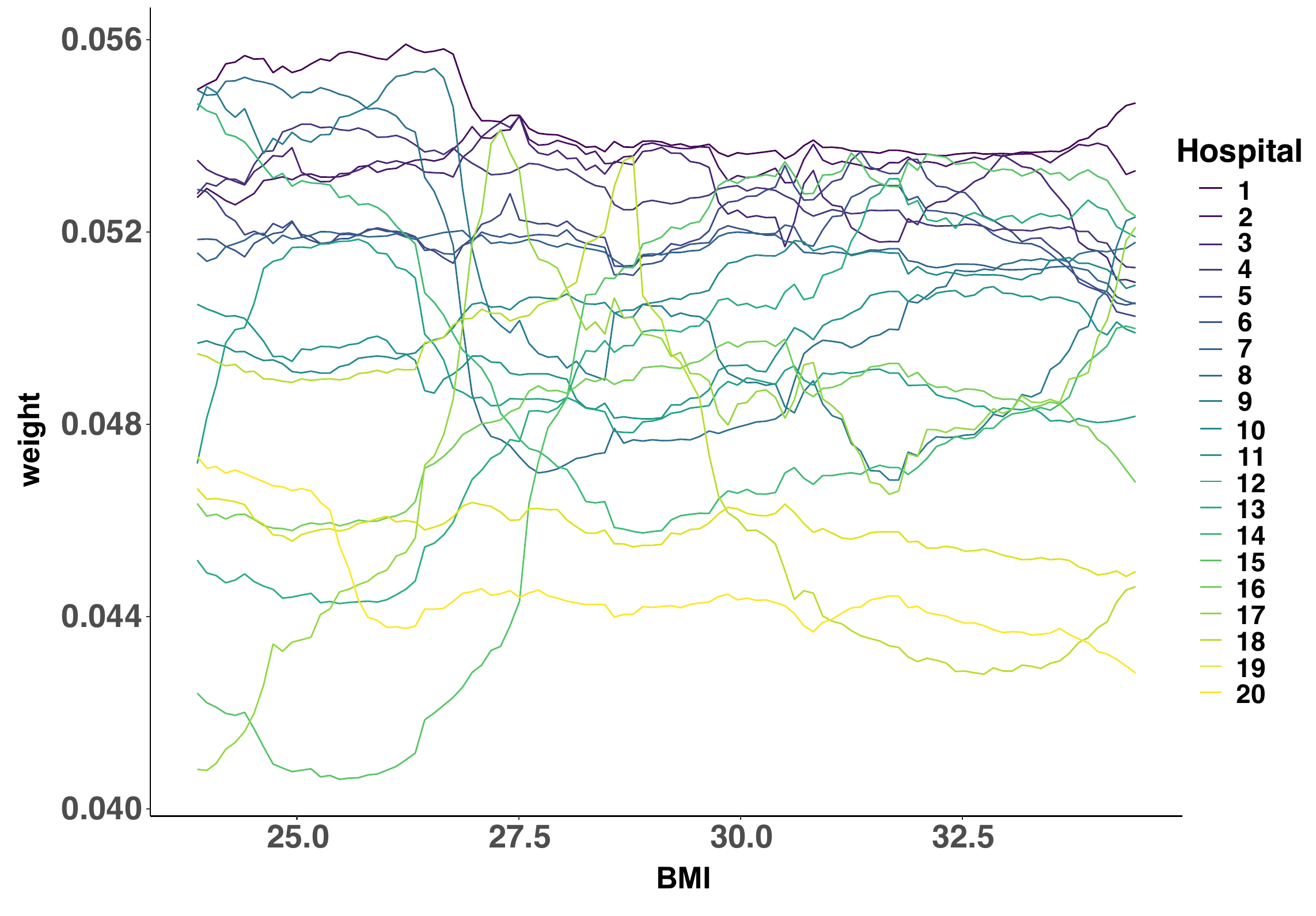}}
  \caption{}
 \end{subfigure}
%  \vspace{-1.4cm}
 \caption{
 Application to estimating treatment effects of oxygen therapy on survival. 
 (a) Expected survival of treatment decision following different estimators. 
 The proposed EF shows the largest gain in improving survival rate, more promising than LOC and baseline. 
 (b) Estimated treatment effects varying duration and BMI, two important features in the fitted EF. 
 Patients with a BMI around 35, and a duration above 400 benefited the most. 
 (c) Visualization of data-adaptive weights in the estimated EF varying BMI. 
 Hospitals with a larger bed capacity tend to contribute more, the data of which might be more similar to hospital 1.
}
\label{fig:real}
\end{figure*}

In the fitted EF, 
the hospital indicator is the most important, explaining about 50\% of the decrease in training error. 
Figure~\ref{fig:real}(b) shows the estimated CATE varying two important features, BMI and oxygen therapy duration. 
Patients with BMI around 36 and duration above 400 show the most benefit from oxygen therapy in the target SpO$_2$ range.  
Patients with BMI between 20 and 30 and duration around 200 may not benefit from such alteration. 
Figure~\ref{fig:real}(c) visualizes 
the 
data-adaptive weights $\omega_{k}(\bx)$ in the fitted EF with respect to BMI for different models, while holding other variables constant. The weights of hospital 1 are quite stable while models from other sites may have different contribution to the weighted estimator for different values of BMI. 
Judging from hospital information 
in Appendix~\ref{suppl-sec:real}, hospitals with a larger bed capacity tend to be similar to hospital 1, and are shown to provide larger contributions.

In this distributed research network, different hospitals have a different sample size. 
For sensitivity analysis, we consider a weighting strategy to adjust for the sample size of site $k$. 
Results show similar patterns as in Figure \ref{fig:real}. Detailed results 
are provided in Appendix~\ref{suppl-sec:real}. 
The real-data access is provided in Appendix~\ref{suppl-sec:code}.

\section{Discussion} 
\label{sec:disc}

We have proposed an efficient and interpretable tree-based model averaging framework for enhancing treatment effect estimation at a target site by borrowing information from potentially heterogeneous data sources. We generalize standard model averaging scheme in a data-adaptive way such that the generated weights depend on subject-level features. 
This work makes multi-site collaborations and especially treatment effect estimation more practical by avoiding the need to share subject-level data. 
Our approach extends beyond causal inference to estimating a general $f(\bx)$ from heterogeneous data.

Unlike in classic model averaging where prediction performance can be assessed against observed outcomes or labels, treatment effects are not directly observed. 
While our approach is guaranteed to be consistent under randomized studies, 
the weights are estimated based on expected treatment effects, hence relying on  Assumption~\ref{assump:unconfounded} (unconfoundedness) to hold. It may be a strong assumption in observational studies with unmeasured confounding.

\chapter{Robust Individualized Decision Learning with Sensitive Variables}\label{chap:itr}

\section{Introduction}\label{sec:intro-rise}

Recently, there has been a widespread interest in developing methodology for individualized decision rules (IDRs) based on observational data. 
When deriving IDRs, some collectible data are important to the intervention decision, while their inclusion in decision making is prohibited due to reasons such as delayed availability or fairness concerns. 
For example, sensitive characteristics of subjects regarding their income, sex, race and ethnicity may not be appropriate to be used directly for decision making due to fairness concerns. 
In the medical field especially for patients in severe life-threatening conditions such as sepsis, timely bedside intervention decisions have to be made before lab measurements are ordered, assayed and returned to the attending physicians. However, due to the delayed availability of lab results, most of the decisions are made with great uncertainty and bias due to partial information at hand. 
We define \textit{sensitive variables} as variables whose inclusion into decision rules is prohibited.  
The formal definition of sensitive variables will be given in Section~\ref{ssec:method}.

In this work, we propose RISE (\textbf{R}obust \textbf{I}ndividualized decision learning with \textbf{SE}nsitive variables)\footnote{Python code is available at \url{https://github.com/ellenxtan/rise}.}, a robust IDR framework to improve the outcome of individuals when there are informative yet sensitive variables that are either not available or prohibited from using during IDR deployment. 
To achieve this, we propose to estimate the optimal IDR by optimizing a quantile- or infimum-based objective, respectively, for continuous or discrete sensitive variables. This optimization problem is then shown to be equivalent to a weighted classification problem where \textit{most existing machine learning classifiers can be readily applied}.
Our idea falls along the lines of work that considers algorithmic fairness \citep{dwork2012fairness} while extending it to the setting of causal inference \citep{rubin2005causal} in the sense that decisions are driven by causality rather than a general utility function. 
We show in our empirical analyses that this leads to fairer and safer real-life decisions with little sacrifice of the overall performance.

Assuming that a larger outcome value is preferable, optimal IDRs are traditionally derived through maximizing the mean outcome of the sample population. 
In this paper, we are interested in a specific yet broadly applicable setting of learning that involves sensitive variables. 
We consider offline learning where sensitive variables are collected and \textit{can be used in training the IDRs}, but they \textit{cannot be used as input in the resulting IDRs}. 
This is a setting commonly considered in the fairness and privacy-related literature for classification (e.g., \cite{kamiran2012data}), but not from a causal standpoint.
When there exist important variables that are simply left out from training, the estimated IDR will be biased.
This bias can be removed if all important variables are used during training, which we will show in Section~\ref{sec:assump} a mean-optimal approach.
The optimal action maximizes the mean outcome where the mean is taken over the sensitive variables, conditioning on other variables. 
This method, however, has no control of the disparity in sensitive variables. Subjects with different sensitive values may report large outcome differences, hence unfairly or unsafely treated. 
Therefore, objective functions with robustness guarantees for sensitive variables are preferred, since they offer protection to subjects in the lower tail of the outcome distribution with regards to the sensitive values.

For illustration, we consider a toy example with binary actions, $A \in \{-1, 1\}$. 
We remark that the decision can only be made based on the variable $X$ whereas $S$ is a sensitive variable.
The setup is shown in Table~\ref{t:toy_setup} and the oracle values under the mean-optimal rule and RISE\footnote{For the mean-optimal rule, overall average reward is calculated by $(30+0+5+27)/4=15.5$, reward among vulnerable subjects is calculated by $(0+5)/2=2.5$; for RISE, overall average reward is calculated by $(11+13+15+13)/4=13$, reward among vulnerable subjects is calculated by $(13+15)/2=14$.}
are given in Table~\ref{t:toy_res}. 
The detailed setup can be found in Section~\ref{ssec:simulation} under Example 1. 
We consider \textit{vulnerable subjects} as those with low outcome values, as highlighted in red in Table~\ref{t:toy_setup} (A full definition is given in Section~\ref{ssec:vul}). 
For $X \leq 0.5$, the mean-optimal rule would assign action $A=1$ as it tries to achieve the largest average reward across $S=0$ and $S=1${. Recall that $S$ is not available at the time of decision}. However, this action results in great harm for subjects with $S=1$ as they could get the worst expected outcome of 0. On the contrary, RISE improves the worst-case outcome by assigning $A=-1$, protecting the vulnerable subjects. 
Likewise, for $X > 0.5$, the mean-optimal rule assigns $A=-1$ while RISE assigns $A=1$ protecting those with $S=0$ that could have experienced an outcome of $5$. 
Compared to the mean-optimal rule, the proposed rule achieves a larger reward among vulnerable subjects while maintaining a comparable overall reward.

\begin{table}[hbt!]%[!htb]
\caption{Toy example setup {of $E(Y|X,S,A)$.}}
\centering
\label{t:toy_setup}
\begin{tabular}{cccccc}
\toprule
 & \multicolumn{2}{c}{$X \leq 0.5$} & & \multicolumn{2}{c}{$X > 0.5$} \\ \cmidrule(l){2-3} \cmidrule(l){5-6} 
 & $S=0$ & $S=1$ & & $S=0$ & $S=1$ \\ \midrule
$A=-1$ & {11} & 13 & & \red{5} & 27 \\
$A=1$ & 30 & \red{0} & & 15 & {13} \\ \bottomrule
\end{tabular}
\end{table}

\begin{table}[hbt!]%[!htb]
\caption{Toy example results.}
\centering
\label{t:toy_res}
\begin{tabular}{ccc}
\toprule
  & \multicolumn{2}{c}{Average reward} \\ \cmidrule(l){2-3}  
 & Overall & Vulnerable \\ \midrule
Mean-optimal rule & {\textbf{15.5}} & {2.5} \\
RISE &  13 & \textbf{14} \\ \bottomrule
\end{tabular}
\end{table}

\vspace{1cm}
\textbf{Main contributions. } 
\textit{Methodology-wise,}  
\textit{1)} we propose a novel framework, RISE, to handle sensitive variables in causality-driven decision making.
Robustness is introduced to improve the worst-case outcome caused by sensitive variables, and as a result, it reduces the outcome variation across subjects. The latter is directly associated with fairness and safety in decision making. 
To the best of our knowledge, we are among the first to propose a robust-type fairness criterion under causal inference.
\textit{2)} We introduce a classification-based optimization framework that can easily leverage most existing classification tools catered to different functional classes, including state-of-the-art random forest, boosting, or neural network models. 
\textit{Application-wise,} 
\textit{3)} the consideration of sensitive variables in decision learning is important to applications in policy, education, healthcare, etc. 
Specifically, we illustrate the application of RISE using three real-world examples from fairness and safety perspectives where robust decision rules are needed, across which we have observed robust performance of the proposed approach. 
From a fairness perspective, we consider a job training program where age is considered as a sensitive variable. 
From a safety perspective, we consider two applications to healthcare where lab measurements are considered as sensitive variables.

The remaining chapter is organized as follows.  Section~\ref{ssec:related} discusses related work. 
Section~\ref{ssec:method} describes the proposed RISE framework in detail. 
The performance of the proposed framework is evaluated by simulation studies and applied to three real-data applications in  Section~\ref{sec:numerical}. 
We conclude and discuss future work in Section~\ref{ssec:disc}.

\section{Related Work}\label{ssec:related}

Our work focuses on individualized decision rules, which aim at assigning treatment decision based on subject characteristics. 
Existing methods for deriving IDRs include model-based methods such as Q-learning \citep{watkins1992q,murphy2003optimal,moodie2007demystifying} and A-learning \citep{robins2000marginal,shi2018high}, model-free policy search methods \citep{zhang2012robust,zhao2012estimating,zhao2015doubly}, and contextual bandit methods \citep{bietti2021contextual,li2011unbiased}. 
In Appendix~\ref{suppl:lit}, we provide additional literature review on general IDRs under causal settings. 
Fairness, safety and robustness are topics of interest that extend well beyond causal inference. In the following, we provide a review of these areas, with focus given to work related to causal inference and IDRs.

\textbf{Fairness and safety in IDRs. } 
The consideration of fairness and safety in machine learning has seen an explosion of interest in the past few years {\citep{dwork2012fairness,varshney2016engineering,barocas2017fairness,nabi2018fair,hashimoto2018fairness,chouldechova2020snapshot,mehrabi2021survey,pessach2022review}, especially for solving classification and regression problems to help derive decisions that are not only accurate but also fair.}  
In these work, sensitive variables are also referred to as sensitive, protected, or auxiliary attributes. 
We extend the definition of sensitive variables to include delayed information that is not available at deployment as it is also suitable for this framework.

Among earlier work, preprocessing \citep{kamiran2012data,feldman2015certifying,creager2019flexibly,sattigeri2019fairness} 
and inprocess training approaches \citep{beutel2017data,hashimoto2018fairness,lahoti2020fairness} 
consider disentangling the input $X$ from a known or unknown sensitive variable $S$ so that the transformed $X$ does not contain any information related to $S$. 
Due to the causal nature of IDRs, effect of IDRs cannot be estimated consistently when an informative $S$ is left out and the resulting rule is suboptimal.
This follows from the classic argument that any unmeasured confounding (i.e., $S$), if not accounted for, would lead to bias.
Similar issues persist in contextual bandits \citep{joseph2016fairness,patil2020achieving}. 
{\cite{makar2022causally} considers reducing the impact of auxiliary variables on prediction under distributional shift. Although it is motivated from a causal idea, its main focus is still on prediction.} 
Inside the causal framework,
\cite{zhang2018fairness,nabi2019learning} extend fairness from prediction to policy learning using causal graphical models by incorporating fairness constraints. 
\cite{chen2022learning} considers counterfactual fairness that seeks to achieve conditional independence of the decisions via data preprocessing. 
Despite earlier efforts in bringing fairness into the causal framework, most of these approaches 
only ensure mean zero disparity in $S$ but do not have robustness guarantees in the sense that the variance of the disparity in $S$ is not controlled. 
Besides, most examples consider a single categorical sensitive variable, but not multiple or continuous ones.

\textbf{Robustness in IDRs. } 
Recently the statistical literature has witnessed a growing interest in developing robust methods for estimating IDRs. They introduce robustness into the objective function by using quantile-optimal treatment regimes or mean-optimal treatment regimes under certain constraints to improve the gain of individuals at the lower tail of the reward spectrum \citep{wang2018quantile,wang2018learning,qi2019estimating,qi2019estimation,fang2021fairness}. 
In particular, \cite{wang2018quantile,qi2019estimating} propose to estimate quantile- or tail-optimal treatment regimes. 
\cite{wang2018learning} studies the mean-optimal treatment regime under a constraint to control for the average potential risk. 
\cite{qi2019estimation} proposes a decision-rule-based optimized covariates dependent equivalent for tasks of  individualized decision making. 
\cite{fang2021fairness} considers mean and quantile objectives simultaneously by maximizing the average value with the guarantee that its tail performance exceeds a prespecified threshold. 
Robustness, in their sense, pertains to the outcome distribution subject to the sampling error.
When sensitive variables are present, we consider instead the robustness of the outcome distribution subject to the uncertainty due to sensitive variables, providing a more targeted way of ensuring robustness, which is directly related to fairness and safety. 
Compared to algorithms based on explicit fairness constraints (for example \cite{zafar2017fairness,zhang2018mitigating} in classification and \cite{zhang2018fairness,chen2022learning} in causal inference) that seek to remove the disparity across different values of $S$, our method reduces the variance of disparity across $S$. In addition, constraint-based approaches typically require specialized optimization procedures whereas our approach presents an elegant and systematic way for optimization. 
To our knowledge, we are the first few to consider decision fairness via a robust objective under the causal framework.

\section{Robust Decision Learning Framework with Sensitive Variables} \label{ssec:method}

\subsection{Preliminaries}\label{sec:assump}

\textbf{Notation. } We let random variables be represented by upper-case letters, and their realizations be represented by lower-case letters. 
Suppose there are $n$ independent subjects sampled from a given population. For subject $i$, 
let $A_i \in \{-1,1\}$ denote a binary treatment assignment and $Y_i$ denote the corresponding outcome. Without loss of generality, we assume a larger value of outcome is desirable. 
Under the potential outcomes framework \citep{rubin1978bayesian,splawa1990application}, let $Y_i(a)$ be the potential outcome had the subject been assigned to $A = a$ for $a = 1$ or $-1$. 
Let $X_i \in \mathbb{X}$ be the feature vector and, for now, $S_i$ be a single sensitive variable. 
We consider $S \in \mathbb{S}$ where $\mathbb{S} = \{1,\ldots,K\}$ if $S$ is discrete and $\mathbb{S} = \mathbb{R}$ if $S$ is continuous. 
The extension to multiple sensitive variables is presented in Section~\ref{sec:exten}. 

\textbf{Definition of sensitive variables. } We define sensitive variables $S$ as variables that are important to the intervention decision, but their inclusion in decision making is prohibited. 
Formally, consider variables $X$ and $S$ that are both available during model training and are both determinants of conditional average treatment effect \citep{rubin1974estimating}. 
While $S$ may be involved in training, the derived decision rule $d(\cdot)$ precludes the input of $S$ due to sensitive concerns. Hence, the derived IDR is a function of the form $d(X): \mathbb{X} \mapsto \mathbb{A}$.  
Following the above definition, we consider an offline learning framework where sensitive variables are collected and can be used in obtaining the IDRs, but they cannot be used in the resulting IDRs. 
A causal diagram and a decision diagram are provided in Figure~\ref{fig:diagram}. 
As it shows in Figure~\ref{fig:diagram}a, both $X$ and $S$ confound the effect of treatment $A$ on outcome $Y$. The arrows represent the causal relationship between variables. Note that $X$ and $S$ can be correlated.  
This causal diagram is formalized in Assumption~\ref{assump:rise-unconf} below. 
On the other hand in Figure~\ref{fig:diagram}b, in the decision diagram under our setting, $S$ is shown in a dotted circle as $S$ may not be readily available at the time of decision making. We connect $S$ and $A$ with a dotted arrow to indicate that $S$ is incorporated in the training of the decision rule, but it is not required at deployment.

\vspace{0.5em}
\begin{figure}[!htb]
\centering
 \begin{subfigure}{0.25\textwidth}
  \centerline{\includegraphics[width=\linewidth]{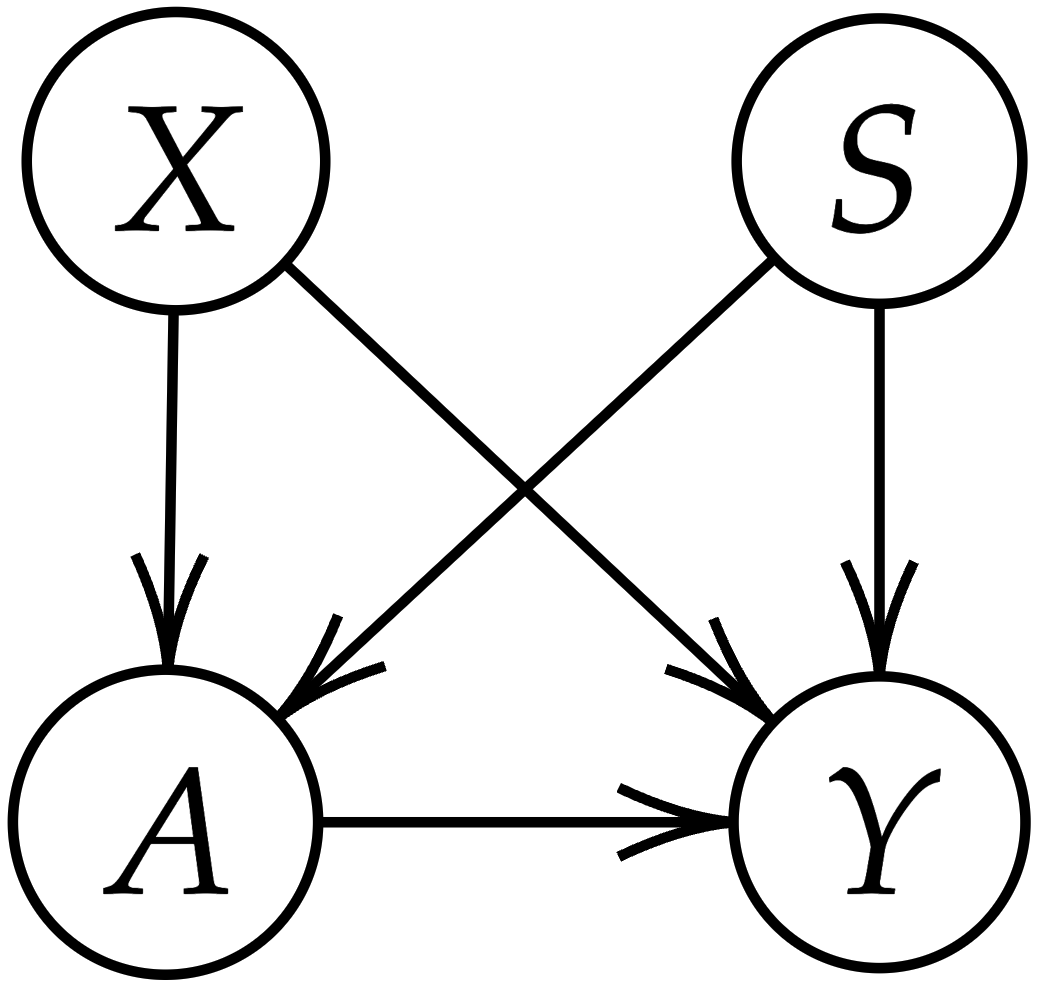}}
  \caption{}
 \end{subfigure}
 \qquad
 \begin{subfigure}{0.25\textwidth}
  \centerline{\includegraphics[width=\linewidth]{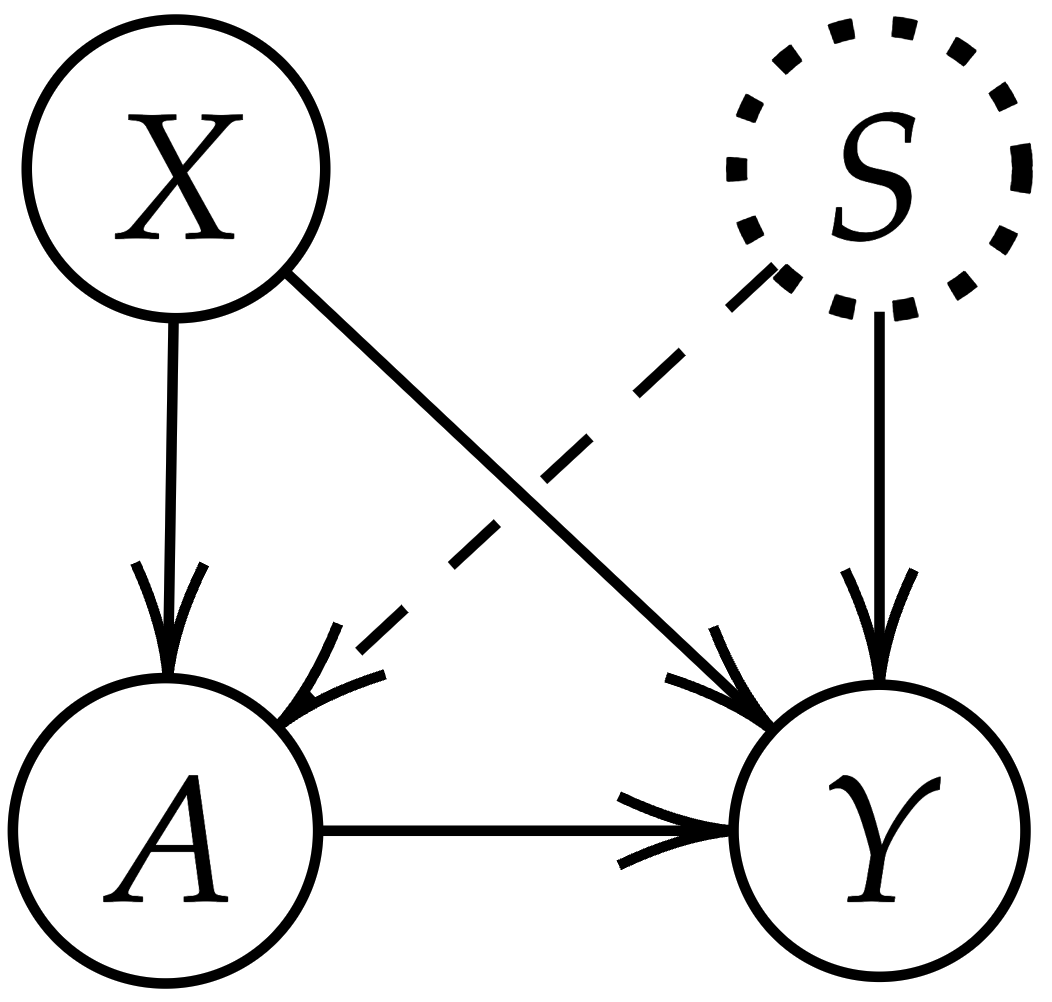}}
  \caption{}
 \end{subfigure}
 \vspace{-1.3cm}
 \caption{
 (a) A causal diagram.   
 (b) A decision diagram. 
 } 
 \label{fig:diagram}
\end{figure}

\begin{assumption}[Consistency] \label{assump:rise-consist}
$$ Y=Y(-1)\mathbbm{1}(A=-1)+Y(1)\mathbbm{1}(A=1).$$ 
\end{assumption}

\begin{assumption}[Positivity] \label{assump:rise-pos}
$$ 0 < Pr(A=1 | X,S) < 1.$$ 
\end{assumption}

\begin{assumption}[Unconfoundedness] \label{assump:rise-unconf}
$$ \{Y(-1), Y(1)\} \perp A | \{X,S\} ~\text{and}~ \{Y(-1), Y(1)\} \not\perp A | X.$$ 
\end{assumption}

Assumption~\ref{assump:rise-consist} is the standard consistency assumption in causal inference and  
Assumption~\ref{assump:rise-pos} states that every subject has a nonzero probability of getting the treatment. 
Assumption~\ref{assump:rise-unconf} states that given $X$ and $S$, the potential outcomes are independent of the treatment assignments. Besides, unconfoundedness does not hold when only $X$ is given, signifying the important role of $S$. 
Under causal settings, Assumption~\ref{assump:rise-unconf} implies that treatment effects cannot be non-parametrically identifiable without $S$ \citep{neyman1923applications,rubin1974estimating}. Approaches such as ones that disentangle $X$ from $S$ under supervised learning settings mentioned in Section~\ref{ssec:related} will introduce bias towards estimating IDRs.

Before introducing the proposed method, we first discuss two kinds of approaches to deal with sensitive variables under causal settings. 
The first kind is naive approaches that omit sensitive variables. 
When $S$ is not available for future deployment, a naive approach is to maximize $E_{X} \{E(Y|X,A=d(X))\}$ over $d$ using $(X, A, Y)$ during the training procedure. This approach will introduce bias in the estimation of potential outcomes and lead to a suboptimal IDR due to the unmeasured confounder $S$. 
It is thus important that one includes $S$ into the training procedure. For example, if we consider the value function framework (i.e., expected outcome) used by most existing works such as \cite{manski2004statistical,qian2011performance}, we can show that 
\begin{align}\label{eq:obj_mean}
    E\{Y(d)\} 
    &= E_{X,S} \big[ E(Y(d)|X,S) \big] = E_{X} \big[ E_{S|X} \{ E(Y(d)|X,S) \}   \big]  \\
    &= E_{X} \big[ E_{S|X}\{E(Y|X,S,A=d(X))\}   \big] \neq E_{X} \big[ E(Y|X,A=d(X))\big], \nonumber
\end{align}
where the third equality in \eqref{eq:obj_mean} holds by Assumptions in Section~\ref{sec:assump} and the last inequality also indicates the naive approaches without using $S$ will in general fail. Then one valid approach is the mean-optimal approach that uses the sensitive variables. That is, to maximize $E_{X} \big[ E_{S|X}\{E(Y|X,S,A=d(X))\}   \big]$ over $d$ using $(X, S, A, Y)$. The optimal IDR under this criterion is, for every $X \in \mathbb{X}$,
\begin{align*}
    \tilde{d}(X) \in \operatorname{sgn}(E_{S|X}\{E(Y|X,S,A=1)\} - E_{S|X}\{E(Y|X,S,A=-1)\}), 
\end{align*}
which guarantees to find the treatment that maximizes the conditional expected outcome given each $X$ by averaging out the effect of the sensitive variable $S$. 
The mean-optimal approaches, however, fail to control the disparities across realizations of the sensitive variables due to the integration over $S$, which may lead to unsatisfactory decisions to certain subgroups, as illustrated in the toy example in Section~\ref{sec:intro-rise}.

\subsection{Robust Optimality with Sensitive Variables}\label{ssec:vul}

Driven by the limitation of existing approaches, our goal is to derive a robust decision rule that maximizes the worst-case scenarios of subjects when some sensitive information is not available at the time of deploying the decision rule. Specifically, our robust decision learning framework draws decisions based on individuals' available characteristics summarized in the vector $X$ without the sensitive variable $S$, while improving the worst-case outcome of subjects in terms of the sensitive variable in the population. Formally, given a collection $\mathbb{D}$ of all treatment decision rules depending only on $X$, the proposed RISE approach estimates the following IDR, which is defined as 
\begin{align}\label{eq:obj_max} 
    d^\ast \in {\arg\max}_{d\in\mathbb{D}} E_X \big[ G_{S|X} \{E (Y|X,S,A=d(X)) \} \big],
\end{align}
where $G_{S|X}(\cdot)$ could be chosen as some risk measure for evaluating $E (Y|X,S,A=d(X))$ for each $S \in \mathbb{S}$. Examples include variance, conditional value at risk, quantiles, etc. In this paper, we consider $G_{S|X}$ as the conditional quantiles (for a continuous $S$) or the infimum (for a discrete $S$) over $\mathbb{S}$. 

Specifically, for a discrete $S$, $G_{S|X}$ 
is consider as an infimum operator of $E(Y|X,S,A=d(X))$ over $S$. We thus aim to find
$$%\text{maximize}_{d\in\mathbb{D}} E_X \big[ 
d^\ast \in {\arg\max}_{d\in\mathbb{D}} E_X\big[\operatorname{inf}_{s \in \mathbb{S}} \{E (Y|X,S=s,A=d(X)) \} \big]
,$$
where $\operatorname{inf}$ is the infimum taken with respect to
$E (Y|X,s,A=d(X))$ over $s \in \mathbb{S}$. This implies that for a given $X$, $d^\ast(X)$ assigns the treatment that yields the best worst-case scenario among all possible values of $S$ for every $X \in \mathbb{X}$, or equivalently,
$$
d^\ast(X) \in \operatorname{sgn}(\operatorname{inf}_{s \in \mathbb{S}} \{E (Y|X,S=s,A=1) - \operatorname{inf}_{s \in \mathbb{S}} \{E (Y|X,S=s,A=-1)\}).
$$

For a continuous $S$, we consider $G_{S|X} \{E (Y|X,S,A=d(X)) \}$ as $Q_{S|X}^{\tau} \{E (Y|X,S,A=d(X)) \},$
which is the $\tau$-th quantile of $\{E (Y|X,S,A=d(X)) \}$ and $\tau \in (0,1)$ is the quantile level of interest. 
Specifically, $Q_{S|X}^{\tau} \{E (Y|X,S,A=d(X)) \} = \inf\{ t: F(t) \geq \tau\}$ with $F$ denoting the conditional distribution function of $E (Y|X,S,A=d(X))$ over $\mathbb{S}$ given $X$ and $d$. Note the randomness behind $E (Y|X,S,A=d(X))$ given $X$ and $d$ is fully determined by the sensitive variable $S$. Then optimal IDR under this criterion is defined as
$$%\text{maximize}_{d\in\mathbb{D}} E_X \big[ 
d^\ast \in {\arg\max}_{\mathbb{D}} E_X\big[Q_{S|X}^{\tau} \{E (Y|X,S,A=d(X)) \} \big]
.$$ 
This implies that for a given $X$, $d^\ast(X)$ assigns a treatment that yields the largest $\tau$-th quantile of the outcome over the distribution related to $S$, or equivalently,
$$
d^\ast(X) \in \operatorname{sgn}( \{Q_{S|X}^{\tau} \{E (Y|X,S,A=1) \}  - Q_{S|X}^{\tau} \{E (Y|X,S,A=-1) \} ).
$$ 
We let $\tau=0.25$ throughout the paper and suppress $\tau$ for simplicity. Results on varying the value of $\tau$ is provided in Appendix; see Section~\ref{ssec:simulation} for details.

\subsection{Identifying Vulnerable Subjects } 

Our RISE framework provides a natural way to define \textit{vulnerable groups}. Specifically, for a discrete $S$, if $\inf_S \{E(Y | X,S, A = 1)\} > \inf_S \{E(Y | X,S, A = 0)\}$, then $\arg \inf_S \{E(Y | X,S, A = 0)\}$ is vulnerable given $X$, otherwise is $\arg \inf_S \{E(Y | X,S, A = 1)\}$. 
In other words, the vulnerable subjects are those in the worst-off group that needs protection.
Similarly, for a continuous $S$, if $Q_S \{E(Y | X,S, A = 1)\} > Q_S \{E(Y | X,S, A = 0)\}$, then the set $\{S: E(Y | X,S, A = 0) \leq Q_S \{E(Y | X,S, A = 0)\} \}$ defines the vulnerable subjects given $X$, otherwise this group is defined as $\{S: E(Y | X,S, A = 1) \leq Q_S \{E(Y | X,S, A = 1)\} \}$.

\subsection{Estimation and Algorithm}\label{sec:diff_g}

Here we provide a transformation of the proposed RISE from an optimization problem to a weighted classification problem. There are several advantages to this conversion: 
1) The optimization problem defined in \eqref{eq:obj_max} involves a nonsmooth and nonconvex objective function that could lead to computational challenges. 
2) With multiple powerful statistical and machine learning toolbox to choose from, a classification problem can be more readily solved in practice. Hyperparameter tuning and model selection could be conducted to further boost performance. 
3) Compared to a direct optimization of \eqref{eq:obj_max}, a classification-based optimizer allows the use of off-the-shelf software packages that can be tailored to different functional classes or incorporate different properties such as model sparsity.

\begin{proposition}\label{prop:obj}
    Maximizing the objective function in~\eqref{eq:obj_max} is equivalent to maximizing $$E_X \big\{ \mathbbm{1}(d(X) = 1) [G_{S|X} \{E (Y|X,S,A=1) \}- G_{S|X} \{E (Y|X,S,A=-1) \}] \big\}.$$
\end{proposition}

With Proposition~\ref{prop:obj} and a proper estimator of the outcome model $E (Y|X,S,A)$ using our training data $\mathcal{D}_n = \{X_i, S_i, A_i, Y_i\}_{i=1}^n$, we replace the expectation of $Y_i$ by its estimate $\hat Y_i$ and solve the following problem. 
\begin{align} \label{eq:obj2}
     {\arg\max}_{d\in\mathbb{D}}   n^{-1} \sum_{i = 1}^{n} [ \mathbbm{1}(d(x_i)=1) \{g_1(x_i)-g_2(x_i)\} ],
\end{align}
where
$ g_1(x_i) = G_{s|x} \{\hat Y_i(x_i,s,a_i=1)  \} $ and $ g_2(x_i) = G_{s|x} \{\hat Y_i(x_i,s,a_i=-1) \}$.  
We have the following proposition to address noncontinuity in \eqref{eq:obj2} and transform it into a classification problem. 
Define $\mathbb{F}$ as a class of all measurable functions over $\mathbb{X}$.

\begin{proposition}\label{prop:obj2}
    Let $f(x)$ to be a smooth function. 
    Maximizing the empirical objective in~\eqref{eq:obj2} is equivalent to a weighted classification problem of minimizing over $f \in \mathbb{F}$,
    \begin{align} \label{eq:obj_min}
        n^{-1} \sum_{i = 1}^{n} \mathbbm{1}[ \operatorname{sgn}\{g_1(x_i)-g_2(x_i)\} \cdot f(x_i) <0 ] \cdot |g_1(x_i)-g_2(x_i)|,  
    \end{align}
    with features $x_i$, the true label $\operatorname{sgn}\{g_1(x_i)-g_2(x_i)\}$, and the sample weight $|g_1(x_i)-g_2(x_i)|$, for subject $i$, $i = 1,\dots, n$. 
\end{proposition}

With Proposition~\ref{prop:obj2}, we have transformed the optimization problem \eqref{eq:obj_max} into a weighted classification problem \eqref{eq:obj_min} where for subject $i$ with features $x_i$, the true label is $\operatorname{sgn}\{g_1(x_i)-g_2(x_i)\}$ and the sample weight is $|g_1(x_i)-g_2(x_i)|$. 
The estimated optimal decision rule by \eqref{eq:obj_min} is then given by $\hat d(x) = \operatorname{sgn}\{\hat f(x)\}$\footnote{If $\hat f(x)=0$, assign a random treatment.}. 
The proof of Proposition~\ref{prop:obj} and Proposition~\ref{prop:obj2} is presented in Appendix~\ref{suppl:proof}.

Algorithm~\ref{algo:code} provides an algorithmic overview of RISE. The inner expectation $E (Y|X,S,A)$ can be modeled as $\hat Y(X,S,A)$ using a twin model separated by the treatment and control groups (i.e., a T-learner as in \cite{kunzel2019metalearners}). 
For a continuous $S$, $G(X,A)= Q_{S|X,A} \{ E (Y|X,S,A) \}$ is estimated via a quantile regression of $\hat Y$ on $X$ but without $S$. 
For a discrete $S$, an estimate of $G(X,A) = \inf_{S} \{ E (Y|X,S,A) \}$ is obtained by finding the minimum among $\{E (Y|X,S=1,A),\ldots,E (Y|X,S=K,A)\}$. 
The estimated decision rule can then be obtained from the weighted classification. 
In our implementation, neural networks are used to fit models in the training data sets. The details on modeling and hyperparameter tuning via cross-validations are given in Appendix~\ref{suppl:tune}. 
A Python package \texttt{rise} based on neural networks is available on GitHub 
(\url{https://github.com/ellenxtan/rise}). 
Note that the model choices are flexible.

\clearpage
\begin{algorithm}
% \small
  \caption{RISE (Robust individualized decision learning with sensitive variables)}
  \label{algo:code-rise}
  \hspace*{\algorithmicindent} \textbf{Input} Training data $\mathcal{D}_n = \{Y_i,A_i,X_i,S_i\}_{i=1}^n$ \\
  \hspace*{\algorithmicindent} \textbf{Output} Estimated decision rule $\hat{d}$
  \begin{algorithmic}[1]
    \For{$i=1$ {\bfseries to} $n$} 
    \State $\hat Y_i(x_i,s_i,a_i) \gets$ Model $E(Y |X, S, A=a)$ using $\mathcal{D}_n$ with $a=1$ and $a=-1$, respectively. 
    \If {$S$ is continuous}
    \State $ g_1(x_i) \gets $ Model $Q_{S|X,A} \{ E (Y|X,S,A=a) \}$ via quantile regressions of $\hat Y_i(x_i,s_i,a_i)$ on $x_i$, for $\mathcal{D}_n$ with $a=1$. 
    \State $ g_2(x_i) \gets $ Model $Q_{S|X,A} \{ E (Y|X,S,A=a) \}$ via quantile regressions of $\hat Y_i(x_i,s_i,a_i)$ on $x_i$, for $\mathcal{D}_n$ with $a=-1$. 
    \EndIf
    \If {$S$ is discrete}
    \State $ g_1(x_i) \gets $ Compute $\inf_{s \in \mathbb{S}} \{ \hat Y_i(x_i,s,a_i=1) \}$. 
    \State $ g_2(x_i) \gets $ Compute $\inf_{s \in \mathbb{S}} \{ \hat Y_i(x_i,s,a_i=-1) \}$. 
    \EndIf
    \EndFor
    \State $\hat{d} \gets$ Build a weighted classification model with features $x_i$, label $\operatorname{sgn}\{g_1(x_i)-g_2(x_i)\}$, and sample weight $|g_1(x_i)-g_2(x_i)|$ for $1\leq i \leq n$. 
    \State \textbf{Return} $\hat{d}$
  \end{algorithmic}
\end{algorithm}

\subsection{Extension to Multiple Sensitive Variables}\label{sec:exten}

The extension from $S$ being a single continuous variable to multiple continuous variables is straightforward in  Algorithm~\ref{algo:code}.
For multiple discrete sensitive variables, similar estimation procedure can be conducted as outlined in Section~\ref{sec:diff_g}.  Suppose there are $L$ discrete sensitive variables, i.e., $\mathcal{S} = \{S_1, S_2, \ldots, S_L\}$. 
The inner expectation $E(Y|X,S_1,\ldots,S_L,A)$ can be obtained with a twin model of $Y$ on $X$ and all $\mathcal{S}$ for each treatment level. 
The infimum over $\mathbb{S}$ is obtained by finding the minimum iterating space of possible parameter values for each sensitive variable. 
See Section~\ref{ssec:app} for an example of using multiple discrete sensitive variables. 
We will discuss in Section~\ref{ssec:disc} the challenges and future work related to the scenario with a mixture of continuous and discrete sensitive variables and the identification of vulnerable subjects under these cases.

\section{Numerical Studies}\label{sec:numerical}

In this section, we perform extensive numerical experiments to investigate the merit of robustness of the proposed framework via simulations and three real-data applications.
The results demonstrate that the proposed rules achieve a robust objective with sensitive variables unavailable at the time of decision while maintaining comparable mean outcomes.

For comparison, we consider the naive and mean-optimal approaches described in Section~\ref{sec:assump}, which correspond to different choices of $G(\cdot)$ functions. The naive decision rule that simply disregard information of $S$, denoted as \textbf{Base}, can be formulated in our optimization framework of \eqref{eq:obj_max} by letting $G(X,A) = E (Y|X,A)$. The IDR can be estimated directly by fitting a model of $Y$ on $X$ in each treatment arm. 
The resulting IDR is not sensitive variables-aware and is biased due to confounding, as discussed. 
Another IDR that resembles traditional mean-optimal decision rules, denoted as \textbf{Exp}, can be formulated as $G(X,S,A) = E (Y|X,S,A)$. This can be obtained by training a classification model without $S$, i.e., only using $X$, after obtaining an outcome model for the inner expectation $E(Y|X,S,A)$. Note that this approach is not robust to extreme behaviors in $S$. 
The modeling approaches described in Appendix~\ref{suppl:tune} apply to here. 
We also include the \textit{double robust} \citep{Chernozhukov2018dml} versions of Base and Exp, respectively, by adapting Policytree \citep[PT,][]{sverdrup2020policytree,athey2021policy}, the latest state-of-the-art policy learning method for maximizing the expected values. The two new methods are termed \textbf{PT-Base} and \textbf{PT-Exp}.

We consider the following evaluation metrics. 
\textit{1) Objective:} the quantile objective is estimated and reported for a continuous $S$ and the infimum objective is for a discrete $S$. 
The objective, when $\tau < 0.5$, (here $\tau = 0.25$) represents the value of the ``low performers'' among all possible value of $S$ under a given $d$.
\textit{2) Value:} the value function, or expected reward used by the most existing methods, such as \cite{manski2004statistical,qian2011performance}, is defined as $V(d) = E\{Y(d)\}$. 
It represents the ``average performers''. 
For randomized trials, an unbiased estimator of $V(d)$ is given by $\hat{V}(d) = \{\sum_{i=1}^T Y_i {\mathbbm{1}}(A_i = \hat d(X_i)) / \pi(A_i,X_i) \} / \{\sum_{i=1}^T {\mathbbm{1}}(A_i = \hat d(X_i)) / \pi(A_i,X_i) \}$ \citep{murphy2001marginal}, where $T$ is the sample size of the test data and $\pi(A,X)$ is propensity score.  For observational studies, the value is estimated with $\hat{V}(d) = T^{-1}\sum_{i=1}^T \hat{Y_i}(x_i,s_i,a_i=\hat{d})$.
We report the metrics among all subjects and among the potential vulnerable subgroup, respectively. 
For simulation, we consider training data and testing data with sample sizes of 8,000 and 2,000, respectively. 
For real-data applications, we consider a 80-20 split of the dataset into a training data and a testing data. 
Continuous covariates are standardized before the estimation. 
All results are based on 100 replications.

\subsection{Simulation Studies}\label{ssec:simulation}

\textit{Example 1. } 
Here we provide the detail for the simulation of the motivating example introduced in Section~\ref{sec:intro-rise}. 
The outcome is generated using the following model: 
$Y_i = \mathbbm{1}(X_i>0.5)\{5 + 10\mathbbm{1}(A_i=1) + 22S_i - 24\mathbbm{1}(A_i=1) S_i\} + \mathbbm{1}(X_i\leq0.5)\{11 + 19\mathbbm{1}(A_i=1) + 2S_i - 32\mathbbm{1}(A_i=1) S_i\} + \epsilon_i$, where the covariate $X_i \sim Unif(0,1)$, treatment assignment $A_i \sim Bernoulli(0.5)$, and the noise $\epsilon_i \sim N(0,1)$. For a discrete type $S$, the sensitive variable $S_i \sim Bernoulli(0.5)$. For a continuous type $S$, $S_i$ is generated from a mixture of beta distributions, $Beta(4, 1)$ and $Beta(1, 4)$, with equal mixing proportions.

\textit{Example 2. } %($S$ as a sensitive variable). } 
We generate the outcome $Y$ using the following model: 
$Y_i = \{0.5 + \mathbbm{1}(A_i=1) + \exp(S_i) - 2.5S_i\mathbbm{1}(A_i=1)\}  \{1+X_{i1} -X_{i2} +X_{i3}^2 +\exp(X_{i4})\} + \{1 + 2\mathbbm{1}(A_i=1) + 0.2\exp(S_i) - 3.5S_i\mathbbm{1}(A_i=1)\} \{1+5X_{i1} -2X_{i2} +3X_{i3} +2\exp(X_{i4})\} + \epsilon$, 
where $X_{ij} \sim Unif(0,1), ~j=1,\ldots,6$, $A$ satisfies $\log \{P(A_i = 1 |X_i)/P(A = 0 |X_i)\} = 0.6(-S_i + X_{i1} - X_{i2} + X_{i3} - X_{i4} + X_{i5} - X_{i6})$, and $\epsilon_i \sim N(0,1)$. For a continuous type $S$, $S_i$ is generated from a mixture of beta distributions, $Beta(4,1)$ and $Beta(1,4)$, with equal mixing proportions;  
for a discrete type $S$, we consider a binary $S_i$ that satisfies $\log \{P(S_i = 1 |X_i)/P(S_i = 0 |X_i)\} = -2.5 + 0.8(X_{i1} + X_{i2} + X_{i3} + X_{i4} + X_{i5} + X_{i6})$.

Table~\ref{t:sim_toy_full} summarizes the performance of the proposed IDRs compared to the mean criterion for Example 1 and Example 2. The proposed RISE achieves the largest objectives and improves the value among vulnerable subjects, while maintaining comparative overall values. 
As for the objective, intuitively, the proposed rule is expected to achieve a larger objective than all other methods uniformly in $\mathbb{X}$. 
We also point out that there is no direct relationship between the objective among all subjects versus the objective among vulnerable subjects. For example, using the toy example with setup in Table~\ref{t:toy_setup}, and limiting to subjects with $X \leq 0.5$ only, $S=1$ is vulnerable and is assigned $A=-1$ by the proposed RISE. 
The objective among $S=1$ is 13 but the objective among both $S=0$ and $S=1$ is $12=(11+13)/2$, which is smaller than that among the vulnerable group. 
In other words, by protecting the vulnerable subjects, the proposed rule may lead to an increase in the outcome of the vulnerable group, and the gain may result in a higher outcome than the overall mean outcome. 
PT-Exp tends to show the best improvement in terms of the overall value, as the doubly robust-based estimators tend to reduce variance in value estimation. However, PT-Exp is shown to have minimal benefits for vulnerable subjects. RISE still shows the largest gain in the objective and value among vulnerable subjects among all compared methods.

In the appendix, we consider a continuous $S$ for different quantile criteria $\tau=0.1$ and $0.5$ to test the robustness of RISE. 
Results show that when $\tau$ is small, there is more strength in the proposed method, as the algorithm aims to improve the worst-outcome scenarios. The proposed RISE has the largest gain in objective and value among vulnerable subjects when $\tau$ is 0.1, and has similar performance as the compared approaches when $\tau$ is 0.5. 
We also consider a scenario where $S$ is not involved in the data generation of $Y$, i.e., Assumption~\ref{assump:rise-unconf} is simplified as $\{Y(-1), Y(1)\} \perp A | X$. 
The estimated objective and value function are similar across all compared approaches, which indicates the robustness of RISE. 
Finally, we study the performances of our method when Assumption~\ref{assump:rise-pos} is nearly violated or Assumption~\ref{assump:rise-unconf} is violated. Similar patterns have been observed that the proposed RISE achieves the largest objectives and improves the value among vulnerable subjects, while maintaining comparable overall values. 
The details can be found in Appendix~\ref{suppl:sim}.

\clearpage
\begin{table}[!htb]
\centering
\caption{Simulation results for Example 1 and Example 2. Standard error in parenthesis. }
\label{t:sim_toy_full}
\resizebox{0.99\columnwidth}{!}{
% \begin{scriptsize}
\begin{tabular}{@{}ccccccc@{}}
\toprule
Example & Type of $S$ & IDR & Obj. (all) & Obj. (vulnerable) & Value (all) & Value (vulnerable) \\ \midrule
\multirow{10}{*}{1} & \multirow{5}{*}{Disc.} 
     & Base & 7.03 (0.03) & 7.01 (0.04) & 14.3 (0.05) & 7.92 (0.06) \\
 &   & Exp  & 6.39 (0.03) & 6.39 (0.04) & {14.4} (0.05) & 7.14 (0.06) \\
 &    & {PT-Base}   & {2.66 (0.02)} & {2.65 (0.02)} & {15.4 (0.05)} & {2.58 (0.02)} \\
 &    & {PT-Exp}   & {2.62 (0.02)} & {2.62 (0.02)} & {\textbf{15.5} (0.05)} & {2.55 (0.02)} \\
 &   & RISE  & \textbf{12.0} (0.01) & \textbf{12.0} (0.01) & 13.0 (0.01) & \textbf{14.0} (0.01) \\
 \cmidrule(l){2-7}
 & \multirow{5}{*}{Cont.} 
     & Base & 9.12 (0.03) & 9.14 (0.04) & 14.5 (0.08) & 8.25 (0.11) \\
 &   & Exp  & 8.75 (0.03) & 8.75 (0.04) & {14.6} (0.08) & 7.58 (0.06) \\
 &    & {PT-Base}   & {6.71 (0.03)} & {6.72 (0.03)} & {15.3 (0.05)} & {4.52 (0.02)} \\
 &    & {PT-Exp}   & {6.68 (0.02)} & {6.67 (0.02)} & {\textbf{15.4} (0.05)} & {4.47 (0.02)} \\
 &   & RISE  & \textbf{12.2} (0.02) & \textbf{12.2} (0.03) & 13.0 (0.01) & \textbf{13.7} (0.01) \\ 
 \midrule
\multirow{10}{*}{2} & \multirow{5}{*}{Disc.} 
      & Base & 7.79 (0.02) & 8.66 (0.03) & 19.4 (0.04) & 11.4 (0.06) \\
 &    & Exp   & 9.12 (0.03) & 10.1 (0.03) & \textbf{19.5} (0.04) & 14.4 (0.05) \\
 &    & {PT-Base}   & {7.19 (0.03)} & {7.77 (0.03)} & {19.0 (0.05)} & {9.71 (0.05)} \\
 &    & {PT-Exp}   & {8.30 (0.02)} & {9.03 (0.03)} & {{19.1} (0.04)} & {12.2 (0.05)} \\
 &    & RISE   & \textbf{13.5} (0.01) & \textbf{14.0} (0.01) & 17.4 (0.02) & \textbf{22.1} (0.02) \\
 \cmidrule(l){2-7}
& \multirow{5}{*}{Cont.} 
      & Base & 9.89 (0.02) & 9.87 (0.03) & 17.6 (0.02) & 9.09 (0.04) \\
 &    & Exp   & 11.1 (0.02) & 11.1 (0.02) & {17.8} (0.02) & 12.2 (0.04) \\
 &    & {PT-Base}   & {9.30 (0.02)} & {9.29 (0.03)} & {18.0 (0.03)} & {7.61 (0.04)} \\
 &    & {PT-Exp}   & {9.41 (0.02)} & {9.41 (0.02)} & {\textbf{18.1} (0.02)} & {7.92 (0.04)} \\
 &    & RISE   & \textbf{14.1} (0.01) & \textbf{14.2} (0.02) & 17.0 (0.01) & \textbf{20.3} (0.03) \\ 
 \bottomrule
%  \vspace{-1cm}
\end{tabular}
% \end{scriptsize}
}
\end{table}

\clearpage
\subsection{Real-data Applications}\label{ssec:app} 

We present three real-data examples to showcase the robust performance of RISE. These applications consider either fairness or safety in the context of policy \citep{lalonde1986evaluating} and healthcare \citep{hammer1996trial,seymour2016assessment} where sensitive variables commonly exist.

\subsubsection{Fairness in a Job Training Program }
To illustrate the implication of the proposed method from a fairness perspective, we consider the National Supported Work (NSW) program \citep{lalonde1986evaluating} for improving personalized recommendations of a job training program on increasing incomes. This program intended to provide a 6 to 18-month training for individuals in face of economic and social problems such as former drug addicts and juvenile delinquents. 
The original experimental dataset consists of 185 individuals who received the job training program ($A = 1$) and 260 individuals who did not ($A = -1$). 
The baseline covariates are age, years of schooling, race (1 = African Americans or Hispanics, 0 = others), married (1 = yes, 0 = no), high school diploma (1 = yes, 0 = no), earning in 1974, and earning in 1975. 
The outcome variable is the earning in 1978. 
In the exploratory analysis using causal forest \citep{wager2018estimation}, we observe that age may play an important role in the causal effect of the job training program on the long-term post-market earning. In the following data example we use age as the sensitive variable $S$ and other baseline covariates as $X$. 
The earnings in years 1974, 1975, and 1978 are transformed by taking the logarithm of the earning plus one.

\subsubsection{Improvement of HIV Treatment }
To illustrate the implication of the proposed method from a safety perspective when there is delayed information, we consider the ACTG175 dataset among HIV positive patients \citep{hammer1996trial}. 
The original study considers a total of 2,139 patients who were randomly assigned into four treatment groups. In this data application, we focus on finding the optimal IDRs between two treatments: zidovudine combined with didanosine ($A=-1$) and zidovudine combined with zalcitabine ($A=1$). The total number of patients receiving these two treatments is 1,046. 
The baseline covariates we consider are 
age, weight, CD4 T-cell amount at baseline, hemophilia (1 = yes, 0 = no), homosexual activity (1 = yes, 0 = no), Karnofsky score, history of intravenous drug use (1 = yes, 0 = no), gender (1 = male, 0 = female), CD8 T-cell amount at baseline, race (1 = non-Caucasian, 0 = Caucasian), number of days of previously received antiretroviral therapy, 
use of zidovudine in the 30 days prior to treatment initiation (1 = yes, 0 = no), and
symptomatic indicator (1 = symptomatic, 0 = asymptomatic). 
The outcome variable is the CD4 T-cell amount at $96\pm5$ weeks from the baseline. 
We consider CD8 T-cell amount at baseline as the sensitive variable. 
The response of CD8 T-cell among HIV positive patients has not been fully understood \citep{boppana2018understanding}. 
Clinically, it is plausible that only CD4 is measured in clinical visits where treatments are based on, hence CD8 might not be measured and not used in decision making. 
As our exploratory analysis using causal forest shows, CD8 T-cell amount may play an important part in the treatment effect of the outcome.

\subsubsection{Safe Resuscitation for Patients with Sepsis }

For this application, we apply the proposed method to treating sepsis, a life-threatening disease. 
This application intends to provide an example to apply our method with multiple categorical sensitive variables in the scenario where there is missing yet important information at the time of decision making. 
We apply the proposed method to a sepsis study from the University of Pittsburgh Medical Center (UPMC). 
The original study cohort includes 30,687 patients with Sepsis-3 \citep{seymour2016assessment} within 6 hours of hospital arrival from 14 UPMC hospitals between 2013 and 2017. 
For our data analysis, we consider $X$ to be baseline patient characteristics 4 hours before sepsis onset, which includes patient demographics of age, gender (1 = male, 0 = female), race (1 = Caucasian, 0 = others), and weight, and vital signs of usage of mechanical ventilation (1 = yes, 0 = no), respiratory rate, temperature, intravenous fluids (1 = yes, 0 = no), Glasgow Coma Scale score, platelets, blood urea nitrogen, white blood cell counts, glucose, creatinine. 
% 'age', 'mechvent', 'rr', 'temp', 'fluid', 'race', 'gcs', 'plt', 'bun', 'wbc', 'gluc', 'creat', 'weight', 'gender
We consider two sensitive variables, lactate and Sequential Organ Failure Assessment (SOFA) score 4 hours before sepsis onset. 
Lactate and SOFA score have been two important indicators of sepsis severity \citep{howell2007occult, krishna2009evaluation,shankar2016developing,singh2022acute}. 
Different from the baseline patient demographics or common vital signs that are typically obtained at the admission of patients, SOFA score combines performance of several organ systems in the body \citep{seymour2016assessment}, which requires additional calculation and cannot be obtained directly. Lactate labs measures the level of lactic acid in the blood \citep{andersen2013etiology} and are less common in routine examination, which could be delayed in ordering. 
Hence, their measurements are obtained retrospectively after treatment decisions have been made and are not available at times of decision. 
We dichotomize lactate level at clinically meaningful value of 2 mmol/L \citep{shankar2016developing}, and SOFA score at value of 6 for analysis \citep{vincent1996sofa,ferreira2001serial}. 
The treatment option is whether the patient took any vasopressors during the first 24 hours after sepsis onset. The outcome is patient survival at day 90. The analysis cohort contains 6,539 patients in total. 
We are interested in making decision about whether to treat patients with vasopressors in the first 24 hours after sepsis onset given the measurements of lactate and SOFA are not available at the time of decision making. 
Additional rationale and background on this example is provided in Appendix~\ref{suppl:real}.

\subsubsection{Results of Real-data Applications } 
Table~\ref{t:res_data} presents the performance of various IDRs on the three applications. As expected, RISE has the largest objective as well as value among vulnerable subjects. The patterns are similar to that in the synthetic experiments in Section~\ref{ssec:simulation}. 
In applications to the job training data and the sepsis study, results show that RISE has a larger value among all subjects than other IDRs. This is possible when there are more gains in the vulnerable subjects than other subjects, which further demonstrate the superiority of the proposed approach in improving worst-case outcomes caused by sensitive variables. 
We provide visualizations by Shapley additive explanations (SHAP) \citep{lundberg2017unified} for RISE and Exp, respectively, in Appendix~\ref{suppl:real} about feature importance in the final classification models to help interpret important covariates in making the decisions. 
The SHAP approach provides united values to describe the correlation between each feature and the predicted decision rule, respectively \citep{lundberg2017unified}. 
Overall, the direction of correlations is similar for RISE and Exp, but their
ranking of feature importance may be different.

\begin{table}[!htb]
\centering
\caption{Estimated objective and value of different IDRs for the three data applications. Standard error in parenthesis. The outcome of each study is italicized. }
\label{t:res_data}
\resizebox{0.99\columnwidth}{!}{
% \begin{scriptsize}
\begin{tabular}{@{}cccccc@{}}
\toprule
Dataset & IDR & Obj. (all) & Obj. (vulnerable) & Value (all) & Value (vulnerable) \\ \midrule
\multirow{5}{*}{\begin{tabular}[c]{@{}c@{}}NSW\\ \textit{log(income+1)}\end{tabular}} 
 & Base  & 5.26 (0.04) & 5.28 (0.05) & 6.32 (0.05) & 6.33 (0.07) \\
 & Exp   & 5.22 (0.04) & 5.24 (0.05) & 6.37 (0.05) & 6.37 (0.07) \\
 & {PT-Base}   & {4.97 (0.04)} & {5.08 (0.06)} & {6.40 (0.03)} & {6.38 (0.05)} \\
 & {PT-Exp}   & {5.03 (0.04)} & {5.11 (0.05)} & {\textbf{6.43} (0.03)} & {6.40 (0.05)} \\
 & RISE  & \textbf{5.43} (0.04) & \textbf{5.44} (0.04) & {6.42} (0.04) & \textbf{6.42} (0.06) \\ 
 \midrule
\multirow{5}{*}{\begin{tabular}[c]{@{}c@{}}ACTG175\\ \textit{CD4 T-cell amount}\end{tabular}} 
 & Base  & 336.9 (1.65) & 338.1 (2.23) & 353.5 (1.86) & 357.5 (2.24) \\ 
 & Exp   & 337.5 (1.65) & 338.9 (1.80) & {355.9} (1.95) & 359.1 (2.21) \\ 
  & {PT-Base}   & {299.7 (1.01)} & {299.5 (1.91)} & {356.9 (1.72)} & {350.7 (2.54)} \\
 & {PT-Exp}   & {300.1 (0.99)} & {299.9 (1.83)} & {\textbf{357.1} (1.55)} & {352.7 (2.61)} \\
 & RISE  & \textbf{351.5} (1.67) & \textbf{351.2} (1.80) & 351.8 (1.88) & \textbf{363.1} (2.19) \\
 \midrule
\multirow{5}{*}{\begin{tabular}[c]{@{}c@{}}Sepsis\\ \textit{survival rate}\end{tabular}} 
 & Base  & 0.803 (0.001) & 0.822 (0.001) & 0.965 (0.001) & 0.905 (0.002) \\ 
 & Exp   & 0.803 (0.001) & 0.822 (0.002) & 0.966 (0.001) & 0.908 (0.002) \\ 
& {PT-Base}   & {0.758 (0.001)} & {0.771 (0.002)} & {0.981 (0.001)} & {0.848 (0.003)} \\
 & {PT-Exp}   & {0.758 (0.001)} & {0.772 (0.002)} & {\textbf{0.984} (0.001)} & {0.875 (0.003)} \\
 & RISE  & \textbf{0.836} (0.001) & \textbf{0.833} (0.001) & {0.972} (0.001) & \textbf{0.923} (0.002) \\ 
 \bottomrule
  \vspace{0.1cm}
\end{tabular}
% \end{scriptsize}
}
\end{table}

\section{Discussion and Future Work} \label{ssec:disc}

We have proposed RISE, a robust decision learning framework with a novel quantile- or infimum-optimal treatment objective intended to improve the worst-case scenarios of individuals when decisions with uncertainty need to be made, but with sensitive yet important information missing. 
Our approach can be applied to a broad range of applications, including but not limited to policy, education, healthcare, etc.   
For a mixture of continuous and discrete sensitive variables, the estimated rule can be obtained by first taking the infimum over the discrete ones as in Section~\ref{sec:exten}, then obtaining the quantile over the continuous ones. However, challenges remain in finding the vulnerable subjects described in Section~\ref{ssec:vul} under these settings as it may be computationally difficult to find a vulnerable set of $S$ when it is multi-dimensional. 
Another future work includes the extension of the current binary treatment option to a multi-treatment option. 
It is also worth mentioning that our work can be naturally extended to the scenario where there exist unmeasured confounders. As long as the conditional average outcome given observed covariates can be identified (via instrumental variables such as \cite{wang2018bounded} or negative control variables such as \cite{qi2021proximal}), our method can be applied.

\chapter{Summary and Future Work}

This dissertation is motivated by challenges in causal inference under practical data restrictions. 
% traditional
Traditional randomized clinical trials typically require long years of data collection, which leads to loss of follow-up, poor compliance, or other issues and subsequently affects the estimation of treatment effects. 
In the more recent neoadjuvant trials, on the other hand, the efficacy of a treatment can be estimated early with an intermediate post-treatment response. However, the clinical implication of this intermediate post-treatment response has not yet been understood. 
This idea, along with real data from a neoadjuvant clinical trial, has motivated the development of methods in Chapter~\ref{chap:ps}. 
% modern
In the modern context, new challenges arise with growing concerns such as data privacy and operational feasibility in distributed research networks, and the timeliness and fairness of individualized decision rules. 
Driven by these concerns, we develop the privacy-protecting method for improving the estimation of conditional average treatment effects in Chapter~\ref{chap:hetero} and the fairness-aware decision learning framework in Chapter~\ref{chap:itr} with board applications, including but not limited to politics, education, healthcare, etc. 
The three proposed methods in this dissertation respectively address important challenges in causal inference, which includes: identification of principal stratum treatment effects, enhancement of treatment effect estimation via heterogeneous data integration, and derivation of robust individualized decision rules. 
Each of the methods can be further improved and extended along in their own framework and settings as have been discussed in each of the chapters.

One promising direction for future research is to generalize the toolbox of privacy-protecting analytic and data-sharing methods and to construct a unified framework that is applicable to a broader range of problems pertaining to modern distributed data networks. 
It is of great interest to borrow the strengths of each of the proposed methods, Chapter~\ref{chap:hetero} in particular, pairing with new methodologies developed by others to develop, test, and distribute open-source statistical software packages, maximizing value of methodology research to practical applications.

Driven by the recent initiatives of collaboratory distributed research networks, another promising direction and natural extension of current work, Chapter~\ref{chap:itr} in particular, is to develop a generalizable recommendation system for treatment that is robust to population heterogeneity across multiple sites. 
A more robust treatment recommendation system that jointly take into account the population heterogeneity due to observed and/or unobserved confounding can enhance treatment gain from across all sites, hence more general and widely applicable.

Last but not least, it is of great interest to investigate the unconfoundedness assumption that is typically adopted in causal inference for observational studies, as mentioned in Chapter~\ref{chap:hetero} and Chapter~\ref{chap:itr}. 
Unconfoundedness is a strong and untestable assumption for observational studies where unmeasured confounding could exist. 
It is therefore important to stress how this assumption breaks down when there is unobserved confounding. 
For example, little has been understood about how bad things could get and when do things cancel out in cases of assumption violations. 
We would like to address these important issues by leveraging the knowledge accumulated through the development of this dissertation.

% This is a very complicated topic and we shall discuss it in our next paper.%\cite{DUMMY:11}
% \footnote{Test}

% \chapter{Conclusions}

%==========================================================================================%
% APPENDIX
%==========================================================================================%
\appendix
%After this command, chapters will be formatted as appendices. 
\chapter{for Chapter 2}

\section{Estimation of \texorpdfstring{$\Pr \{S_i(0) = 1|X_i=x\}$}{Pr\{Si(0)=1|Xi=x\}} and \texorpdfstring{$\Pr \{S_i(1) = 1|X_i=x\}$}{Pr\{Si(1)=1|Xi=x\}}  } \label{supplA}

We use the maximum likelihood approach to estimate $\Pr \{S_i(0) = 0|X_i=x\}$, $\Pr \{S_i(0) = 1|X_i=x\}$ and $\Pr \{S_i(1) = 1|X_i=x\}$. Let 
\begin{equation*}
    E_{jkx} = \{ i: S_i(0)=j, S_i(1)=k|X_i=x \},~~ j,k=0,1, x \in \Gamma
\end{equation*}
be the principal stratum under each category $X=x$. Because of the monotonicity assumption, $E_{10x}$ is empty. Let
\begin{equation*}
    p_{jkx} = \Pr\{E_{jkx}\} = \Pr\{S_i(0)=j, S_i(1)=k|X_i=x\},~~ j,k=0,1, x \in \Gamma
\end{equation*}
Therefore, ${p}_{00x}+p_{01x}+{p}_{11x}=1$ for all $x\in\Gamma$. For each $x$, $\Pr\{E_{jkx}\}$ can be estimated from the observed data $\{Z_i,X_i,S_i(Z_i),i=1,2,\ldots,n\}$ via maximum likelihood.  
Let $N_{zsx}$ be the total number of subjects with $Z=z, S(Z)=s$ and baseline category $x$ with $\sum_{Z;S=0,1;X}N_{zsx}=n$. Then the likelihood function for $(p_{00x}, p_{01x}, p_{11x})$ is given by
\begin{align*}
    L&(p_{00x}, p_{01x}, p_{11x}|N_{00x},N_{01x},N_{10x},N_{11x}) 
    \propto f(N_{zsx}) \\
    & \propto \Pr\{S(0)=0|X=x\}^{N_{00x}} \cdot \Pr\{S(0)=1|X=x\}^{N_{01x}} \\
    & \quad \cdot \Pr\{S(1)=0|X=x\}^{N_{10x}} \cdot \Pr\{S(1)=1|X=x\}^{N_{11x}} \\
    &= (p_{00x}+p_{01x})^{N_{00x}} \cdot p_{11x}^{N_{01x}} \cdot p_{00x}^{N_{10x}} \cdot (p_{01x}+p_{11x})^{N_{11x}} 
~~\mbox{(by monotonicity assumption)}\\
    &= (1-p_{11x})^{N_{00x}} \cdot p_{11x}^{N_{01x}} \cdot p_{00x}^{N_{10x}} \cdot (1-p_{00x})^{N_{11x}} \\
    &= (1-p_{11x})^{N_{00x}} \cdot p_{11x}^{N_{01x}} \cdot (1-p_{+1x})^{N_{10x}} \cdot p_{+1x}^{N_{11x}}
\end{align*}

1) When $N_{00x} \cdot N_{11x} \geq N_{01x} \cdot N_{10x}$, the resulting MLEs for $(p_{00x}, p_{01x}, p_{11x})$ are given by
\begin{align*}
    &\widehat{p}_{00x} = \widehat{\Pr}\{S_i(0) = 0, S_i(1) = 0|X_i = x\} = 1 - \widehat{p}_{+1x} \\
    &\quad ~~ = \frac{N_{10x}}{N_{10x}+N_{11x}} 
    = \frac{\sum_i \mathbbm{1}(Z_i = 1, S_i(1) = 0, X_i = x)}{\sum_i \mathbbm{1}(Z_i = 1,X_i = x)}
\end{align*}
\begin{align*}
    &\widehat{p}_{11x} = \widehat{Pr}\{S_i(0) = 1, S_i(1) = 1|X_i = x\} \\
    &\quad ~~ = \frac{N_{01x}}{N_{00x}+N_{01x}}  
    = \frac{\sum_i \mathbbm{1}(Z_i = 0, S_i(0) = 1, X_i = x)}{\sum_i \mathbbm{1}(Z_i = 0,X_i = x)}
\end{align*}
\begin{align*}
    \widehat{p}_{01x} = \widehat{Pr}\{S_i(0) = 0, S_i(1) = 1|X_i = x\} = 1 - \widehat{p}_{00x} - \widehat{p}_{11x}
\end{align*}

Obviously for each $x\in \Gamma$, $\widehat{p}_{00x}$ is the proportion of non-respondents in the treatment arm with $X=x$; 
$\widehat{p}_{11x}$ is the proportion of respondents in the control arm with $X=x$.

2) When $N_{00x} \cdot N_{11x} < N_{01x} \cdot N_{10x}$, $\widehat{p}_{11x}=\widehat{p}_{+1x}$. The likelihood function is given by
\begingroup
\allowdisplaybreaks
\begin{align*}
    &L(p_{00x}, p_{01x}, p_{11x}|N_{00x},N_{01x},N_{10x},N_{11x}) \\
    &= (1-p_{11x})^{N_{00x}} \cdot p_{11x}^{N_{01x}} \cdot (1-p_{11x})^{N_{10x}} \cdot p_{11x}^{N_{11x}} \\
    &= (1-p_{11x})^{N_{00x}+N_{10x}} \cdot p_{11x}^{N_{01x}+N_{11x}}
\end{align*}
\endgroup

The resulting MLEs for $(p_{00x}, p_{01x}, p_{11x})$ are given by
\begingroup
\allowdisplaybreaks
\begin{align*}
    &\widehat{p}_{01x} = \widehat{Pr}\{S_i(0) = 0, S_i(1) = 1|X_i = x\} = 0 \\
    &\widehat{p}_{00x} = \widehat{\Pr}\{S_i(0) = 0, S_i(1) = 0|X_i = x\} 
    = \frac{N_{+0x}}{N_{++x}}
    = \frac{\sum_i \mathbbm{1}(S_i = 0, X_i = x)}{\sum_i \mathbbm{1}(X_i = x)}  \\
    &\widehat{p}_{11x} = \widehat{Pr}\{S_i(0) = 1, S_i(1) = 1|X_i = x\} 
    = \frac{N_{+1x}}{N_{++x}}
    = \frac{\sum_i \mathbbm{1}(S_i = 1, X_i = x)}{\sum_i \mathbbm{1}(X_i = x)}
\end{align*}
\endgroup

Then $\widehat{p}_{00x}$ is the proportion of non-respondents among all subjects with $X=x$; $\widehat{p}_{11x}$ is the proportion of respondents among all subjects with $X=x$.

% \section{Proof of Consistency of \texorpdfstring{$\widehat{\bbeta}$}{TEXT} and \texorpdfstring{$\widehat{\btheta}$}{TEXT}} \label{supplB}
\section{Proofs of Consistency of Model Parameters and Causal Estimands} \label{supplB}

Here we show our estimator $\widehat{\bbeta}$ is a consistent estimator for $\bbeta$. We first show that $\widehat{\bbeta}$ can be considered as an extremum estimator as defined by \citet{hayashi2000econometrics}. Then we prove that the conditions set forth by \citet{hayashi2000econometrics} for consistency of an extremum estimator are satisfied by our estimator. Then by Slutsky's theorem, the causal estimand $\widehat{\theta}$ is a consistent estimator for $\theta$.

\begin{definition}[Extremum Estimator]
\textit{An estimator $\widehat{\eta}$ is an extremum estimator if there is a function $Q_n(\eta)$ such that}
\citep{hayashi2000econometrics}
\begin{align*}
    \widehat{\eta} = \operatorname*{\arg\,max}_{\eta} Q_n(\eta); ~~\eta \in H .
\end{align*}
\end{definition}

One example of an extremum estimator is the maximum likelihood estimator where
\begin{align*}
    Q_n(\eta) = \prod_{i=1}^n f(x_i|\eta).
\end{align*}
Here we minimize the objective function,
\begin{align*}
    Q_n(\bbeta) &= \sum_{x=0}^K Q_n^{(x)}(\bbeta) \\
    &= \sum_{x=0}^K \{\widehat{G}_L(x) - \sum_{y=0}^1 G_M(x,y;\bbeta) \cdot \widehat{G}_R(x,y)\}^2; \, x \in \Gamma
\end{align*}
which is equivalent to maximizing $-Q_n(\bbeta)$. Therefore $\widehat{\bbeta}$ is an extremum estimator.

Let
\begin{align*}
    Q_0(\bbeta) &= \sum_{x=0}^K Q_0^{(x)}(\bbeta); \, x \in \Gamma
\end{align*}
where $Q_0^{(x)}(\bbeta)= \{{G}_L(x) - \displaystyle{\sum_{y=0}^1} G_M(x,y;\bbeta) \cdot {G}_R(x,y)\}^2$. We present sufficient conditions for the existence of a unique local minimizer of $Q_0(\bbeta)$ in Lemma~\ref{lemma2}. 

\begin{lemma}\label{lemma2}
There exists a unique local minimizer $\bbeta_0$ for $Q_0(\bbeta)$ if:
\begin{enumerate}
    \item[(a)] $Q_0^{(x)}(\bbeta_0)=0$, $\forall x\in \Gamma=\{0,1,2,\ldots,K\}$.
    \item[(b)] $\rank \bigg |\displaystyle{\frac{\partial \tilde{Q_0}(\bbeta)}{\partial \bbeta}} \bigg |_{\bbeta=\bbeta_0} \geq \dim(\bbeta)$ where $\tilde{Q_0}(\bbeta)=\{Q_0^{(0)}(\bbeta), Q_0^{(1)}(\bbeta), \dots, Q_0^{(K)}(\bbeta)\}^T$.
\end{enumerate}
\end{lemma}

\begin{proof}
From (a) we have that $\bbeta_0$ minimizes $Q_0(\bbeta)$ since $Q_0(\bbeta) \geq 0$, $\forall \bbeta$ and $Q_0(\bbeta_0)=0$.

Then from (b) and the Implicit Function Theorem, there exists a unique function $g\{\boldsymbol{G_L(x)}$, $\boldsymbol{G_R(x, y)}\}$ such that $g\{\boldsymbol{G_L(x)}$, $\boldsymbol{G_R(x, y)}\}=\bbeta_0$, in the neighborhood of $\{\boldsymbol{G_L(x)}$, $\boldsymbol{G_R(x, y)}\}$ where $\{\boldsymbol{G_L(x)}$, $\boldsymbol{G_R(x, y)}\} = [G_L(x), G_R(x, y); x \in \{0,1, \ldots, K\}, y=0, 1]$. 
Thus, $\bbeta_0$ is a unique local minimizer for $Q_0(\bbeta)$. 
% This completes the proof of Lemma \ref{lemma2}.
\end{proof}

The proof of Theorem \ref{Theorem} is given as below.
\begin{proof}
From Proposition 7.1 in \citep{hayashi2000econometrics}: an extremum estimator $\widehat{\eta}$ is a consistent estimator for $\eta$ if there is a function $Q_0(\eta)$ satisfying the following two conditions:
\begin{enumerate}
    \item[(I)] Identification: $Q_0(\eta)$ is uniquely maximized on $H$ at $\eta_0 \in H$.
    \item[(II)] Uniform convergence: $Q_n(\cdot)$ converges uniformly in probability to $Q_0(\cdot)$.
\end{enumerate}

The condition (I) is satisfied according to Lemma \ref{lemma2}. To show that the condition (II) is satisfied here, let
\begingroup
\allowdisplaybreaks
\begin{align*}
    \begin{split}
        Q_n(\bbeta) &=\sum_{x=0}^K Q_n^{(x)}(\bbeta)^2 \\
        &= \sum_{x=0}^K \{\widehat{G}_L(x) - \sum_{y=0}^1 G_M(x,y;\bbeta) \cdot \widehat{G}_R(x,y)\}^2
    \end{split} \\
    \begin{split}
        Q_0(\bbeta) &=\sum_{x=0}^K Q_0^{(x)}(\bbeta)^2 \\
        &= \sum_{x=0}^K \{{G}_L(x) - \sum_{y=0}^1 G_M(x,y;\bbeta) \cdot {G}_R(x,y)\}^2. 
    \end{split}
\end{align*}
\endgroup

From
\begin{align*}
    |Q_n(\bbeta) - Q_0(\bbeta)|
    &=|\sum_{x=0}^K Q_n^{(x)}(\bbeta)^2 - \sum_{x=0}^K Q_0^{(x)}(\bbeta)^2| \\
    &\leq \sum_{x=0}^K |Q_n^{(x)}(\bbeta)^2 - Q_0^{(x)}(\bbeta)^2| \\
    &= \sum_{x=0}^K |Q_n^{(x)}(\bbeta) - Q_0^{(x)}(\bbeta)| \cdot |Q_n^{(x)}(\bbeta) + Q_0^{(x)}(\bbeta)| \\
    &\leq \sum_{x=0}^K 2 \cdot |Q_n^{(x)}(\bbeta) - Q_0^{(x)}(\bbeta)|, \quad x \in \Gamma
\end{align*}
because $0 \leq |Q_n^{(x)}(\bbeta)| \leq 1$ and $0 \leq |Q_0^{(x)}(\bbeta)| \leq 1$, each of which is a difference of two probability estimates. 

Therefore,
\begin{align}
    &|Q_n(\bbeta) - Q_0(\bbeta)| \nonumber \\
    &\leq \sum_{x=0}^K 2 \cdot \big \{ |\widehat{G}_L(x)-G_L(x)| + \sum_{y=0}^1 G_M(x, y; \bbeta) \cdot |\widehat{G}_R(x,y)-G_R(x,y)| \big \} \nonumber \\
    &\leq \sum_{x=0}^K 2 \cdot \big \{ |\widehat{G}_L(x)-G_L(x)| + \sum_{y=0}^1 |\widehat{G}_R(x,y)-G_R(x,y)| \big \} \label{inequal}
\end{align}
because $G_M(x, y; \bbeta)$ is a probability bounded between 0 and 1.

Since $\widehat{G}_L(x)$ and  $\widehat{G}_R(x, y)$ are either sample proportions or their ratios, 
\begin{gather*}
    \widehat{G}_L(x) \overset{p}{\to} G_L(x), \text{ as } n \to \infty \\
    \widehat{G}_R(x, y) \overset{p}{\to} G_R(x, y), \text{ as } n \to \infty
\end{gather*}

As $\widehat{G}_L(x)$ and  $\widehat{G}_R(x, y)$ do not involve $\bbeta$, from (\ref{inequal}) we have 
\begin{gather*}
    Q_n(\bbeta) \overset{p}{\Longrightarrow} Q_0(\bbeta), \text{ as } n \to \infty
\end{gather*}
where $\overset{p}{\Longrightarrow}$ denotes uniform convergence in probability.  This confirms condition  (II) and completes the proof of $\widehat{\bbeta}\overset{p}{\to} \bbeta$ as $n \to \infty$.

Because the causal estimate $\widehat{\theta}$ is a continuously differentiable function of $\widehat{\bbeta}$ and relevant sample proportions, by Slutsky's theorem,  $\widehat{\theta} \overset{p}{\to} \theta$  as $n \to \infty$. 
% This completes the proof of Theorem \ref{Theorem}.
\end{proof}

\section{Calculation of True Principal Stratum Causal Effects} \label{supplC}

For the simulated data, the true average causal effect for principal stratum $S_i(1) = 1$ can be calculated by
\begin{align*}
    \mathbb{E}\{Y_i(1)-Y_i(0)|S_i(1) = 1\} &= \mathbb{E}\{Y_i(1)=1|S_i(1) = 1\} - \mathbb{E}\{Y_i(0)=1|S_i(1) = 1\} \\
    &= \frac{\Pr \{Y_i(1)=1, S_i(1)=1\} - \Pr \{ Y_i(0)=1, S_i(1)=1 \}}{\Pr \{ S_i(1)=1 \}}
\end{align*}
where
\begin{align*}
    \Pr \{ S_i(1)=1 \} &= \sum_x \Big \{ \Pr \{S_i(0) = 1|X_i = x\} \cdot \Pr \{X_i = x\} \\ 
        & + \sum_{y} \big [ \Pr \{X_i=x\} \cdot \Pr \{ S_i(0)=0|X_i=x \} \\
        & \cdot \Pr \{ Y_i(0)=y|S_i(0)=0, X_i=x \} \\
        & \cdot \Pr \{ S_i(1)=1|S_i(0)=0, Y_i(0)=y, X_i=x \} \big ] \Big \}
\end{align*}
\begin{align*}
    \Pr \{ Y_i(0)=1, S_i(1)=1 \} &= \sum_x \Big [ \Pr \{X_i=x\} \cdot \Pr \{ S_i(0)=1|X_i=x \} \\
        & \cdot \Pr \{ Y_i(0)=1|S_i(0)=1, X_i=x \} \\
        & + \Pr \{ X_i=x \} \cdot \Pr \{ S_i(0)=0|X_i=x \} \\
        & \cdot \Pr \{ Y_i(0)=1|S_i(0)=0, X_i=x \} \\
        & \cdot \Pr \{ S_i(1)=1|S_i(0)=0, Y_i(0)=1, X_i=x \} \Big ]
\end{align*}
\begin{align*}
    \Pr \{ Y_i(1)=1, S_i(1)=1 \} &= \sum_x \sum_{y} \Big [ \Pr \{X_i=x\} \cdot \Pr \{ S_i(0)=1|X_i=x \} \\
        & \cdot \Pr \{ Y_i(0)=y|S_i(0)=1, X_i=x \} \\
        & \cdot \Pr \{ Y_i(1)=1|Y_i(0)=y, S_i(0)=1, X_i=x \} \\
        & + \Pr\{ X_i=x \} \cdot \Pr \{ S_i(0)=0|X_i=x \} \\
        & \cdot \Pr \{ Y_i(0)=y|S_i(0)=0, X_i=x \} \\
        & \cdot \Pr \{ S_i(1)=1|S_i(0)=0, Y_i(0)=y, X_i=x \} \\
        & \cdot \Pr \{ Y_i(1)=1|S_i(0)=0, S_i(1)=1, Y_i(0)=y, X_i=x \} \Big ]
\end{align*}

\chapter{for Chapter 3}

\section{Related Topics and Distinctions}
\label{suppl-related}

In Section~\ref{sec:related}, we focused on the literature review of model averaging for ease of exposition, because the most innovated part of our method is motivated directly from this class of work. 
Here we clarify the main differences among model averaging, meta-analysis, federated learning, as well as super learner. 

Model averaging: a convex averaging of models via model-specific weights \citep{raftery1997bayesian,yang2001adaptive, dai2011greedy,yao2018using,dai2018bayesian}. The extension of the weights from scalars to functions provides the best motivation for our approach. 

Meta-analysis: classic in the way that it describes the site-level heterogeneity using modeling assumptions \citep{whitehead2002meta,sutton2000methods}, rather than a more data-driven approach such as tree models. It can be either frequentist or Bayesian, the latter of which tends to be more useful under limited sample sizes. However, its main interest is typically the overall effect rather than the site-level heterogeneity, which is usually modeled by a nuisance parameter \citep{borenstein2011introduction,riley2011interpretation,tan2018changepoint,rover2020dynamically}. 

Federated learning: originated from the field of computer science \citep{mcmahan2017communication}, federated learning is a collaborative learning procedure that ensures data privacy by exchanging model parameters only. Federated learning methods often involves iterative updating \citep{fallah2020personalized,cho2021personalized,smith2017federated,yang2019federated}, rather than a one-shot procedure, which could be hard to apply to nonautomated distributed research networks.  
It has been developed mainly to estimate a global prediction model by leveraging distributed data \citep{li2020federated,kairouz2019advances,zhao2018federated,hard2018federated}, 
and is not designed to target any specific site. 

Super learner: an ensemble of multiple statistical and machine learning models \citep{van2007super}. It learns an optimal weighted average of those candidate models by minimizing the cross-validated risk, and assigns higher weights to more accurate models \citep{polley2010super}.  The final prediction on an independent testing data is the weighted combination of the predictions of those models. 
Super learner has been showed empirically to improve treatment effect estimation via the modeling of propensity score in observational studies \citep{pirracchio2015improving,wyss2018using,shortreed2019challenges,ju2019propensity,tan2022doubly}. 

Mixture of experts: an ensemble learning technique that decomposes a task into multiple subtasks with domain knowledge, followed by using multiple expert models to handle each subtask. A gating model is then used to decide which expert to use to make future prediction \citep{masoudnia2014mixture}. It differs from other ensemble methods typically in that often only a few experts will be selected for predictions, rather than combining results from all experts \citep{masoudnia2014mixture}.

\section{Proof of Theorem~\ref{Theorem}}\label{suppl-sec:consist}

The proof of Theorem~\ref{Theorem} closely follows arguments given in \citet{wager2018estimation}. Suppose the subsamples for building each tree in an ensemble forest are drawn from different subjects in the augmented site 1 data. Specifically, in one round of EF, we draw $m$ samples from the augmented data, where $m$ is less than the rows in the augmented data, i.e., $m < (n_1 \cdot K)$. By randomly picking $m$ unique subjects from site 1 and then randomly picking a site indicator $k$ out of $K$ sites for each of the $m$ subjects. The resulted $m$ subsamples should not be from the same subject and are hence independent and identically distributed. 
As long as $m < n_1$, we can ensure that all the subsamples are independent. 
In practice, when the ratio of $n_1 / K$ is relatively large, the probability of obtaining samples from the same subject is small. 

Assume that subject features $\bX_i$ and the site indicator $S_i$ are independent and have a density that is bounded away from 0 and infinity. Suppose moreover that the conditional mean function $\mathbb{E}[ \Tau|\bX=\bx,S=k]$ is Lipschitz continuous. We adopt the honesty definition in \citet{athey2016recursive} when building trees in a random forest. Honest approaches separate the training sample into two halves, one half for building the tree model, and another half for estimating treatment effects within the leaves \citep{athey2016recursive}. Following Definitions 1-5 and Theorem 3.1 in \citet{wager2018estimation}, the proposed estimator $\widehat \Tau_{\text{EF}}(\bx, 1)$ is a consistent estimator of the true treatment effect function $\tau_1(\bx)$ for site 1.

\section{Additional Simulation Results}
\label{suppl-sec:sim}

\subsection{Connection to Supervised Learning} 
Similar to ET-oracle and EF-oracle whose weights are built on the ground truth CATE functions $\tau_k$'s, we also consider for EWMA and STACK under a similar hypothetical setting.
Specifically, we assume the true $\tau_1$ is known and use it to compute the weights. 
% EWMA
This version of EWMA estimator is denoted as EWMA-oracle and its weight is given by $${\omega}_{k}^{\text{EWMA-oracle}} =\frac{\exp\{- \sum_{i \in \mathcal{I}_1^{(2)}}(\widehat\tau_k(\bx_i) - \tau_1(\bx_i))^2\} }{ \sum_{\ell=1}^{K} \exp\{- \sum_{i \in \mathcal{I}_1^{(2)}}(\widehat\tau_\ell(\bx_i) - \tau_1(\bx_i))^2\} }.$$
% STACK
Similarly, the corresponding linear stacking approach, denoted as STACK-oracle, regresses the ground truth $\tau_1(\bx)$ on the predictions of the estimation set in site 1 from each local model, $\{\widehat\tau_1(\bx), \dots, \widehat\tau_k(\bx) \}$. 
We compare the proposed model averaging estimators with the local estimator, MA, two versions of modified EWMA, as well as two versions of the linear stacking approach. 
We present simulation results using CT as the local model and the sample size at local sites to be $n=500$. 
Figure~\ref{web:sim_fig_ct500_full} presents the performance of the proposed estimators along with other competing estimators. Each series of boxes corresponds to a different strength of global heterogeneity $c$. 
Table~\ref{tab:sim_res} reports the ratio between MSE of the estimator and MSE of the local model in terms of average and standard deviation of MSE, respectively, over 1000 replicates.
Our proposed estimators ET and EF shows the best performance overall in terms of the mean and variation of MSE among the estimators without using the information of ground truth $\tau_1(\bx)$. Comparing with ET, EF has a slightly smaller MSE when $c$ is large, which is expected because forest models tend to be more stable and accurate than a single tree.  
ET-oracle achieves minimal MSE for low and moderate degrees of heterogeneity while EF-oracle has the minimal MSE under all settings.
The local estimator (LOC) in general shows the largest MSE compared to other estimators, as it does not leverage information from other sites. By borrowing information from additional sites, variances are greatly reduced, resulting in a small MSE of ensemble estimators. 
MA that naively adopts the inverse of sample size as weights performs well under low levels of heterogeneity, but suffers from a huge MSE with large variation as $c$ increases. 
EWMA estimators perform slightly better and are more stable than LOC and MA. EWMA-oracle has better performance than EWMA in all settings as the information of true CATE is used for weight construction. STACK estimators performs better than EWMA estimators. 
Similarly, STACK-oracle performs better than STACK in all settings. STACK-oracle, with ground truth $\tau_1(\bx)$ available, outperforms ET and EF when there exists a moderate to high level of heterogeneity across sites.

\subsection{Various Sample Sizes in Local Sites}
 
We provide detailed simulation results varying $n$ $(100, 500, 1000)$ with CT as the local model. 
Figure~\ref{web:sim_fig_ct100} and Figure~\ref{web:sim_fig_ct1000} show box plots of simulation results with a sample size of 100 and 1000, respectively, at each site. Our proposed methods ET and EF show robust performance in all settings. ET-oracle and EF-oracle achieve close-to-zero MSE with very small spreads in some settings.  Figure~\ref{fig:sim_vary_n} shows plots of the bias and MSE of EF-oracle varying sample size at each site ($n = 100, 500, 1000$). As the sample size increases, both bias and MSE of EF-oracle reduce to zero. Consistency of EF-oracle can be shown via simulation when perfect estimates are obtained from local models. 
Meanwhile, our proposed method greatly reduce MSE by selectively borrowing information from multiple sites.

\subsection{Simulations under Observational Studies}

We also consider the treatment generation mechanism under an observational design. 
Specifically, the propensity is given as $e(\bx) = \text{expit}(0.6x_1)$. We consider both a correctly specified propensity model using a logistic regression of $Z$ on $X_1$ and a misspecified propensity model with a logistic regression of $Z$ on all $\bX$. 
Figure~\ref{web:sim_obs_correct} and Figure~\ref{web:sim_obs_misspecified} show box plots of simulation results. 
In general, the proposed estimators obtain the best performance with similar results are obtained as in the Figure~\ref{fig:sim_box}. With the correctly specified propensity score model, the local estimator is consistent in estimating $\tau_k(\bx)$, the proposed framework is valid. When the propensity model misspecified, extra uncertainty is carried forward from the local estimates, but the proposed estimators can improve upon the local models. This is due to a bias-and-variance trade-off that guarantees small MSE in prediction, which remains smaller than those from local estimators.

\subsection{Covariate Dimensions}

We consider various choices of covariate dimensions besides $D=5$. Specifically, we also try $D=20$ and $D=50$. 
Figure~\ref{web:sim_p20} and Figure~\ref{web:sim_p50} show box plots of simulation results. 
With a higher dimension of variables, the MSE ratio between the proposed estimates and LOC estimates increases than that in the scenario with a small dimension.

\subsection{Unequal Sample Size at Each Site}

In the distributed date network, different sites may have a different sample size $n_k$. Those with a smaller sample size may not be representative of their population, leading to an uneven level of precision for local causal estimates. We consider a simulation setting where site 1 has a sample size of $n_1=500$ while other site $n_2,\ldots,n_K$ has a sample size of 200. 
Figure~\ref{web:sim_diffN} shows box plots of simulation results. 
Results show that the MSE ratio between the proposed estimates and LOC estimates increases compared to the scenario where the sample size in all sites are 500. However, the proposed estimators still enjoy the most robust performance via bias-and-variance trade-off. 
This also shows our method is robust to the existence of local uncertainty.

\subsection{Different Local Estimators} 

We explore another option for the local model using the causal forest (CF) \citep{wager2018estimation} varying the sample size at local sites. 
A causal forest is a stochastic averaging of multiple causal trees \citep{athey2016recursive}, and hence is more powerful in estimating treatment effects. In each tree of the causal forest, MSE of treatment effect is used to select the feature and cutoff point in each split \citep{wager2018estimation}. 
CF is implemented in the R packages \verb|grf|.
Figure~\ref{web:sim_fig_cf100}, Figure~\ref{web:sim_fig_cf500}, and Figure~\ref{web:sim_fig_cf1000} show box plots of simulation results with a sample size of 100, 500, and 1000, respectively, at each site. Our proposed methods ET and EF show robust performance in all settings regardless of the use of information of the ground truth $\tau_1(\bx)$.

\subsection{Further Comparisons to Non-adaptive Ensemble} 
We provide simulation results to compare the proposed methods to the non-adaptive method STACK. 
Consider the following setting where the heterogeneity is continuous and nonlinear: $\tau(\bx, k) = \mathbbm{1}\{x_1 > 0\} \cdot x_1 +(x_1 - 3) \cdot {(U_k)}^c,$
with $U_k\sim Unif[0,3]$, $\bX_i\sim {N}(\boldsymbol{0},\boldsymbol{I}_5)$, and $c=(1,2,3,4)$. 
As $c$ increases, the heterogeneity across sites gets larger, reducing the influence of $x_1$ on heterogeneity, hence the weights become more non-adaptive. 
For $c=(1,2,3,4)$, 
the one-SD ranges of MSE ratios of EF over STACK are [0.73,0.82], [0.86,0.87], [0.99,1.04], [0.87,1.07], respectively. 
When $c$ is relatively small, the proposed EF has a smaller MSE compared to STACK. As $c$ increases, the performance of EF is similar to that of STACK, in the case of a large global heterogeneity. This further indicates the robustness of the proposed methods.

%%%%%%%%%% CT500
\clearpage
\begin{figure}[!htb]%[hbt!]%[!h]%tp]
\centering
% \vspace{0.1cm}
 \begin{subfigure}{0.49\textwidth}
  \centerline{\includegraphics[width=\linewidth]{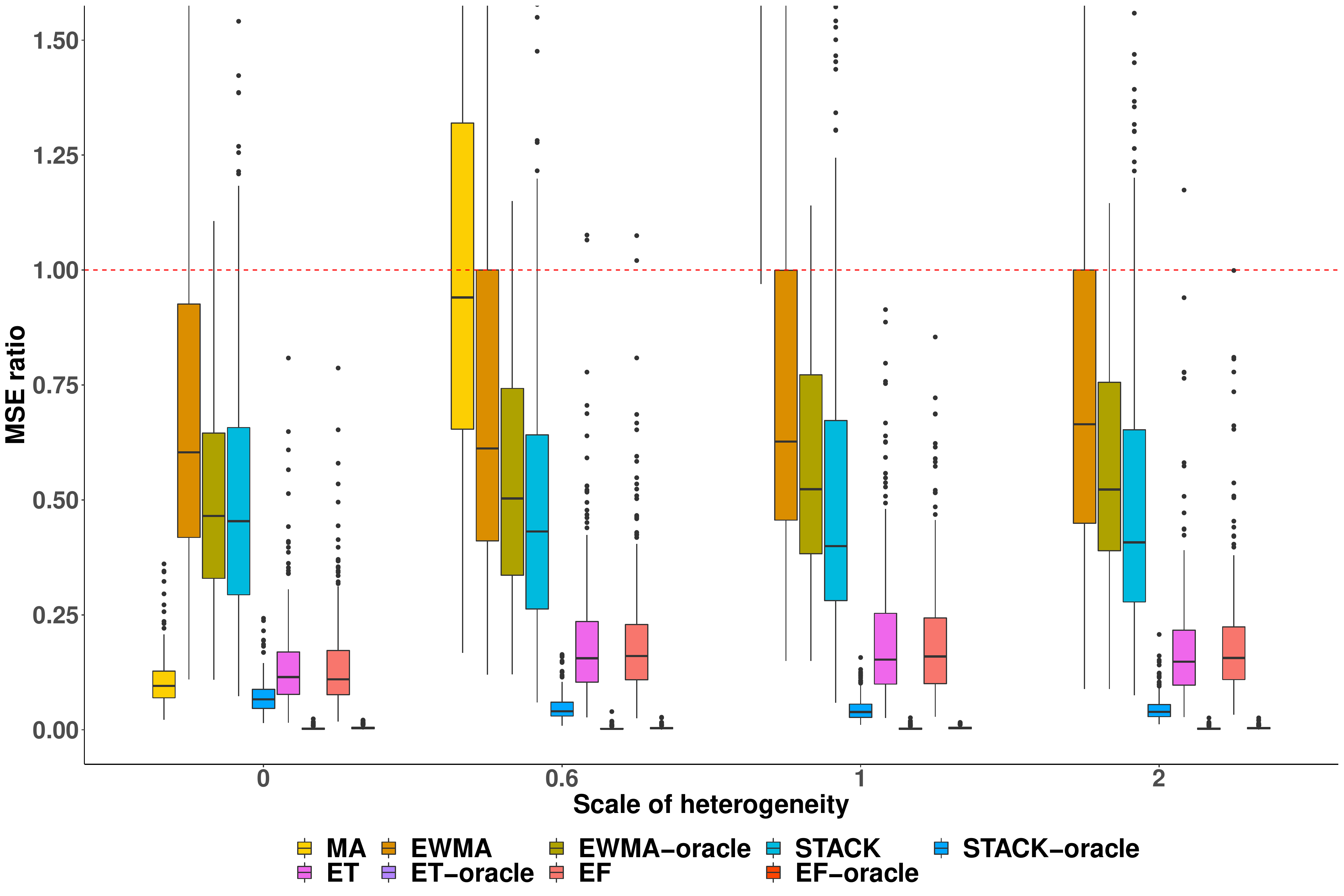}}
  \caption{}
%   \label{fig:disc}
 \end{subfigure}
 \begin{subfigure}{0.49\textwidth}
  \centerline{\includegraphics[width=\linewidth]{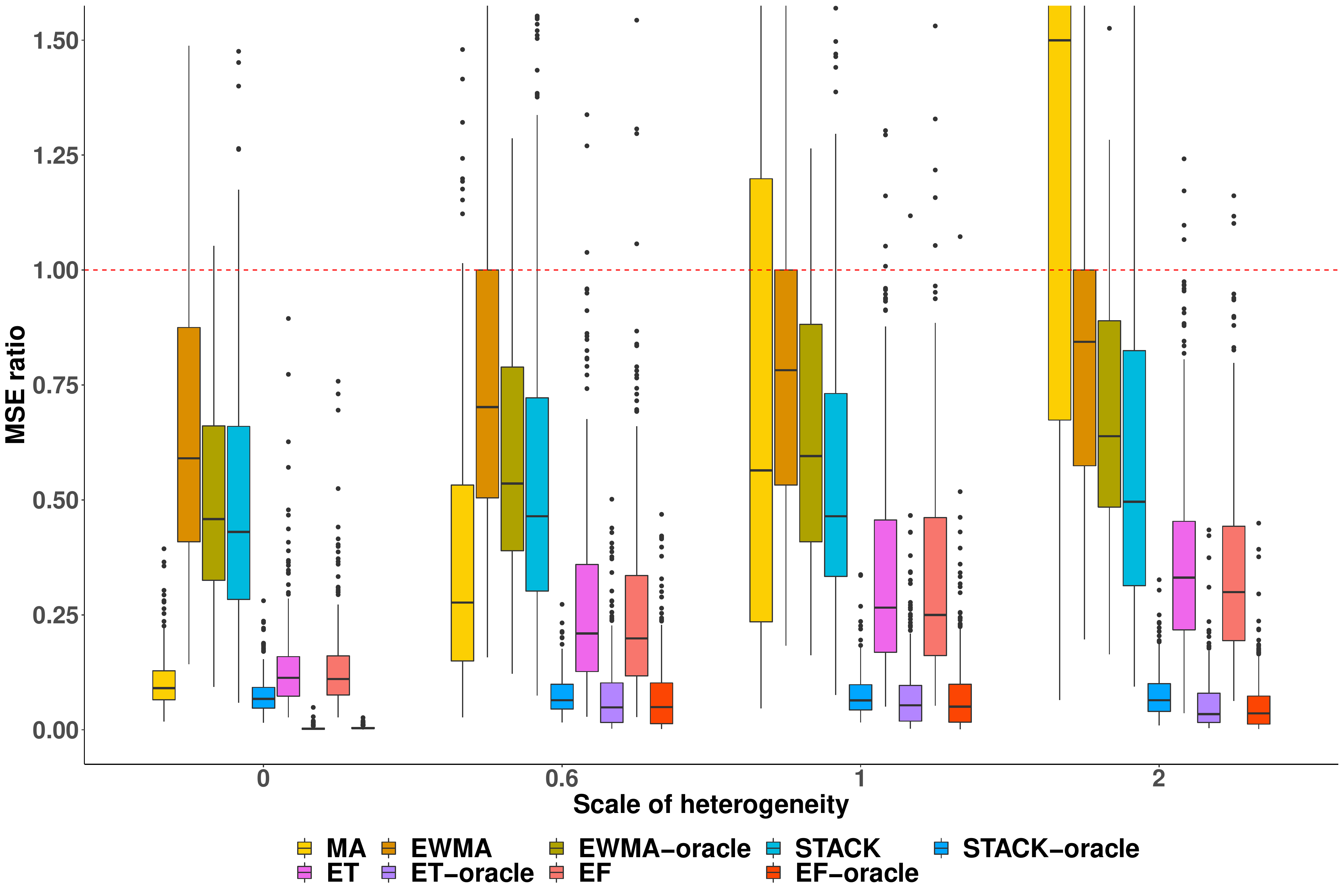}}
  \caption{}
%   \label{fig:cont}
 \end{subfigure}
%  \vspace{-1cm}
 \caption{
Box plots of the MSE ratios of CATE estimators, respectively, over LOC (\textbf{CT}) and a sample size of \textbf{500} at each site for (a) discrete grouping and (b) continuous grouping across site, respectively, varying scale of global heterogeneity. 
Estimators ending with ``-oracle" makes use of ground truth treatment effects. 
Different colors imply different estimators, and x-axis, i.e., the value of $c$, differentiates the scale of global heterogeneity. The red dotted line denotes an MSE ratio of 1. 
MA performance is truncated due to large MSE ratios. 
The proposed ET and EF achieve competitive performance compared to standard model averaging or ensemble methods and are robust to heterogeneity across settings. 
Note that ET-oracle and EF-oracle achieve close-to-zero MSE ratios with very small spreads in some settings. 
}
 \label{web:sim_fig_ct500_full}
\end{figure}

%%%%%%%%%% CT500 table
\clearpage
\begin{table}[!htb]%[hbt!]%[!h]%tp]
\centering
% \small%\small
\caption{Simulation results for ratio between MSE of the estimator and MSE of LOC  (\textbf{CT}) with a sample size of \textbf{500} at each site. A smaller number indicates larger improvement over the local model. 
Estimators ending with ``-oracle" makes use of ground truth treatment effects.
% and estimators ending with ``" uses approximated treatment effects from the estimation set in site 1. 
Our proposed methods ET and EF shows robust performance in all settings whether or not using the information of ground truth $\tau_1(\bx)$.
}
\label{tab:sim_res}
% \vspace{-1cm}
\resizebox{0.9\columnwidth}{!}{
\begin{tabular}{@{}ccccccccc@{}}
\toprule
 & \multicolumn{4}{c}{Discrete grouping} & \multicolumn{4}{c}{Continuous grouping} \\ \cmidrule(l){2-9} 
Estimator & $c=0$ & $c=0.2$ & $c=0.6$ & $c=1$ & $c=0$ & $c=0.2$ & $c=0.6$ & $c=1$ \\ \midrule
\multicolumn{9}{l}{\textit{Ratio of average of MSEs over 1000 replicates}} \\ %$<$0.01
MA & 0.09 & 0.91 & 2.4 & 9.87 & 0.08 & 0.32 & 0.65 & 1.78 \\
EWMA & 0.57 & 0.62 & 0.61 & 0.62 & 0.56 & 0.65 & 0.7 & 0.77 \\
EWMA-oracle & 0.42 & 0.5 & 0.49 & 0.5 & 0.42 & 0.49 & 0.53 & 0.59 \\
STACK & 0.44 & 0.45 & 0.44 & 0.45 & 0.45 & 0.45 & 0.48 & 0.54 \\
STACK-oracle & 0.06 & 0.04 & 0.04 & 0.04 & 0.06 & 0.06 & 0.06 & 0.07 \\
ET & 0.12 & 0.17 & 0.16 & 0.16 & 0.13 & 0.24 & 0.29 & 0.37 \\
ET-oracle & $<$0.01 & $<$0.01 & $<$0.01 & $<$0.01 & $<$0.01 & 0.08 & 0.1 & 0.07 \\
EF & 0.1 & 0.13 & 0.13 & 0.13 & 0.1 & 0.19 & 0.25 & 0.3 \\
EF-oracle & $<$0.01 & $<$0.01 & $<$0.01 & $<$0.01 & $<$0.01 & 0.06 & 0.06 & 0.05 \\
\midrule
\multicolumn{9}{l}{\textit{Ratio of standard deviation of MSEs over 1000 replicates}} \\
MA & 0.15 & 0.35 & 0.76 & 3.05 & 0.14 & 0.24 & 0.38 & 0.81 \\
EWMA & 0.61 & 0.65 & 0.67 & 0.66 & 0.58 & 0.65 & 0.69 & 0.75 \\
EWMA-oracle & 0.46 & 0.52 & 0.54 & 0.54 & 0.44 & 0.52 & 0.55 & 0.6 \\
STACK & 0.47 & 0.46 & 0.47 & 0.47 & 0.45 & 0.49 & 0.52 & 0.6 \\
STACK-oracle & 0.1 & 0.08 & 0.08 & 0.08 & 0.09 & 0.11 & 0.12 & 0.14 \\
ET & 0.18 & 0.23 & 0.22 & 0.22 & 0.18 & 0.26 & 0.32 & 0.43 \\
ET-oracle & 0.02 & 0.03 & 0.02 & 0.02 & 0.02 & 0.06 & 0.07 & 0.07 \\
EF & 0.17 & 0.19 & 0.19 & 0.2 & 0.17 & 0.23 & 0.29 & 0.39 \\
EF-oracle & 0.03 & 0.03 & 0.03 & 0.03 & 0.03 & 0.06 & 0.07 & 0.08\\ \bottomrule
\end{tabular}
% \vspace{1cm}
}
\end{table}

%%%%%%%%%% CT100
\clearpage
\begin{figure}[!htb]%[hbt!]%[!h]%tp]
\centering
% \vspace{1cm}
 \begin{subfigure}{0.49\textwidth}
  \centerline{\includegraphics[width=\linewidth]{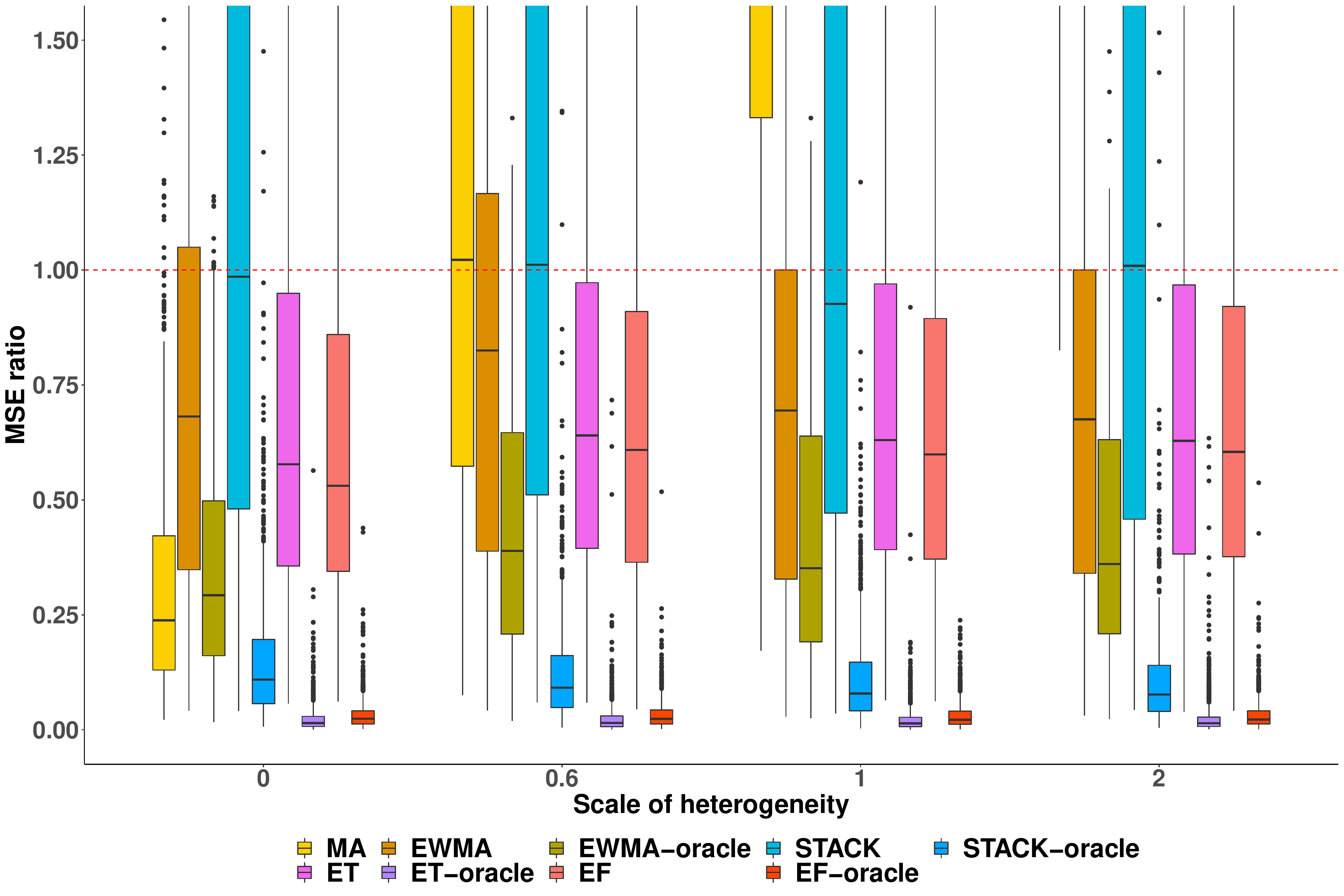}}
  \caption{}
 \end{subfigure}
 \begin{subfigure}{0.49\textwidth}
  \centerline{\includegraphics[width=\linewidth]{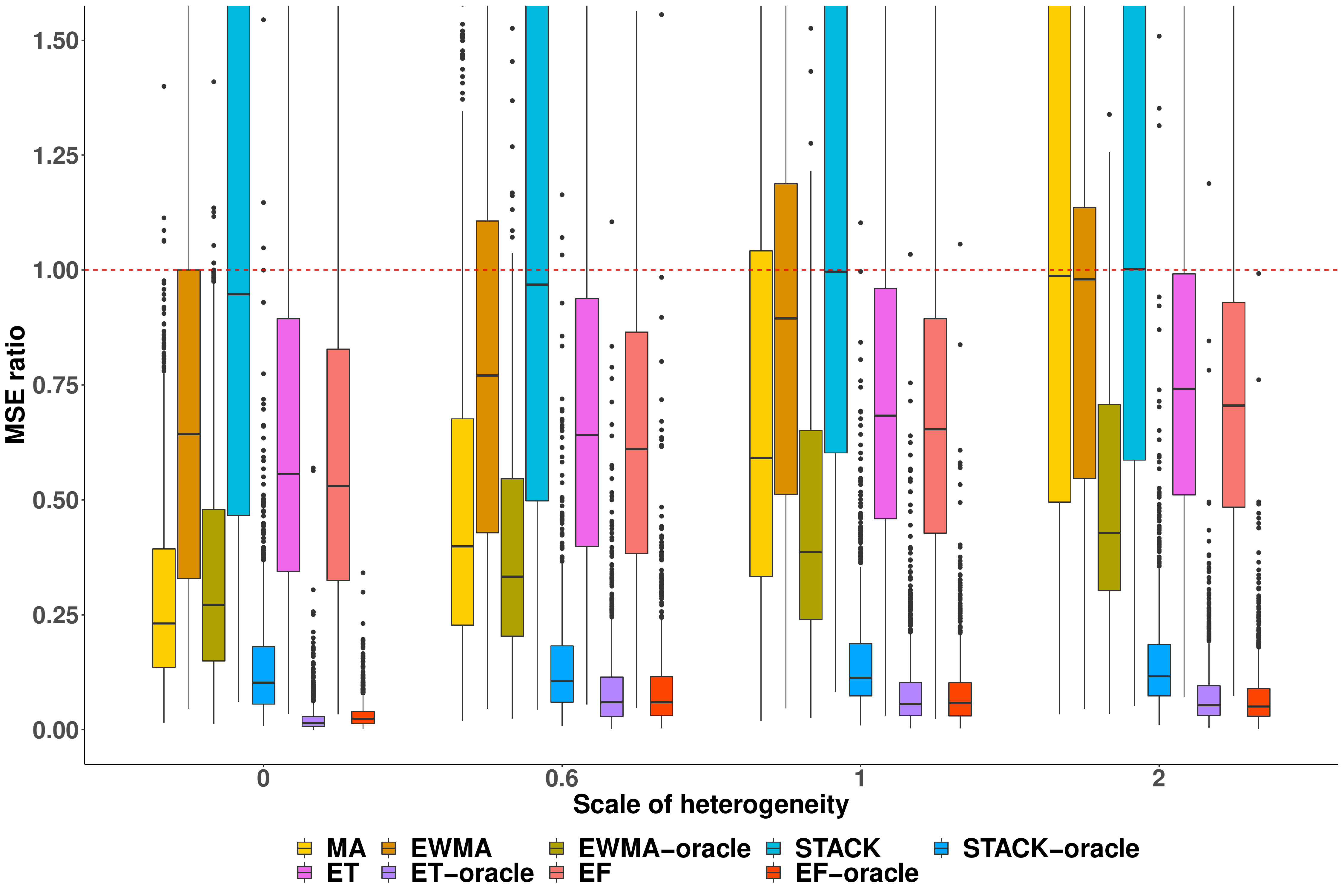}}
  \caption{}
 \end{subfigure}
%  \vspace{-1.4cm}
 \caption{
Box plots of the MSE ratios of CATE estimators, respectively, over LOC (\textbf{CT}) and a sample size of \textbf{100} at each site for (a) discrete grouping and (b) continuous grouping across site, respectively, varying scale of global heterogeneity. 
Estimators ending with ``-oracle" makes use of ground truth treatment effects. 
Different colors imply different estimators, and x-axis, i.e., the value of $c$, differentiates the scale of global heterogeneity. The red dotted line denotes an MSE ratio of 1. 
MA performance is truncated due to large MSE ratios. 
The proposed ET and EF achieve competitive performance compared to standard model averaging or ensemble methods and are robust to heterogeneity across settings. 
Note that ET-oracle and EF-oracle achieve close-to-zero MSE ratios with very small spreads in some settings. 
 }
 \label{web:sim_fig_ct100}
\end{figure}

%%%%%%%%%% CT1000
% \FloatBarrier
\clearpage
\begin{figure}[!h]%[hbt!]%[!h]%tp]
\centering
% \vspace{1cm}
 \begin{subfigure}{0.49\textwidth}
  \centerline{\includegraphics[width=\linewidth]{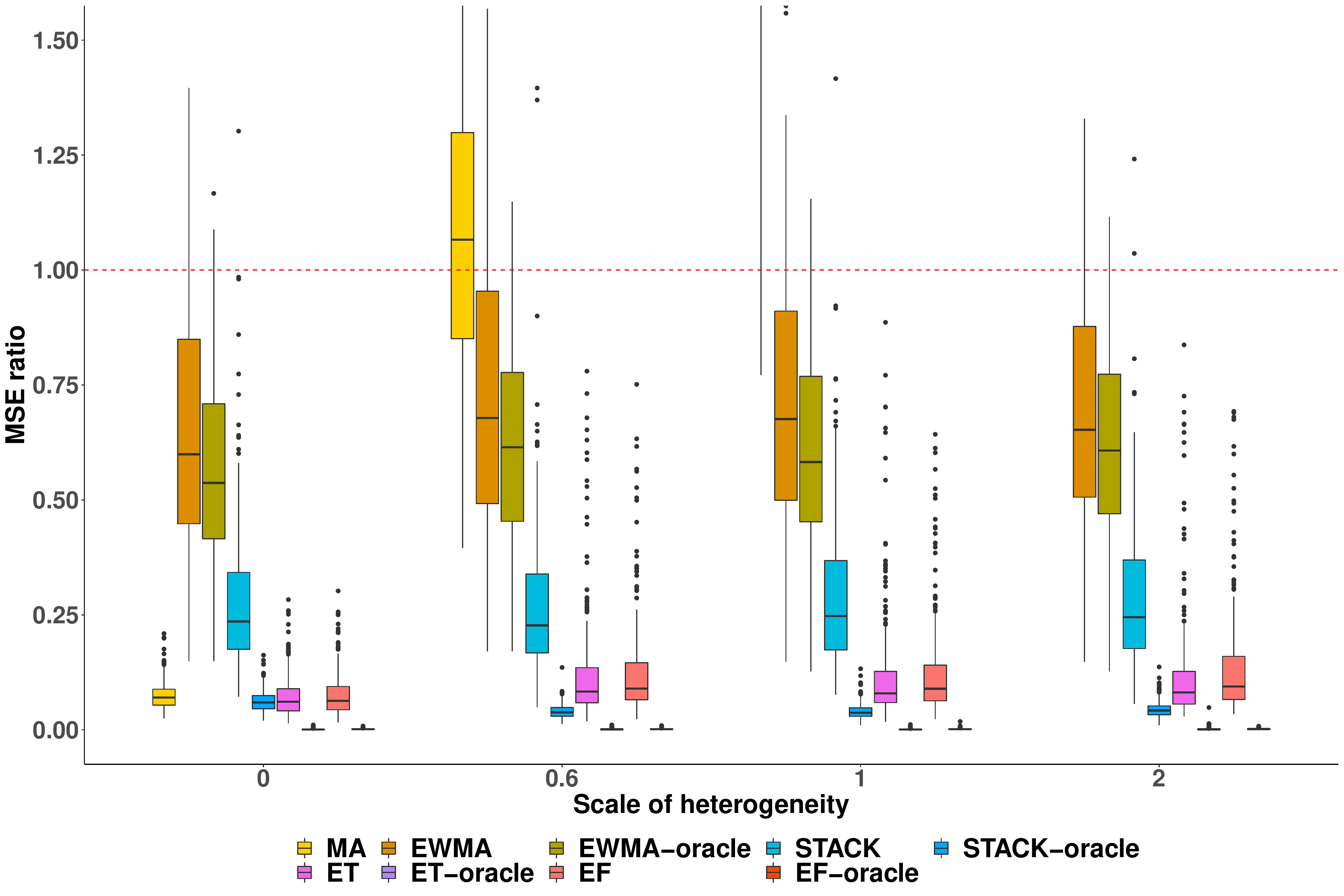}}
  \caption{}
 \end{subfigure}
 \begin{subfigure}{0.49\textwidth}
  \centerline{\includegraphics[width=\linewidth]{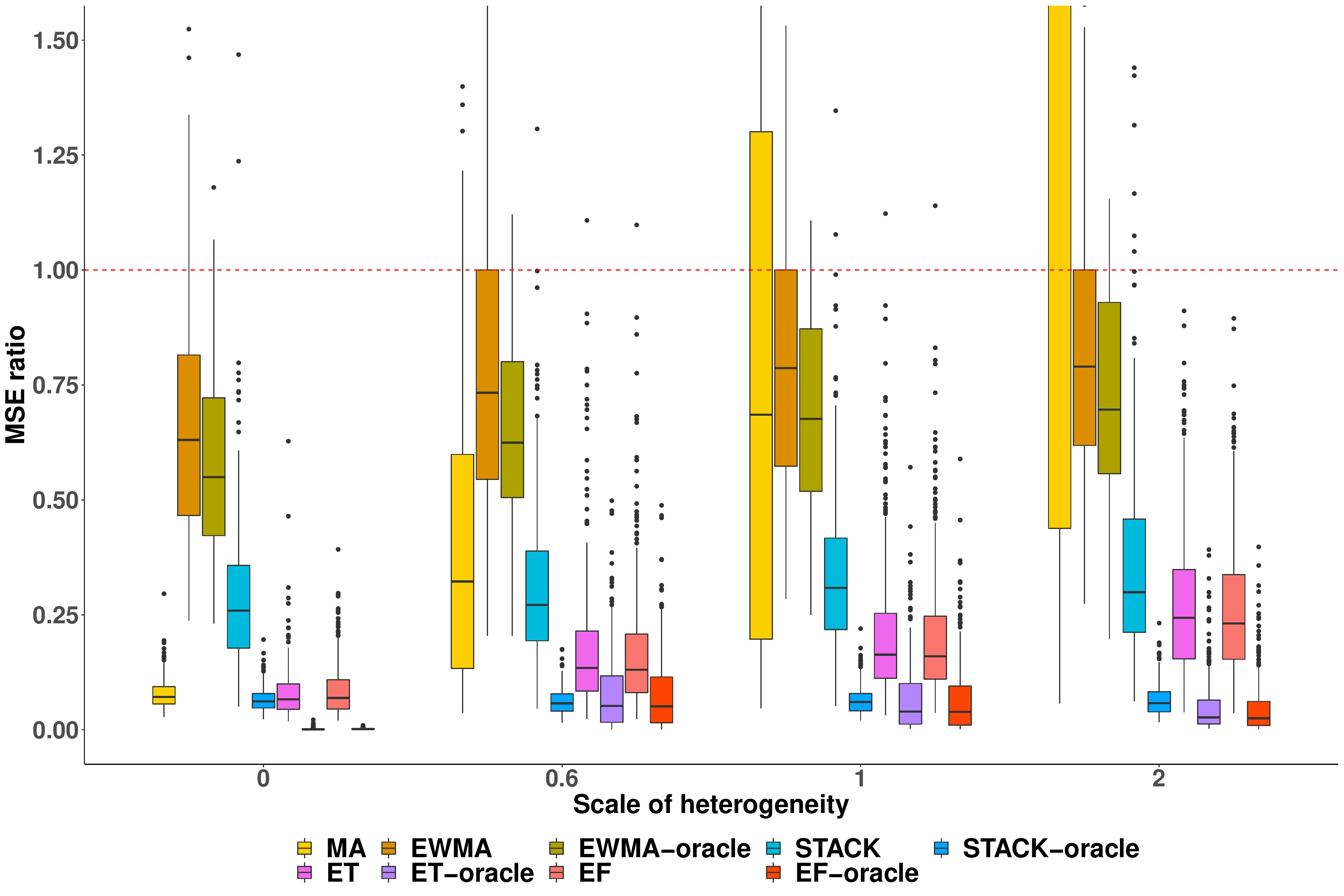}}
  \caption{}
 \end{subfigure}
%  \vspace{-1cm}
 \caption{
Box plots of the MSE ratios of CATE estimators, respectively, over LOC (\textbf{CT}) and a sample size of \textbf{1000} at each site for (a) discrete grouping and (b) continuous grouping across site, respectively, varying scale of global heterogeneity. 
Estimators ending with ``-oracle" makes use of ground truth treatment effects. 
Different colors imply different estimators, and x-axis, i.e., the value of $c$, differentiates the scale of global heterogeneity. The red dotted line denotes an MSE ratio of 1. 
MA performance is truncated due to large MSE ratios. 
The proposed ET and EF achieve competitive performance compared to standard model averaging or ensemble methods and are robust to heterogeneity across settings. 
Note that ET-oracle and EF-oracle achieve close-to-zero MSE ratios with very small spreads in some settings. 
 }
 \label{web:sim_fig_ct1000}
\end{figure}

%%%%%%%%%% CT 100-500-1000
% \afterpage{
\clearpage
\begin{figure}[hbt!]%[hbt!]%[h]%tp]
\centering
% \vspace{1cm}
 \begin{subfigure}{0.49\textwidth}
  \centerline{\includegraphics[width=\linewidth]{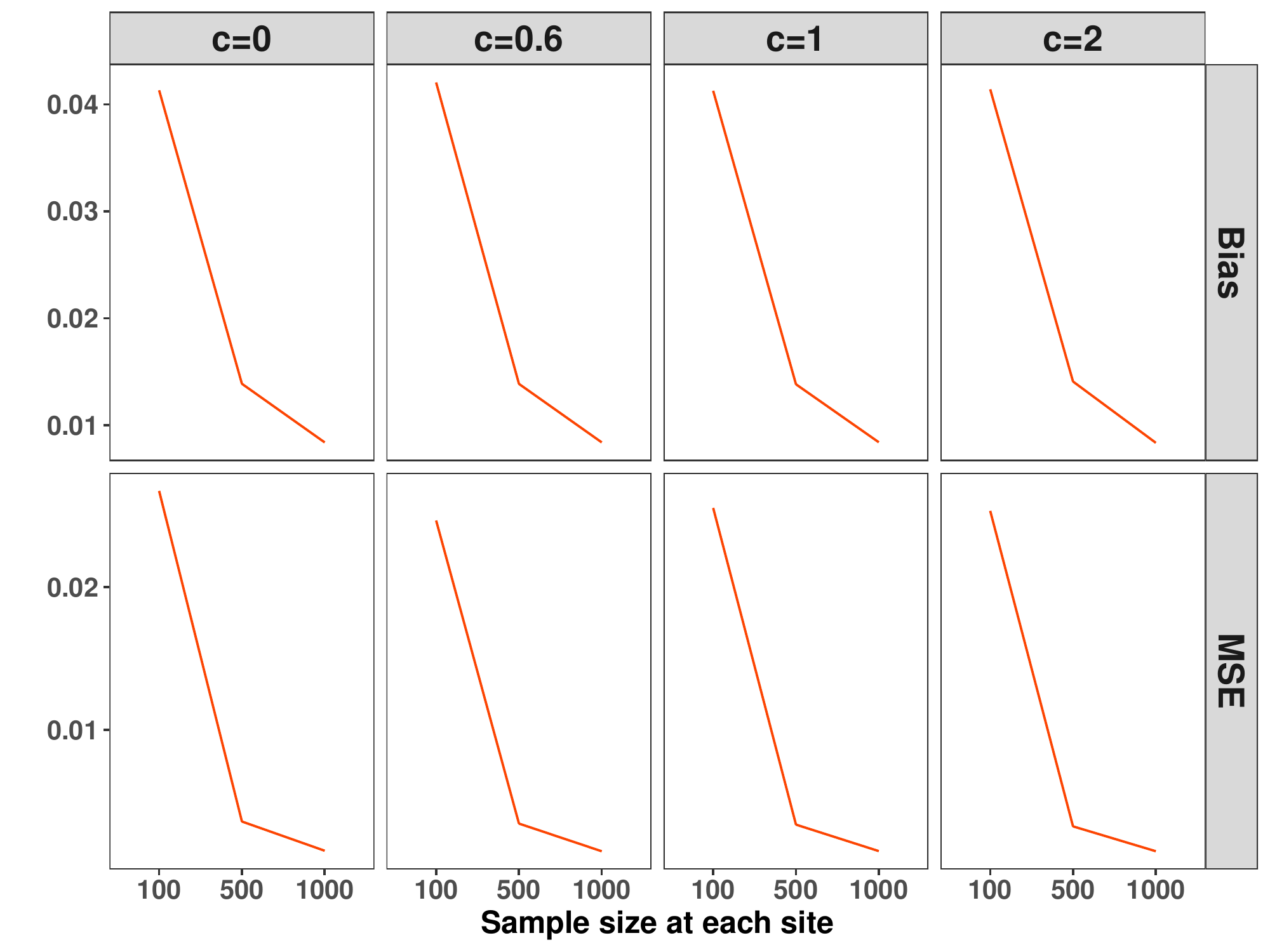}}
  \caption{}
 \end{subfigure}
 \begin{subfigure}{0.49\textwidth}
  \centerline{\includegraphics[width=\linewidth]{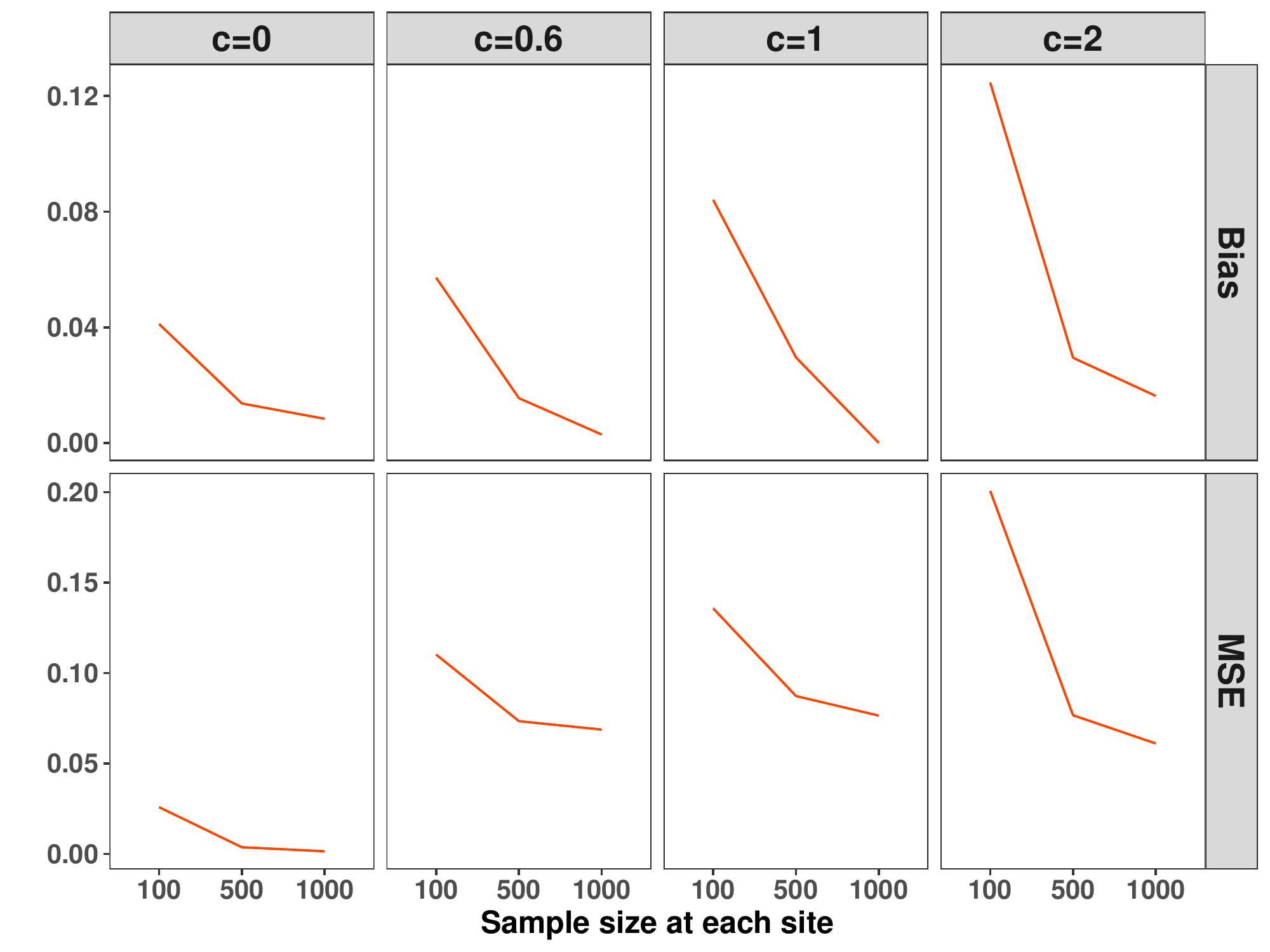}}
  \caption{}
 \end{subfigure}
%  \vspace{-1cm}
 \caption{Plots of the bias and MSE of \textbf{EF-oracle} varying sample site at each site for (a) discrete grouping and (b) continuous grouping across site, varying scale of global heterogeneity. 
 Both bias and MSE reduces to zero as the sample size increases. 
 }
 \label{fig:sim_vary_n}
\end{figure}
% }

%%%%%%%%%% obs-correct
\clearpage
\begin{figure}[hbt!]%[!h]%tp]
\centering
% \vspace{1cm}
 \begin{subfigure}{0.49\textwidth}
  \centerline{\includegraphics[width=\linewidth]{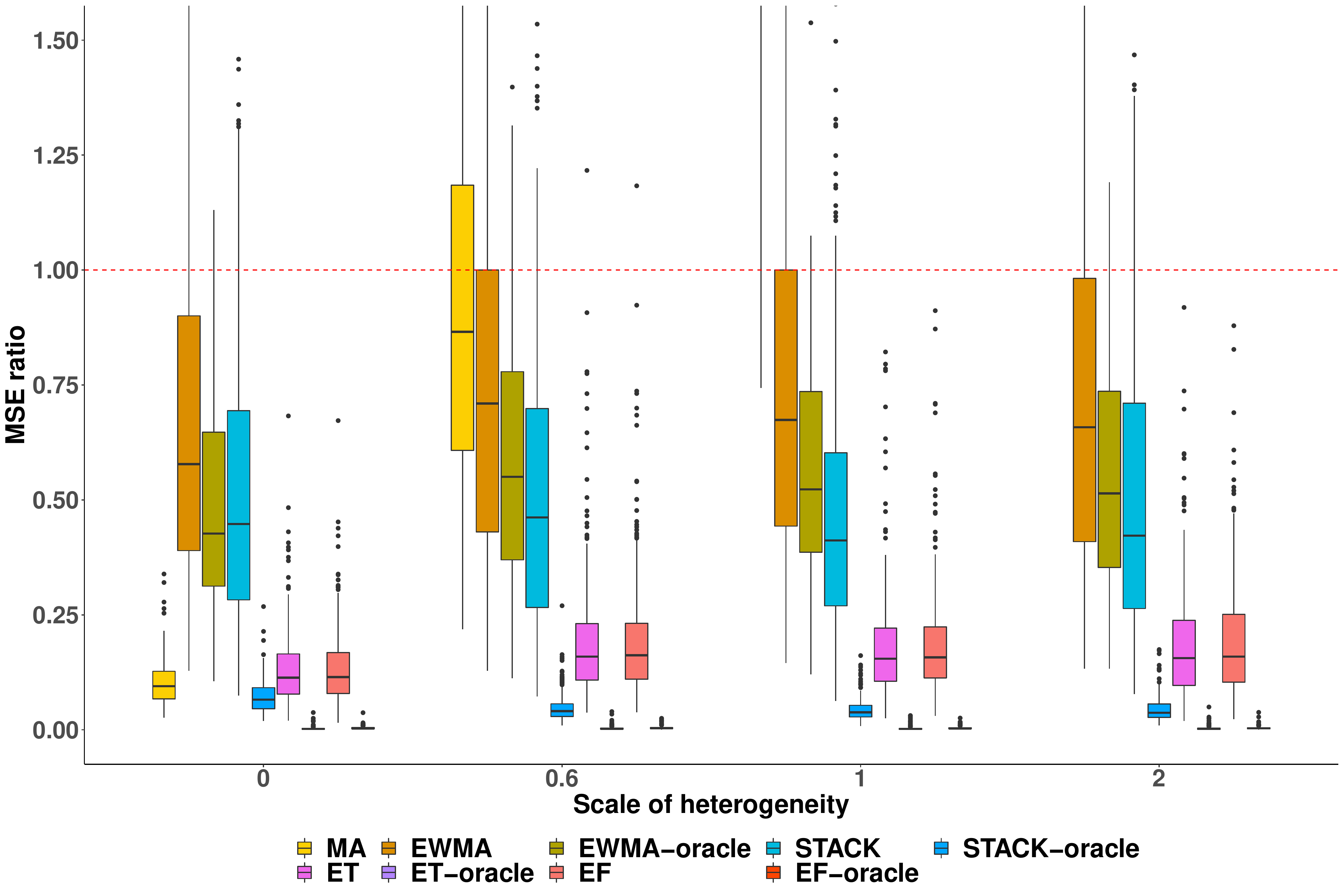}}
  \caption{}
 \end{subfigure}
 \begin{subfigure}{0.49\textwidth}
  \centerline{\includegraphics[width=\linewidth]{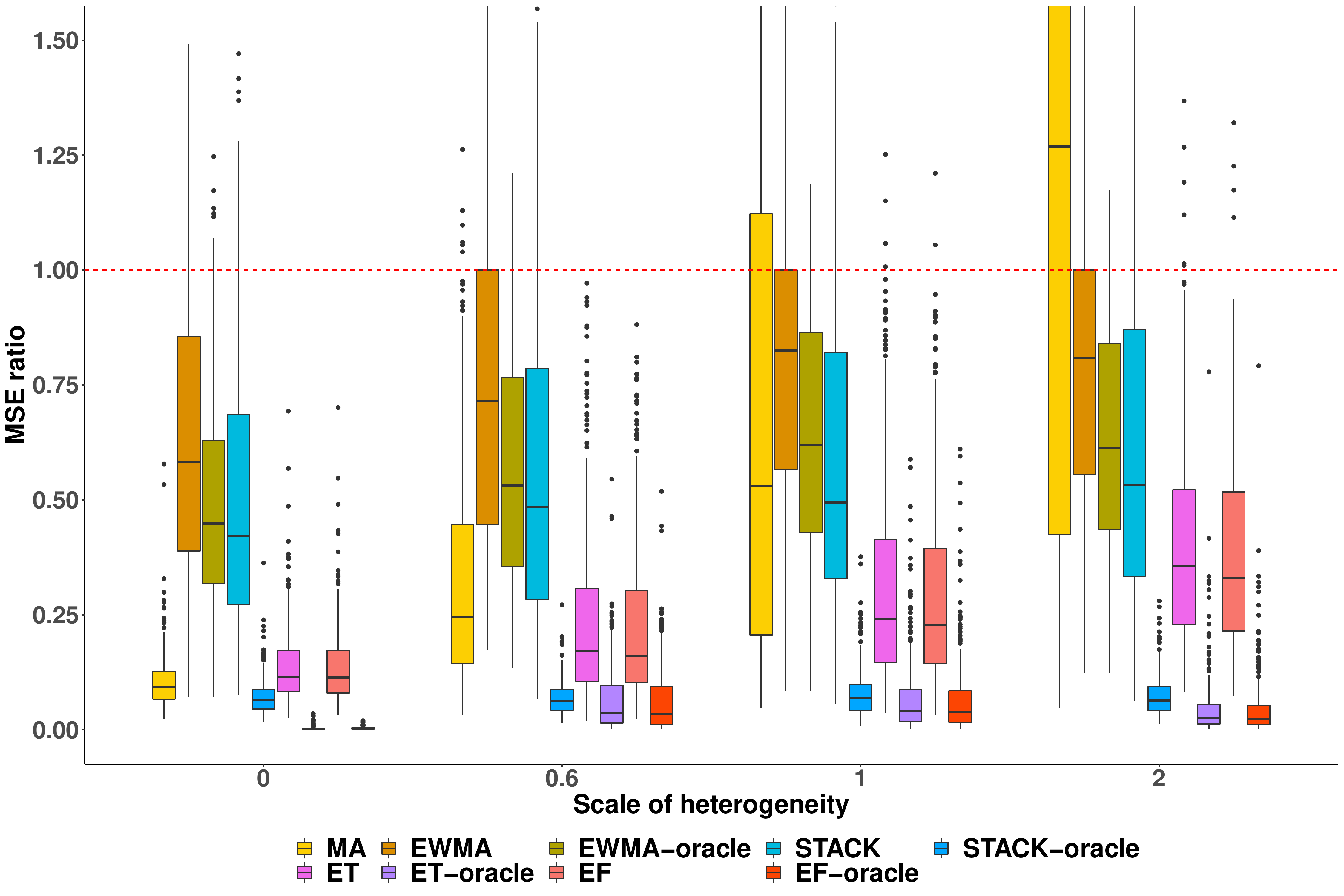}}
  \caption{}
 \end{subfigure}
%  \vspace{-1cm}
 \caption{
Box plots of the MSE ratios of CATE estimators, respectively, over LOC (\textbf{CT}) and a sample size of \textbf{500} at each site under \textbf{observational design with a correctly specified propensity score model} for (a) discrete grouping and (b) continuous grouping across site, respectively, varying scale of global heterogeneity. 
Estimators ending with ``-oracle" makes use of ground truth treatment effects. 
Different colors imply different estimators, and x-axis, i.e., the value of $c$, differentiates the scale of global heterogeneity. The red dotted line denotes an MSE ratio of 1. 
MA performance is truncated due to large MSE ratios. 
The proposed ET and EF achieve competitive performance compared to standard model averaging or ensemble methods and are robust to heterogeneity across settings. 
Note that ET-oracle and EF-oracle achieve close-to-zero MSE ratios with very small spreads in some settings. 
 }
 \label{web:sim_obs_correct}
\end{figure}

%%%%%%%%%% obs-misspecified
\clearpage
\begin{figure}[hbt!]%[!h]%tp]
\centering
% \vspace{1cm}
 \begin{subfigure}{0.49\textwidth}
  \centerline{\includegraphics[width=\linewidth]{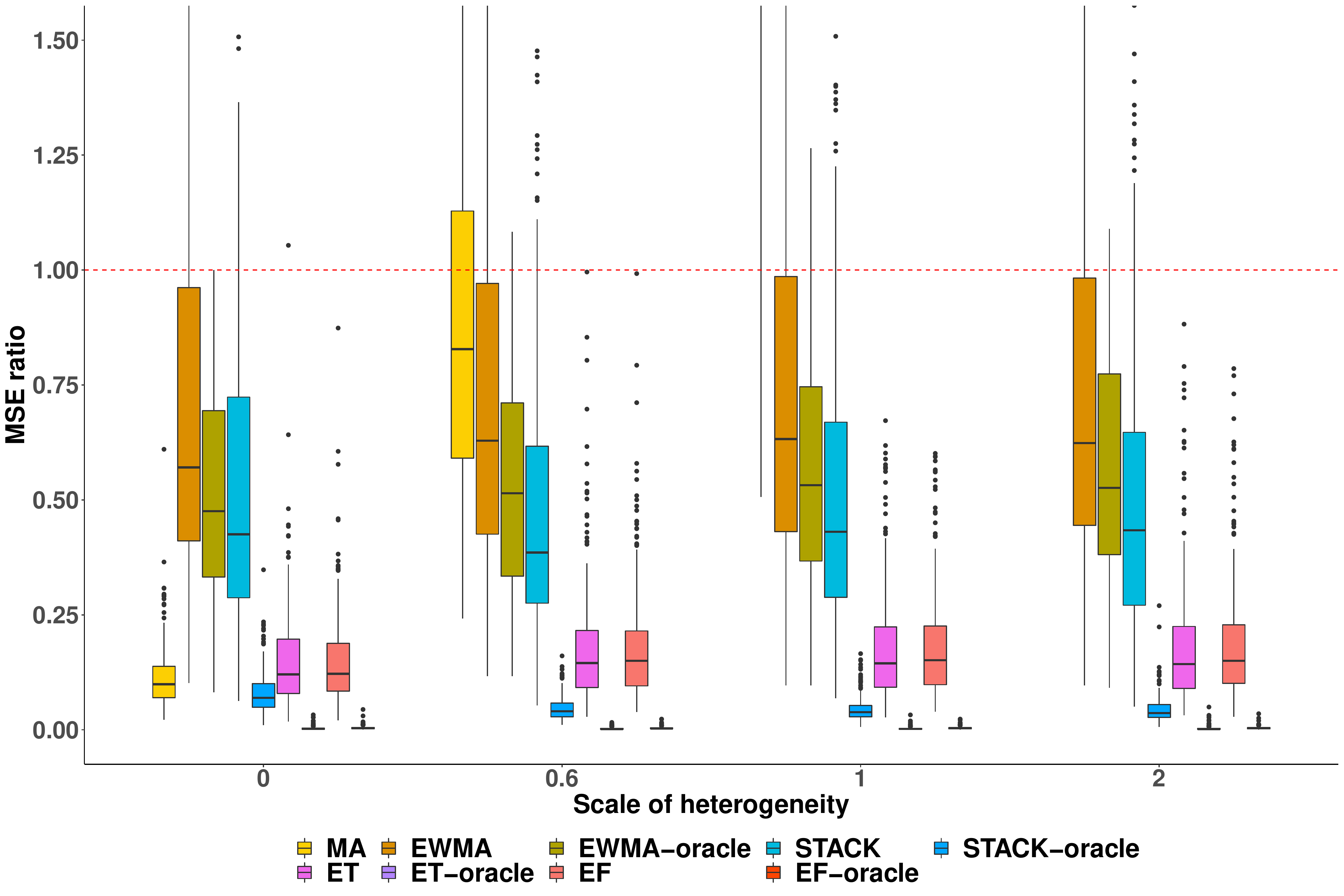}}
  \caption{}
%   \label{fig:disc}
 \end{subfigure}
 \begin{subfigure}{0.49\textwidth}
  \centerline{\includegraphics[width=\linewidth]{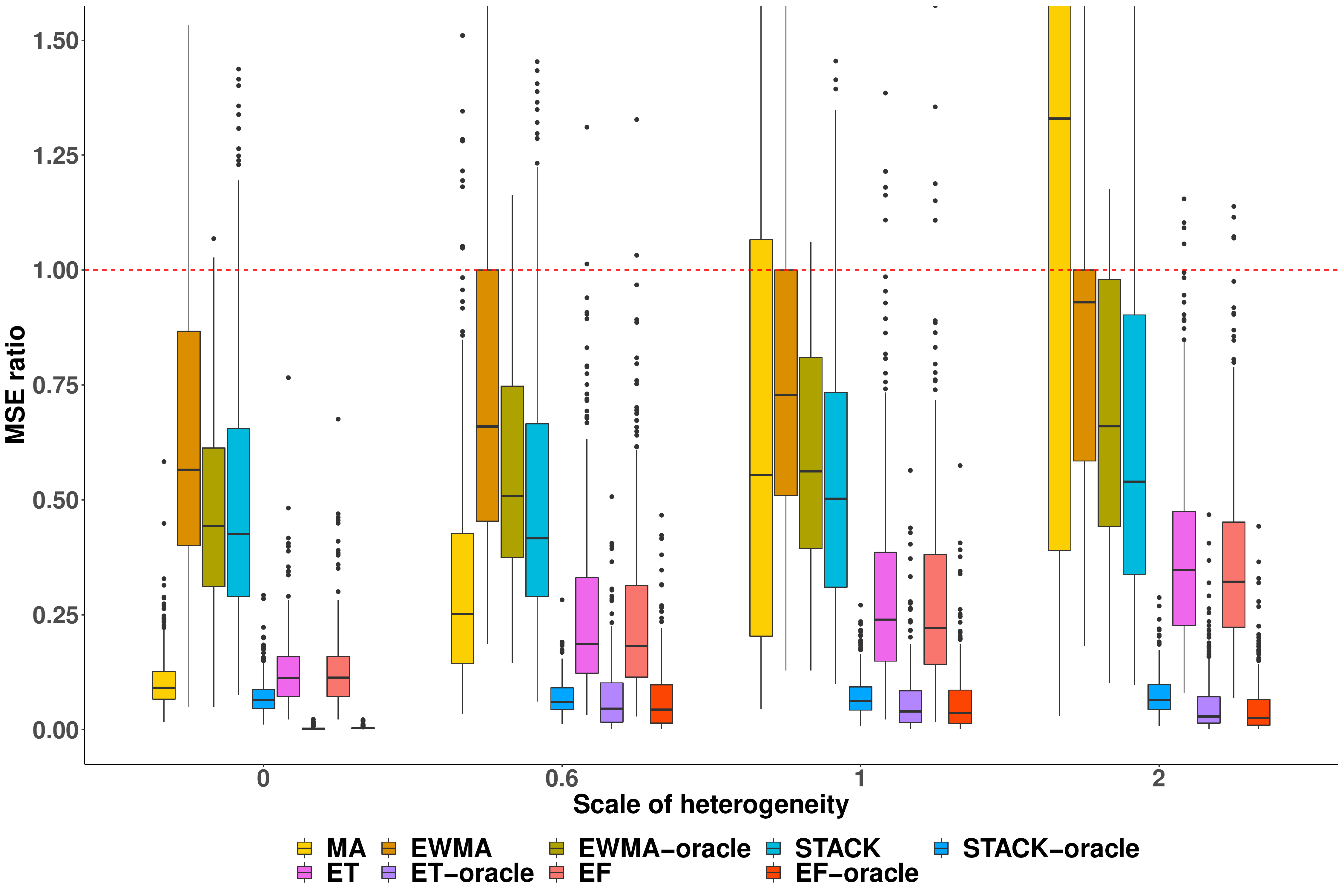}}
  \caption{}
%   \label{fig:cont}
 \end{subfigure}
%  \vspace{-1cm}
 \caption{
Box plots of the MSE ratios of CATE estimators, respectively, over LOC (\textbf{CT}) and a sample size of \textbf{500} at each site under \textbf{observational design with a misspecified propensity score model} for (a) discrete grouping and (b) continuous grouping across site, respectively, varying scale of global heterogeneity. 
Estimators ending with ``-oracle" makes use of ground truth treatment effects. 
Different colors imply different estimators, and x-axis, i.e., the value of $c$, differentiates the scale of global heterogeneity. The red dotted line denotes an MSE ratio of 1. 
MA performance is truncated due to large MSE ratios. 
The proposed ET and EF achieve competitive performance compared to standard model averaging or ensemble methods and are robust to heterogeneity across settings. 
Note that ET-oracle and EF-oracle achieve close-to-zero MSE ratios with very small spreads in some settings. 
 }
 \label{web:sim_obs_misspecified}
\end{figure}

%%%%%%%%%% p20
\clearpage
\begin{figure}[!htb]%[hbt!]%[!h]%tp]
\centering
% \vspace{1cm}
 \begin{subfigure}{0.49\textwidth}
  \centerline{\includegraphics[width=\linewidth]{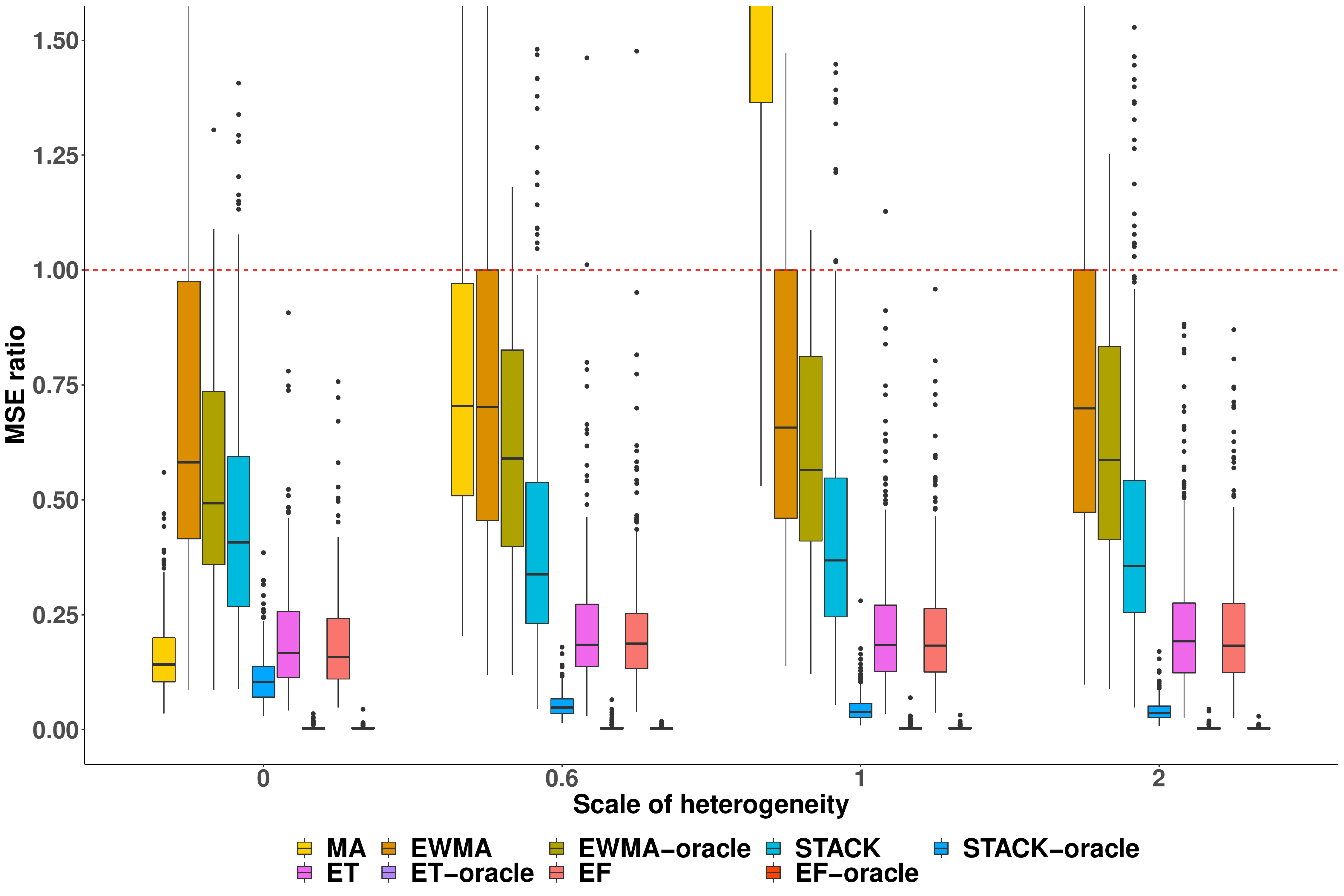}}
  \caption{}
 \end{subfigure}
 \begin{subfigure}{0.49\textwidth}
  \centerline{\includegraphics[width=\linewidth]{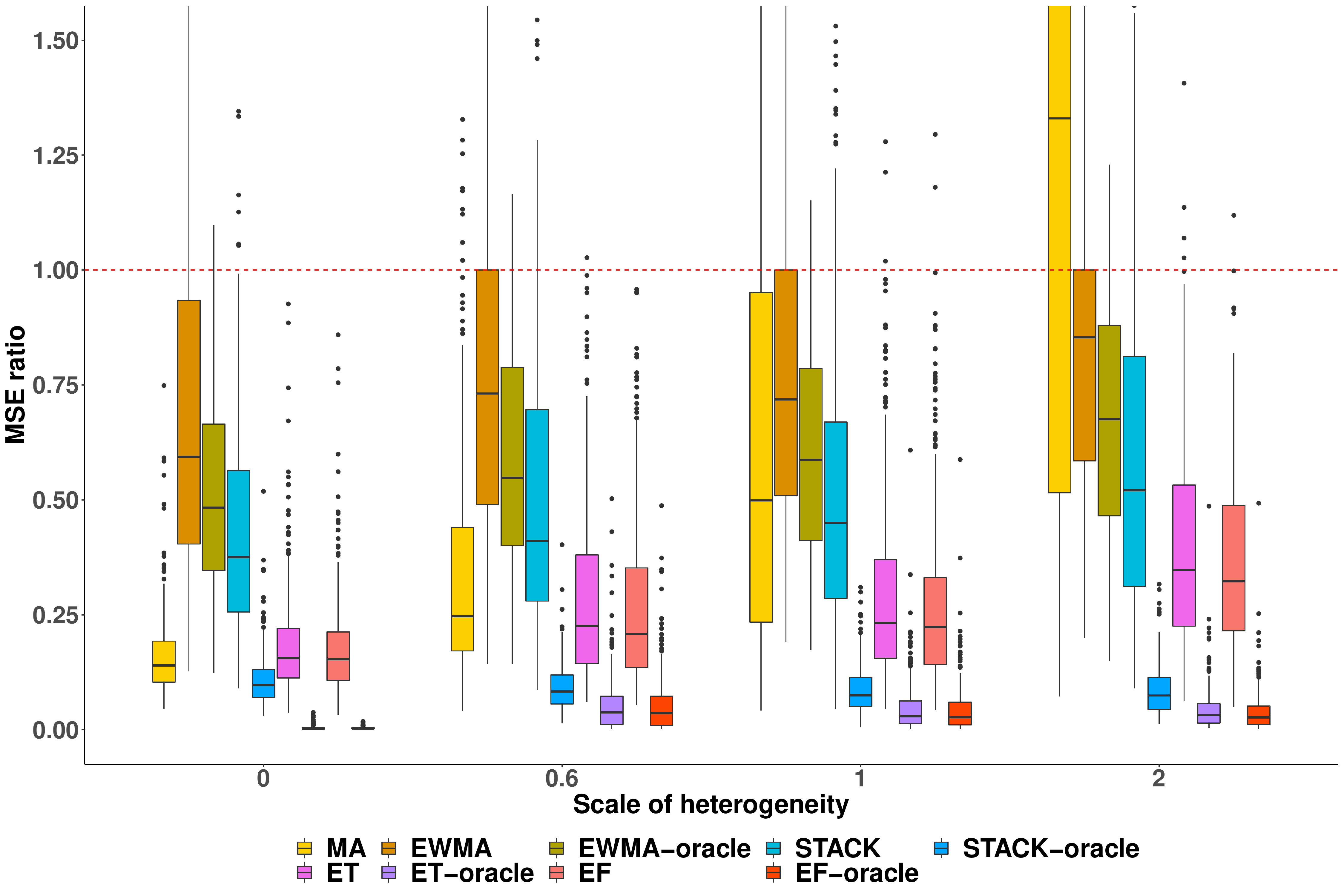}}
  \caption{}
 \end{subfigure}
%  \vspace{-1cm}
 \caption{
Box plots of the MSE ratios of CATE estimators, respectively, over LOC (\textbf{CT}) and a sample size of \textbf{500} at each site, and covariate dimension of \textbf{20} for (a) discrete grouping and (b) continuous grouping across site, respectively, varying scale of global heterogeneity. 
Estimators ending with ``-oracle" makes use of ground truth treatment effects. 
Different colors imply different estimators, and x-axis, i.e., the value of $c$, differentiates the scale of global heterogeneity. The red dotted line denotes an MSE ratio of 1. 
MA performance is truncated due to large MSE ratios. 
The proposed ET and EF achieve competitive performance compared to standard model averaging or ensemble methods and are robust to heterogeneity across settings. 
Note that ET-oracle and EF-oracle achieve close-to-zero MSE ratios with very small spreads in some settings. 
 }
 \label{web:sim_p20}
\end{figure}

%%%%%%%%%% p50
\clearpage
\begin{figure}[!htb]%[hbt!]%[!h]%tp]
\centering
% \vspace{0.5cm}
 \begin{subfigure}{0.49\textwidth}
  \centerline{\includegraphics[width=\linewidth]{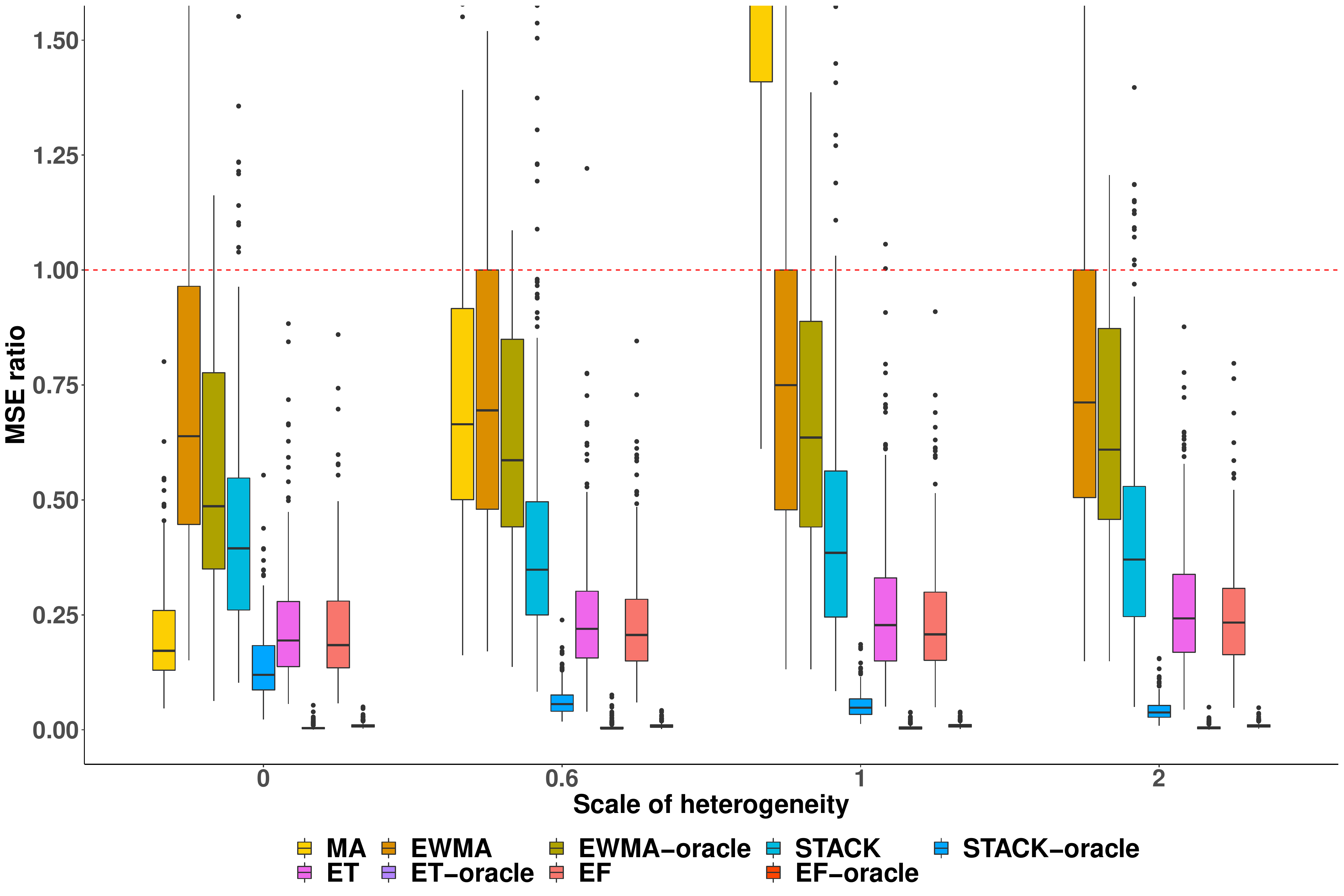}}
  \caption{}
 \end{subfigure}
 \begin{subfigure}{0.49\textwidth}
  \centerline{\includegraphics[width=\linewidth]{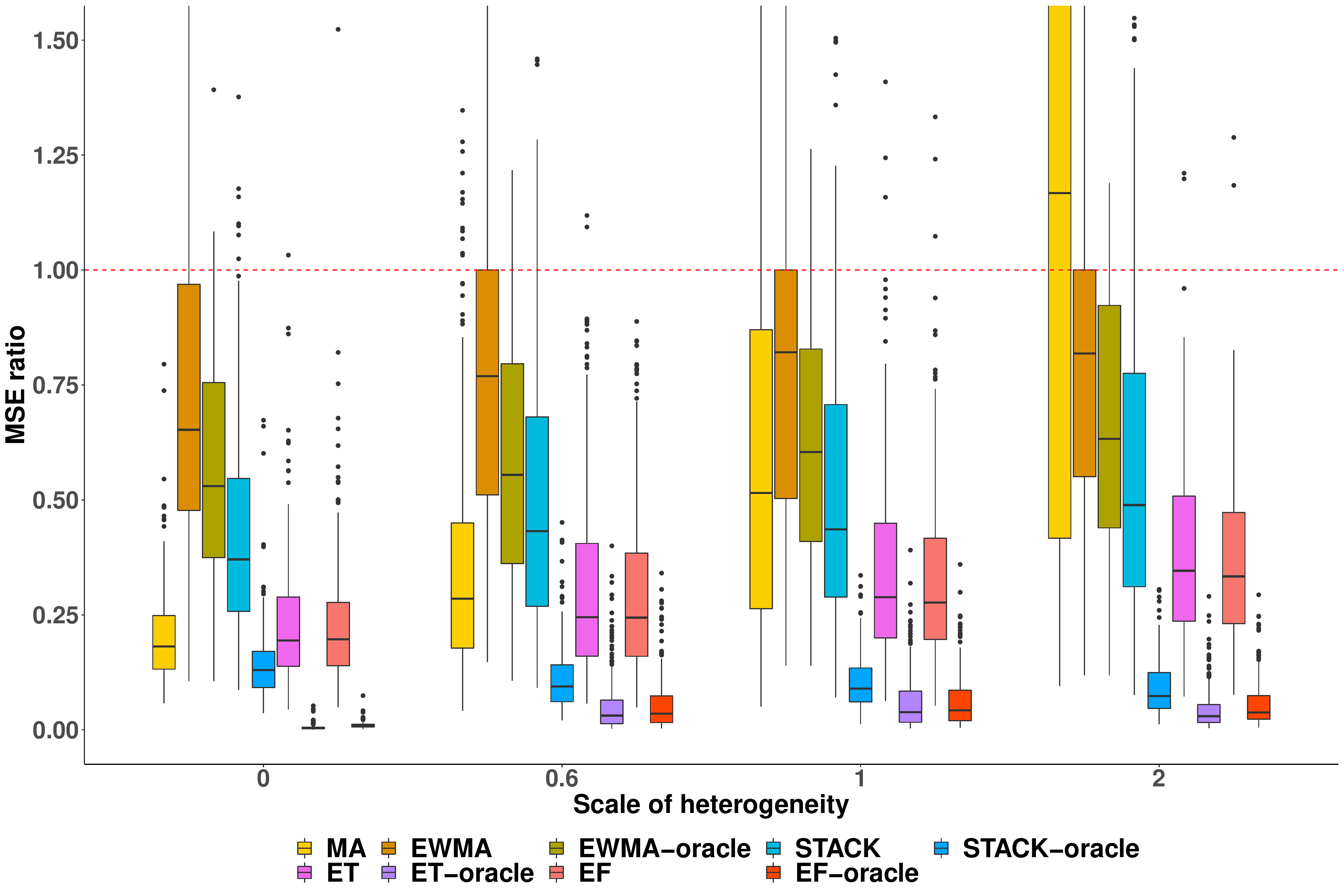}}
  \caption{}
 \end{subfigure}
%  \vspace{-1cm}
 \caption{
Box plots of the MSE ratios of CATE estimators, respectively, over LOC (\textbf{CT}) and a sample size of \textbf{500} at each site, and covariate dimension of \textbf{50} for (a) discrete grouping and (b) continuous grouping across site, respectively, varying scale of global heterogeneity. 
Estimators ending with ``-oracle" makes use of ground truth treatment effects. 
Different colors imply different estimators, and x-axis, i.e., the value of $c$, differentiates the scale of global heterogeneity. The red dotted line denotes an MSE ratio of 1. 
MA performance is truncated due to large MSE ratios. 
The proposed ET and EF achieve competitive performance compared to standard model averaging or ensemble methods and are robust to heterogeneity across settings. 
Note that ET-oracle and EF-oracle achieve close-to-zero MSE ratios with very small spreads in some settings. 
 }
 \label{web:sim_p50}
\end{figure}

%%%%%%%%%% diffN
\clearpage
\begin{figure}[hbt!]%[!h]%tp]
\centering
% \vspace{1cm}
 \begin{subfigure}{0.49\textwidth}
  \centerline{\includegraphics[width=\linewidth]{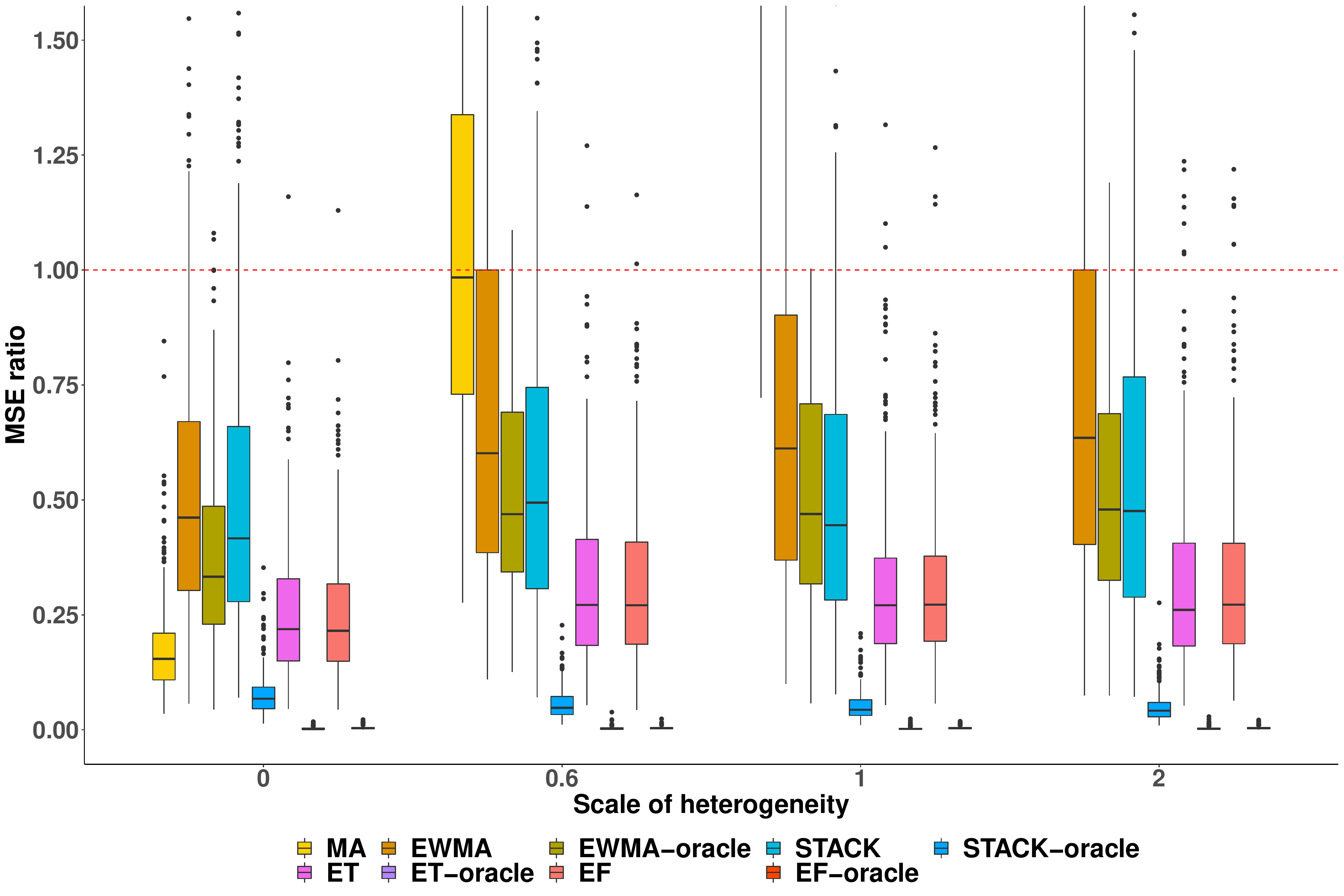}}
  \caption{}
 \end{subfigure}
 \begin{subfigure}{0.49\textwidth}
  \centerline{\includegraphics[width=\linewidth]{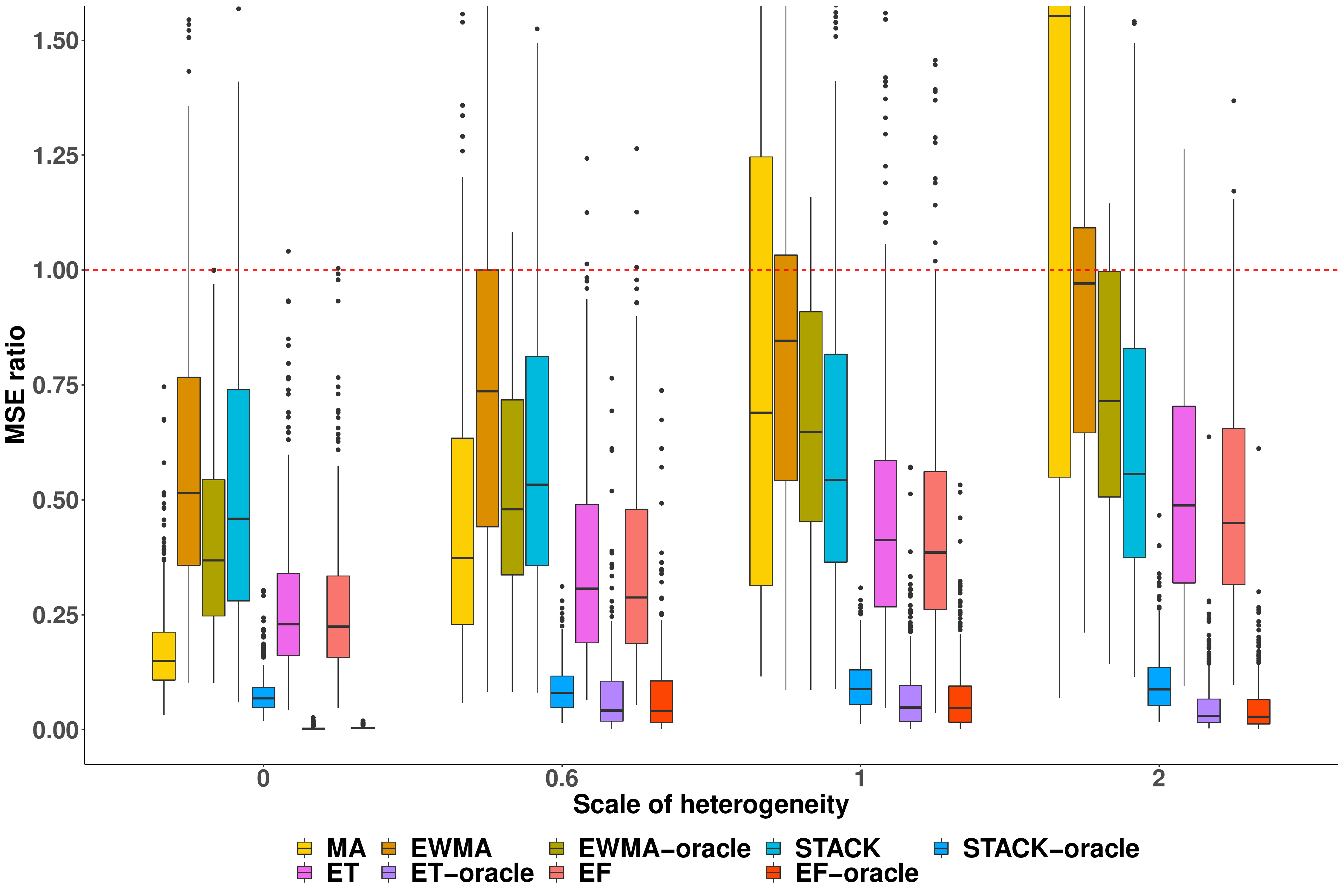}}
  \caption{}
 \end{subfigure}
%  \vspace{-1cm}
 \caption{
Box plots of the MSE ratios of CATE estimators, respectively, over LOC (\textbf{CT}) and a sample size of \textbf{500} at site 1, and a sample size of \textbf{200} at other sites for (a) discrete grouping and (b) continuous grouping across site, respectively, varying scale of global heterogeneity. 
Estimators ending with ``-oracle" makes use of ground truth treatment effects. 
Different colors imply different estimators, and x-axis, i.e., the value of $c$, differentiates the scale of global heterogeneity. The red dotted line denotes an MSE ratio of 1. 
MA performance is truncated due to large MSE ratios. 
The proposed ET and EF achieve competitive performance compared to standard model averaging or ensemble methods and are robust to heterogeneity across settings. 
Note that ET-oracle and EF-oracle achieve close-to-zero MSE ratios with very small spreads in some settings. 
 }
 \label{web:sim_diffN}
\end{figure}

%%%%%%%%%% CF100
\clearpage
\begin{figure}[hbt!]%[h]%tp]
\centering
% \vspace{1cm}
 \begin{subfigure}{0.49\textwidth}
  \centerline{\includegraphics[width=\linewidth]{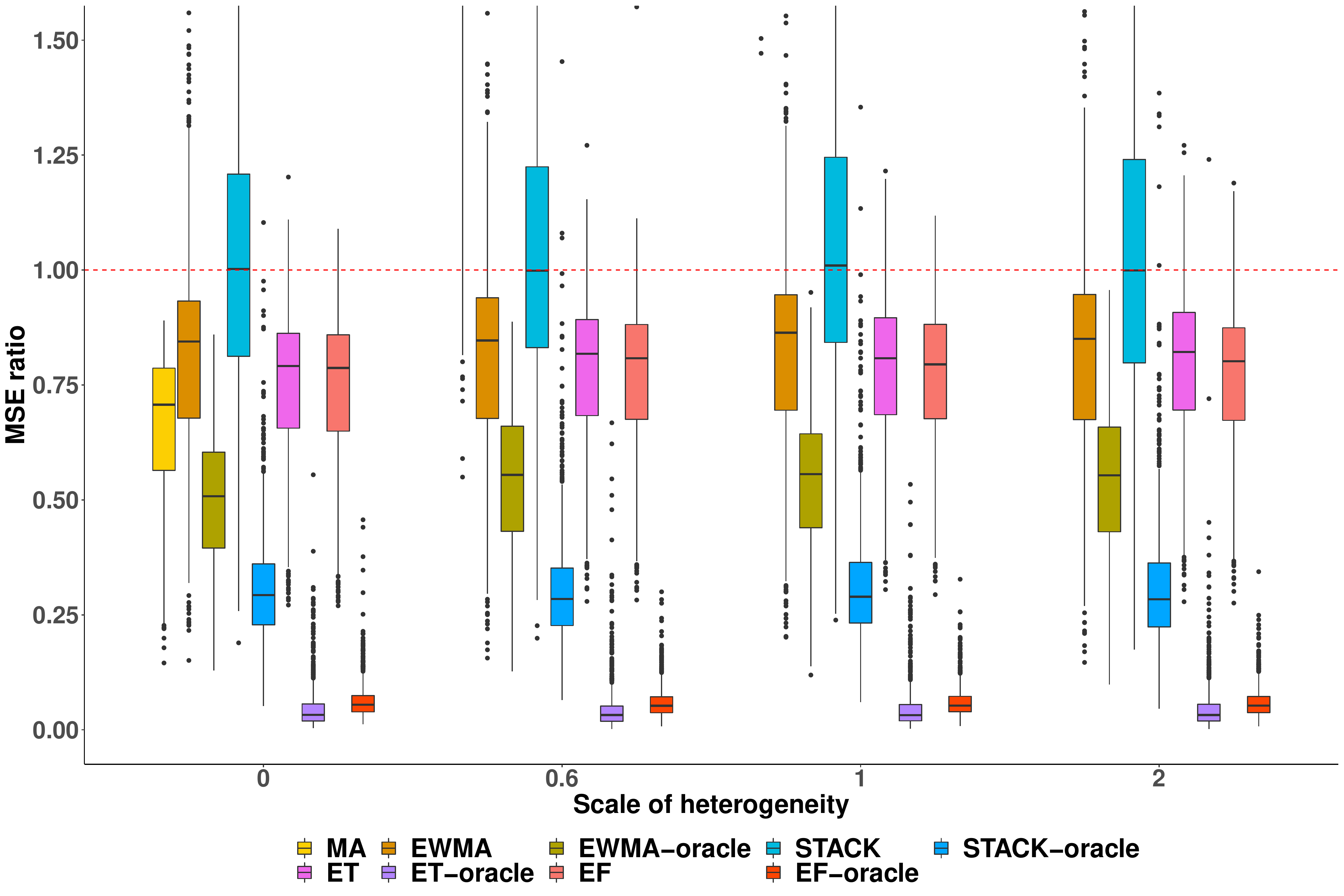}}
  \caption{}
%   \label{fig:disc}
 \end{subfigure}
 \begin{subfigure}{0.49\textwidth}
  \centerline{\includegraphics[width=\linewidth]{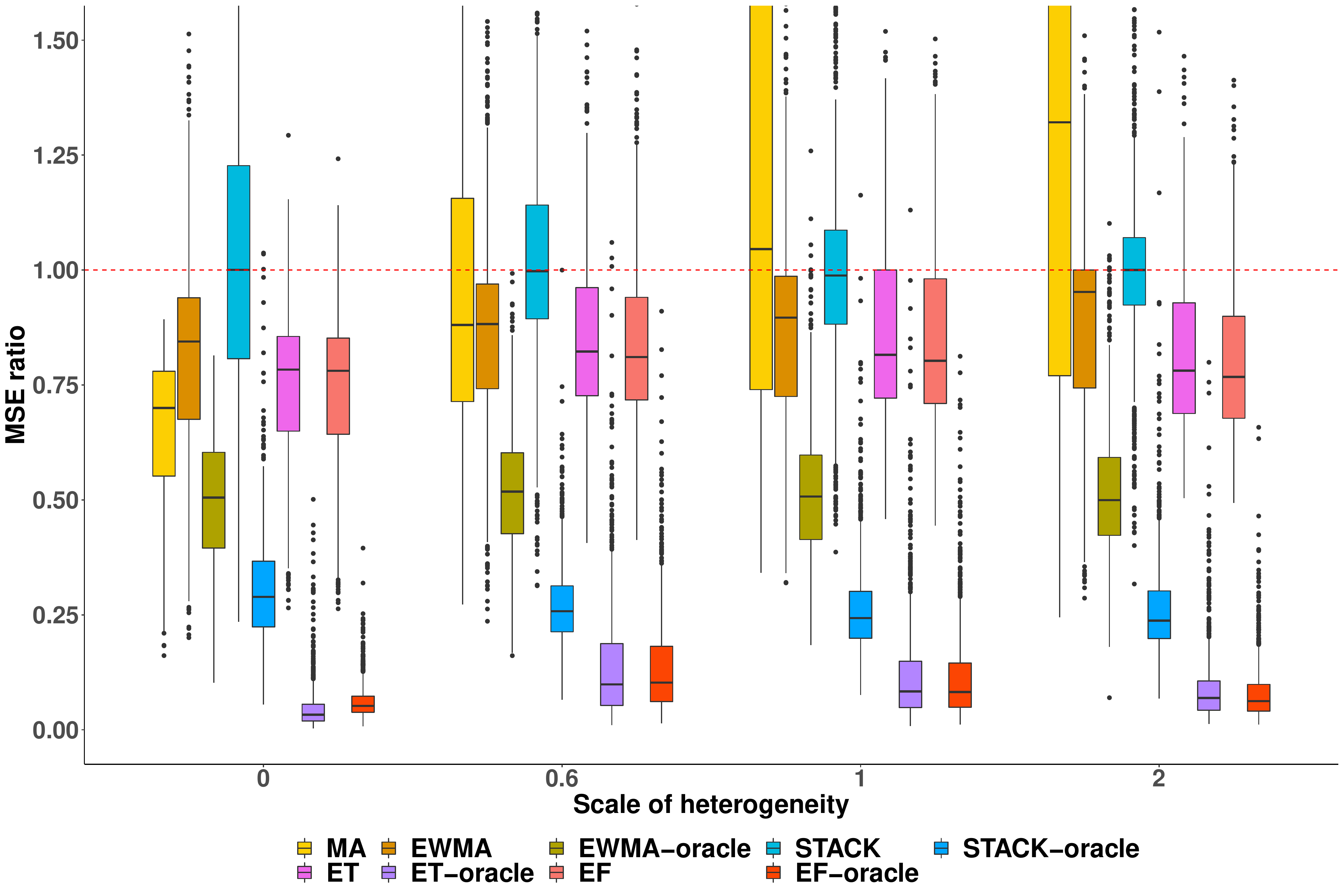}}
  \caption{}
%   \label{fig:cont}
 \end{subfigure}
%  \vspace{-1cm}
 \caption{
Box plots of the MSE ratios of CATE estimators, respectively, over LOC (\textbf{CF}) and a sample size of \textbf{100} at each site for (a) discrete grouping and (b) continuous grouping across site, respectively, varying scale of global heterogeneity. 
Estimators ending with ``-oracle" makes use of ground truth treatment effects. 
Different colors imply different estimators, and x-axis, i.e., the value of $c$, differentiates the scale of global heterogeneity. The red dotted line denotes an MSE ratio of 1. 
MA performance is truncated due to large MSE ratios. 
The proposed ET and EF achieve competitive performance compared to standard model averaging or ensemble methods and are robust to heterogeneity across settings. 
Note that ET-oracle and EF-oracle achieve close-to-zero MSE ratios with very small spreads in some settings. 
 }
 \label{web:sim_fig_cf100}
\end{figure}

%%%%%%%%%% CF500
\clearpage
\begin{figure}[hbt!]%[!h]%tp]
\centering
% \vspace{1cm}
 \begin{subfigure}{0.49\textwidth}
  \centerline{\includegraphics[width=\linewidth]{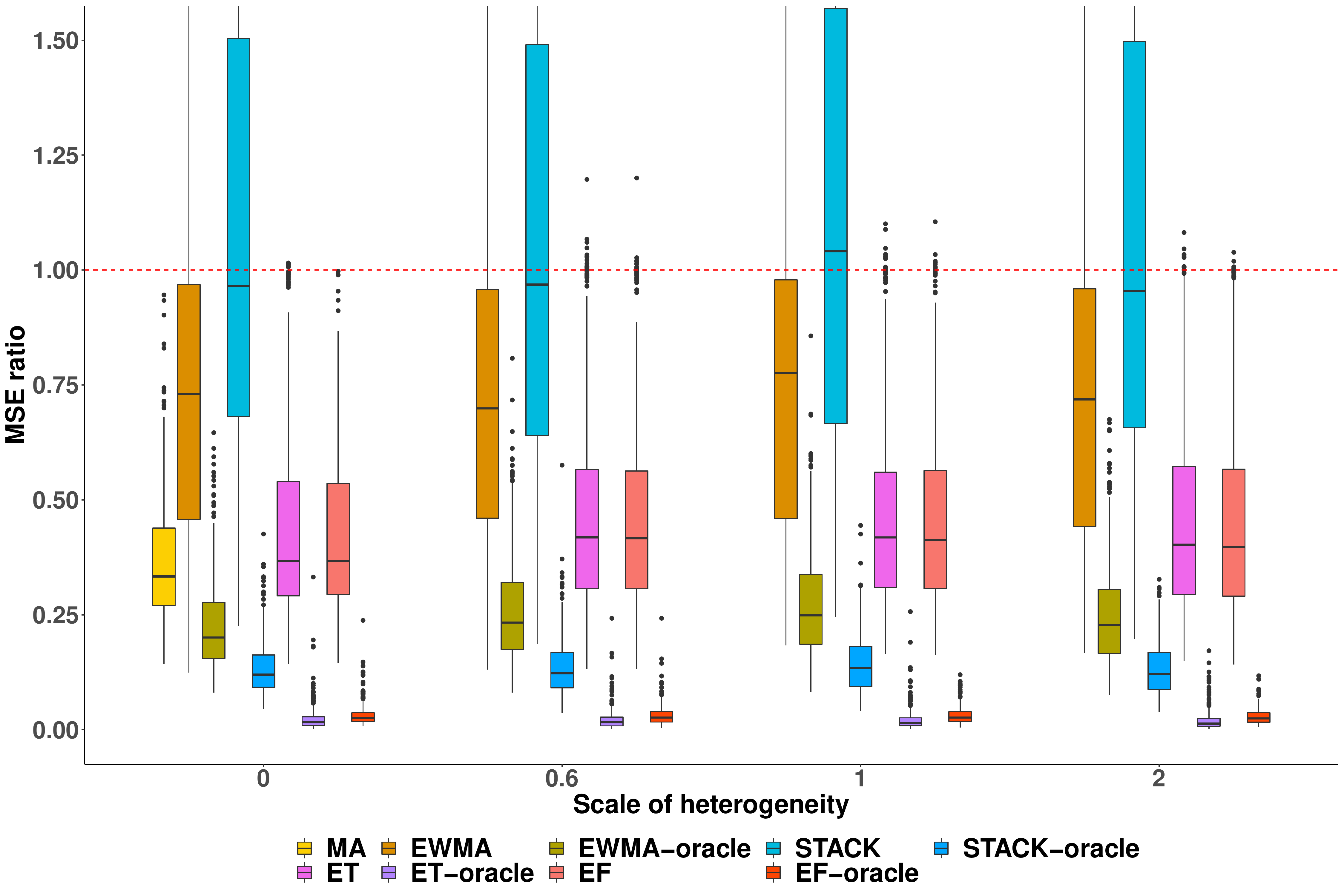}}
  \caption{}
%   \label{fig:disc}
 \end{subfigure}
 \begin{subfigure}{0.49\textwidth}
  \centerline{\includegraphics[width=\linewidth]{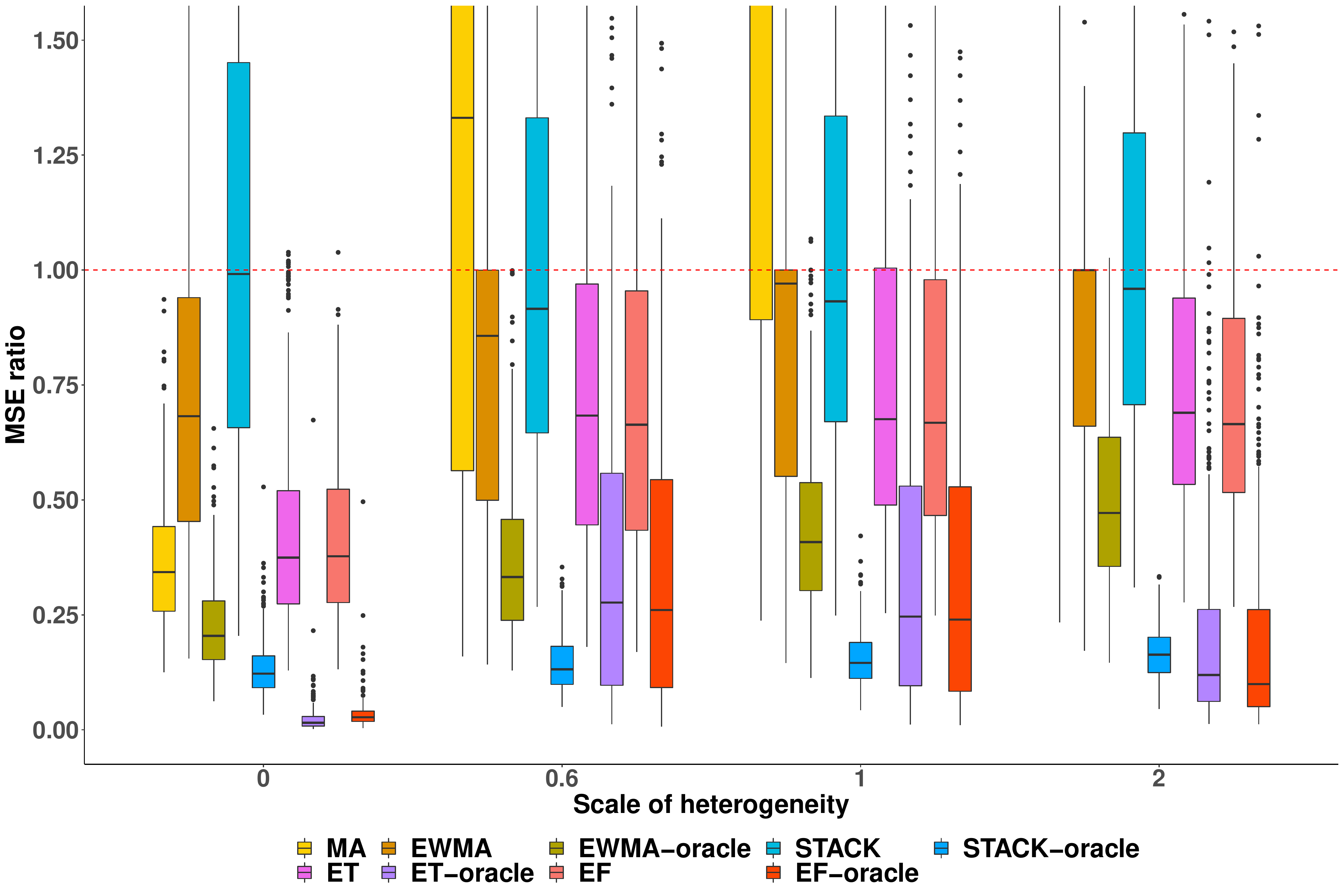}}
  \caption{}
%   \label{fig:cont}
 \end{subfigure}
%  \vspace{-1cm}
 \caption{
Box plots of the MSE ratios of CATE estimators, respectively, over LOC (\textbf{CF}) and a sample size of \textbf{500} at each site for (a) discrete grouping and (b) continuous grouping across site, respectively, varying scale of global heterogeneity. 
Estimators ending with ``-oracle" makes use of ground truth treatment effects. 
Different colors imply different estimators, and x-axis, i.e., the value of $c$, differentiates the scale of global heterogeneity. The red dotted line denotes an MSE ratio of 1. 
MA performance is truncated due to large MSE ratios. 
The proposed ET and EF achieve competitive performance compared to standard model averaging or ensemble methods and are robust to heterogeneity across settings. 
Note that ET-oracle and EF-oracle achieve close-to-zero MSE ratios with very small spreads in some settings. 
 }
 \label{web:sim_fig_cf500}
\end{figure}

%%%%%%%%%% CF1000
\clearpage
\begin{figure}[hbt!]%[!h]%tp]
\centering
% \vspace{1cm}
 \begin{subfigure}{0.49\textwidth}
  \centerline{\includegraphics[width=\linewidth]{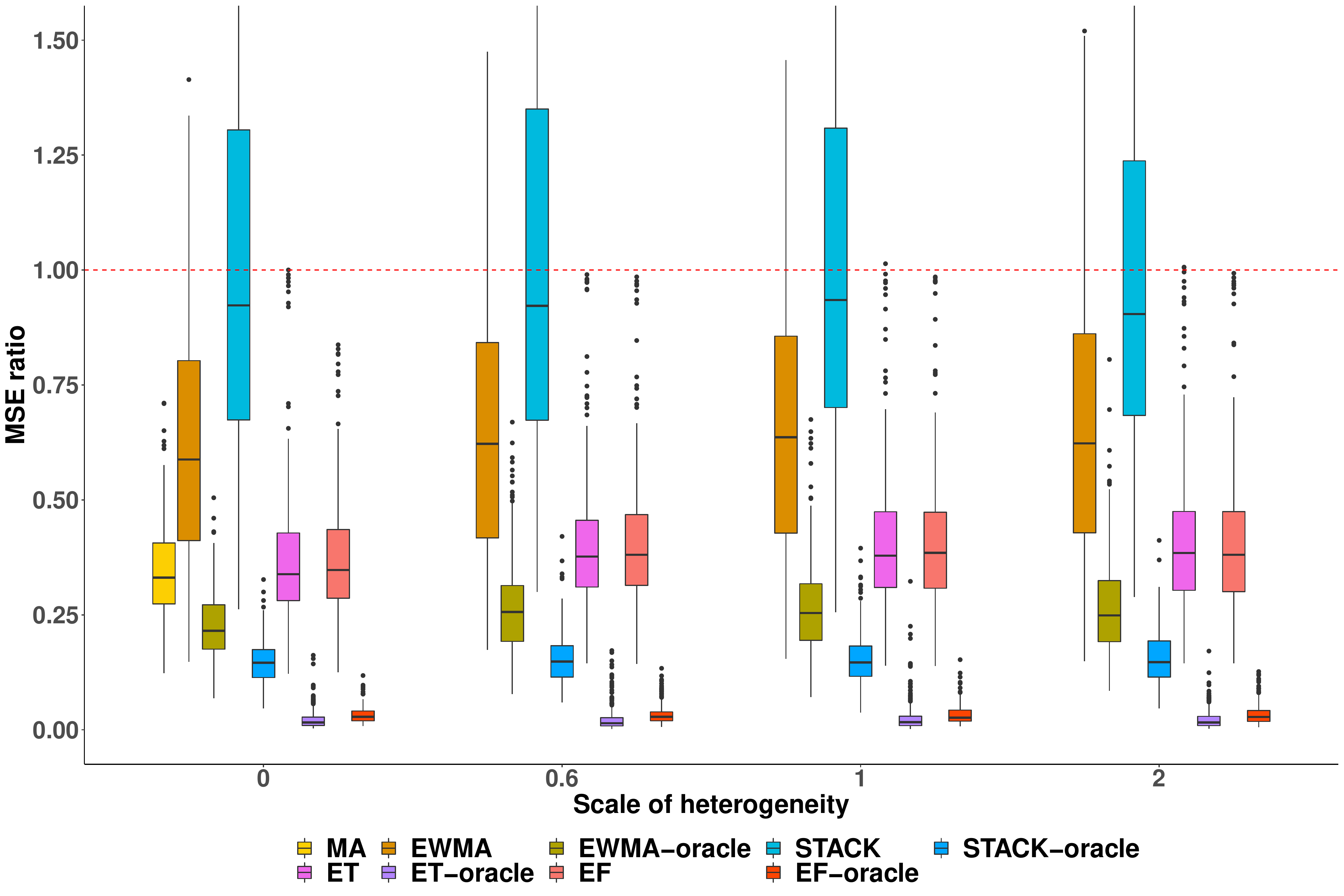}}
  \caption{}
%   \label{fig:disc}
 \end{subfigure}
 \begin{subfigure}{0.49\textwidth}
  \centerline{\includegraphics[width=\linewidth]{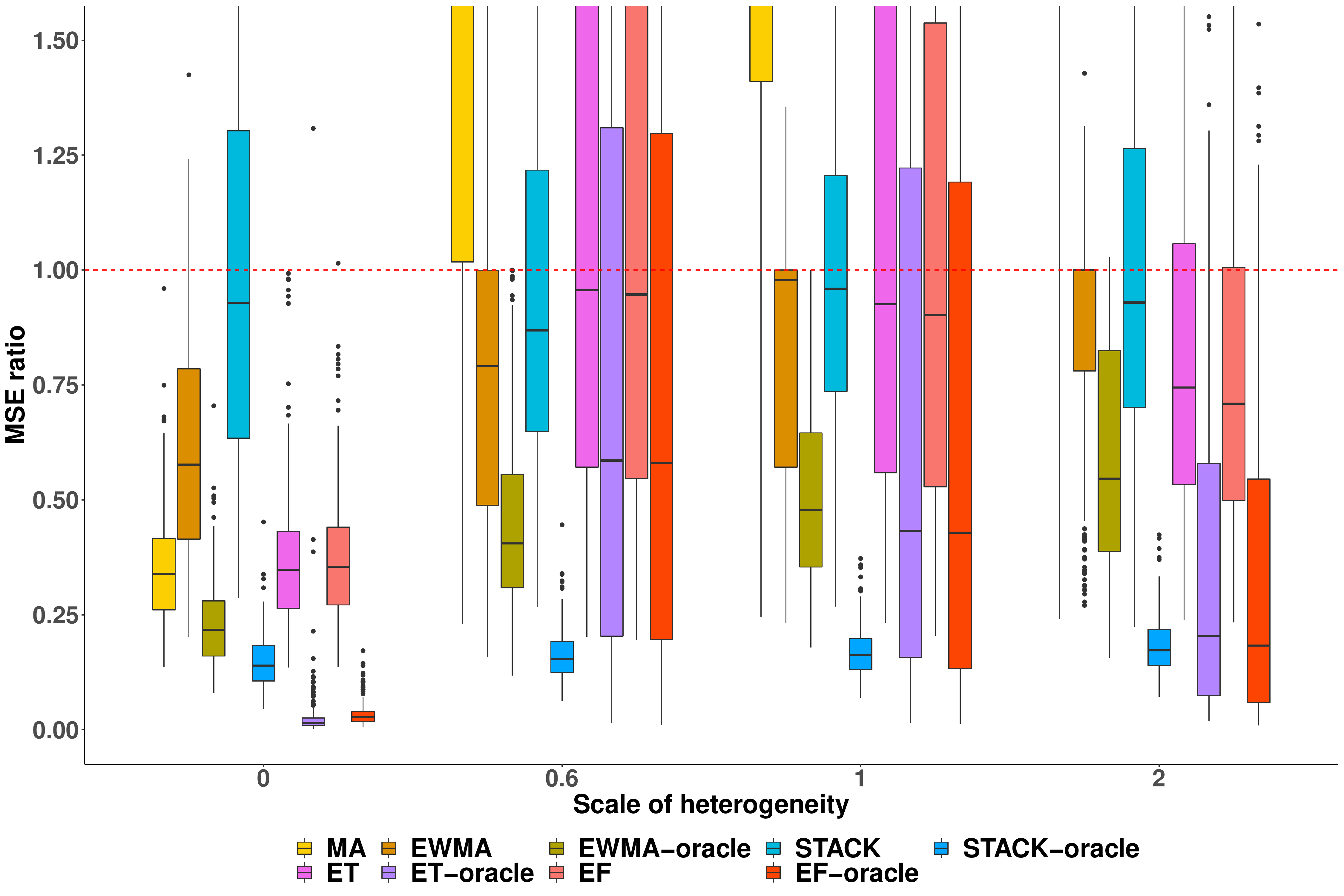}}
  \caption{}
%   \label{fig:cont}
 \end{subfigure}
%  \vspace{-1cm}
 \caption{
Box plots of the MSE ratios of CATE estimators, respectively, over LOC (\textbf{CF}) and a sample size of \textbf{1000} at each site for (a) discrete grouping and (b) continuous grouping across site, respectively, varying scale of global heterogeneity. 
Estimators ending with ``-oracle" makes use of ground truth treatment effects. 
Different colors imply different estimators, and x-axis, i.e., the value of $c$, differentiates the scale of global heterogeneity. The red dotted line denotes an MSE ratio of 1. 
MA performance is truncated due to large MSE ratios. 
The proposed ET and EF achieve competitive performance compared to standard model averaging or ensemble methods and are robust to heterogeneity across settings. 
Note that ET-oracle and EF-oracle achieve close-to-zero MSE ratios with very small spreads in some settings. 
 }
 \label{web:sim_fig_cf1000}
\end{figure}

\clearpage
\section{Additional Results for Data Application}
\label{suppl-sec:real}

In real-life applications, hospitals may have different sample sizes $n_k$ that may affect the accuracy of the estimation of $\tau_k$. Table \ref{web:hosp_smry} shows hospital-level information for the 20 hospitals where the number of patients across sites varies. Information includes the region of the U.S. where the hospital is located, whether it is a teaching hospital, the bed capacity, and the number of patients within the hospital.

Hospitals with a smaller sample size may not be representative of the population, leading to an uneven level of precision for local causal estimates. To account for different sample sizes at each hospital, we consider a basic weighting strategy where we add weights to each observation $\widehat \tau_k(\bx)$ in the augmented site 1 data adjusting for the sample size of site $k$. The weights are defined as
$
    \eta_k(\bx) = K n_k \{\sum_{j=1}^K n_j\}^{-1}.
$

\begin{table}[!htb]%[hbt!]%[h]%tp]
\centering
% \small
\caption{Hospital-level information of our analysis cohort in eICU database. Hospitals are relabeled according to their average contribution to the estimation task at hospital 1, the target site.}
% \vspace{-1cm}
\label{web:hosp_smry}
\resizebox{0.95\columnwidth}{!}{
\begin{tabular}{@{}rrrrlll@{}}
\toprule
\multicolumn{1}{l}{Hospital} &
  \multicolumn{1}{l}{Number of} &
  \multicolumn{1}{l}{Number of} &
  \multicolumn{1}{l}{Number of} &
  Bed &
  Teaching &
  \multirow{2}{*}{Region} \\
\multicolumn{1}{l}{site} &
  \multicolumn{1}{l}{patients} &
  \multicolumn{1}{l}{control} &
  \multicolumn{1}{l}{treated} &
  capacity &
  status &
   \\ \midrule
1  & 477 & 205 & 272 & $\geq$ 500 & False & South     \\
2  & 297 & 109 & 188 & $\geq$ 500 & True  & West      \\
3 & 163 & 58  & 105 & $\geq$ 500 & True  & Midwest   \\
4 & 222 & 58  & 164 & $\geq$ 500 & False & South     \\ %250 - 499
5 & 659 & 165 & 494 & $\geq$ 500 & True  & Midwest   \\
6  & 305 & 174 & 131 & $\geq$ 500 & False & South     \\
7 & 347 & 109 & 238 & $\geq$ 500 & True  & Midwest   \\
8  & 523 & 162 & 361 & $\geq$ 500 & False & South     \\
9  & 210 & 78  & 132 & Unknown             & False & Unknown   \\
10 & 379 & 161 & 218 & $\geq$ 500 & True  & Midwest   \\
11 & 234 & 70  & 164 & $\geq$ 500 & True  & Midwest   \\
12 & 747 & 185 & 562 & $\geq$ 500 & True  & Northeast \\
13  & 464 & 129 & 335 & $\geq$ 500 & True  & South     \\
14 & 474 & 229 & 245 & $\geq$ 500 & False & South     \\
15 & 166 & 64  & 102 & 100 - 249           & False & Midwest   \\
16  & 388 & 94  & 294 & $\geq$ 500 & False & Midwest   \\
17 & 435 & 240 & 195 & $\geq$ 500 & True  & South     \\
18 & 200 & 55  & 145 & 250 - 499           & False & South     \\
19  & 183 & 52  & 131 & 250 - 499           & False & West      \\
20  & 149 & 71  & 78  & 250 - 499           & False & South     \\
\bottomrule
\end{tabular}
}
\end{table}

Figure~\ref{web:real_wt} visualizes the performance of oxygen therapy on hospital survival with the weighting strategy adopted. 
CT is used as the local model with propensity score modeled by a logistic regression. 
Figure~\ref{web:real_wt}(a) shows the propensity score-weighted average survival for those whose received treatment is consistent with the estimated decision. 
Treatment rule based on our method can increase survival by 4\%, more promising than the EF estimates without the weighting strategy and the LOC and the baseline. The weighting strategy takes account into the unequal sample size among the hospital network, and assign weights based on precision of local estimates. 

In the fitted EF, 
the hospital indicator remains the most important, explaining about 48\% of the decrease in training error. 
Figure~\ref{web:real_wt}(b) shows the estimated CATEs varying two important features, BMI and oxygen therapy duration. 
Patients with BMI between 36 and 40 and duration above 400 show the most benefit from oxygen therapy in the target SpO$_2$ range. 
Patients with BMI between 20 and 30 and duration between 100 and 400 may not benefit from such alteration. 
The treatment estimates are similar to that in Figure~\ref{fig:real}(b)
Figure~\ref{web:real_wt}(c) visualizes 
the proposed model averaging scheme with
data-adaptive weights $\omega_{k}(\bx)$ in the fitted EF with respect to BMI for different models, while holding other variables constant. The weights of hospital 1 are quite stable while models from other sites may have different contribution to the weighted estimator for different values of BMI. 
Similar to Figure~\ref{fig:real}(c), hospitals with a larger bed capacity tend to be similar to hospital 1, and are shown to provide larger contributions. 
In general, the weighting strategy helps further improve the expected survival rate. The patterns in each subfigure are similar to Figure~\ref{fig:real}, which indicates the robustness of our proposed estimators.
We do stress that improvements to the weighting strategy for different sample sizes at each site are needed. 
A strategy considering both treatment proportion as well as covariate distributions across sites may further enhance the data-adaptive model averaging estimator.

\begin{figure*}[hbt!]%[htp]%[htp]
\centering
\vspace{1cm}
 \begin{subfigure}{0.2\textwidth}
  \centerline{\includegraphics[width=\linewidth]{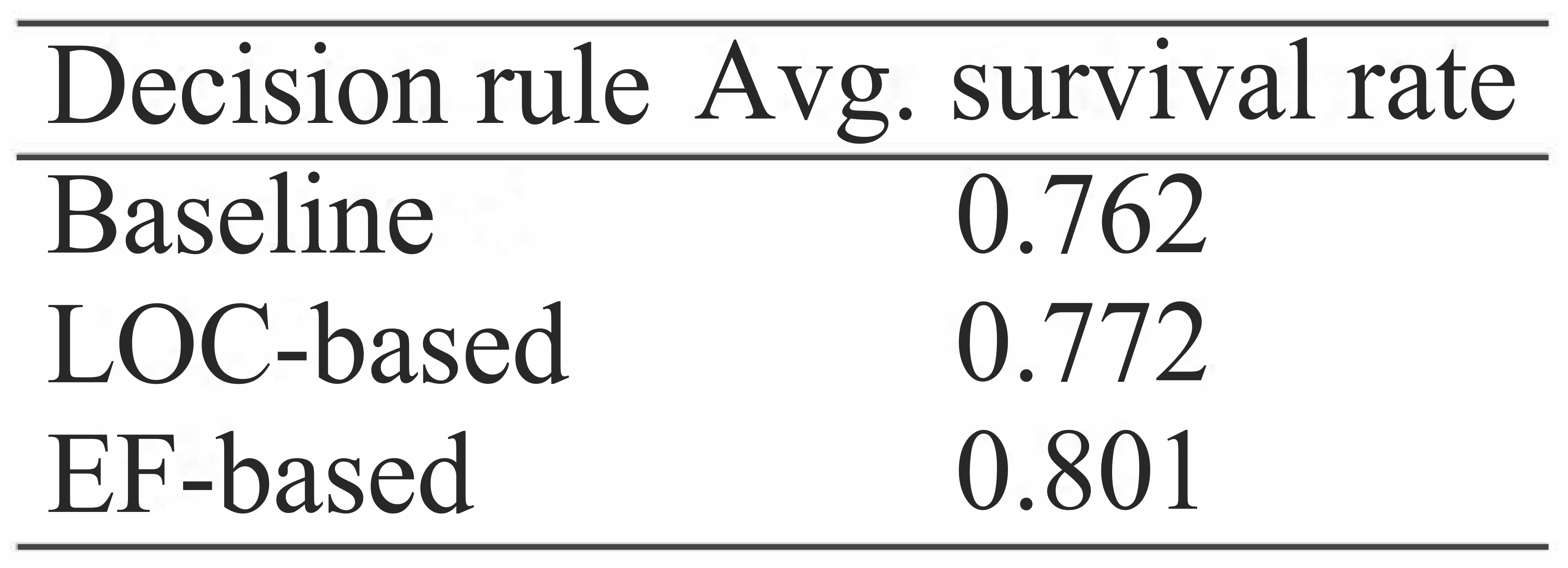}}
  \caption{}
 \end{subfigure}
 \begin{subfigure}{0.39\textwidth}
  \centerline{\includegraphics[width=0.85\linewidth]{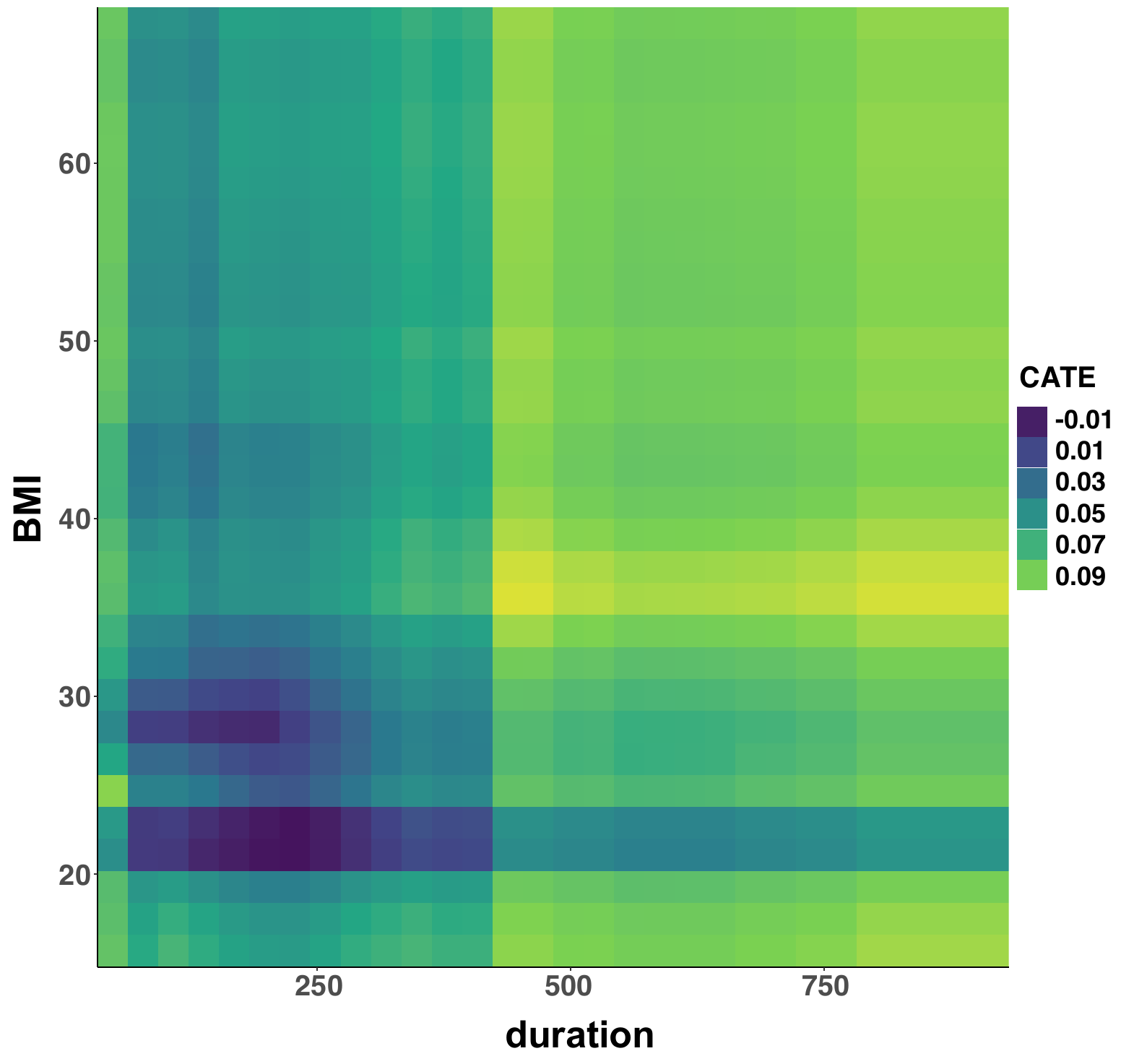}}
  \caption{}
 \end{subfigure}
 \begin{subfigure}{0.39\textwidth}
  \centerline{\includegraphics[width=0.85\linewidth]{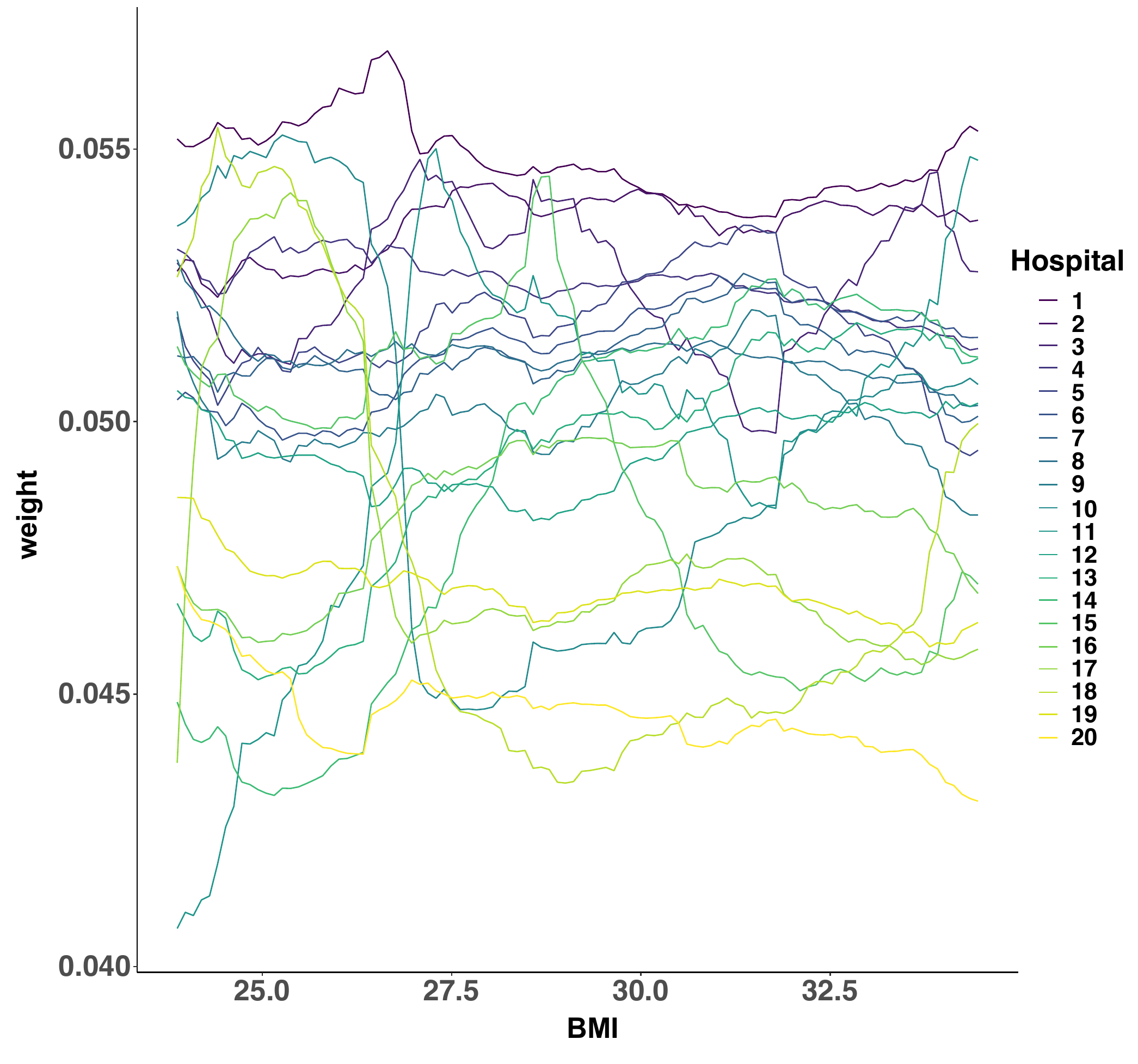}}
  \caption{}
 \end{subfigure}
 \vspace{-1cm}
 \caption{
 Application to estimating treatment effects of oxygen therapy on survival with a \textbf{sample size weighting strategy}. 
 (a) Expected survival of treatment decision following different estimators. 
 (b) Estimated treatment effects varying duration and BMI, two important features in the fitted EF. 
 (c) Visualization of data-adaptive weights in EF varying BMI. 
}
\label{web:real_wt}
\end{figure*}

\section{Real Data Access}
\label{suppl-sec:code}

Although the eICU data used in our application example cannot be shared subject to the data use agreement, access can be individually requested at \url{https://eicu-crd.mit.edu/gettingstarted/access/}.

\chapter{for Chapter 4}

\section{Additional Literature Review}\label{suppl:lit}

% IDR literature in general
Typical model-based methods for deriving IDRs include Q-learning such as \citep{watkins1992q,murphy2003optimal,moodie2007demystifying} and A-learning such as \citep{robins2000marginal,murphy2005experimental} where a model of responses is imposed and the optimal decision rule is obtained by optimizing value function derived from the model. 
Model-based methods posit a model of responses given observed covariates and treatment assignments, and obtain the optimal IDR by optimizing the corresponding value function derived from the model. 
Q-learning optimizes the corresponding value function derived from a parametric model of responses given observed covariates and treatment assignments, and it results in an optimal decision rule. 
A-learning is a semiparametric method, which derives from a model that directly describes the difference between treatments, with the baseline remaining unspecified. 
On the other hand, model-free methods such as \cite{zhang2012robust,zhao2012estimating,zhao2015doubly} assign values to actions simply through trial and error without pre-specifying a model. 
Besides, contextual bandit methods (see \cite{bietti2021contextual} and references therein) test out different actions and automatically learn which one has the most rewarding outcome for a given situation. 
See \citep{chakraborty2010inference,chakraborty2013statistical,laber2014dynamic,kosorok2015adaptive} and references therein for a comprehensive review on general IDRs under causal settings.

\section{Proofs of Propositions}\label{suppl:proof}

Here we show proofs of Propositions \ref{prop:obj} and \ref{prop:obj2} in Section~\ref{sec:diff_g}. 

\begin{proof}[Proof of Proposition \ref{prop:obj}]
We observe that to maximize the objective function in~\eqref{eq:obj_max} is equivalent to maximizing 
\begin{align*}
    & E_X \big[ G_{S|X} \{E (Y|X,S,A=d(X)) \} |X \big] \\
    &= E_X \big[ G_{S|X} \{E (Y|X,S,A=1) \}\mathbbm{1}(d(X) = 1)   \\
    &~~~~~~~~ + G_{S|X} \{E(Y|X,S,A=-1)\mathbbm{1}(d(X) = -1) \} \big] \\
    &= E_X \big\{ \mathbbm{1}(d(X) = 1) [G_{S|X} \{E (Y|X,S,A=1) \} - G_{S|X} \{E (Y|X,S,A=-1) \} ] \\
    &~~~~~~~~ + G_{S|X} \{E (Y|X,S,A=-1)\} \big\} \\
    & \propto E_X \big\{ \mathbbm{1}(d(X) = 1) [G_{S|X} \{E (Y|X,S,A=1) \}- G_{S|X} \{E (Y|X,S,A=-1) \}] \big\}.
\end{align*}
\end{proof}

\begin{proof}[Proof of Proposition \ref{prop:obj2}]

Let $d(x) = \operatorname{sgn}\{f(x)\}$, by this transformation, we consider the following objective on a smooth function $f(x)$, 
\begin{align*}
    &  {\arg\max}_{d\in\mathbb{D}}   \frac{1}{n} \sum_{i = 1}^{n} \big\{ \mathbbm{1}(d(x_i)=1) [g_1(x_i)-g_2(x_i)] \big\} \\
    & = {\arg\max}_{f}  \frac{1}{n} \sum_{i = 1}^{n} \mathbbm{1}[ \operatorname{sgn}\{f(x_i)\} = 1] \cdot [g_1(x_i)-g_2(x_i)] \\
    & =  {\arg\min}_{f}  \frac{1}{n} \sum_{i = 1}^{n} \mathbbm{1}\{ 1 \cdot f(x_i) < 0 \} \cdot [g_1(x_i)-g_2(x_i)] \\
    & =  {\arg\min}_{f}  \frac{1}{n} \sum_{i = 1}^{n} \mathbbm{1}[ \operatorname{sgn}\{g_1(x_i)-g_2(x_i)\} \cdot f(x_i) <0 ] \cdot |g_1(x_i)-g_2(x_i)|.
\end{align*}
The sign of the estimated $f$ above is a $d$ to \eqref{eq:obj2}. 

Hence, the proposed classification-based objective is to minimize
\begin{align*} %\label{eq:obj_min}
    \frac{1}{n} \sum_{i = 1}^{n} \mathbbm{1}[ \operatorname{sgn}\{g_1(x_i)-g_2(x_i)\} \cdot f(x_i) <0 ] \cdot |g_1(x_i)-g_2(x_i)|.  
\end{align*}
To this point, we have transformed the optimization problem \eqref{eq:obj_max} into a weighted classification problem where for subject $i$ with features $x_i$, the true label is $\operatorname{sgn}\{g_1(x_i)-g_2(x_i)\}$ and the sample weight is $|g_1(x_i)-g_2(x_i)|$. 

\end{proof}

\section{Details on Modeling and Hyperparameter Tuning}\label{suppl:tune}

In our implementation, neural networks with mean or quantile losses are used to fit the models with hyperparameters tuned via a 5-fold cross validation in the training data sets. 
Specifically, implemented in TensorFlow \citep{abadi2016tensorflow}, neural networks with mean squared loss is used to model $E(Y|X,S,A)$ separated by the control arm and the treatment arm, respectively. 
For continuous $S$, to model $Q_{S|X,A} \{ E (Y|X,S,A) \}$, neural networks with quantile loss is used with a prespecified $\tau$, for the control arm and the treatment arm, respectively. 
In the final weighted classification model, neural networks with cross-entropy loss is used. 
Note that the model choices here are flexible. One can perform model selection if they would like to.

Hyperparameter tuning helps prevent overfitting and is essential in machine learning methods or other black-box algorithms such as neural networks. In our implementation, the optimal hyperparameters are obtained via a 5-fold cross validation in the training data sets. Specifically, we consider 
the number of hidden layers (1, 2, and 3 layers), 
the number of hidden units in each layer (256, 512, and 1024 nodes), 
activation function (RELU, Sigmoid, and Tanh), 
optimizer (Adam, Nadam, and Adadelta), 
dropout rate (0.1, 0.2, and 0.3), 
number of epochs (50, 100, and 200), 
and batch size (32, 64, and 128).

\section{Additional Simulations}\label{suppl:sim}

\subsection{Different Quantile Criteria} 
For the quantile criteria, we also consider $\tau$ = 0.1 and 0.5, respectively. Table~\ref{t:sim_exp2_q0.1_0.5} presents the simulation results for Example 2 with continuous $S$ using 0.1 quantile criterion and 0.5 quantile criterion, respectively. 
Results show that when $\tau$ is small, there is more strength in the proposed method, as the algorithm aims to improve the worst-outcome scenarios. The proposed RISE has the largest gain in objective and value among vulnerable subjects when $\tau$ is 0.1, and has similar performance as the compared approaches when $\tau$ is 0.5.

\subsection{$S$ as a Noise Variable}
We generate the outcome $Y$ using the following model where $S$ is not involved: 
$Y = \mathbbm{1}(X_1\leq0.5)\{8+12\mathbbm{1}(A=1)+16\exp(X_2)-26\mathbbm{1}(A=1)X_2\} + \mathbbm{1}(X_1>0.5)\{13+3\mathbbm{1}(A_i=1)+2\exp(X_2)-8\mathbbm{1}(A=1)X_2\} + \epsilon$, 
where $X_j \sim U(0,1), ~j=1,2$, $A \sim Bernoulli(0.5)$, and $\epsilon \sim N(0,1)$. For continuous $S$, $S = \expit\{-2.5(1-X_1-X_2)\}$; for discrete $S$, we consider a binary $S$ that satisfies $\log \{P(S = 1 |X)/P(S = 0 |X)\} = -2.5(1-X_1-X_2)$. 
Table~\ref{t:sim_notimp} summarizes the performance of the proposed IDRs compared to the mean criterion for Example 2. The estimated objective and value function are similar for the compared IDRs, which indicates the robustness of the proposed RISE. 
% As $S$ does not contribute to the generation of $Y$. All IDRs achieve very similar performance. This indicates the robustness of the proposed IDRs. 

\subsection{Violations of Causal Assumptions } 

To further test the robustness of the proposed RISE, we consider scenarios where the causal ssumptions in Section~\ref{sec:assump} may not hold. 
To test the violation of positivity assumption in Assumption~\ref{assump:rise-pos}, using the same setting as in Example 2, we consider an extreme propensity score, or the probability of being treated given $X$ and $S$. 
Specifically, we let $A$ satisfy $\log \{P(A_i = 1 |X_i)/P(A = 0 |X_i)\} = -1.2(-S_i + X_{i1} - X_{i2} + X_{i3} - X_{i4} + X_{i5} - X_{i6})$. 
To test the unconfoundedness assumption in Assumption~\ref{assump:rise-unconf}, a random normal noise, $e \sim N(0,1)$ is added to $X_1$ in the setting of Example 2. 
The simulation results are presented in Table~\ref{t:sim_pos} and Table~\ref{t:sim_unconf} respectively.

\clearpage
\begin{table}[!htb]
\centering
\caption{Simulation results for Example 2 with continuous $S$ using 0.1 quantile criterion and 0.5 quantile criterion, respectively. Standard error in parenthesis. 
The proposed RISE has more strengths when $\tau$ is small, as the algorithm aims to improve the worst-outcome scenarios. 
}
\label{t:sim_exp2_q0.1_0.5}
\resizebox{0.98\columnwidth}{!}{
% \begin{footnotesize}
\begin{tabular}{@{}ccccccc@{}}
\toprule
Type of $S$ & $\tau$ & IDR & Obj. (all) & Obj. (vulnerable) & Value (all) & Value (vulnerable) \\ \midrule
\multirow{5}{*}{Cont.} & \multirow{5}{*}{0.1}
    & Base & 7.93 (0.03) & 7.92 (0.03) & 17.7 (0.02) & 8.64 (0.07) \\
    &   & Exp  & 8.88 (0.05) & 8.85 (0.05) & {17.8} (0.02) & 10.6 (0.12) \\
    &  & {PT-Base}   & {6.97 (0.02)} & {6.95 (0.02)} & {17.9 (0.03)} & {6.65 (0.04)} \\
    & & {PT-Exp}   & {7.11 (0.02)} & {7.08 (0.03)} & {\textbf{18.0} (0.03)} & {6.96 (0.05)} \\
    &   & RISE & \textbf{13.8} (0.01) & \textbf{13.7} (0.02) & 16.9 (0.01) & \textbf{20.9} (0.03) \\
  \midrule
\multirow{5}{*}{Cont.} & \multirow{5}{*}{0.5}
    & Base & 17.3 (0.04) & 17.2 (0.04) & 17.7 (0.02) & 23.8 (0.19) \\
    &    & Exp  & 17.2 (0.03) & \textbf{17.4} (0.03) & {17.8} (0.02) & 22.1 (0.17) \\
    &  & {PT-Base}   & {17.3 (0.05)} & {17.3 (0.05)} & {18.0 (0.03)} & {23.9 (0.26)} \\
    &  & {PT-Exp}   & {\textbf{17.4} (0.05)} & {\textbf{17.4} (0.05)} & {\textbf{18.1} (0.03)} & {\textbf{24.0} (0.25)} \\
    &    & RISE & \textbf{17.4} (0.04) & \textbf{17.4} (0.04) & {17.8} (0.02) & \textbf{24.0} (0.22) \\
 \bottomrule
 \vspace{0.5cm}
\end{tabular}
% \end{footnotesize}
}
\end{table}

\clearpage
\begin{table}[!htb]
\centering
\caption{Simulation results for scenario when $S$ is a noise variable. Vulnerable subjects cannot be defined as $S$ is not important in the example. The estimated objective and value function are similar for the compared IDRs, which indicates the robustness of the proposed RISE. }
\label{t:sim_notimp}
\resizebox{0.98\columnwidth}{!}{
% \begin{scriptsize}
\begin{tabular}{@{}cccccc@{}}
\toprule
Type of $S$ & IDR & Obj. (all) & Obj. (vulnerable) & Value (all) & Value (vulnerable) \\ \midrule
\multirow{5}{*}{Disc.} 
 & Base & 27.5 (0.03) & - & 27.5 (0.06) & - \\
 & Exp   & 27.5 (0.03) & - & 27.5 (0.06) & - \\
  & {PT-Base}   & {27.5 (0.02)} & {-} & {27.5 (0.03)} & {-} \\
 & {PT-Exp}   & {27.5 (0.02)} & {-} & {27.5 (0.03)} & {-} \\
 & RISE   & 27.5 (0.03) & - & 27.5 (0.06) & - \\
 \midrule
\multirow{5}{*}{Cont.} 
 & Base & 27.2 (0.04) & - & 27.3 (0.07) & - \\
 & Exp   & 27.2 (0.04) & - & 27.3 (0.07) & - \\
 & {PT-Base}   & {27.2 (0.04)} & {-} & {27.3 (0.07)} & {-} \\
 & {PT-Exp}   & {27.2 (0.04)} & {-} & {27.3 (0.07)} & {-} \\
 & RISE   & 27.2 (0.04) & - & 27.3 (0.07) & - \\ 
 \bottomrule
\end{tabular}
% \end{scriptsize}
}
\end{table}

\clearpage
\begin{table}[!htb]
\centering
\caption{Simulation results for Example 2 where the positivity assumption in Assumption~\ref{assump:rise-pos} is nearly violated. Standard error in parenthesis. }
\label{t:sim_pos}
\resizebox{0.98\columnwidth}{!}{
% \begin{scriptsize}
\begin{tabular}{@{}cccccc@{}}
\toprule
Type of $S$ & IDR & Obj. (all) & Obj. (vulnerable) & Value (all) & Value (vulnerable) \\ \midrule
\multirow{5}{*}{Disc.} 
 & Base & 10.0 (0.03) & 11.1 (0.03) & 19.3 (0.02) & 16.1 (0.04) \\
 & Exp  & 8.80 (0.03) & 9.77 (0.04) & \textbf{19.5} (0.02) & 13.6 (0.04) \\
  & {PT-Base} & 9.88 (0.03) & 10.7 (0.04) & 18.9 (0.02) & 15.4 (0.05) \\
 & {PT-Exp}   & 8.42 (0.03) & 9.14 (0.04) & 19.1 (0.02) & 12.4 (0.05) \\
 & RISE & \textbf{13.5} (0.01) & \textbf{14.0} (0.01) & 17.3 (0.01) & \textbf{22.0} (0.02) \\
 \midrule
\multirow{5}{*}{Cont.} 
 & Base  & 11.5 (0.03) & 11.5 (0.04) & 17.5 (0.03) & 13.1 (0.04) \\
 & Exp   & 10.4 (0.04) & 10.4 (0.05) & 17.8 (0.04) & 10.3 (0.05) \\
 & {PT-Base}  & 11.0 (0.04) & 10.9 (0.04) & 17.7 (0.02) & 11.8 (0.04) \\
 & {PT-Exp}   & 9.63 (0.03) & 9.61 (0.03) & \textbf{18.0} (0.02) & 8.38 (0.03) \\
 & RISE  & \textbf{14.3} (0.01) & \textbf{14.3} (0.02) & 16.9 (0.01) & \textbf{20.4} (0.02) \\ 
 \bottomrule
\end{tabular}
% \end{scriptsize}
}
\end{table}

\clearpage
\begin{table}[!htb]
\centering
\caption{Simulation results for Example 2 where the unconfoundedness assumption in Assumption~\ref{assump:rise-unconf} is violated. Standard error in parenthesis. }
\label{t:sim_unconf}
\resizebox{0.98\columnwidth}{!}{
% \begin{scriptsize}
\begin{tabular}{@{}cccccc@{}}
\toprule
Type of $S$ & IDR & Obj. (all) & Obj. (vulnerable) & Value (all) & Value (vulnerable) \\ \midrule
\multirow{5}{*}{Disc.} 
 & Base   & 7.65 (0.04) & 8.44 (0.05) & 19.3 (0.03) & 11.1 (0.06) \\
 & Exp    & 8.94 (0.05) & 9.91 (0.06) & \textbf{19.4} (0.02) & 13.9 (0.06) \\
  & {PT-Base} & 6.84 (0.03) & 7.35 (0.04) & 18.9 (0.03) & 8.91 (0.05) \\
 & {PT-Exp}   & 7.95 (0.05) & 8.62 (0.06) & 19.1 (0.03) & 11.4 (0.06) \\
 & RISE   & \textbf{13.5} (0.01) & \textbf{14.0} (0.01) & 17.4 (0.01) & \textbf{22.1} (0.02) \\
 \midrule
\multirow{5}{*}{Cont.} 
 & Base & 9.58 (0.03) & 9.58 (0.03) & 17.9 (0.02) & 8.33 (0.05) \\
 & Exp  & 10.2 (0.04) & 10.2 (0.04) & 17.8 (0.02) & 9.83 (0.06) \\
 & {PT-Base}  & 9.27 (0.02) & 9.26 (0.03) & 17.9 (0.02)  & 7.51 (0.03) \\
 & {PT-Exp}   & 9.34 (0.02) & 9.34 (0.03) & \textbf{18.0} (0.02)  & 7.72 (0.03) \\
 & RISE & \textbf{14.2} (0.01) & \textbf{14.1} (0.02) & 16.9 (0.01) & \textbf{20.1} (0.03) \\ 
 \bottomrule
\end{tabular}
% \end{scriptsize}
}
\end{table}

\clearpage
\section{Additional Information and Results for Real-data Applications}\label{suppl:real}

\subsection{Data Availability } 

The job training dataset \citep{lalonde1986evaluating} is available at  \url{https://users.nber.org/~rdehejia/data/.nswdata2.html}. The ACTG175 dataset \citep{hammer1996trial} is available in the R package \texttt{speff2trial}. The sepsis dataset \citep{seymour2016assessment} is proprietary and not publicly available. 
All data used in this work are deidentified.

\subsection{Additional Background on the Sepsis Application }

Sepsis is leading cause of acute hospital mortality and commonly results in multi-organ dysfunction among ICU patients \citep{sakr2018sepsis,onyemekwu2022associations}. Clinically, treatment decisions for sepsis patients are needed to be made within a short period of time due to the rapid deterioration of patient conditions. 
Lactate and the Sequential Organ Failure Assessment (SOFA) score have been two important indicators of sepsis severity and has been found to be more useful for predicting the outcome of sepsis than other clinical vitals and comorbidity scores \citep{howell2007occult, krishna2009evaluation,shankar2016developing,machicado2021mortality}. 
Typically, information of baseline patient characteristics such as age, gender, race, and weight, and common vital signs such as usage of mechanical ventilation, respiratory rate, temperature, intravenous fluids, Glasgow Coma Scale score, platelets, blood urea nitrogen, white blood cell counts, glucose, and creatinine are obtained at the admission of patients. 
On the other hand, SOFA score combines performance of several organ systems in the body such as neurologic, blood, liver, and kidney \citep{seymour2016assessment,liu2019regional,tan2018changepoint,koutroumpakis2021serum} and cannot be obtained directly. Lactate labs measures the level of lactic acid in the blood \citep{andersen2013etiology,prathapan2020peripheral,du2017cbinderdb} and are less common in routine examination, which could be delayed in ordering. 
Hence, their information may not be available by the time of treatment decision due to multiple reasons including doctors' delayed ordering, long laboratory processing time, or the rapid deterioration of development of sepsis, which poses tremendous difficulties for early diagnosis and treatment decisions within a short time. 
According to the new definition of Sepsis-3 \citep{shankar2016developing}, a serum lactate level greater than 2 mmol/L is considered to be in critical conditions and is highly likely to indicate a septic shock. Also, a SOFA score greater than 6 has been associated with a higher mortality \citep{vincent1996sofa,ferreira2001serial}.

\subsection{Visualizations } 
Here we provide visualizations of features that are important in the estimated decision rules for the three real-data applications in Section~\ref{ssec:app}. The Shapely additive explanations (SHAP) \citep{lundberg2017unified} is considered to be a united approach to explaining the predictions of any machine learning or black-box models. Figure~\ref{fig:job_shap}, Figure~\ref{fig:actg_shap}, and Figure~\ref{fig:sepsis_shap} presents the SHAP variable importance plots in the final weighted classification model by RISE and Exp, respectively, for the three real-data applications. Correlations between the feature and their SHAP value are highlighted in color. The red color means a feature is positively correlated with assigning treatment A = 1 and the blue indicates a negative correlation. 
Overall, the direction of correlation is similar for RISE and Exp, but their ranking of feature importance may be different.

\textbf{Fairness in a job training program. } Figure~\ref{fig:job_shap} presents the SHAP variable importance plots in the final weighted classification model by RISE and Exp, respectively. 
We observe that whether having a high school diploma and income in 1974 are two important features in the variable important plot by RISE, while incomes in 1974 and in 1975 are important by Exp. It seems that being no degree and low income in 1974 has a higher chance of assigning $A=1$ (to receive the job training program) by RISE, while low income in 1974 and but a higher income in 1975 may be associated with assigning $A=1$ by Exp.

\textbf{Improvement of HIV treatment. }
Figure~\ref{fig:actg_shap} presents the SHAP variable importance plots in the final weighted classification model by RISE and Exp, respectively. 
We observe that age and CD4 T-cell counts are two important features in the variable important plot by RISE, while weight and number of days of previously received antiretroviral therapy are important by Exp. 
It seems that being of a younger age and high CD4 T-cell count has a higher chance of assigning $A=1$ (zidovudine combined with didanosine) by RISE, while being of a larger weight and few days of previously received antiretroviral therapy may be associated with assigning the treatment by Exp.

\textbf{Safe resuscitation for patients with sepsis. }
Figure~\ref{fig:sepsis_shap} presents the SHAP variable importance plots in the final weighted classification model by RISE and Exp, respectively. 
We observe that Glasgow Coma Scale score, age, and platelets appears to be important features in both the plot by RISE and that by Exp. 
Other important features in the plot by RISE include temperature and blood urea nitrogen, where in the plot by Exp, respiratory rate and white blood cell counts are of top importance. 
Being in a low temperature with a high blood urea nitrogen tends to be predicted as $A=1$ (to assign vasopressors) by RISE while being of higher respiratory rate with high white blood cell counts tends to be predicted as $A=1$ by Exp.

\clearpage
\begin{figure}[!htb]
\centering
 \begin{subfigure}{0.48\textwidth}
  \centerline{\includegraphics[width=\linewidth]{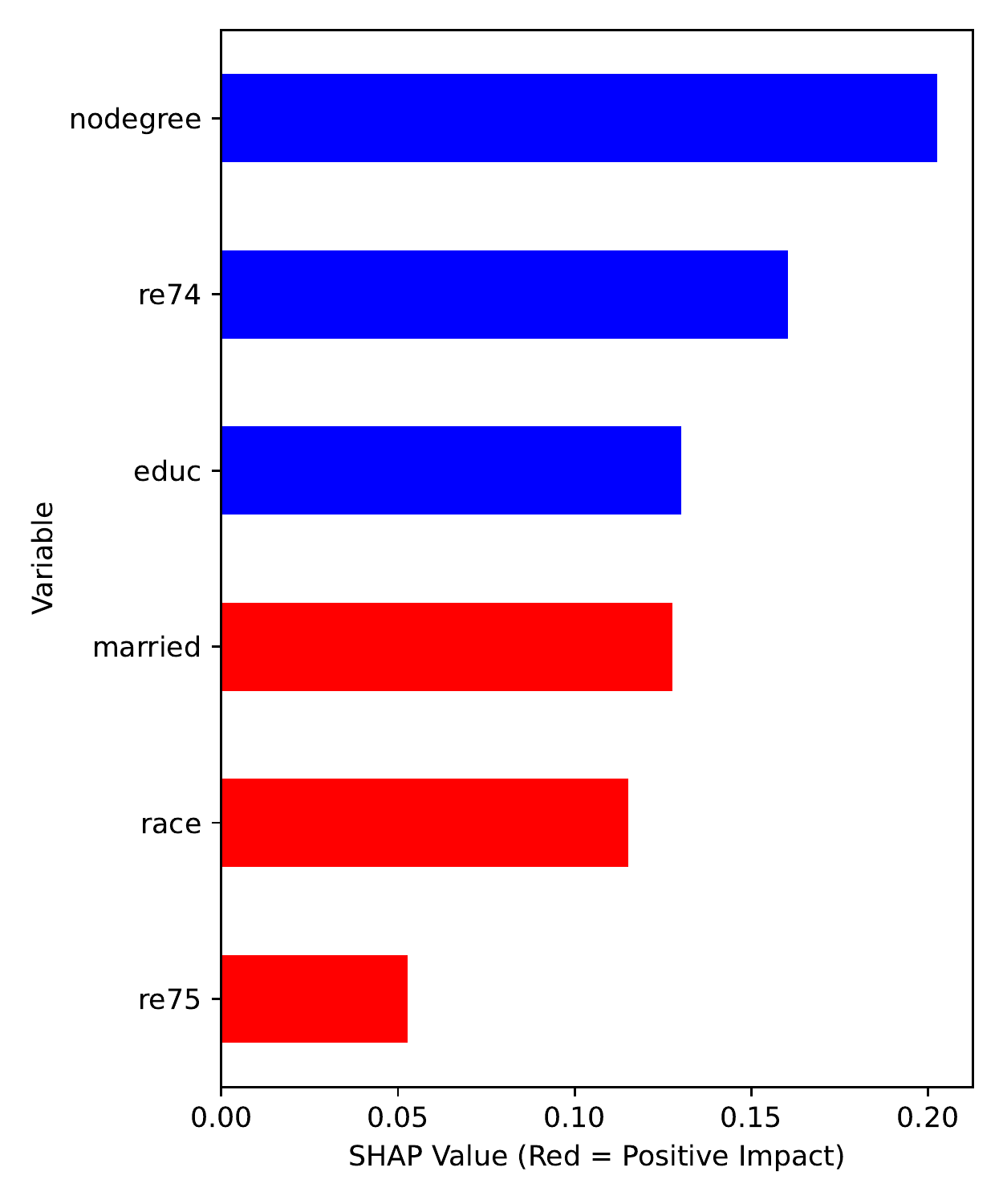}}
  \caption{}
 \end{subfigure}
 \begin{subfigure}{.48\textwidth}
  \centerline{\includegraphics[width=\linewidth]{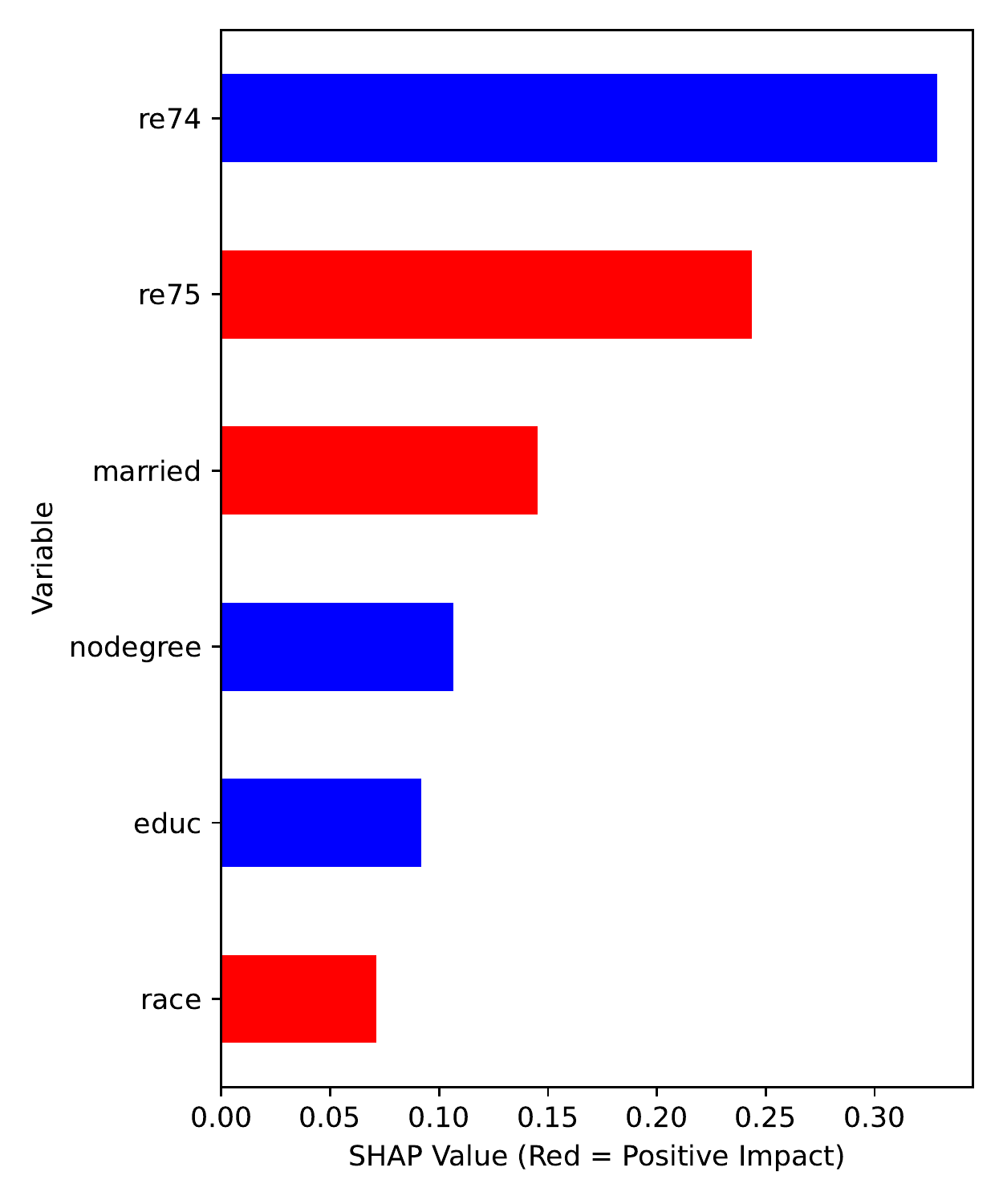}}
  \caption{}
 \end{subfigure}
 \caption{Visualization for the job training program: SHAP variable importance plots for decision rules RISE (a) and Exp (b), respectively. 
 Covariates ($X$) are ranked by variable importance in descending order. 
 Correlations between the feature and their SHAP value are highlighted in color. The red color means a feature is positively correlated with assigning treatment $A=1$ and the blue indicates a negative correlation. 
 } 
 \label{fig:job_shap}
\end{figure}

\clearpage
\begin{figure}[!htb]
\centering
 \begin{subfigure}{0.48\textwidth}
  \centerline{\includegraphics[width=\linewidth]{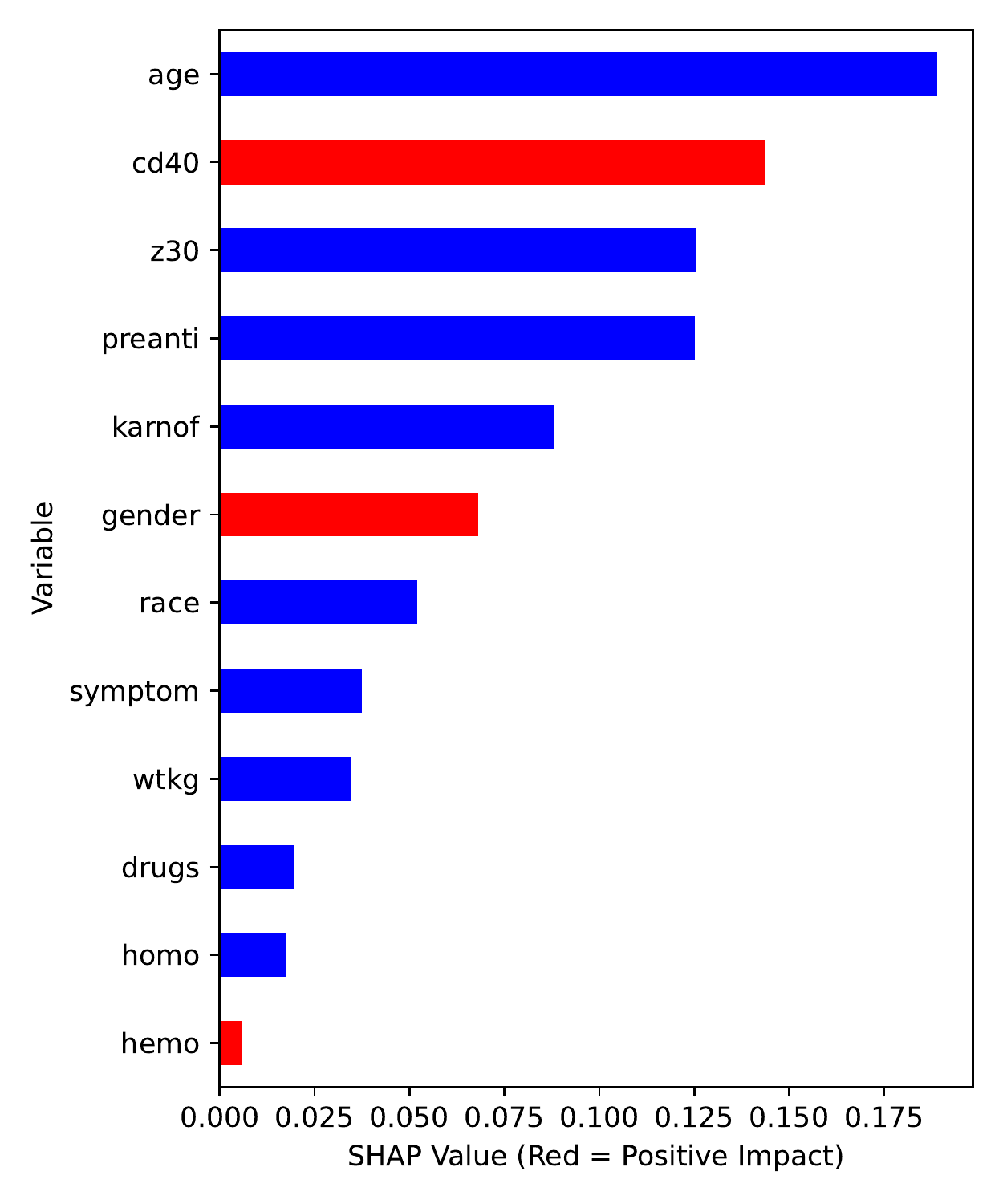}}
  \caption{}
 \end{subfigure}
 \begin{subfigure}{.48\textwidth}
  \centerline{\includegraphics[width=\linewidth]{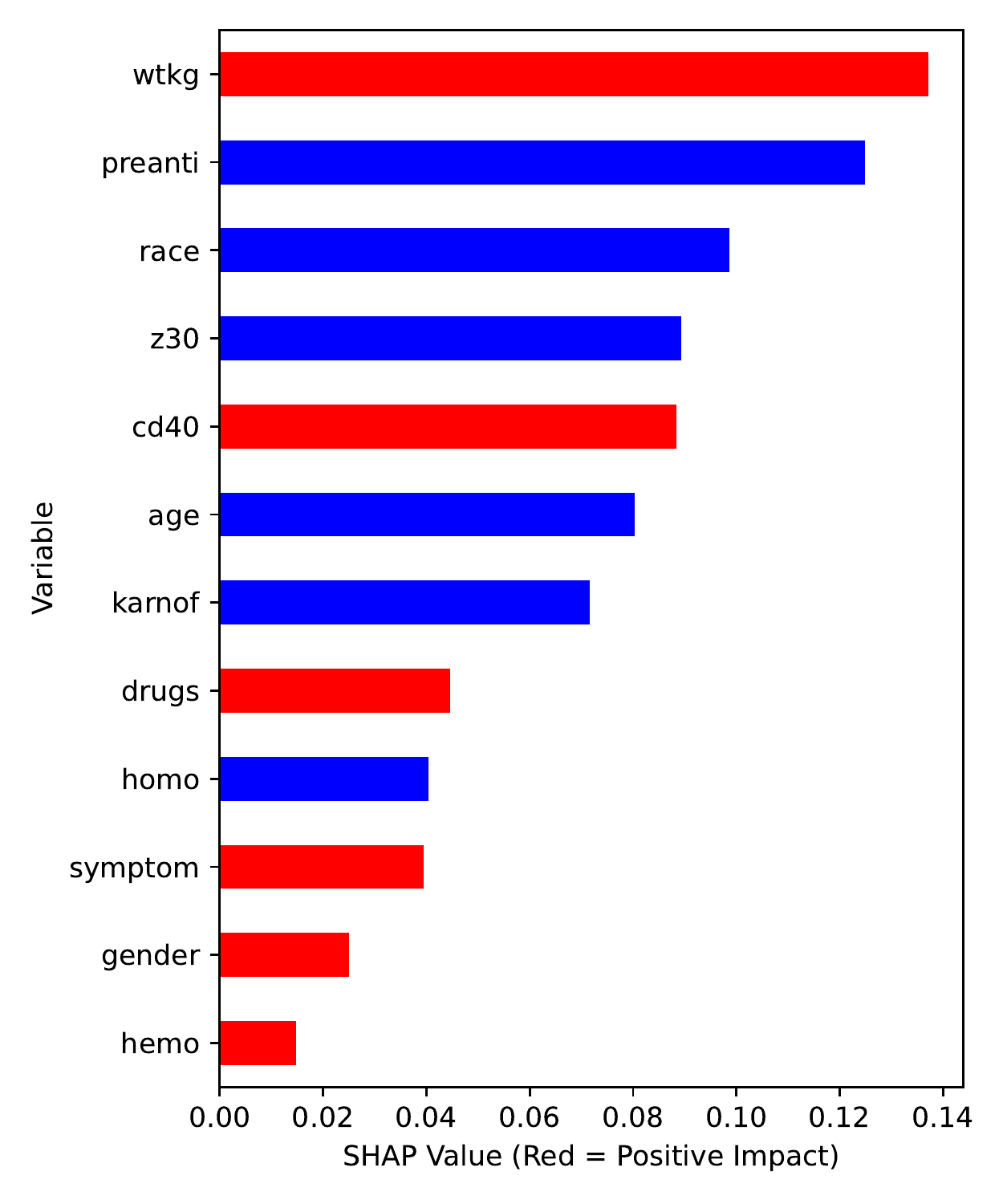}}
  \caption{}
 \end{subfigure}
 \caption{Visualization for the ACTG175 dataset: SHAP variable importance plots for decision rules RISE (a) and Exp (b), respectively. 
 Covariates ($X$) are ranked by variable importance in descending order. 
 Correlations between the feature and their SHAP value are highlighted in color. The red color means a feature is positively correlated with assigning treatment $A=1$ and the blue indicates a negative correlation. 
 } 
 \label{fig:actg_shap}
\end{figure}

\clearpage
\begin{figure}[!htb]
\centering
 \begin{subfigure}{0.48\textwidth}
  \centerline{\includegraphics[width=\linewidth]{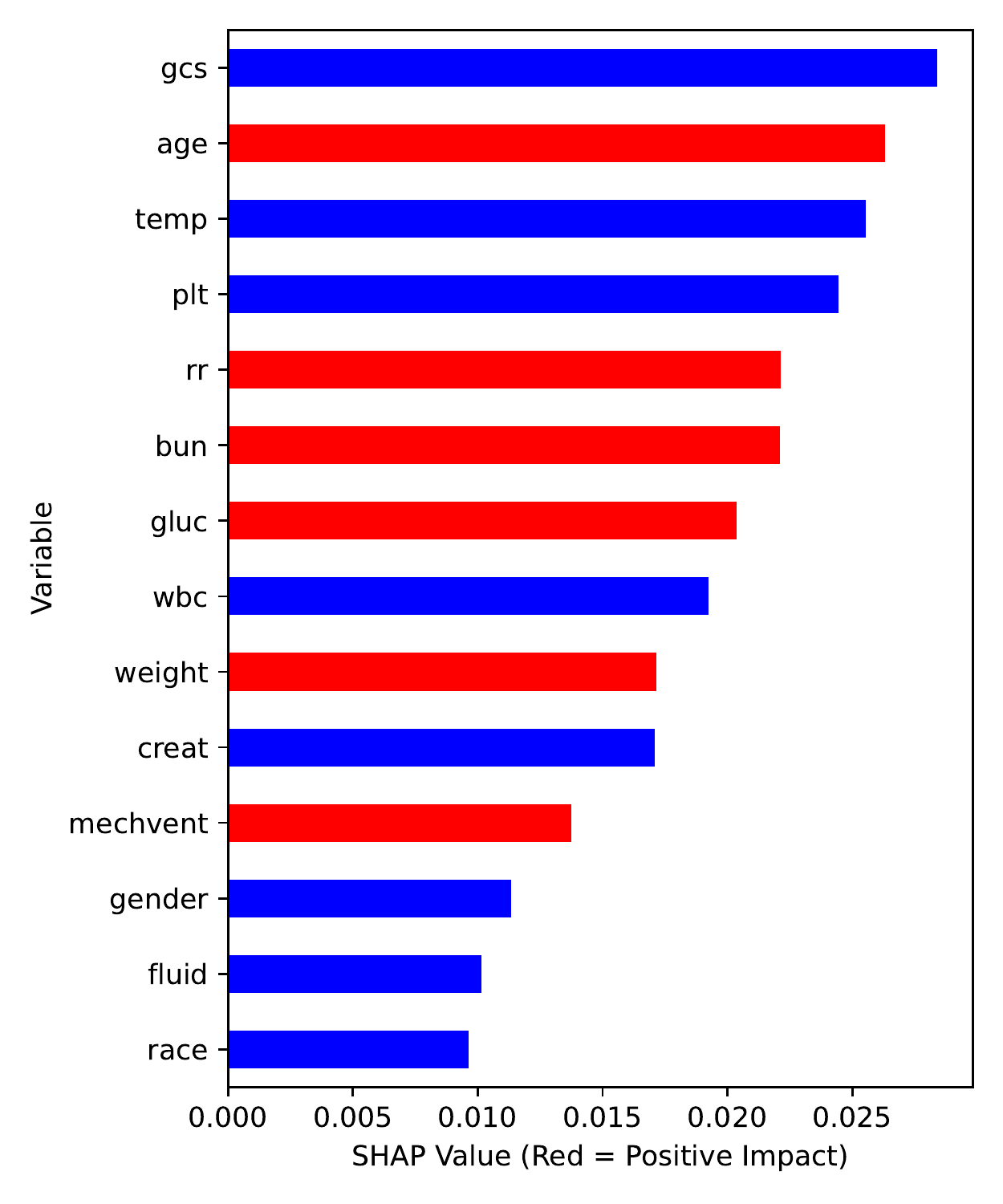}}
  \caption{}
 \end{subfigure}
 \begin{subfigure}{.48\textwidth}
  \centerline{\includegraphics[width=\linewidth]{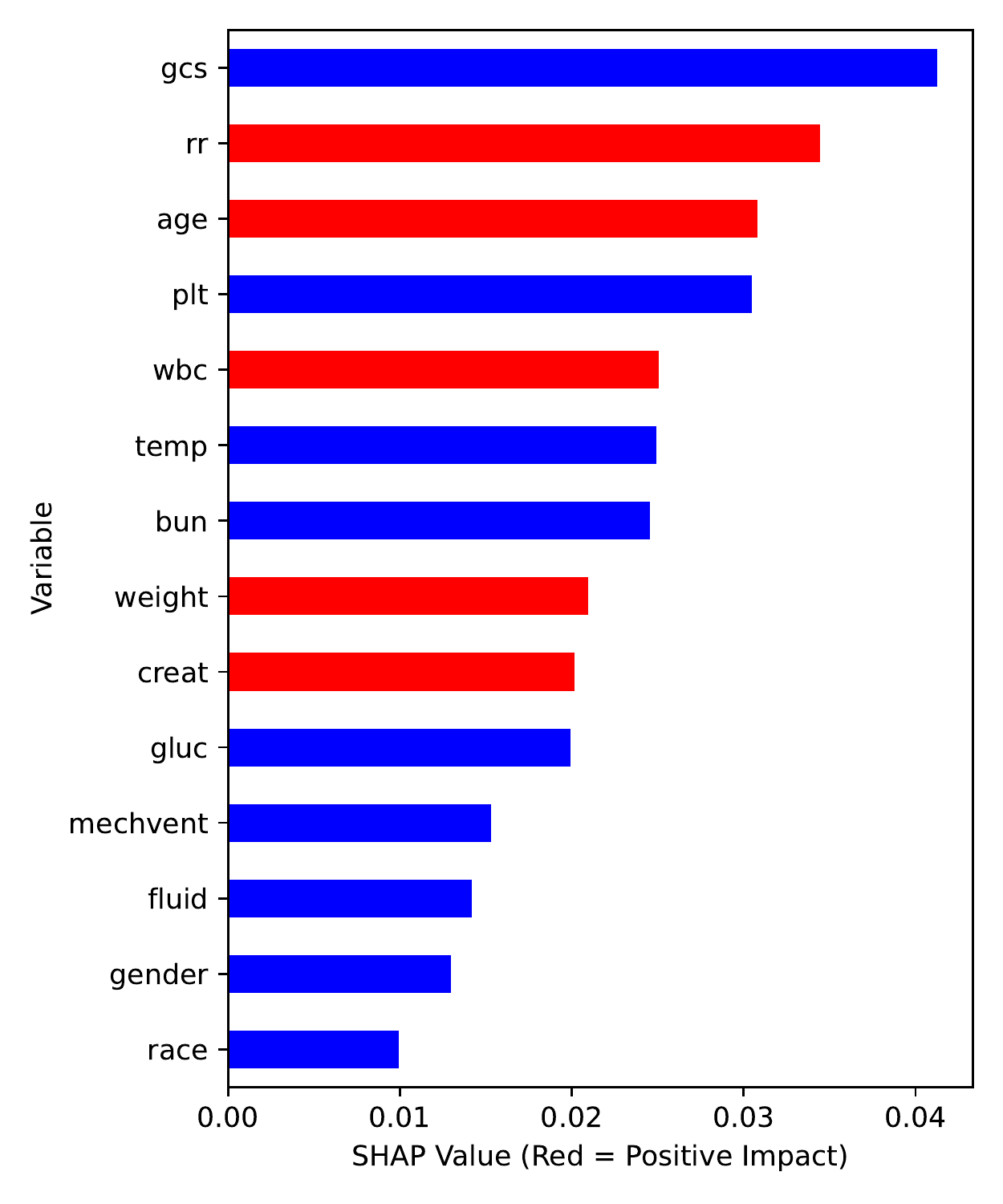}}
  \caption{}
 \end{subfigure}
 \caption{Visualization for the sepsis data: SHAP variable importance plots for decision rules RISE (a) and Exp (b), respectively. 
 Covariates ($X$) are ranked by variable importance in descending order. 
 Correlations between the feature and their SHAP value are highlighted in color. The red color means a feature is positively correlated with assigning treatment $A=1$ and the blue indicates a negative correlation. 
 } 
 \label{fig:sepsis_shap}
\end{figure}

% \chapter{Examples and Results}
% %==========================================================================================%
% %==========================================================================================%
% \begin{figure}[t]
%     \centering
%     \includegraphics[width=0.9\textwidth]{figures/hist_EFS_OS.eps}
%     \caption{Caption: Image Example}
    
%     \label{Reference: Picture Example}
% \end{figure}

%==========================================================================================%
% BIBLIOGRAPHY
%==========================================================================================%
\safebibliography{etdbib}
\bibliographystyle{apalike} %agu04 %apalike
%==========================================================================================%
%==========================================================================================%

\end{document}